\documentclass{fac}

\usepackage{graphicx}
\usepackage{amssymb}
\usepackage{amsfonts}
\usepackage{amsmath}
\usepackage{dsfont}
\usepackage{MnSymbol}
\usepackage{mathtools}
\usepackage{epsfig}
\usepackage{epstopdf}
\usepackage{lmerge}
\usepackage{tikz}
\usepackage{amsthm}
\ifprodtf \else \usepackage{latexsym}\fi

\usetikzlibrary{calc,decorations.pathmorphing,shapes,graphs}

\newcounter{sarrow}

\newcounter{sarrow1}
\newcommand\xnrsquigarrow[1]{%
\stepcounter{sarrow1}%
\mathrel{\begin{tikzpicture}[baseline= {( $ (current bounding box.south) + (0,-0.5ex) $ )}]
\node[inner sep=.5ex] (\thesarrow) {$\scriptstyle #1$};
\path[draw,<-,decorate,
  decoration={zigzag,amplitude=0.7pt,segment length=1.2mm,pre=lineto,pre length=4pt}]
    (\thesarrow1.south east) -- (\thesarrow1.south west);
    $\slashedarrowfill@\relbar\relbar/$
    \end{tikzpicture}}%
}

\makeatletter
\def\slashedarrowfill@#1#2#3#4#5{%
  $\m@th\thickmuskip0mu\medmuskip\thickmuskip\thinmuskip\thickmuskip
   \relax#5#1\mkern-7mu%
   \cleaders\hbox{$#5\mkern-2mu#2\mkern-2mu$}\hfill
   \mathclap{#3}\mathclap{#2}%
   \cleaders\hbox{$#5\mkern-2mu#2\mkern-2mu$}\hfill
   \mkern-7mu#4$%
}
\def\rightslashedarrowfillb@{%
  \slashedarrowfill@\relbar\relbar/\rightarrow}
\newcommand\xnrightarrow[2][]{%
  \ext@arrow 0055{\rightslashedarrowfillb@}{#1}{#2}}

\def\rightslashedarrowfille@{%
  \slashedarrowfill@\relbar\relbar/\twoheadrightarrow}
\newcommand\xntworightarrow[2][]{%
  \ext@arrow 0055{\rightslashedarrowfille@}{#1}{#2}}

\def\rightslashedarrowfillg@{%
  \slashedarrowfill@\relbar\relbar{\raisebox{.12em}{}}\twoheadrightarrow}
\newcommand\xtworightarrow[2][]{%
  \ext@arrow 0055{\rightslashedarrowfillg@}{#1}{#2}}

\def\rightslashedarrowfillx@{%
  \slashedarrowfill@\Relbar\Relbar/\rightrightarrows}
\newcommand\xnTworightarrow[2][]{%
  \ext@arrow 0055{\rightslashedarrowfillx@}{#1}{#2}}

\def\rightslashedarrowfilly@{%
  \slashedarrowfill@\Relbar\Relbar{\raisebox{.12em}{}}\rightrightarrows}
\newcommand\xTworightarrow[2][]{%
  \ext@arrow 0055{\rightslashedarrowfilly@}{#1}{#2}}

\pgfdeclareshape{slash underlined}
{
  \inheritsavedanchors[from=rectangle] % this is nearly a circle
  \inheritanchorborder[from=rectangle]
  \inheritanchor[from=rectangle]{north}
  \inheritanchor[from=rectangle]{north west}
  \inheritanchor[from=rectangle]{north east}
  \inheritanchor[from=rectangle]{center}
  \inheritanchor[from=rectangle]{west}
  \inheritanchor[from=rectangle]{east}
  \inheritanchor[from=rectangle]{mid}
  \inheritanchor[from=rectangle]{mid west}
  \inheritanchor[from=rectangle]{mid east}
  \inheritanchor[from=rectangle]{base}
  \inheritanchor[from=rectangle]{base west}
  \inheritanchor[from=rectangle]{base east}
  \inheritanchor[from=rectangle]{south}
  \inheritanchor[from=rectangle]{south west}
  \inheritanchor[from=rectangle]{south east}
  \inheritanchorborder[from=rectangle]
  \foregroundpath{
    \southwest \pgf@xa=\pgf@x \pgf@ya=\pgf@y
    \northeast \pgf@xb=\pgf@x \pgf@yb=\pgf@y
    \pgf@xc=\pgf@xa
    \advance\pgf@xc by .5\pgf@xb
    \pgf@yc=\pgf@ya
    \advance\pgf@xc by -1.3pt
    \advance\pgf@yc by -1.8pt
    \pgfpathmoveto{\pgfqpoint{\pgf@xc}{\pgf@yc}}
    \advance\pgf@xc by  2.6pt
    \advance\pgf@yc by  3.6pt
    \pgfpathlineto{\pgfqpoint{\pgf@xc}{\pgf@yc}}
    \pgfpathmoveto{\pgfqpoint{\pgf@xa}{\pgf@ya}}
    \pgfpathlineto{\pgfqpoint{\pgf@xb}{\pgf@ya}}
 }
}
\tikzset{nomorepostaction/.code=\let\tikz@postactions\pgfutil@empty}

\newtheorem{theorem}{Theorem}[section]
\newtheorem{definition}[theorem]{Definition}
\newtheorem{proposition}[theorem]{Proposition}

\title[Draft of Algebraic Laws for True Concurrency]
      {Algebraic Laws for True Concurrency}

\author[Yong Wang]
    {Yong Wang\\
     College of Computer Science and Technology,\\
     Faculty of Information Technology,\\
     Beijing University of Technology, Beijing, China\\
     }

\correspond{Yong Wang, Pingleyuan 100, Chaoyang District, Beijing, China.
            e-mail: wangy@bjut.edu.cn}

\pubyear{2017}
\pagerange{\pageref{firstpage}--\pageref{lastpage}}

\begin{document}
\label{firstpage}

\makecorrespond

\maketitle

\begin{abstract}
We find the algebraic laws for true concurrency. Eventually, we establish a whole axiomatization for true concurrency called $APTC$ (Algebra for Parallelism in True Concurrency). The theory $APTC$ has four modules: $BATC$ (Basic Algebra for True Concurrency), $APTC$ (Algebra for Parallelism in True Concurrency), recursion and abstraction. And also, we show the applications and extensions of $APTC$.
\end{abstract}


\section{Introduction}{\label{int}}

Parallelism and concurrency \cite{CM} are the core concepts within computer science. There are mainly two camps in capturing concurrency: the interleaving concurrency and the true concurrency.

The representative of interleaving concurrency is bisimulation/rooted branching bisimulation equivalences. $CCS$ (Calculus of Communicating Systems) \cite{CCS} is a calculus based on bisimulation semantics model. Hennessy and Milner (HM) logic for bisimulation equivalence is also designed. Later, algebraic laws to capture computational properties modulo bisimulation equivalence was introduced in \cite{ALNC}, this work eventually founded the comprehensive axiomatization modulo bisimulation equivalence -- $ACP$ (Algebra of Communicating Processes) \cite{ACP}.

The other camp of concurrency is true concurrency. The researches on true concurrency are still active. Firstly, there are several truly concurrent bisimulation equivalences, the representatives are: pomset bisimulation equivalence, step bisimulation equivalence, history-preserving (hp-) bisimulation equivalence, and especially hereditary history-preserving (hhp-) bisimulation equivalence \cite{HHP1} \cite{HHP2}, the well-known finest truly concurrent bisimulation equivalence. These truly concurrent bisimulations are studied in different structures \cite{ES1} \cite{ES2} \cite{CM}: Petri nets, event structures, domains, and also a uniform form called TSI (Transition System with Independence) \cite{SFL}. There are also several logics based on different truly concurrent bisimulation equivalences, for example, SFL (Separation Fixpoint Logic) and TFL (Trace Fixpoint Logic) \cite{SFL} are extensions on true concurrency of mu-calculi \cite{MUC} on bisimulation equivalence, and also a logic with reverse modalities \cite {RL1} \cite{RL2} based on the so-called reverse bisimulations with a reverse flavor. It must be pointed out that, a uniform logic for true concurrency \cite{LTC1} \cite{LTC2} was represented several years ago, which used a logical framework to unify several truly concurrent bisimulation equivalences, including pomset bisimulation, step bisimulation, hp-bisimulation and hhp-bisimulation.

There are simple comparisons between HM logic for bisimulation equivalence and the uniform logic \cite{LTC1} \cite{LTC2} for truly concurrent bisimulation equivalences, the algebraic laws \cite{ALNC}, $ACP$ \cite{ACP} for bisimulation equivalence, and \emph{what} for truly concurrent bisimulation equivalences, which is still missing.

Yes, we try to find the algebraic laws for true concurrency following the way paved by $ACP$ for bisimulation equivalence. And finally, we establish a whole axiomatization for true concurrency called $APTC$. The theory $APTC$ has four modules: $BATC$ (Basic Algebra for True Concurrency), $APTC$ (Algebra for Parallelism in True Concurrency), recursion and abstraction. With the help of placeholder in section \ref{ph}, we get an intuitive result for true concurrency: $a \between b = a \cdot b + b \cdot a + a \parallel b + a \mid b$ modulo truly concurrent bisimilarities pomset bisimulation equivalence, step bisimulation equivalence and history-preserving bisimulation equivalence, with $a,b$ are atomic actions (events), $\between$ is the whole true concurrency operator, $\cdot$ is the temporal causality operator, $+$ is the a kind of structured conflict, $\parallel$ is the parallel operator and $\mid$ is the communication merge.

This paper is organized as follows. In section \ref{bg}, we introduce some preliminaries, including a brief introduction to $ACP$, and also preliminaries on true concurrency. We introduce $BATC$ in section \ref{batc}, $APTC$ in section \ref{aptc}, recursion in section \ref{rec}, and abstraction in section \ref{abs}. In section \ref{app}, we show the applications of $APTC$ by an example called alternating bit protocol. We show the modularity and extension mechanism of $APTC$ in section \ref{ext}. We establish an axiomatization for hhp-bisimilarity in section \ref{ahhpb}. Finally, in section \ref{con}, we conclude this paper.

\section{Backgrounds}\label{bg}

\subsection{Process Algebra}\label{pa}

In this subsection, we introduce the preliminaries on process algebra $ACP$ \cite{ACP}, which is based on the interleaving bisimulation semantics. $ACP$ has an almost perfect axiomatization to capture laws on bisimulation equivalence, including equational logic and bisimulation semantics, and also the soundness and completeness bridged between them.

\subsubsection{$ACP$}\label{ACP}

$ACP$ captures several computational properties in the form of algebraic laws, and proves the soundness and completeness modulo bisimulation/rooted branching bisimulation equivalence. These computational properties are organized in a modular way by use of the concept of conservational extension, which include the following modules, note that, every algebra are composed of constants and operators, the constants are the computational objects, while operators capture the computational properties.

\begin{enumerate}
  \item \textbf{$BPA$ (Basic Process Algebras)}. $BPA$ has sequential composition $\cdot$ and alternative composition $+$ to capture sequential computation and nondeterminacy. The constants are ranged over $A$, the set of atomic actions. The algebraic laws on $\cdot$ and $+$ are sound and complete modulo bisimulation equivalence.
  \item \textbf{$ACP$ (Algebra of Communicating Processes)}. $ACP$ uses the parallel operator $\parallel$, the auxiliary binary left merge $\leftmerge$ to model parallelism, and the communication merge $\mid$ to model communications among different parallel branches. Since a communication may be blocked, a new constant called deadlock $\delta$ is extended to $A$, and also a new unary encapsulation operator $\partial_H$ is introduced to eliminate $\delta$, which may exist in the processes. The algebraic laws on these operators are also sound and complete modulo bisimulation equivalence. Note that, these operators in a process can be eliminated by deductions on the process using axioms of $ACP$, and eventually be steadied by $\cdot$ and $+$, this is also why bisimulation is called an \emph{interleaving} semantics.
  \item \textbf{Recursion}. To model infinite computation, recursion is introduced into $ACP$. In order to obtain a sound and complete theory, guarded recursion and linear recursion are needed. The corresponding axioms are $RSP$ (Recursive Specification Principle) and $RDP$ (Recursive Definition Principle), $RDP$ says the solutions of a recursive specification can represent the behaviors of the specification, while $RSP$ says that a guarded recursive specification has only one solution, they are sound with respect to $ACP$ with guarded recursion modulo bisimulation equivalence, and they are complete with respect to $ACP$ with linear recursion modulo bisimulation equivalence.
  \item \textbf{Abstraction}. To abstract away internal implementations from the external behaviors, a new constant $\tau$ called silent step is added to $A$, and also a new unary abstraction operator $\tau_I$ is used to rename actions in $I$ into $\tau$ (the resulted $ACP$ with silent step and abstraction operator is called $ACP_{\tau}$). The recursive specification is adapted to guarded linear recursion to prevent infinite $\tau$-loops specifically. The axioms for $\tau$ and $\tau_I$ are sound modulo rooted branching bisimulation equivalence (a kind of weak bisimulation equivalence). To eliminate infinite $\tau$-loops caused by $\tau_I$ and obtain the completeness, $CFAR$ (Cluster Fair Abstraction Rule) is used to prevent infinite $\tau$-loops in a constructible way.
\end{enumerate}

$ACP$ can be used to verify the correctness of system behaviors, by deduction on the description of the system using the axioms of $ACP$. Base on the modularity of $ACP$, it can be extended easily and elegantly. For more details, please refer to the book of $ACP$ \cite{ACP}.

\subsubsection{Operational Semantics}\label{OS}

The semantics of $ACP$ is based on bisimulation/rooted branching bisimulation equivalences, and the modularity of $ACP$ relies on the concept of conservative extension, for the conveniences, we introduce some concepts and conclusions on them.

\begin{definition}[Bisimulation]
A bisimulation relation $R$ is a binary relation on processes such that: (1) if $p R q$ and $p\xrightarrow{a}p'$ then $q\xrightarrow{a}q'$ with $p' R q'$; (2) if $p R q$ and $q\xrightarrow{a}q'$ then $p\xrightarrow{a}p'$ with $p' R q'$; (3) if $p R q$ and $pP$, then $qP$; (4) if $p R q$ and $qP$, then $pP$. Two processes $p$ and $q$ are bisimilar, denoted by $p\sim_{HM} q$, if there is a bisimulation relation $R$ such that $p R q$.
\end{definition}

\begin{definition}[Congruence]
Let $\Sigma$ be a signature. An equivalence relation $R$ on $\mathcal{T}(\Sigma)$ is a congruence if for each $f\in\Sigma$, if $s_i R t_i$ for $i\in\{1,\cdots,ar(f)\}$, then $f(s_1,\cdots,s_{ar(f)}) R f(t_1,\cdots,t_{ar(f)})$.
\end{definition}

\begin{definition}[Branching bisimulation]
A branching bisimulation relation $R$ is a binary relation on the collection of processes such that: (1) if $p R q$ and $p\xrightarrow{a}p'$ then either $a\equiv \tau$ and $p' R q$ or there is a sequence of (zero or more) $\tau$-transitions $q\xrightarrow{\tau}\cdots\xrightarrow{\tau}q_0$ such that $p R q_0$ and $q_0\xrightarrow{a}q'$ with $p' R q'$; (2) if $p R q$ and $q\xrightarrow{a}q'$ then either $a\equiv \tau$ and $p R q'$ or there is a sequence of (zero or more) $\tau$-transitions $p\xrightarrow{\tau}\cdots\xrightarrow{\tau}p_0$ such that $p_0 R q$ and $p_0\xrightarrow{a}p'$ with $p' R q'$; (3) if $p R q$ and $pP$, then there is a sequence of (zero or more) $\tau$-transitions $q\xrightarrow{\tau}\cdots\xrightarrow{\tau}q_0$ such that $p R q_0$ and $q_0P$; (4) if $p R q$ and $qP$, then there is a sequence of (zero or more) $\tau$-transitions $p\xrightarrow{\tau}\cdots\xrightarrow{\tau}p_0$ such that $p_0 R q$ and $p_0P$. Two processes $p$ and $q$ are branching bisimilar, denoted by $p\approx_{bHM} q$, if there is a branching bisimulation relation $R$ such that $p R q$.
\end{definition}

\begin{definition}[Rooted branching bisimulation]
A rooted branching bisimulation relation $R$ is a binary relation on processes such that: (1) if $p R q$ and $p\xrightarrow{a}p'$ then $q\xrightarrow{a}q'$ with $p'\approx_{bHM} q'$; (2) if $p R q$ and $q\xrightarrow{a}q'$ then $p\xrightarrow{a}p'$ with $p'\approx_{bHM} q'$; (3) if $p R q$ and $pP$, then $qP$; (4) if $p R q$ and $qP$, then $pP$. Two processes $p$ and $q$ are rooted branching bisimilar, denoted by $p\approx_{rbHM} q$, if there is a rooted branching bisimulation relation $R$ such that $p R q$.
\end{definition}

\begin{definition}[Conservative extension]
Let $T_0$ and $T_1$ be TSSs (transition system specifications) over signatures $\Sigma_0$ and $\Sigma_1$, respectively. The TSS $T_0\oplus T_1$ is a conservative extension of $T_0$ if the LTSs (labeled transition systems) generated by $T_0$ and $T_0\oplus T_1$ contain exactly the same transitions $t\xrightarrow{a}t'$ and $tP$ with $t\in \mathcal{T}(\Sigma_0)$.
\end{definition}

\begin{definition}[Source-dependency]
The source-dependent variables in a transition rule of $\rho$ are defined inductively as follows: (1) all variables in the source of $\rho$ are source-dependent; (2) if $t\xrightarrow{a}t'$ is a premise of $\rho$ and all variables in $t$ are source-dependent, then all variables in $t'$ are source-dependent. A transition rule is source-dependent if all its variables are. A TSS is source-dependent if all its rules are.
\end{definition}

\begin{definition}[Freshness]
Let $T_0$ and $T_1$ be TSSs over signatures $\Sigma_0$ and $\Sigma_1$, respectively. A term in $\mathbb{T}(T_0\oplus T_1)$ is said to be fresh if it contains a function symbol from $\Sigma_1\setminus\Sigma_0$. Similarly, a transition label or predicate symbol in $T_1$ is fresh if it does not occur in $T_0$.
\end{definition}

\begin{theorem}[Conservative extension]\label{TCE}
Let $T_0$ and $T_1$ be TSSs over signatures $\Sigma_0$ and $\Sigma_1$, respectively, where $T_0$ and $T_0\oplus T_1$ are positive after reduction. Under the following conditions, $T_0\oplus T_1$ is a conservative extension of $T_0$. (1) $T_0$ is source-dependent. (2) For each $\rho\in T_1$, either the source of $\rho$ is fresh, or $\rho$ has a premise of the form $t\xrightarrow{a}t'$ or $tP$, where $t\in \mathbb{T}(\Sigma_0)$, all variables in $t$ occur in the source of $\rho$ and $t'$, $a$ or $P$ is fresh.
\end{theorem}

\subsubsection{Proof Techniques}\label{PT}

In this subsection, we introduce the concepts and conclusions about elimination, which is very important in the proof of completeness theorem.

\begin{definition}[Elimination property]
Let a process algebra with a defined set of basic terms as a subset of the set of closed terms over the process algebra. Then the process algebra has the elimination to basic terms property if for every closed term $s$ of the algebra, there exists a basic term $t$ of the algebra such that the algebra$\vdash s=t$.
\end{definition}

\begin{definition}[Strongly normalizing]
A term $s_0$ is called strongly normalizing if does not an infinite series of reductions beginning in $s_0$.
\end{definition}

\begin{definition}
We write $s>_{lpo} t$ if $s\rightarrow^+ t$ where $\rightarrow^+$ is the transitive closure of the reduction relation defined by the transition rules of an algebra.
\end{definition}

\begin{theorem}[Strong normalization]\label{SN}
Let a term rewriting system (TRS) with finitely many rewriting rules and let $>$ be a well-founded ordering on the signature of the corresponding algebra. If $s>_{lpo} t$ for each rewriting rule $s\rightarrow t$ in the TRS, then the term rewriting system is strongly normalizing.
\end{theorem}

\subsection{True Concurrency}{\label{tc}}

In this subsection, we introduce the concepts of prime event structure, and also concurrent behavior equivalence \cite{ES1} \cite{ES2} \cite{CM}, and also we extend prime event structure with silent event $\tau$, and explain the concept of weakly true concurrency, i.e., concurrent behaviorial equivalence with considering silent event $\tau$ \cite{WTC}.

\subsubsection{Event Structure}

We give the definition of prime event structure (PES) \cite{ES1} \cite{ES2} \cite{CM} extended with the silent event $\tau$ as follows.

\begin{definition}[Prime event structure with silent event]\label{PES}
Let $\Lambda$ be a fixed set of labels, ranged over $a,b,c,\cdots$ and $\tau$. A ($\Lambda$-labelled) prime event structure with silent event $\tau$ is a tuple $\mathcal{E}=\langle \mathbb{E}, \leq, \sharp, \lambda\rangle$, where $\mathbb{E}$ is a denumerable set of events, including the silent event $\tau$. Let $\hat{\mathbb{E}}=\mathbb{E}\backslash\{\tau\}$, exactly excluding $\tau$, it is obvious that $\hat{\tau^*}=\epsilon$, where $\epsilon$ is the empty event. Let $\lambda:\mathbb{E}\rightarrow\Lambda$ be a labelling function and let $\lambda(\tau)=\tau$. And $\leq$, $\sharp$ are binary relations on $\mathbb{E}$, called causality and conflict respectively, such that:

\begin{enumerate}
  \item $\leq$ is a partial order and $\lceil e \rceil = \{e'\in \mathbb{E}|e'\leq e\}$ is finite for all $e\in \mathbb{E}$. It is easy to see that $e\leq\tau^*\leq e'=e\leq\tau\leq\cdots\leq\tau\leq e'$, then $e\leq e'$.
  \item $\sharp$ is irreflexive, symmetric and hereditary with respect to $\leq$, that is, for all $e,e',e''\in \mathbb{E}$, if $e\sharp e'\leq e''$, then $e\sharp e''$.
\end{enumerate}

Then, the concepts of consistency and concurrency can be drawn from the above definition:

\begin{enumerate}
  \item $e,e'\in \mathbb{E}$ are consistent, denoted as $e\frown e'$, if $\neg(e\sharp e')$. A subset $X\subseteq \mathbb{E}$ is called consistent, if $e\frown e'$ for all $e,e'\in X$.
  \item $e,e'\in \mathbb{E}$ are concurrent, denoted as $e\parallel e'$, if $\neg(e\leq e')$, $\neg(e'\leq e)$, and $\neg(e\sharp e')$.
\end{enumerate}
\end{definition}

The prime event structure without considering silent event $\tau$ is the original one in \cite{ES1} \cite{ES2} \cite{CM}.

\begin{definition}[Configuration]
Let $\mathcal{E}$ be a PES. A (finite) configuration in $\mathcal{E}$ is a (finite) consistent subset of events $C\subseteq \mathcal{E}$, closed with respect to causality (i.e. $\lceil C\rceil=C$). The set of finite configurations of $\mathcal{E}$ is denoted by $\mathcal{C}(\mathcal{E})$. We let $\hat{C}=C\backslash\{\tau\}$.
\end{definition}

A consistent subset of $X\subseteq \mathbb{E}$ of events can be seen as a pomset. Given $X, Y\subseteq \mathbb{E}$, $\hat{X}\sim \hat{Y}$ if $\hat{X}$ and $\hat{Y}$ are isomorphic as pomsets. In the following of the paper, we say $C_1\sim C_2$, we mean $\hat{C_1}\sim\hat{C_2}$.

\begin{definition}[Pomset transitions and step]
Let $\mathcal{E}$ be a PES and let $C\in\mathcal{C}(\mathcal{E})$, and $\emptyset\neq X\subseteq \mathbb{E}$, if $C\cap X=\emptyset$ and $C'=C\cup X\in\mathcal{C}(\mathcal{E})$, then $C\xrightarrow{X} C'$ is called a pomset transition from $C$ to $C'$. When the events in $X$ are pairwise concurrent, we say that $C\xrightarrow{X}C'$ is a step.
\end{definition}

\begin{definition}[Weak pomset transitions and weak step]
Let $\mathcal{E}$ be a PES and let $C\in\mathcal{C}(\mathcal{E})$, and $\emptyset\neq X\subseteq \hat{\mathbb{E}}$, if $C\cap X=\emptyset$ and $\hat{C'}=\hat{C}\cup X\in\mathcal{C}(\mathcal{E})$, then $C\xRightarrow{X} C'$ is called a weak pomset transition from $C$ to $C'$, where we define $\xRightarrow{e}\triangleq\xrightarrow{\tau^*}\xrightarrow{e}\xrightarrow{\tau^*}$. And $\xRightarrow{X}\triangleq\xrightarrow{\tau^*}\xrightarrow{e}\xrightarrow{\tau^*}$, for every $e\in X$. When the events in $X$ are pairwise concurrent, we say that $C\xRightarrow{X}C'$ is a weak step.
\end{definition}

We will also suppose that all the PESs in this paper are image finite, that is, for any PES $\mathcal{E}$ and $C\in \mathcal{C}(\mathcal{E})$ and $a\in \Lambda$, $\{e\in \mathbb{E}|C\xrightarrow{e} C'\wedge \lambda(e)=a\}$ and $\{e\in\hat{\mathbb{E}}|C\xRightarrow{e} C'\wedge \lambda(e)=a\}$ is finite.

\subsubsection{Concurrent Behavioral Equivalence}

\begin{definition}[Pomset, step bisimulation]\label{PSB}
Let $\mathcal{E}_1$, $\mathcal{E}_2$ be PESs. A pomset bisimulation is a relation $R\subseteq\mathcal{C}(\mathcal{E}_1)\times\mathcal{C}(\mathcal{E}_2)$, such that if $(C_1,C_2)\in R$, and $C_1\xrightarrow{X_1}C_1'$ then $C_2\xrightarrow{X_2}C_2'$, with $X_1\subseteq \mathbb{E}_1$, $X_2\subseteq \mathbb{E}_2$, $X_1\sim X_2$ and $(C_1',C_2')\in R$, and vice-versa. We say that $\mathcal{E}_1$, $\mathcal{E}_2$ are pomset bisimilar, written $\mathcal{E}_1\sim_p\mathcal{E}_2$, if there exists a pomset bisimulation $R$, such that $(\emptyset,\emptyset)\in R$. By replacing pomset transitions with steps, we can get the definition of step bisimulation. When PESs $\mathcal{E}_1$ and $\mathcal{E}_2$ are step bisimilar, we write $\mathcal{E}_1\sim_s\mathcal{E}_2$.
\end{definition}

\begin{definition}[Weak pomset, step bisimulation]\label{WPSB}
Let $\mathcal{E}_1$, $\mathcal{E}_2$ be PESs. A weak pomset bisimulation is a relation $R\subseteq\mathcal{C}(\mathcal{E}_1)\times\mathcal{C}(\mathcal{E}_2)$, such that if $(C_1,C_2)\in R$, and $C_1\xRightarrow{X_1}C_1'$ then $C_2\xRightarrow{X_2}C_2'$, with $X_1\subseteq \hat{\mathbb{E}_1}$, $X_2\subseteq \hat{\mathbb{E}_2}$, $X_1\sim X_2$ and $(C_1',C_2')\in R$, and vice-versa. We say that $\mathcal{E}_1$, $\mathcal{E}_2$ are weak pomset bisimilar, written $\mathcal{E}_1\approx_p\mathcal{E}_2$, if there exists a weak pomset bisimulation $R$, such that $(\emptyset,\emptyset)\in R$. By replacing weak pomset transitions with weak steps, we can get the definition of weak step bisimulation. When PESs $\mathcal{E}_1$ and $\mathcal{E}_2$ are weak step bisimilar, we write $\mathcal{E}_1\approx_s\mathcal{E}_2$.
\end{definition}

\begin{definition}[Posetal product]
Given two PESs $\mathcal{E}_1$, $\mathcal{E}_2$, the posetal product of their configurations, denoted $\mathcal{C}(\mathcal{E}_1)\overline{\times}\mathcal{C}(\mathcal{E}_2)$, is defined as

$$\{(C_1,f,C_2)|C_1\in\mathcal{C}(\mathcal{E}_1),C_2\in\mathcal{C}(\mathcal{E}_2),f:C_1\rightarrow C_2 \textrm{ isomorphism}\}.$$

A subset $R\subseteq\mathcal{C}(\mathcal{E}_1)\overline{\times}\mathcal{C}(\mathcal{E}_2)$ is called a posetal relation. We say that $R$ is downward closed when for any $(C_1,f,C_2),(C_1',f',C_2')\in \mathcal{C}(\mathcal{E}_1)\overline{\times}\mathcal{C}(\mathcal{E}_2)$, if $(C_1,f,C_2)\subseteq (C_1',f',C_2')$ pointwise and $(C_1',f',C_2')\in R$, then $(C_1,f,C_2)\in R$.

For $f:X_1\rightarrow X_2$, we define $f[x_1\mapsto x_2]:X_1\cup\{x_1\}\rightarrow X_2\cup\{x_2\}$, $z\in X_1\cup\{x_1\}$,(1)$f[x_1\mapsto x_2](z)=
x_2$,if $z=x_1$;(2)$f[x_1\mapsto x_2](z)=f(z)$, otherwise. Where $X_1\subseteq \mathbb{E}_1$, $X_2\subseteq \mathbb{E}_2$, $x_1\in \mathbb{E}_1$, $x_2\in \mathbb{E}_2$.
\end{definition}

\begin{definition}[Weakly posetal product]
Given two PESs $\mathcal{E}_1$, $\mathcal{E}_2$, the weakly posetal product of their configurations, denoted $\mathcal{C}(\mathcal{E}_1)\overline{\times}\mathcal{C}(\mathcal{E}_2)$, is defined as

$$\{(C_1,f,C_2)|C_1\in\mathcal{C}(\mathcal{E}_1),C_2\in\mathcal{C}(\mathcal{E}_2),f:\hat{C_1}\rightarrow \hat{C_2} \textrm{ isomorphism}\}.$$

A subset $R\subseteq\mathcal{C}(\mathcal{E}_1)\overline{\times}\mathcal{C}(\mathcal{E}_2)$ is called a weakly posetal relation. We say that $R$ is downward closed when for any $(C_1,f,C_2),(C_1',f,C_2')\in \mathcal{C}(\mathcal{E}_1)\overline{\times}\mathcal{C}(\mathcal{E}_2)$, if $(C_1,f,C_2)\subseteq (C_1',f',C_2')$ pointwise and $(C_1',f',C_2')\in R$, then $(C_1,f,C_2)\in R$.

For $f:X_1\rightarrow X_2$, we define $f[x_1\mapsto x_2]:X_1\cup\{x_1\}\rightarrow X_2\cup\{x_2\}$, $z\in X_1\cup\{x_1\}$,(1)$f[x_1\mapsto x_2](z)=
x_2$,if $z=x_1$;(2)$f[x_1\mapsto x_2](z)=f(z)$, otherwise. Where $X_1\subseteq \hat{\mathbb{E}_1}$, $X_2\subseteq \hat{\mathbb{E}_2}$, $x_1\in \hat{\mathbb{E}}_1$, $x_2\in \hat{\mathbb{E}}_2$. Also, we define $f(\tau^*)=f(\tau^*)$.
\end{definition}

\begin{definition}[(Hereditary) history-preserving bisimulation]\label{HHPB}
A history-preserving (hp-) bisimulation is a posetal relation $R\subseteq\mathcal{C}(\mathcal{E}_1)\overline{\times}\mathcal{C}(\mathcal{E}_2)$ such that if $(C_1,f,C_2)\in R$, and $C_1\xrightarrow{e_1} C_1'$, then $C_2\xrightarrow{e_2} C_2'$, with $(C_1',f[e_1\mapsto e_2],C_2')\in R$, and vice-versa. $\mathcal{E}_1,\mathcal{E}_2$ are history-preserving (hp-)bisimilar and are written $\mathcal{E}_1\sim_{hp}\mathcal{E}_2$ if there exists a hp-bisimulation $R$ such that $(\emptyset,\emptyset,\emptyset)\in R$.

A hereditary history-preserving (hhp-)bisimulation is a downward closed hp-bisimulation. $\mathcal{E}_1,\mathcal{E}_2$ are hereditary history-preserving (hhp-)bisimilar and are written $\mathcal{E}_1\sim_{hhp}\mathcal{E}_2$.
\end{definition}

\begin{definition}[Weak (hereditary) history-preserving bisimulation]\label{WHHPB}
A weak history-preserving (hp-) bisimulation is a weakly posetal relation $R\subseteq\mathcal{C}(\mathcal{E}_1)\overline{\times}\mathcal{C}(\mathcal{E}_2)$ such that if $(C_1,f,C_2)\in R$, and $C_1\xRightarrow{e_1} C_1'$, then $C_2\xRightarrow{e_2} C_2'$, with $(C_1',f[e_1\mapsto e_2],C_2')\in R$, and vice-versa. $\mathcal{E}_1,\mathcal{E}_2$ are weak history-preserving (hp-)bisimilar and are written $\mathcal{E}_1\approx_{hp}\mathcal{E}_2$ if there exists a weak hp-bisimulation $R$ such that $(\emptyset,\emptyset,\emptyset)\in R$.

A weakly hereditary history-preserving (hhp-)bisimulation is a downward closed weak hp-bisimulation. $\mathcal{E}_1,\mathcal{E}_2$ are weakly hereditary history-preserving (hhp-)bisimilar and are written $\mathcal{E}_1\approx_{hhp}\mathcal{E}_2$.
\end{definition}

\begin{proposition}[Weakly concurrent behavioral equivalence]\label{WSCBE}
(Strongly) concurrent behavioral equivalences imply weakly concurrent behavioral equivalences. That is, $\sim_p$ implies $\approx_p$, $\sim_s$ implies $\approx_s$, $\sim_{hp}$ implies $\approx_{hp}$, $\sim_{hhp}$ implies $\approx_{hhp}$.
\end{proposition}

\begin{proof}
From the definition of weak pomset transition, weak step transition, weakly posetal product and weakly concurrent behavioral equivalence, it is easy to see that $\xrightarrow{e}=\xrightarrow{\epsilon}\xrightarrow{e}\xrightarrow{\epsilon}$ for $e\in \mathbb{E}$, where $\epsilon$ is the empty event.
\end{proof}

Note that in the above definitions, truly concurrent behavioral equivalences are defined by events $e\in\mathcal{E}$ and prime event structure $\mathcal{E}$, in contrast to interleaving behavioral equivalences by actions $a,b\in\mathcal{P}$ and process (graph) $\mathcal{P}$. Indeed, they have correspondences, in \cite{SFL}, models of concurrency, including Petri nets, transition systems and event structures, are unified in a uniform representation -- TSI (Transition System with Independence). If $x$ is a process, let $C(x)$ denote the corresponding configuration (the already executed part of the process $x$, of course, it is free of conflicts), when $x\xrightarrow{e} x'$, the corresponding configuration $C(x)\xrightarrow{e}C(x')$ with $C(x')=C(x)\cup\{e\}$, where $e$ may be caused by some events in $C(x)$ and concurrent with the other events in $C(x)$, or entirely concurrent with all events in $C(x)$, or entirely caused by all events in $C(x)$. Though the concurrent behavioral equivalences (Definition \ref{PSB}, \ref{WPSB}, \ref{HHPB} and \ref{WHHPB}) are defined based on configurations (pasts of processes), they can also be defined based on processes (futures of configurations), we omit the concrete definitions. One key difference between definitions based on configurations and processes is that, the definitions based on configurations are stressing the structures of two equivalent configurations and the concrete atomic events may be different, but the definitions based on processes require not only the equivalent structures, but also the same atomic events by their labels, since we try to establish the algebraic equations modulo the corresponding concurrent behavior equivalences.

With a little abuse of concepts, in the following of the paper, we will not distinguish actions and events, prime event structures and processes, also concurrent behavior equivalences based on configurations and processes, and use them freely, unless they have specific meanings. Usually, in congruence theorem and soundness, we show them in a structure only flavor (equivalences based on configuration); but in proof of the completeness theorem, we must require not only the equivalent structure, but also the same set of atomic events.

\section{Basic Algebra for True Concurrency}{\label{batc}}

In this section, we will discuss the algebraic laws for prime event structure $\mathcal{E}$, exactly for causality $\leq$ and conflict $\sharp$. We will follow the conventions of process algebra, using $\cdot$ instead of $\leq$ and $+$ instead of $\sharp$. The resulted algebra is called Basic Algebra for True Concurrency, abbreviated $BATC$.

\subsection{Axiom System of $BATC$}

In the following, let $e_1, e_2, e_1', e_2'\in \mathbb{E}$, and let variables $x,y,z$ range over the set of terms for true concurrency, $p,q,s$ range over the set of closed terms. The set of axioms of $BATC$ consists of the laws given in Table \ref{AxiomsForBATC}.

\begin{center}
    \begin{table}
        \begin{tabular}{@{}ll@{}}
            \hline No. &Axiom\\
            $A1$ & $x+ y = y+ x$\\
            $A2$ & $(x+ y)+ z = x+ (y+ z)$\\
            $A3$ & $x+ x = x$\\
            $A4$ & $(x+ y)\cdot z = x\cdot z + y\cdot z$\\
            $A5$ & $(x\cdot y)\cdot z = x\cdot(y\cdot z)$\\
        \end{tabular}
        \caption{Axioms of $BATC$}
        \label{AxiomsForBATC}
    \end{table}
\end{center}

Intuitively, the axiom $A1$ says that the binary operator $+$ satisfies commutative law. The axiom $A2$ says that $+$ satisfies associativity. $A3$ says that $+$ satisfies idempotency. The axiom $A4$ is the right distributivity of the binary operator $\cdot$ to $+$. And $A5$ is the associativity of $\cdot$.

\subsection{Properties of $BATC$}

\begin{definition}[Basic terms of $BATC$]\label{BTBATC}
The set of basic terms of $BATC$, $\mathcal{B}(BATC)$, is inductively defined as follows:
\begin{enumerate}
  \item $\mathbb{E}\subset\mathcal{B}(BATC)$;
  \item if $e\in \mathbb{E}, t\in\mathcal{B}(BATC)$ then $e\cdot t\in\mathcal{B}(BATC)$;
  \item if $t,s\in\mathcal{B}(BATC)$ then $t+ s\in\mathcal{B}(BATC)$.
\end{enumerate}
\end{definition}

\begin{theorem}[Elimination theorem of $BATC$]\label{ETBATC}
Let $p$ be a closed $BATC$ term. Then there is a basic $BATC$ term $q$ such that $BATC\vdash p=q$.
\end{theorem}

\begin{proof}
(1) Firstly, suppose that the following ordering on the signature of $BATC$ is defined: $\cdot > +$ and the symbol $\cdot$ is given the lexicographical status for the first argument, then for each rewrite rule $p\rightarrow q$ in Table \ref{TRSForBATC} relation $p>_{lpo} q$ can easily be proved. We obtain that the term rewrite system shown in Table \ref{TRSForBATC} is strongly normalizing, for it has finitely many rewriting rules, and $>$ is a well-founded ordering on the signature of $BATC$, and if $s>_{lpo} t$, for each rewriting rule $s\rightarrow t$ is in Table \ref{TRSForBATC} (see Theorem \ref{SN}).

\begin{center}
    \begin{table}
        \begin{tabular}{@{}ll@{}}
            \hline No. &Rewriting Rule\\
            $RA3$ & $x+ x \rightarrow x$\\
            $RA4$ & $(x+ y)\cdot z \rightarrow x\cdot z + y\cdot z$\\
            $RA5$ & $(x\cdot y)\cdot z \rightarrow x\cdot(y\cdot z)$\\
        \end{tabular}
        \caption{Term rewrite system of $BATC$}
        \label{TRSForBATC}
    \end{table}
\end{center}

(2) Then we prove that the normal forms of closed $BATC$ terms are basic $BATC$ terms.

Suppose that $p$ is a normal form of some closed $BATC$ term and suppose that $p$ is not a basic term. Let $p'$ denote the smallest sub-term of $p$ which is not a basic term. It implies that each sub-term of $p'$ is a basic term. Then we prove that $p$ is not a term in normal form. It is sufficient to induct on the structure of $p'$:

\begin{itemize}
  \item Case $p'\equiv e, e\in \mathbb{E}$. $p'$ is a basic term, which contradicts the assumption that $p'$ is not a basic term, so this case should not occur.
  \item Case $p'\equiv p_1\cdot p_2$. By induction on the structure of the basic term $p_1$:
      \begin{itemize}
        \item Subcase $p_1\in \mathbb{E}$. $p'$ would be a basic term, which contradicts the assumption that $p'$ is not a basic term;
        \item Subcase $p_1\equiv e\cdot p_1'$. $RA5$ rewriting rule can be applied. So $p$ is not a normal form;
        \item Subcase $p_1\equiv p_1'+ p_1''$. $RA4$ rewriting rule can be applied. So $p$ is not a normal form.
      \end{itemize}
  \item Case $p'\equiv p_1+ p_2$. By induction on the structure of the basic terms both $p_1$ and $p_2$, all subcases will lead to that $p'$ would be a basic term, which contradicts the assumption that $p'$ is not a basic term.
\end{itemize}
\end{proof}

\subsection{Structured Operational Semantics of $BATC$}

In this subsection, we will define a term-deduction system which gives the operational semantics of $BATC$. We give the operational transition rules for operators $\cdot$ and $+$ as Table \ref{SETRForBATC} shows. And the predicate $\xrightarrow{e}\surd$ represents successful termination after execution of the event $e$.

\begin{center}
    \begin{table}
        $$\frac{}{e\xrightarrow{e}\surd}$$
        $$\frac{x\xrightarrow{e}\surd}{x+ y\xrightarrow{e}\surd} \quad\frac{x\xrightarrow{e}x'}{x+ y\xrightarrow{e}x'} \quad\frac{y\xrightarrow{e}\surd}{x+ y\xrightarrow{e}\surd} \quad\frac{y\xrightarrow{e}y'}{x+ y\xrightarrow{e}y'}$$
        $$\frac{x\xrightarrow{e}\surd}{x\cdot y\xrightarrow{e} y} \quad\frac{x\xrightarrow{e}x'}{x\cdot y\xrightarrow{e}x'\cdot y}$$
        \caption{Single event transition rules of $BATC$}
        \label{SETRForBATC}
    \end{table}
\end{center}

\begin{theorem}[Congruence of $BATC$ with respect to bisimulation equivalence]
Bisimulation equivalence $\sim_{HM}$ is a congruence with respect to $BATC$.
\end{theorem}

\begin{proof}
The axioms in Table \ref{AxiomsForBATC} of $BATC$ are the same as the axioms of $BPA$ (Basic Process Algebra) \cite{ALNC} \cite{CC} \cite{ACP}, so, bisimulation equivalence $\sim_{HM}$ is a congruence with respect to $BATC$.
\end{proof}

\begin{theorem}[Soundness of $BATC$ modulo bisimulation equivalence]
Let $x$ and $y$ be $BATC$ terms. If $BATC\vdash x=y$, then $x\sim_{HM} y$.
\end{theorem}

\begin{proof}
The axioms in Table \ref{AxiomsForBATC} of $BATC$ are the same as the axioms of $BPA$ (Basic Process Algebra) \cite{ALNC} \cite{CC} \cite{ACP}, so, $BATC$ is sound modulo bisimulation equivalence.
\end{proof}

\begin{theorem}[Completeness of $BATC$ modulo bisimulation equivalence]
Let $p$ and $q$ be closed $BATC$ terms, if $p\sim_{HM} q$ then $p=q$.
\end{theorem}

\begin{proof}
The axioms in Table \ref{AxiomsForBATC} of $BATC$ are the same as the axioms of $BPA$ (Basic Process Algebra) \cite{ALNC} \cite{CC} \cite{ACP}, so, $BATC$ is complete modulo bisimulation equivalence.
\end{proof}

The pomset transition rules are shown in Table \ref{PTRForBATC}, different to single event transition rules in Table \ref{SETRForBATC}, the pomset transition rules are labeled by pomsets, which are defined by causality $\cdot$ and conflict $+$.

\begin{center}
    \begin{table}
        $$\frac{}{X\xrightarrow{X}\surd}$$
        $$\frac{x\xrightarrow{X}\surd}{x+ y\xrightarrow{X}\surd} (X\subseteq x)\quad\frac{x\xrightarrow{X}x'}{x+ y\xrightarrow{X}x'} (X\subseteq x) \quad\frac{y\xrightarrow{Y}\surd}{x+ y\xrightarrow{Y}\surd} (Y\subseteq y)\quad\frac{y\xrightarrow{Y}y'}{x+ y\xrightarrow{Y}y'}(Y\subseteq y)$$
        $$\frac{x\xrightarrow{X}\surd}{x\cdot y\xrightarrow{X} y} (X\subseteq x)\quad\frac{x\xrightarrow{X}x'}{x\cdot y\xrightarrow{X}x'\cdot y} (X\subseteq x)$$
        \caption{Pomset transition rules of $BATC$}
        \label{PTRForBATC}
    \end{table}
\end{center}

\begin{theorem}[Congruence of $BATC$ with respect to pomset bisimulation equivalence]
Pomset bisimulation equivalence $\sim_{p}$ is a congruence with respect to $BATC$.
\end{theorem}

\begin{proof}
It is easy to see that pomset bisimulation is an equivalent relation on $BATC$ terms, we only need to prove that $\sim_{p}$ is preserved by the operators $\cdot$ and $+$.

\begin{itemize}
  \item Causality operator $\cdot$. Let $x_1,x_2$ and $y_1,y_2$ be $BATC$ processes, and $x_1\sim_{p} y_1$, $x_2\sim_{p} y_2$, it is sufficient to prove that $x_1\cdot x_2\sim_{p} y_1\cdot y_2$.

      By the definition of pomset bisimulation $\sim_p$ (Definition \ref{PSB}), $x_1\sim_p y_1$ means that

      $$x_1\xrightarrow{X_1} x_1' \quad y_1\xrightarrow{Y_1} y_1'$$

      with $X_1\subseteq x_1$, $Y_1\subseteq y_1$, $X_1\sim Y_1$ and $x_1'\sim_p y_1'$. The meaning of $x_2\sim_p y_2$ is similar.

      By the pomset transition rules for causality operator $\cdot$ in Table \ref{PTRForBATC}, we can get

      $$x_1\cdot x_2\xrightarrow{X_1} x_2 \quad y_1\cdot y_2\xrightarrow{Y_1} y_2$$

      with $X_1\subseteq x_1$, $Y_1\subseteq y_1$, $X_1\sim Y_1$ and $x_2\sim_p y_2$, so, we get $x_1\cdot x_2\sim_p y_1\cdot y_2$, as desired.

      Or, we can get

      $$x_1\cdot x_2\xrightarrow{X_1} x_1'\cdot x_2 \quad y_1\cdot y_2\xrightarrow{Y_1} y_1'\cdot y_2$$

      with $X_1\subseteq x_1$, $Y_1\subseteq y_1$, $X_1\sim Y_1$ and $x_1'\sim_p y_1'$, $x_2\sim_p y_2$, so, we get $x_1\cdot x_2\sim_p y_1\cdot y_2$, as desired.
  \item Conflict operator $+$. Let $x_1, x_2$ and $y_1, y_2$ be $BATC$ processes, and $x_1\sim_p y_1$, $x_2\sim_p y_2$, it is sufficient to prove that $x_1+ x_2 \sim_p y_1+ y_2$. The meanings of $x_1\sim_p y_1$ and $x_2\sim_p y_2$ are the same as the above case, according to the definition of pomset bisimulation $\sim_p$ in Definition \ref{PSB}.

      By the pomset transition rules for conflict operator $+$ in Table \ref{PTRForBATC}, we can get four cases:

      $$x_1+ x_2\xrightarrow{X_1} \surd \quad y_1+ y_2\xrightarrow{Y_1} \surd$$

      with $X_1\subseteq x_1$, $Y_1\subseteq y_1$, $X_1\sim Y_1$, so, we get $x_1+ x_2\sim_p y_1+ y_2$, as desired.

      Or, we can get

      $$x_1+ x_2\xrightarrow{X_1} x_1' \quad y_1+ y_2\xrightarrow{Y_1} y_1'$$

      with $X_1\subseteq x_1$, $Y_1\subseteq y_1$, $X_1\sim Y_1$, and $x_1'\sim_p y_1'$, so, we get $x_1+ x_2\sim_p y_1+ y_2$, as desired.

      Or, we can get

      $$x_1+ x_2\xrightarrow{X_2} \surd \quad y_1+ y_2\xrightarrow{Y_2} \surd$$

      with $X_2\subseteq x_2$, $Y_2\subseteq y_2$, $X_2\sim Y_2$, so, we get $x_1+ x_2\sim_p y_1+ y_2$, as desired.

      Or, we can get

      $$x_1+ x_2\xrightarrow{X_2} x_2' \quad y_1+ y_2\xrightarrow{Y_2} y_2'$$

      with $X_2\subseteq x_2$, $Y_2\subseteq y_2$, $X_2\sim Y_2$, and $x_2'\sim_p y_2'$, so, we get $x_1+ x_2\sim_p y_1+ y_2$, as desired.
\end{itemize}
\end{proof}

\begin{theorem}[Soundness of $BATC$ modulo pomset bisimulation equivalence]\label{SBATCPBE}
Let $x$ and $y$ be $BATC$ terms. If $BATC\vdash x=y$, then $x\sim_{p} y$.
\end{theorem}

\begin{proof}
Since pomset bisimulation $\sim_p$ is both an equivalent and a congruent relation, we only need to check if each axiom in Table \ref{AxiomsForBATC} is sound modulo pomset bisimulation equivalence.

\begin{itemize}
  \item \textbf{Axiom $A1$}. Let $p,q$ be $BATC$ processes, and $p+ q=q+ p$, it is sufficient to prove that $p+ q\sim_p q+ p$. By the pomset transition rules for operator $+$ in Table \ref{PTRForBATC}, we get

      $$\frac{p\xrightarrow{P}\surd}{p+ q\xrightarrow{P}\surd} (P\subseteq p) \quad \frac{p\xrightarrow{P}\surd}{q+ p\xrightarrow{P}\surd}(P\subseteq p)$$

      $$\frac{p\xrightarrow{P}p'}{p+ q\xrightarrow{P}p'}(P\subseteq p) \quad \frac{p\xrightarrow{P}p'}{q+ p\xrightarrow{P}p'}(P\subseteq p)$$

      $$\frac{q\xrightarrow{Q}\surd}{p+ q\xrightarrow{Q}\surd}(Q\subseteq q) \quad \frac{q\xrightarrow{Q}\surd}{q+ p\xrightarrow{Q}\surd}(Q\subseteq q)$$

      $$\frac{q\xrightarrow{Q}q'}{p+ q\xrightarrow{Q}q'}(Q\subseteq q) \quad \frac{q\xrightarrow{Q}q'}{q+ p\xrightarrow{Q}q'}(Q\subseteq q)$$

      So, $p+ q\sim_p q+ p$, as desired.
  \item \textbf{Axiom $A2$}. Let $p,q,s$ be $BATC$ processes, and $(p+ q)+ s=p+ (q+ s)$, it is sufficient to prove that $(p+ q)+ s \sim_p p+ (q+ s)$. By the pomset transition rules for operator $+$ in Table \ref{PTRForBATC}, we get

      $$\frac{p\xrightarrow{P}\surd}{(p+ q)+ s\xrightarrow{P}\surd} (P\subseteq p) \quad \frac{p\xrightarrow{P}\surd}{p+ (q+ s)\xrightarrow{P}\surd}(P\subseteq p)$$

      $$\frac{p\xrightarrow{P}p'}{(p+ q)+ s\xrightarrow{P}p'}(P\subseteq p) \quad \frac{p\xrightarrow{P}p'}{p+ (q+ s)\xrightarrow{P}p'}(P\subseteq p)$$

      $$\frac{q\xrightarrow{Q}\surd}{(p+ q)+ s\xrightarrow{Q}\surd}(Q\subseteq q) \quad \frac{q\xrightarrow{Q}\surd}{p+ (q+ s)\xrightarrow{Q}\surd}(Q\subseteq q)$$

      $$\frac{q\xrightarrow{Q}q'}{(p+ q)+ s\xrightarrow{Q}q'}(Q\subseteq q) \quad \frac{q\xrightarrow{Q}q'}{p+ (q+ s)\xrightarrow{Q}q'}(Q\subseteq q)$$

      $$\frac{s\xrightarrow{S}\surd}{(p+ q)+ s\xrightarrow{S}\surd}(S\subseteq s) \quad \frac{s\xrightarrow{S}\surd}{p+ (q+ s)\xrightarrow{S}\surd}(S\subseteq s)$$

      $$\frac{s\xrightarrow{S}s'}{(p+ q)+ s\xrightarrow{S}s'}(S\subseteq s) \quad \frac{s\xrightarrow{S}s'}{p+ (q+ s)\xrightarrow{S}s'}(S\subseteq s)$$

      So, $(p+ q)+ s\sim_p p+ (q+ s)$, as desired.
  \item \textbf{Axiom $A3$}. Let $p$ be a $BATC$ process, and $p+ p=p$, it is sufficient to prove that $p+ p\sim_p p$. By the pomset transition rules for operator $+$ in Table \ref{PTRForBATC}, we get

      $$\frac{p\xrightarrow{P}\surd}{p+ p\xrightarrow{P}\surd} (P\subseteq p) \quad \frac{p\xrightarrow{P}\surd}{p\xrightarrow{P}\surd}(P\subseteq p)$$

      $$\frac{p\xrightarrow{P}p'}{p+ p\xrightarrow{P}p'}(P\subseteq p) \quad \frac{p\xrightarrow{P}p'}{p\xrightarrow{P}p'}(P\subseteq p)$$

      So, $p+ p\sim_p p$, as desired.
  \item \textbf{Axiom $A4$}. Let $p,q,s$ be $BATC$ processes, and $(p+ q)\cdot s=p\cdot s + q\cdot s$, it is sufficient to prove that $(p+ q)\cdot s \sim_p p\cdot s + q\cdot s$. By the pomset transition rules for operators $+$ and $\cdot$ in Table \ref{PTRForBATC}, we get

      $$\frac{p\xrightarrow{P}\surd}{(p+ q)\cdot s\xrightarrow{P}s} (P\subseteq p) \quad \frac{p\xrightarrow{P}\surd}{p\cdot s + q\cdot s\xrightarrow{P}s}(P\subseteq p)$$

      $$\frac{p\xrightarrow{P}p'}{(p+ q)\cdot s\xrightarrow{P}p'\cdot s}(P\subseteq p) \quad \frac{p\xrightarrow{P}p'}{p\cdot s+ q\cdot s\xrightarrow{P}p'\cdot s}(P\subseteq p)$$

      $$\frac{q\xrightarrow{Q}\surd}{(p+ q)\cdot s\xrightarrow{Q}s}(Q\subseteq q) \quad \frac{q\xrightarrow{Q}\surd}{p\cdot s+ q\cdot s\xrightarrow{Q}s}(Q\subseteq q)$$

      $$\frac{q\xrightarrow{Q}q'}{(p+ q)\cdot s\xrightarrow{Q}q'\cdot s}(Q\subseteq q) \quad \frac{q\xrightarrow{Q}q'}{p\cdot s+ q\cdot s\xrightarrow{Q}q'\cdot s}(Q\subseteq q)$$

      So, $(p+ q)\cdot s\sim_p p\cdot s+ q\cdot s$, as desired.
  \item \textbf{Axiom $A5$}. Let $p,q,s$ be $BATC$ processes, and $(p\cdot q)\cdot s=p\cdot (q\cdot s)$, it is sufficient to prove that $(p\cdot q)\cdot s \sim_p p\cdot (q\cdot s)$. By the pomset transition rules for operator $\cdot$ in Table \ref{PTRForBATC}, we get

      $$\frac{p\xrightarrow{P}\surd}{(p\cdot q)\cdot s\xrightarrow{P}q\cdot s} (P\subseteq p) \quad \frac{p\xrightarrow{P}\surd}{p\cdot (q\cdot s)\xrightarrow{P}q\cdot s}(P\subseteq p)$$

      $$\frac{p\xrightarrow{P}p'}{(p\cdot q)\cdot s\xrightarrow{P}(p'\cdot q)\cdot s}(P\subseteq p) \quad \frac{p\xrightarrow{P}p'}{p\cdot (q\cdot s)\xrightarrow{P}p'\cdot (q\cdot s)}(P\subseteq p)$$

      With an assumption $(p'\cdot q)\cdot s=p'\cdot(q\cdot s)$, so, $(p\cdot q)\cdot s\sim_p p\cdot (q\cdot s)$, as desired.
\end{itemize}
\end{proof}

\begin{theorem}[Completeness of $BATC$ modulo pomset bisimulation equivalence]\label{CBATCPBE}
Let $p$ and $q$ be closed $BATC$ terms, if $p\sim_{p} q$ then $p=q$.
\end{theorem}

\begin{proof}
Firstly, by the elimination theorem of $BATC$, we know that for each closed $BATC$ term $p$, there exists a closed basic $BATC$ term $p'$, such that $BATC\vdash p=p'$, so, we only need to consider closed basic $BATC$ terms.

The basic terms (see Definition \ref{BTBATC}) modulo associativity and commutativity (AC) of conflict $+$ (defined by axioms $A1$ and $A2$ in Table \ref{AxiomsForBATC}), and this equivalence is denoted by $=_{AC}$. Then, each equivalence class $s$ modulo AC of $+$ has the following normal form

$$s_1+\cdots+ s_k$$

with each $s_i$ either an atomic event or of the form $t_1\cdot t_2$, and each $s_i$ is called the summand of $s$.

Now, we prove that for normal forms $n$ and $n'$, if $n\sim_{p} n'$ then $n=_{AC}n'$. It is sufficient to induct on the sizes of $n$ and $n'$.

\begin{itemize}
  \item Consider a summand $e$ of $n$. Then $n\xrightarrow{e}\surd$, so $n\sim_p n'$ implies $n'\xrightarrow{e}\surd$, meaning that $n'$ also contains the summand $e$.
  \item Consider a summand $t_1\cdot t_2$ of $n$. Then $n\xrightarrow{t_1}t_2$, so $n\sim_p n'$ implies $n'\xrightarrow{t_1}t_2'$ with $t_2\sim_p t_2'$, meaning that $n'$ contains a summand $t_1\cdot t_2'$. Since $t_2$ and $t_2'$ are normal forms and have sizes smaller than $n$ and $n'$, by the induction hypotheses $t_2\sim_p t_2'$ implies $t_2=_{AC} t_2'$.
\end{itemize}

So, we get $n=_{AC} n'$.

Finally, let $s$ and $t$ be basic terms, and $s\sim_p t$, there are normal forms $n$ and $n'$, such that $s=n$ and $t=n'$. The soundness theorem of $BATC$ modulo pomset bisimulation equivalence (see Theorem \ref{SBATCPBE}) yields $s\sim_p n$ and $t\sim_p n'$, so $n\sim_p s\sim_p t\sim_p n'$. Since if $n\sim_p n'$ then $n=_{AC}n'$, $s=n=_{AC}n'=t$, as desired.
\end{proof}

The step transition rules are defined in Table \ref{STRForBATC}, different to pomset transition rules, the step transition rules are labeled by steps, in which every event is pairwise concurrent.

\begin{center}
    \begin{table}
        $$\frac{}{X\xrightarrow{X}\surd}(\forall e_1, e_2\in X \textrm{ are pairwise concurrent.})$$
        $$\frac{x\xrightarrow{X}\surd}{x+ y\xrightarrow{X}\surd} (X\subseteq x, \forall e_1, e_2\in X \textrm{ are pairwise concurrent.})$$
        $$\frac{x\xrightarrow{X}x'}{x+ y\xrightarrow{X}x'} (X\subseteq x, \forall e_1, e_2\in X \textrm{ are pairwise concurrent.})$$ $$\frac{y\xrightarrow{Y}\surd}{x+ y\xrightarrow{Y}\surd} (Y\subseteq y, \forall e_1, e_2\in Y \textrm{ are pairwise concurrent.})$$
        $$\frac{y\xrightarrow{Y}y'}{x+ y\xrightarrow{Y}y'}(Y\subseteq y, \forall e_1, e_2\in Y \textrm{ are pairwise concurrent.})$$
        $$\frac{x\xrightarrow{X}\surd}{x\cdot y\xrightarrow{X} y} (X\subseteq x, \forall e_1, e_2\in X \textrm{ are pairwise concurrent.})$$
        $$\frac{x\xrightarrow{X}x'}{x\cdot y\xrightarrow{X}x'\cdot y} (X\subseteq x, \forall e_1, e_2\in X \textrm{ are pairwise concurrent.})$$
        \caption{Step transition rules of $BATC$}
        \label{STRForBATC}
    \end{table}
\end{center}

\begin{theorem}[Congruence of $BATC$ with respect to step bisimulation equivalence]
Step bisimulation equivalence $\sim_s$ is a congruence with respect to $BATC$.
\end{theorem}

\begin{proof}
It is easy to see that step bisimulation is an equivalent relation on $BATC$ terms, we only need to prove that $\sim_{s}$ is preserved by the operators $\cdot$ and $+$.

\begin{itemize}
  \item Causality operator $\cdot$. Let $x_1,x_2$ and $y_1,y_2$ be $BATC$ processes, and $x_1\sim_{s} y_1$, $x_2\sim_{s} y_2$, it is sufficient to prove that $x_1\cdot x_2\sim_{s} y_1\cdot y_2$.

      By the definition of step bisimulation $\sim_s$ (Definition \ref{PSB}), $x_1\sim_s y_1$ means that

      $$x_1\xrightarrow{X_1} x_1' \quad y_1\xrightarrow{Y_1} y_1'$$

      with $X_1\subseteq x_1$, $\forall e_1, e_2\in X_1 \textrm{ are pairwise concurrent}$, $Y_1\subseteq y_1$, $\forall e_1, e_2\in Y_1 \textrm{ are pairwise concurrent}$, $X_1\sim Y_1$ and $x_1'\sim_s y_1'$. The meaning of $x_2\sim_s y_2$ is similar.

      By the step transition rules for causality operator $\cdot$ in Table \ref{STRForBATC}, we can get

      $$x_1\cdot x_2\xrightarrow{X_1} x_2 \quad y_1\cdot y_2\xrightarrow{Y_1} y_2$$

      with $X_1\subseteq x_1$, $\forall e_1, e_2\in X_1 \textrm{ are pairwise concurrent}$, $Y_1\subseteq y_1$, $\forall e_1, e_2\in Y_1 \textrm{ are pairwise concurrent}$, $X_1\sim Y_1$ and $x_2\sim_s y_2$, so, we get $x_1\cdot x_2\sim_s y_1\cdot y_2$, as desired.

      Or, we can get

      $$x_1\cdot x_2\xrightarrow{X_1} x_1'\cdot x_2 \quad y_1\cdot y_2\xrightarrow{Y_1} y_1'\cdot y_2$$

      with $X_1\subseteq x_1$, $\forall e_1, e_2\in X_1 \textrm{ are pairwise concurrent.}$, $Y_1\subseteq y_1$, $\forall e_1, e_2\in Y_1 \textrm{ are pairwise concurrent.}$, $X_1\sim Y_1$ and $x_1'\sim_s y_1'$, $x_2\sim_s y_2$, so, we get $x_1\cdot x_2\sim_s y_1\cdot y_2$, as desired.
  \item Conflict operator $+$. Let $x_1, x_2$ and $y_1, y_2$ be $BATC$ processes, and $x_1\sim_s y_1$, $x_2\sim_s y_2$, it is sufficient to prove that $x_1+ x_2 \sim_s y_1+ y_2$. The meanings of $x_1\sim_s y_1$ and $x_2\sim_s y_2$ are the same as the above case, according to the definition of step bisimulation $\sim_s$ in Definition \ref{PSB}.

      By the step transition rules for conflict operator $+$ in Table \ref{STRForBATC}, we can get four cases:

      $$x_1+ x_2\xrightarrow{X_1} \surd \quad y_1+ y_2\xrightarrow{Y_1} \surd$$

      with $X_1\subseteq x_1$, $\forall e_1, e_2\in X_1 \textrm{ are pairwise concurrent}$, $Y_1\subseteq y_1$, $\forall e_1, e_2\in Y_1 \textrm{ are pairwise concurrent}$, $X_1\sim Y_1$, so, we get $x_1+ x_2\sim_s y_1+ y_2$, as desired.

      Or, we can get

      $$x_1+ x_2\xrightarrow{X_1} x_1' \quad y_1+ y_2\xrightarrow{Y_1} y_1'$$

      with $X_1\subseteq x_1$, $\forall e_1, e_2\in X_1 \textrm{ are pairwise concurrent}$, $Y_1\subseteq y_1$, $\forall e_1, e_2\in Y_1 \textrm{ are pairwise concurrent}$, $X_1\sim Y_1$, and $x_1'\sim_s y_1'$, so, we get $x_1+ x_2\sim_s y_1+ y_2$, as desired.

      Or, we can get

      $$x_1+ x_2\xrightarrow{X_2} \surd \quad y_1+ y_2\xrightarrow{Y_2} \surd$$

      with $X_2\subseteq x_2$, $\forall e_1, e_2\in X_2 \textrm{ are pairwise concurrent}$, $Y_2\subseteq y_2$, $\forall e_1, e_2\in Y_2 \textrm{ are pairwise concurrent}$, $X_2\sim Y_2$, so, we get $x_1+ x_2\sim_s y_1+ y_2$, as desired.

      Or, we can get

      $$x_1+ x_2\xrightarrow{X_2} x_2' \quad y_1+ y_2\xrightarrow{Y_2} y_2'$$

      with $X_2\subseteq x_2$, $\forall e_1, e_2\in X_2 \textrm{ are pairwise concurrent}$, $Y_2\subseteq y_2$, $\forall e_1, e_2\in Y_2 \textrm{ are pairwise concurrent}$, $X_2\sim Y_2$, and $x_2'\sim_s y_2'$, so, we get $x_1+ x_2\sim_s y_1+ y_2$, as desired.
\end{itemize}
\end{proof}

\begin{theorem}[Soundness of $BATC$ modulo step bisimulation equivalence]\label{SBATCSBE}
Let $x$ and $y$ be $BATC$ terms. If $BATC\vdash x=y$, then $x\sim_{s} y$.
\end{theorem}

\begin{proof}
Since step bisimulation $\sim_s$ is both an equivalent and a congruent relation, we only need to check if each axiom in Table \ref{AxiomsForBATC} is sound modulo step bisimulation equivalence.

\begin{itemize}
  \item \textbf{Axiom $A1$}. Let $p,q$ be $BATC$ processes, and $p+ q=q+ p$, it is sufficient to prove that $p+ q\sim_s q+ p$. By the step transition rules for operator $+$ in Table \ref{STRForBATC}, we get

      $$\frac{p\xrightarrow{P}\surd}{p+ q\xrightarrow{P}\surd} (P\subseteq p,\forall e_1, e_2\in P \textrm{ are pairwise concurrent.})$$

      $$\frac{p\xrightarrow{P}\surd}{q+ p\xrightarrow{P}\surd}(P\subseteq p,\forall e_1, e_2\in P \textrm{ are pairwise concurrent.})$$

      $$\frac{p\xrightarrow{P}p'}{p+ q\xrightarrow{P}p'}(P\subseteq p,\forall e_1, e_2\in P \textrm{ are pairwise concurrent.})$$

      $$\frac{p\xrightarrow{P}p'}{q+ p\xrightarrow{P}p'}(P\subseteq p,\forall e_1, e_2\in P \textrm{ are pairwise concurrent.})$$

      $$\frac{q\xrightarrow{Q}\surd}{p+ q\xrightarrow{Q}\surd}(Q\subseteq q,\forall e_1, e_2\in Q \textrm{ are pairwise concurrent.})$$

      $$\frac{q\xrightarrow{Q}\surd}{q+ p\xrightarrow{Q}\surd}(Q\subseteq q,\forall e_1, e_2\in Q \textrm{ are pairwise concurrent.})$$

      $$\frac{q\xrightarrow{Q}q'}{p+ q\xrightarrow{Q}q'}(Q\subseteq q,\forall e_1, e_2\in Q \textrm{ are pairwise concurrent.})$$

      $$\frac{q\xrightarrow{Q}q'}{q+ p\xrightarrow{Q}q'}(Q\subseteq q,\forall e_1, e_2\in Q \textrm{ are pairwise concurrent.})$$

      So, $p+ q\sim_s q+ p$, as desired.
  \item \textbf{Axiom $A2$}. Let $p,q,s$ be $BATC$ processes, and $(p+ q)+ s=p+ (q+ s)$, it is sufficient to prove that $(p+ q)+ s \sim_s p+ (q+ s)$. By the step transition rules for operator $+$ in Table \ref{STRForBATC}, we get

      $$\frac{p\xrightarrow{P}\surd}{(p+ q)+ s\xrightarrow{P}\surd} (P\subseteq p,\forall e_1, e_2\in P \textrm{ are pairwise concurrent.})$$

      $$\frac{p\xrightarrow{P}\surd}{p+ (q+ s)\xrightarrow{P}\surd}(P\subseteq p,\forall e_1, e_2\in P \textrm{ are pairwise concurrent.})$$

      $$\frac{p\xrightarrow{P}p'}{(p+ q)+ s\xrightarrow{P}p'}(P\subseteq p,\forall e_1, e_2\in P \textrm{ are pairwise concurrent.})$$

      $$\frac{p\xrightarrow{P}p'}{p+ (q+ s)\xrightarrow{P}p'}(P\subseteq p,\forall e_1, e_2\in P \textrm{ are pairwise concurrent.})$$

      $$\frac{q\xrightarrow{Q}\surd}{(p+ q)+ s\xrightarrow{Q}\surd}(Q\subseteq q,\forall e_1, e_2\in Q \textrm{ are pairwise concurrent.})$$

      $$\frac{q\xrightarrow{Q}\surd}{p+ (q+ s)\xrightarrow{Q}\surd}(Q\subseteq q,\forall e_1, e_2\in Q \textrm{ are pairwise concurrent.})$$

      $$\frac{q\xrightarrow{Q}q'}{(p+ q)+ s\xrightarrow{Q}q'}(Q\subseteq q,\forall e_1, e_2\in Q \textrm{ are pairwise concurrent.})$$

      $$\frac{q\xrightarrow{Q}q'}{p+ (q+ s)\xrightarrow{Q}q'}(Q\subseteq q,\forall e_1, e_2\in Q \textrm{ are pairwise concurrent.})$$

      $$\frac{s\xrightarrow{S}\surd}{(p+ q)+ s\xrightarrow{S}\surd}(S\subseteq s,\forall e_1, e_2\in S \textrm{ are pairwise concurrent.})$$

      $$\frac{s\xrightarrow{S}\surd}{p+ (q+ s)\xrightarrow{S}\surd}(S\subseteq s,\forall e_1, e_2\in S \textrm{ are pairwise concurrent.})$$

      $$\frac{s\xrightarrow{S}s'}{(p+ q)+ s\xrightarrow{S}s'}(S\subseteq s,\forall e_1, e_2\in S \textrm{ are pairwise concurrent.})$$

      $$\frac{s\xrightarrow{S}s'}{p+ (q+ s)\xrightarrow{S}s'}(S\subseteq s,\forall e_1, e_2\in S \textrm{ are pairwise concurrent.})$$

      So, $(p+ q)+ s\sim_s p+ (q+ s)$, as desired.
  \item \textbf{Axiom $A3$}. Let $p$ be a $BATC$ process, and $p+ p=p$, it is sufficient to prove that $p+ p\sim_s p$. By the step transition rules for operator $+$ in Table \ref{STRForBATC}, we get

      $$\frac{p\xrightarrow{P}\surd}{p+ p\xrightarrow{P}\surd} (P\subseteq p,\forall e_1, e_2\in P \textrm{ are pairwise concurrent.})$$

      $$\frac{p\xrightarrow{P}\surd}{p\xrightarrow{P}\surd}(P\subseteq p,\forall e_1, e_2\in P \textrm{ are pairwise concurrent.})$$

      $$\frac{p\xrightarrow{P}p'}{p+ p\xrightarrow{P}p'}(P\subseteq p,\forall e_1, e_2\in P \textrm{ are pairwise concurrent.})$$

      $$\frac{p\xrightarrow{P}p'}{p\xrightarrow{P}p'}(P\subseteq p,\forall e_1, e_2\in P \textrm{ are pairwise concurrent.})$$

      So, $p+ p\sim_s p$, as desired.
  \item \textbf{Axiom $A4$}. Let $p,q,s$ be $BATC$ processes, and $(p+ q)\cdot s=p\cdot s + q\cdot s$, it is sufficient to prove that $(p+ q)\cdot s \sim_s p\cdot s + q\cdot s$. By the step transition rules for operators $+$ and $\cdot$ in Table \ref{STRForBATC}, we get

      $$\frac{p\xrightarrow{P}\surd}{(p+ q)\cdot s\xrightarrow{P}s} (P\subseteq p,\forall e_1, e_2\in P \textrm{ are pairwise concurrent.})$$

      $$\frac{p\xrightarrow{P}\surd}{p\cdot s + q\cdot s\xrightarrow{P}s}(P\subseteq p,\forall e_1, e_2\in P \textrm{ are pairwise concurrent.})$$

      $$\frac{p\xrightarrow{P}p'}{(p+ q)\cdot s\xrightarrow{P}p'\cdot s}(P\subseteq p,\forall e_1, e_2\in P \textrm{ are pairwise concurrent.})$$

      $$\frac{p\xrightarrow{P}p'}{p\cdot s+ q\cdot s\xrightarrow{P}p'\cdot s}(P\subseteq p,\forall e_1, e_2\in P \textrm{ are pairwise concurrent.})$$

      $$\frac{q\xrightarrow{Q}\surd}{(p+ q)\cdot s\xrightarrow{Q}s}(Q\subseteq q,\forall e_1, e_2\in Q \textrm{ are pairwise concurrent.})$$

      $$\frac{q\xrightarrow{Q}\surd}{p\cdot s+ q\cdot s\xrightarrow{Q}s}(Q\subseteq q,\forall e_1, e_2\in Q \textrm{ are pairwise concurrent.})$$

      $$\frac{q\xrightarrow{Q}q'}{(p+ q)\cdot s\xrightarrow{Q}q'\cdot s}(Q\subseteq q,\forall e_1, e_2\in Q \textrm{ are pairwise concurrent.})$$

      $$\frac{q\xrightarrow{Q}q'}{p\cdot s+ q\cdot s\xrightarrow{Q}q'\cdot s}(Q\subseteq q,\forall e_1, e_2\in Q \textrm{ are pairwise concurrent.})$$

      So, $(p+ q)\cdot s\sim_s p\cdot s+ q\cdot s$, as desired.
  \item \textbf{Axiom $A5$}. Let $p,q,s$ be $BATC$ processes, and $(p\cdot q)\cdot s=p\cdot (q\cdot s)$, it is sufficient to prove that $(p\cdot q)\cdot s \sim_s p\cdot (q\cdot s)$. By the step transition rules for operator $\cdot$ in Table \ref{STRForBATC}, we get

      $$\frac{p\xrightarrow{P}\surd}{(p\cdot q)\cdot s\xrightarrow{P}q\cdot s} (P\subseteq p,\forall e_1, e_2\in P \textrm{ are pairwise concurrent.})$$

      $$\frac{p\xrightarrow{P}\surd}{p\cdot (q\cdot s)\xrightarrow{P}q\cdot s}(P\subseteq p,\forall e_1, e_2\in P \textrm{ are pairwise concurrent.})$$

      $$\frac{p\xrightarrow{P}p'}{(p\cdot q)\cdot s\xrightarrow{P}(p'\cdot q)\cdot s}(P\subseteq p,\forall e_1, e_2\in P \textrm{ are pairwise concurrent.})$$

      $$\frac{p\xrightarrow{P}p'}{p\cdot (q\cdot s)\xrightarrow{P}p'\cdot (q\cdot s)}(P\subseteq p,\forall e_1, e_2\in P \textrm{ are pairwise concurrent.})$$

      With an assumption $(p'\cdot q)\cdot s=p'\cdot(q\cdot s)$, so, $(p\cdot q)\cdot s\sim_s p\cdot (q\cdot s)$, as desired.
\end{itemize}
\end{proof}

\begin{theorem}[Completeness of $BATC$ modulo step bisimulation equivalence]\label{CBATCSBE}
Let $p$ and $q$ be closed $BATC$ terms, if $p\sim_{s} q$ then $p=q$.
\end{theorem}

\begin{proof}
Firstly, by the elimination theorem of $BATC$, we know that for each closed $BATC$ term $p$, there exists a closed basic $BATC$ term $p'$, such that $BATC\vdash p=p'$, so, we only need to consider closed basic $BATC$ terms.

The basic terms (see Definition \ref{BTBATC}) modulo associativity and commutativity (AC) of conflict $+$ (defined by axioms $A1$ and $A2$ in Table \ref{AxiomsForBATC}), and this equivalence is denoted by $=_{AC}$. Then, each equivalence class $s$ modulo AC of $+$ has the following normal form

$$s_1+\cdots+ s_k$$

with each $s_i$ either an atomic event or of the form $t_1\cdot t_2$, and each $s_i$ is called the summand of $s$.

Now, we prove that for normal forms $n$ and $n'$, if $n\sim_{s} n'$ then $n=_{AC}n'$. It is sufficient to induct on the sizes of $n$ and $n'$.

\begin{itemize}
  \item Consider a summand $e$ of $n$. Then $n\xrightarrow{e}\surd$, so $n\sim_s n'$ implies $n'\xrightarrow{e}\surd$, meaning that $n'$ also contains the summand $e$.
  \item Consider a summand $t_1\cdot t_2$ of $n$. Then $n\xrightarrow{t_1}t_2$($\forall e_1, e_2\in t_1$ are pairwise concurrent), so $n\sim_s n'$ implies $n'\xrightarrow{t_1}t_2'$($\forall e_1, e_2\in t_1$ are pairwise concurrent) with $t_2\sim_s t_2'$, meaning that $n'$ contains a summand $t_1\cdot t_2'$. Since $t_2$ and $t_2'$ are normal forms and have sizes smaller than $n$ and $n'$, by the induction hypotheses if $t_2\sim_s t_2'$ then $t_2=_{AC} t_2'$.
\end{itemize}

So, we get $n=_{AC} n'$.

Finally, let $s$ and $t$ be basic terms, and $s\sim_s t$, there are normal forms $n$ and $n'$, such that $s=n$ and $t=n'$. The soundness theorem of $BATC$ modulo step bisimulation equivalence (see Theorem \ref{SBATCSBE}) yields $s\sim_s n$ and $t\sim_s n'$, so $n\sim_s s\sim_s t\sim_s n'$. Since if $n\sim_s n'$ then $n=_{AC}n'$, $s=n=_{AC}n'=t$, as desired.
\end{proof}

The transition rules for (hereditary) hp-bisimulation of $BATC$ are defined in Table \ref{HHPTRForBATC}, they are the same as single event transition rules in Table \ref{SETRForBATC}.

\begin{center}
    \begin{table}
        $$\frac{}{e\xrightarrow{e}\surd}$$
        $$\frac{x\xrightarrow{e}\surd}{x+ y\xrightarrow{e}\surd} \quad\frac{x\xrightarrow{e}x'}{x+ y\xrightarrow{e}x'} \quad\frac{y\xrightarrow{e}\surd}{x+ y\xrightarrow{e}\surd} \quad\frac{y\xrightarrow{e}y'}{x+ y\xrightarrow{e}y'}$$
        $$\frac{x\xrightarrow{e}\surd}{x\cdot y\xrightarrow{e} y} \quad\frac{x\xrightarrow{e}x'}{x\cdot y\xrightarrow{e}x'\cdot y}$$
        \caption{(Hereditary) hp-transition rules of $BATC$}
        \label{HHPTRForBATC}
    \end{table}
\end{center}

\begin{theorem}[Congruence of $BATC$ with respect to hp-bisimulation equivalence]
Hp-bisimulation equivalence $\sim_{hp}$ is a congruence with respect to $BATC$.
\end{theorem}

\begin{proof}
It is easy to see that history-preserving bisimulation is an equivalent relation on $BATC$ terms, we only need to prove that $\sim_{hp}$ is preserved by the operators $\cdot$ and $+$.

\begin{itemize}
  \item Causality operator $\cdot$. Let $x_1,x_2$ and $y_1,y_2$ be $BATC$ processes, and $x_1\sim_{hp} y_1$, $x_2\sim_{hp} y_2$, it is sufficient to prove that $x_1\cdot x_2\sim_{hp} y_1\cdot y_2$.

      By the definition of hp-bisimulation $\sim_{hp}$ (Definition \ref{HHPB}), $x_1\sim_{hp} y_1$ means that there is a posetal relation $(C(x_1),f,C(y_1))\in\sim_{hp}$, and

      $$x_1\xrightarrow{e_1} x_1' \quad y_1\xrightarrow{e_2} y_1'$$

      with $(C(x_1'),f[e_1\mapsto e_2],C(y_1'))\in\sim_{hp}$. The meaning of $x_2\sim_{hp} y_2$ is similar.

      By the hp-transition rules for causality operator $\cdot$ in Table \ref{HHPTRForBATC}, we can get

      $$x_1\cdot x_2\xrightarrow{e_1} x_2 \quad y_1\cdot y_2\xrightarrow{e_2} y_2$$

      with $x_2\sim_{hp} y_2$, so, we get $x_1\cdot x_2\sim_{hp} y_1\cdot y_2$, as desired.

      Or, we can get

      $$x_1\cdot x_2\xrightarrow{e_1} x_1'\cdot x_2 \quad y_1\cdot y_2\xrightarrow{e_2} y_1'\cdot y_2$$

      with $x_1'\sim_{hp} y_1'$, $x_2\sim_{hp} y_2$, so, we get $x_1\cdot x_2\sim_{hp} y_1\cdot y_2$, as desired.
  \item Conflict operator $+$. Let $x_1, x_2$ and $y_1, y_2$ be $BATC$ processes, and $x_1\sim_{hp} y_1$, $x_2\sim_{hp} y_2$, it is sufficient to prove that $x_1+ x_2 \sim_{hp} y_1+ y_2$. The meanings of $x_1\sim_{hp} y_1$ and $x_2\sim_{hp} y_2$ are the same as the above case, according to the definition of hp-bisimulation $\sim_{hp}$ in Definition \ref{HHPB}.

      By the hp-transition rules for conflict operator $+$ in Table \ref{HHPTRForBATC}, we can get four cases:

      $$x_1+ x_2\xrightarrow{e_1} \surd \quad y_1+ y_2\xrightarrow{e_2} \surd$$

      so, we get $x_1+ x_2\sim_{hp} y_1+ y_2$, as desired.

      Or, we can get

      $$x_1+ x_2\xrightarrow{e_1} x_1' \quad y_1+ y_2\xrightarrow{e_2} y_1'$$

      with $x_1'\sim_{hp} y_1'$, so, we get $x_1+ x_2\sim_{hp} y_1+ y_2$, as desired.

      Or, we can get

      $$x_1+ x_2\xrightarrow{e_1'} \surd \quad y_1+ y_2\xrightarrow{e_2'} \surd$$

      so, we get $x_1+ x_2\sim_{hp} y_1+ y_2$, as desired.

      Or, we can get

      $$x_1+ x_2\xrightarrow{e_1'} x_2' \quad y_1+ y_2\xrightarrow{e_2'} y_2'$$

      with $x_2'\sim_{hp} y_2'$, so, we get $x_1+ x_2\sim_{hp} y_1+ y_2$, as desired.
\end{itemize}
\end{proof}

\begin{theorem}[Soundness of $BATC$ modulo hp-bisimulation equivalence]\label{SBATCHPBE}
Let $x$ and $y$ be $BATC$ terms. If $BATC\vdash x=y$, then $x\sim_{hp} y$.
\end{theorem}

\begin{proof}
Since hp-bisimulation $\sim_{hp}$ is both an equivalent and a congruent relation, we only need to check if each axiom in Table \ref{AxiomsForBATC} is sound modulo hp-bisimulation equivalence.

\begin{itemize}
  \item \textbf{Axiom $A1$}. Let $p,q$ be $BATC$ processes, and $p+ q=q+ p$, it is sufficient to prove that $p+ q\sim_{hp} q+ p$. By the hp-transition rules for operator $+$ in Table \ref{HHPTRForBATC}, we get

      $$\frac{p\xrightarrow{e_1}\surd}{p+ q\xrightarrow{e_1}\surd} \quad \frac{p\xrightarrow{e_1}\surd}{q+ p\xrightarrow{e_1}\surd}$$

      $$\frac{p\xrightarrow{e_1}p'}{p+ q\xrightarrow{e_1}p'} \quad \frac{p\xrightarrow{e_1}p'}{q+ p\xrightarrow{e_1}p'}$$

      $$\frac{q\xrightarrow{e_2}\surd}{p+ q\xrightarrow{e_2}\surd} \quad \frac{q\xrightarrow{e_2}\surd}{q+ p\xrightarrow{e_2}\surd}$$

      $$\frac{q\xrightarrow{e_2}q'}{p+ q\xrightarrow{e_2}q'} \quad \frac{q\xrightarrow{e_2}q'}{q+ p\xrightarrow{e_2}q'}$$

      So, for $(C(p+ q),f,C(q+ p))\in\sim_{hp}$, $(C((p+ q)'),f[e_1\mapsto e_1],C((q+ p)'))\in\sim_{hp}$ and $(C((p+ q)'),f[e_2\mapsto e_2],C((q+ p)'))\in\sim_{hp}$, that is, $p+ q\sim_{hp} q+ p$, as desired.
  \item \textbf{Axiom $A2$}. Let $p,q,s$ be $BATC$ processes, and $(p+ q)+ s=p+ (q+ s)$, it is sufficient to prove that $(p+ q)+ s \sim_{hp} p+ (q+ s)$. By the hp-transition rules for operator $+$ in Table \ref{HHPTRForBATC}, we get

      $$\frac{p\xrightarrow{e_1}\surd}{(p+ q)+ s\xrightarrow{e_1}\surd} \quad \frac{p\xrightarrow{e_1}\surd}{p+ (q+ s)\xrightarrow{e_1}\surd}$$

      $$\frac{p\xrightarrow{e_1}p'}{(p+ q)+ s\xrightarrow{e_1}p'} \quad \frac{p\xrightarrow{e_1}p'}{p+ (q+ s)\xrightarrow{e_1}p'}$$

      $$\frac{q\xrightarrow{e_2}\surd}{(p+ q)+ s\xrightarrow{e_2}\surd} \quad \frac{q\xrightarrow{e_2}\surd}{p+ (q+ s)\xrightarrow{e_2}\surd}$$

      $$\frac{q\xrightarrow{e_2}q'}{(p+ q)+ s\xrightarrow{e_2}q'} \quad \frac{q\xrightarrow{e_2}q'}{p+ (q+ s)\xrightarrow{e_2}q'}$$

      $$\frac{s\xrightarrow{e_3}\surd}{(p+ q)+ s\xrightarrow{e_3}\surd} \quad \frac{s\xrightarrow{e_3}\surd}{p+ (q+ s)\xrightarrow{e_3}\surd}$$

      $$\frac{s\xrightarrow{e_3}s'}{(p+ q)+ s\xrightarrow{e_3}s'} \quad \frac{s\xrightarrow{e_3}s'}{p+ (q+ s)\xrightarrow{e_3}s'}$$

      So, for $(C((p+ q)+ s),f,C(p+ (q+ s)))\in\sim_{hp}$, $(C(((p+ q)+ s)'),f[e_1\mapsto e_1],C((p+ (q+ s))'))\in\sim_{hp}$\\ and $(C(((p+ q)+ s)'),f[e_2\mapsto e_2],C((p+ (q+ s))'))\in\sim_{hp}$ and $(C(((p+ q)+ s)'),f[e_3\mapsto e_3],C((p+ (q+ s))'))\in\sim_{hp}$, that is, $(p+ q)+ s\sim_{hp} p+ (q+ s)$, as desired.
  \item \textbf{Axiom $A3$}. Let $p$ be a $BATC$ process, and $p+ p=p$, it is sufficient to prove that $p+ p\sim_{hp} p$. By the hp-transition rules for operator $+$ in Table \ref{HHPTRForBATC}, we get

      $$\frac{p\xrightarrow{e_1}\surd}{p+ p\xrightarrow{e_1}\surd} \quad \frac{p\xrightarrow{e_1}\surd}{p\xrightarrow{e_1}\surd}$$

      $$\frac{p\xrightarrow{e_1}p'}{p+ p\xrightarrow{e_1}p'} \quad \frac{p\xrightarrow{e_1}p'}{p\xrightarrow{e_1}p'}$$

      So, for $(C(p+ p),f,C(p))\in\sim_{hp}$, $(C((p+ p)'),f[e_1\mapsto e_1],C((p)'))\in\sim_{hp}$, that is, $p+ p\sim_{hp} p$, as desired.
  \item \textbf{Axiom $A4$}. Let $p,q,s$ be $BATC$ processes, and $(p+ q)\cdot s=p\cdot s + q\cdot s$, it is sufficient to prove that $(p+ q)\cdot s \sim_{hp} p\cdot s + q\cdot s$. By the hp-transition rules for operators $+$ and $\cdot$ in Table \ref{HHPTRForBATC}, we get

      $$\frac{p\xrightarrow{e_1}\surd}{(p+ q)\cdot s\xrightarrow{e_1}s} \quad \frac{p\xrightarrow{e_1}\surd}{p\cdot s + q\cdot s\xrightarrow{e_1}s}$$

      $$\frac{p\xrightarrow{e_1}p'}{(p+ q)\cdot s\xrightarrow{e_1}p'\cdot s} \quad \frac{p\xrightarrow{e_1}p'}{p\cdot s+ q\cdot s\xrightarrow{e_1}p'\cdot s}$$

      $$\frac{q\xrightarrow{e_2}\surd}{(p+ q)\cdot s\xrightarrow{e_2}s} \quad \frac{q\xrightarrow{e_2}\surd}{p\cdot s+ q\cdot s\xrightarrow{e_2}s}$$

      $$\frac{q\xrightarrow{e_2}q'}{(p+ q)\cdot s\xrightarrow{e_2}q'\cdot s} \quad \frac{q\xrightarrow{e_2}q'}{p\cdot s+ q\cdot s\xrightarrow{Q}q'\cdot s}$$

      So, for $(C((p+ q)\cdot s),f,C(p\cdot s+ q\cdot s))\in\sim_{hp}$, $(C(((p+ q)\cdot s)'),f[e_1\mapsto e_1],C((p\cdot s+ q\cdot s)'))\in\sim_{hp}$ and $(C(((p+ q)\cdot s)'),f[e_2\mapsto e_2],C((p\cdot s+ q\cdot s)'))\in\sim_{hp}$, that is, $(p+ q)\cdot s\sim_{hp} p\cdot s+ q\cdot s$, as desired.
  \item \textbf{Axiom $A5$}. Let $p,q,s$ be $BATC$ processes, and $(p\cdot q)\cdot s=p\cdot (q\cdot s)$, it is sufficient to prove that $(p\cdot q)\cdot s \sim_{hp} p\cdot (q\cdot s)$. By the hp-transition rules for operator $\cdot$ in Table \ref{HHPTRForBATC}, we get

      $$\frac{p\xrightarrow{e_1}\surd}{(p\cdot q)\cdot s\xrightarrow{e_1}q\cdot s} \quad \frac{p\xrightarrow{e_1}\surd}{p\cdot (q\cdot s)\xrightarrow{e_1}q\cdot s}$$

      $$\frac{p\xrightarrow{e_1}p'}{(p\cdot q)\cdot s\xrightarrow{e_1}(p'\cdot q)\cdot s} \quad \frac{p\xrightarrow{e_1}p'}{p\cdot (q\cdot s)\xrightarrow{e_1}p'\cdot (q\cdot s)}$$

      With an assumption $(p'\cdot q)\cdot s=p'\cdot(q\cdot s)$, for $(C((p\cdot q)\cdot s),f,C(p\cdot(q\cdot s)))\in\sim_{hp}$, $(C(((p\cdot q)\cdot s)'),f[e_1\mapsto e_1],C((p\cdot(q\cdot s))'))\in\sim_{hp}$, that is, so, $(p\cdot q)\cdot s\sim_{hp} p\cdot (q\cdot s)$, as desired.
\end{itemize}
\end{proof}

\begin{theorem}[Completeness of $BATC$ modulo hp-bisimulation equivalence]\label{CBATCHPBE}
Let $p$ and $q$ be closed $BATC$ terms, if $p\sim_{hp} q$ then $p=q$.
\end{theorem}

\begin{proof}
Firstly, by the elimination theorem of $BATC$, we know that for each closed $BATC$ term $p$, there exists a closed basic $BATC$ term $p'$, such that $BATC\vdash p=p'$, so, we only need to consider closed basic $BATC$ terms.

The basic terms (see Definition \ref{BTBATC}) modulo associativity and commutativity (AC) of conflict $+$ (defined by axioms $A1$ and $A2$ in Table \ref{AxiomsForBATC}), and this equivalence is denoted by $=_{AC}$. Then, each equivalence class $s$ modulo AC of $+$ has the following normal form

$$s_1+\cdots+ s_k$$

with each $s_i$ either an atomic event or of the form $t_1\cdot t_2$, and each $s_i$ is called the summand of $s$.

Now, we prove that for normal forms $n$ and $n'$, if $n\sim_{hp} n'$ then $n=_{AC}n'$. It is sufficient to induct on the sizes of $n$ and $n'$.

\begin{itemize}
  \item Consider a summand $e$ of $n$. Then $n\xrightarrow{e}\surd$, so $n\sim_{hp} n'$ implies $n'\xrightarrow{e}\surd$, meaning that $n'$ also contains the summand $e$.
  \item Consider a summand $e\cdot s$ of $n$. Then $n\xrightarrow{e}s$, so $n\sim_{hp} n'$ implies $n'\xrightarrow{e}t$ with $s\sim_{hp} t$, meaning that $n'$ contains a summand $e\cdot t$. Since $s$ and $t$ are normal forms and have sizes smaller than $n$ and $n'$, by the induction hypotheses $s\sim_{hp} t$ implies $s=_{AC} t$.
\end{itemize}

So, we get $n=_{AC} n'$.

Finally, let $s$ and $t$ be basic terms, and $s\sim_{hp} t$, there are normal forms $n$ and $n'$, such that $s=n$ and $t=n'$. The soundness theorem of $BATC$ modulo hp-bisimulation equivalence (see Theorem \ref{SBATCHPBE}) yields $s\sim_{hp} n$ and $t\sim_{hp} n'$, so $n\sim_{hp} s\sim_{hp} t\sim_{hp} n'$. Since if $n\sim_{hp} n'$ then $n=_{AC}n'$, $s=n=_{AC}n'=t$, as desired.
\end{proof}

\begin{theorem}[Congruence of $BATC$ with respect to hhp-bisimulation equivalence]
Hhp-bisimulation equivalence $\sim_{hhp}$ is a congruence with respect to $BATC$.
\end{theorem}

\begin{proof}
It is easy to see that hhp-bisimulation is an equivalent relation on $BATC$ terms, we only need to prove that $\sim_{hhp}$ is preserved by the operators $\cdot$ and $+$.

\begin{itemize}
  \item Causality operator $\cdot$. Let $x_1,x_2$ and $y_1,y_2$ be $BATC$ processes, and $x_1\sim_{hhp} y_1$, $x_2\sim_{hhp} y_2$, it is sufficient to prove that $x_1\cdot x_2\sim_{hhp} y_1\cdot y_2$.

      By the definition of hhp-bisimulation $\sim_{hhp}$ (Definition \ref{HHPB}), $x_1\sim_{hhp} y_1$ means that there is a posetal relation $(C(x_1),f,C(y_1))\in\sim_{hhp}$, and

      $$x_1\xrightarrow{e_1} x_1' \quad y_1\xrightarrow{e_2} y_1'$$

      with $(C(x_1'),f[e_1\mapsto e_2],C(y_1'))\in\sim_{hhp}$. The meaning of $x_2\sim_{hhp} y_2$ is similar.

      By the hhp-transition rules for causality operator $\cdot$ in Table \ref{HHPTRForBATC}, we can get

      $$x_1\cdot x_2\xrightarrow{e_1} x_2 \quad y_1\cdot y_2\xrightarrow{e_2} y_2$$

      with $x_2\sim_{hhp} y_2$, so, we get $x_1\cdot x_2\sim_{hhp} y_1\cdot y_2$, as desired.

      Or, we can get

      $$x_1\cdot x_2\xrightarrow{e_1} x_1'\cdot x_2 \quad y_1\cdot y_2\xrightarrow{e_2} y_1'\cdot y_2$$

      with $x_1'\sim_{hhp} y_1'$, $x_2\sim_{hhp} y_2$, so, we get $x_1\cdot x_2\sim_{hhp} y_1\cdot y_2$, as desired.
  \item Conflict operator $+$. Let $x_1, x_2$ and $y_1, y_2$ be $BATC$ processes, and $x_1\sim_{hhp} y_1$, $x_2\sim_{hhp} y_2$, it is sufficient to prove that $x_1+ x_2 \sim_{hhp} y_1+ y_2$. The meanings of $x_1\sim_{hhp} y_1$ and $x_2\sim_{hhp} y_2$ are the same as the above case, according to the definition of hhp-bisimulation $\sim_{hhp}$ in Definition \ref{HHPB}.

      By the hhp-transition rules for conflict operator $+$ in Table \ref{HHPTRForBATC}, we can get four cases:

      $$x_1+ x_2\xrightarrow{e_1} \surd \quad y_1+ y_2\xrightarrow{e_2} \surd$$

      so, we get $x_1+ x_2\sim_{hhp} y_1+ y_2$, as desired.

      Or, we can get

      $$x_1+ x_2\xrightarrow{e_1} x_1' \quad y_1+ y_2\xrightarrow{e_2} y_1'$$

      with $x_1'\sim_{hhp} y_1'$, so, we get $x_1+ x_2\sim_{hhp} y_1+ y_2$, as desired.

      Or, we can get

      $$x_1+ x_2\xrightarrow{e_1'} \surd \quad y_1+ y_2\xrightarrow{e_2'} \surd$$

      so, we get $x_1+ x_2\sim_{hhp} y_1+ y_2$, as desired.

      Or, we can get

      $$x_1+ x_2\xrightarrow{e_1'} x_2' \quad y_1+ y_2\xrightarrow{e_2'} y_2'$$

      with $x_2'\sim_{hhp} y_2'$, so, we get $x_1+ x_2\sim_{hhp} y_1+ y_2$, as desired.
\end{itemize}
\end{proof}

\begin{theorem}[Soundness of $BATC$ modulo hhp-bisimulation equivalence]\label{SBATCHHPBE}
Let $x$ and $y$ be $BATC$ terms. If $BATC\vdash x=y$, then $x\sim_{hhp} y$.
\end{theorem}

\begin{proof}
Since hhp-bisimulation $\sim_{hhp}$ is both an equivalent and a congruent relation, we only need to check if each axiom in Table \ref{AxiomsForBATC} is sound modulo hhp-bisimulation equivalence.

\begin{itemize}
  \item \textbf{Axiom $A1$}. Let $p,q$ be $BATC$ processes, and $p+ q=q+ p$, it is sufficient to prove that $p+ q\sim_{hhp} q+ p$. By the hhp-transition rules for operator $+$ in Table \ref{HHPTRForBATC}, we get

      $$\frac{p\xrightarrow{e_1}\surd}{p+ q\xrightarrow{e_1}\surd} \quad \frac{p\xrightarrow{e_1}\surd}{q+ p\xrightarrow{e_1}\surd}$$

      $$\frac{p\xrightarrow{e_1}p'}{p+ q\xrightarrow{e_1}p'} \quad \frac{p\xrightarrow{e_1}p'}{q+ p\xrightarrow{e_1}p'}$$

      $$\frac{q\xrightarrow{e_2}\surd}{p+ q\xrightarrow{e_2}\surd} \quad \frac{q\xrightarrow{e_2}\surd}{q+ p\xrightarrow{e_2}\surd}$$

      $$\frac{q\xrightarrow{e_2}q'}{p+ q\xrightarrow{e_2}q'} \quad \frac{q\xrightarrow{e_2}q'}{q+ p\xrightarrow{e_2}q'}$$

      So, for $(C(p+ q),f,C(q+ p))\in\sim_{hhp}$, $(C((p+ q)'),f[e_1\mapsto e_1],C((q+ p)'))\in\sim_{hhp}$ and $(C((p+ q)'),f[e_2\mapsto e_2],C((q+ p)'))\in\sim_{hhp}$, that is, $p+ q\sim_{hhp} q+ p$, as desired.
  \item \textbf{Axiom $A2$}. Let $p,q,s$ be $BATC$ processes, and $(p+ q)+ s=p+ (q+ s)$, it is sufficient to prove that $(p+ q)+ s \sim_{hhp} p+ (q+ s)$. By the hhp-transition rules for operator $+$ in Table \ref{HHPTRForBATC}, we get

      $$\frac{p\xrightarrow{e_1}\surd}{(p+ q)+ s\xrightarrow{e_1}\surd} \quad \frac{p\xrightarrow{e_1}\surd}{p+ (q+ s)\xrightarrow{e_1}\surd}$$

      $$\frac{p\xrightarrow{e_1}p'}{(p+ q)+ s\xrightarrow{e_1}p'} \quad \frac{p\xrightarrow{e_1}p'}{p+ (q+ s)\xrightarrow{e_1}p'}$$

      $$\frac{q\xrightarrow{e_2}\surd}{(p+ q)+ s\xrightarrow{e_2}\surd} \quad \frac{q\xrightarrow{e_2}\surd}{p+ (q+ s)\xrightarrow{e_2}\surd}$$

      $$\frac{q\xrightarrow{e_2}q'}{(p+ q)+ s\xrightarrow{e_2}q'} \quad \frac{q\xrightarrow{e_2}q'}{p+ (q+ s)\xrightarrow{e_2}q'}$$

      $$\frac{s\xrightarrow{e_3}\surd}{(p+ q)+ s\xrightarrow{e_3}\surd} \quad \frac{s\xrightarrow{e_3}\surd}{p+ (q+ s)\xrightarrow{e_3}\surd}$$

      $$\frac{s\xrightarrow{e_3}s'}{(p+ q)+ s\xrightarrow{e_3}s'} \quad \frac{s\xrightarrow{e_3}s'}{p+ (q+ s)\xrightarrow{e_3}s'}$$

      So, for $(C((p+ q)+ s),f,C(p+ (q+ s)))\in\sim_{hhp}$, $(C(((p+ q)+ s)'),f[e_1\mapsto e_1],C((p+ (q+ s))'))\in\sim_{hhp}$\\ and $(C(((p+ q)+ s)'),f[e_2\mapsto e_2],C((p+ (q+ s))'))\in\sim_{hhp}$ and $(C(((p+ q)+ s)'),f[e_3\mapsto e_3],C((p+ (q+ s))'))\in\sim_{hhp}$, that is, $(p+ q)+ s\sim_{hhp} p+ (q+ s)$, as desired.
  \item \textbf{Axiom $A3$}. Let $p$ be a $BATC$ process, and $p+ p=p$, it is sufficient to prove that $p+ p\sim_{hhp} p$. By the hhp-transition rules for operator $+$ in Table \ref{HHPTRForBATC}, we get

      $$\frac{p\xrightarrow{e_1}\surd}{p+ p\xrightarrow{e_1}\surd} \quad \frac{p\xrightarrow{e_1}\surd}{p\xrightarrow{e_1}\surd}$$

      $$\frac{p\xrightarrow{e_1}p'}{p+ p\xrightarrow{e_1}p'} \quad \frac{p\xrightarrow{e_1}p'}{p\xrightarrow{e_1}p'}$$

      So, for $(C(p+ p),f,C(p))\in\sim_{hhp}$, $(C((p+ p)'),f[e_1\mapsto e_1],C((p)'))\in\sim_{hhp}$, that is, $p+ p\sim_{hhp} p$, as desired.
  \item \textbf{Axiom $A4$}. Let $p,q,s$ be $BATC$ processes, and $(p+ q)\cdot s=p\cdot s + q\cdot s$, it is sufficient to prove that $(p+ q)\cdot s \sim_{hhp} p\cdot s + q\cdot s$. By the hhp-transition rules for operators $+$ and $\cdot$ in Table \ref{HHPTRForBATC}, we get

      $$\frac{p\xrightarrow{e_1}\surd}{(p+ q)\cdot s\xrightarrow{e_1}s} \quad \frac{p\xrightarrow{e_1}\surd}{p\cdot s + q\cdot s\xrightarrow{e_1}s}$$

      $$\frac{p\xrightarrow{e_1}p'}{(p+ q)\cdot s\xrightarrow{e_1}p'\cdot s} \quad \frac{p\xrightarrow{e_1}p'}{p\cdot s+ q\cdot s\xrightarrow{e_1}p'\cdot s}$$

      $$\frac{q\xrightarrow{e_2}\surd}{(p+ q)\cdot s\xrightarrow{e_2}s} \quad \frac{q\xrightarrow{e_2}\surd}{p\cdot s+ q\cdot s\xrightarrow{e_2}s}$$

      $$\frac{q\xrightarrow{e_2}q'}{(p+ q)\cdot s\xrightarrow{e_2}q'\cdot s} \quad \frac{q\xrightarrow{e_2}q'}{p\cdot s+ q\cdot s\xrightarrow{e_2}q'\cdot s}$$

      So, for $(C((p+ q)\cdot s),f,C(p\cdot s+ q\cdot s))\in\sim_{hhp}$, $(C(((p+ q)\cdot s)'),f[e_1\mapsto e_1],C((p\cdot s+ q\cdot s)'))\in\sim_{hhp}$ and $(C(((p+ q)\cdot s)'),f[e_2\mapsto e_2],C((p\cdot s+ q\cdot s)'))\in\sim_{hhp}$, that is, $(p+ q)\cdot s\sim_{hhp} p\cdot s+ q\cdot s$, as desired.
  \item \textbf{Axiom $A5$}. Let $p,q,s$ be $BATC$ processes, and $(p\cdot q)\cdot s=p\cdot (q\cdot s)$, it is sufficient to prove that $(p\cdot q)\cdot s \sim_{hhp} p\cdot (q\cdot s)$. By the hhp-transition rules for operator $\cdot$ in Table \ref{HHPTRForBATC}, we get

      $$\frac{p\xrightarrow{e_1}\surd}{(p\cdot q)\cdot s\xrightarrow{e_1}q\cdot s} \quad \frac{p\xrightarrow{e_1}\surd}{p\cdot (q\cdot s)\xrightarrow{e_1}q\cdot s}$$

      $$\frac{p\xrightarrow{e_1}p'}{(p\cdot q)\cdot s\xrightarrow{e_1}(p'\cdot q)\cdot s} \quad \frac{p\xrightarrow{e_1}p'}{p\cdot (q\cdot s)\xrightarrow{e_1}p'\cdot (q\cdot s)}$$

      With an assumption $(p'\cdot q)\cdot s=p'\cdot(q\cdot s)$, for $(C((p\cdot q)\cdot s),f,C(p\cdot(q\cdot s)))\in\sim_{hhp}$, $(C(((p\cdot q)\cdot s)'),f[e_1\mapsto e_1],C((p\cdot(q\cdot s))'))\in\sim_{hhp}$, that is, so, $(p\cdot q)\cdot s\sim_{hhp} p\cdot (q\cdot s)$, as desired.
\end{itemize}
\end{proof}

\begin{theorem}[Completeness of $BATC$ modulo hhp-bisimulation equivalence]\label{CBATCHHPBE}
Let $p$ and $q$ be closed $BATC$ terms, if $p\sim_{hhp} q$ then $p=q$.
\end{theorem}

\begin{proof}
Firstly, by the elimination theorem of $BATC$, we know that for each closed $BATC$ term $p$, there exists a closed basic $BATC$ term $p'$, such that $BATC\vdash p=p'$, so, we only need to consider closed basic $BATC$ terms.

The basic terms (see Definition \ref{BTBATC}) modulo associativity and commutativity (AC) of conflict $+$ (defined by axioms $A1$ and $A2$ in Table \ref{AxiomsForBATC}), and this equivalence is denoted by $=_{AC}$. Then, each equivalence class $s$ modulo AC of $+$ has the following normal form

$$s_1+\cdots+ s_k$$

with each $s_i$ either an atomic event or of the form $t_1\cdot t_2$, and each $s_i$ is called the summand of $s$.

Now, we prove that for normal forms $n$ and $n'$, if $n\sim_{hhp} n'$ then $n=_{AC}n'$. It is sufficient to induct on the sizes of $n$ and $n'$.

\begin{itemize}
  \item Consider a summand $e$ of $n$. Then $n\xrightarrow{e}\surd$, so $n\sim_{hhp} n'$ implies $n'\xrightarrow{e}\surd$, meaning that $n'$ also contains the summand $e$.
  \item Consider a summand $e\cdot s$ of $n$. Then $n\xrightarrow{e}s$, so $n\sim_{hhp} n'$ implies $n'\xrightarrow{e}t$ with $s\sim_{hhp} t$, meaning that $n'$ contains a summand $e\cdot t$. Since $s$ and $t$ are normal forms and have sizes smaller than $n$ and $n'$, by the induction hypotheses $s\sim_{hhp} t$ implies $s=_{AC} t$.
\end{itemize}

So, we get $n=_{AC} n'$.

Finally, let $s$ and $t$ be basic terms, and $s\sim_{hhp} t$, there are normal forms $n$ and $n'$, such that $s=n$ and $t=n'$. The soundness theorem of $BATC$ modulo history-preserving bisimulation equivalence (see Theorem \ref{SBATCHHPBE}) yields $s\sim_{hhp} n$ and $t\sim_{hhp} n'$, so $n\sim_{hhp} s\sim_{hhp} t\sim_{hhp} n'$. Since if $n\sim_{hhp} n'$ then $n=_{AC}n'$, $s=n=_{AC}n'=t$, as desired.
\end{proof}

\section{Algebra for Parallelism in True Concurrency}\label{aptc}

In this section, we will discuss parallelism in true concurrency. We know that parallelism can be modeled by left merge and communication merge in $ACP$ (Algebra of Communicating Process) \cite{ALNC} \cite{ACP} with an interleaving bisimulation semantics. Parallelism in true concurrency is quite different to that in interleaving bisimulation: it is a fundamental computational pattern (modeled by parallel operator $\parallel$) and cannot be merged (replaced by other operators). The resulted algebra is called Algebra for Parallelism in True Concurrency, abbreviated $APTC$.

\subsection{Parallelism as a Fundamental Computational Pattern}

Through several propositions, we show that parallelism is a fundamental computational pattern. Firstly, we give the transition rules for parallel operator $\parallel$ as follows, it is suitable for all truly concurrent behavioral equivalence, including pomset bisimulation, step bisimulation, hp-bisimulation and hhp-bisimulation.

\begin{center}
    \begin{table}
        $$\frac{x\xrightarrow{e_1}\surd\quad y\xrightarrow{e_2}\surd}{x\parallel y\xrightarrow{\{e_1,e_2\}}\surd} \quad\frac{x\xrightarrow{e_1}x'\quad y\xrightarrow{e_2}\surd}{x\parallel y\xrightarrow{\{e_1,e_2\}}x'}$$
        $$\frac{x\xrightarrow{e_1}\surd\quad y\xrightarrow{e_2}y'}{x\parallel y\xrightarrow{\{e_1,e_2\}}y'} \quad\frac{x\xrightarrow{e_1}x'\quad y\xrightarrow{e_2}y'}{x\parallel y\xrightarrow{\{e_1,e_2\}}x'\parallel y'}$$
        \caption{Transition rules of parallel operator $\parallel$}
        \label{TRForParallel}
    \end{table}
\end{center}

We will show that Milner's expansion law \cite{ALNC} does not hold modulo any truly concurrent behavioral equivalence, as the following proposition shows.

\begin{proposition}[Milner's expansion law modulo truly concurrent behavioral equivalence]\label{MEL}
Milner's expansion law does not hold modulo any truly concurrent behavioral equivalence, that is:
\begin{enumerate}
  \item For atomic event $e_1$ and $e_2$,
  \begin{enumerate}
    \item $e_1\parallel e_2\nsim_p e_1\cdot e_2 + e_2\cdot e_1$;
    \item $e_1\parallel e_2\nsim_s e_1\cdot e_2 + e_2\cdot e_1$;
    \item $e_1\parallel e_2\nsim_{hp} e_1\cdot e_2 + e_2\cdot e_1$;
    \item $e_1\parallel e_2\nsim_{hhp} e_1\cdot e_2 + e_2\cdot e_1$;
  \end{enumerate}
  \item Specially, for auto-concurrency, let $e$ be an atomic event,
    \begin{enumerate}
    \item $e\parallel e\nsim_p e\cdot e$;
    \item $e\parallel e\nsim_s e\cdot e$;
    \item $e\parallel e\nsim_{hp} e\cdot e$;
    \item $e\parallel e\nsim_{hhp} e\cdot e$.
    \end{enumerate}
\end{enumerate}
\end{proposition}

\begin{proof}
In nature, it is caused by $e_1\parallel e_2$ and $e_1\cdot e_2 + e_2\cdot e_1$ (specially $e\parallel e$ and $e\cdot e$) having different causality structure. They are based on the following obvious facts according to transition rules for parallel operator in Table \ref{TRForParallel}:

\begin{enumerate}
  \item $e_1\parallel e_2\xrightarrow{\{e_1,e_2\}}\surd$, while $e_1\cdot e_2+ e_2\cdot e_1\nrightarrow^{\{e_1,e_2\}}$;
  \item specially, $e\parallel e\xrightarrow{\{e,e\}}\surd$, while $e\cdot e\nrightarrow^{\{e,e\}}$.
\end{enumerate}
\end{proof}

In the following, we show that the elimination theorem does not hold for truly concurrent processes combined the operators $\cdot$, $+$ and $\parallel$. Firstly, we define the basic terms for $APTC$.

\begin{definition}[Basic terms of $APTC$]\label{BTAPTC}
The set of basic terms of $APTC$, $\mathcal{B}(APTC)$, is inductively defined as follows:
\begin{enumerate}
  \item $\mathbb{E}\subset\mathcal{B}(APTC)$;
  \item if $e\in \mathbb{E}, t\in\mathcal{B}(APTC)$ then $e\cdot t\in\mathcal{B}(APTC)$;
  \item if $t,s\in\mathcal{B}(APTC)$ then $t+ s\in\mathcal{B}(APTC)$;
  \item if $t,s\in\mathcal{B}(APTC)$ then $t\parallel s\in\mathcal{B}(APTC)$.
\end{enumerate}
\end{definition}

\begin{proposition}[About elimination theorem of $APTC$]\label{ETAPTC}
\begin{enumerate}
    \item Let $p$ be a closed $APTC$ term. Then there may not be a closed $BATC$ term $q$ such that $APTC\vdash p=q$;
    \item Let $p$ be a closed $APTC$ term. Then there may not be a closed basic $APTC$ term $q$ such that $APTC\vdash p=q$.
\end{enumerate}
\end{proposition}

\begin{proof}
\begin{enumerate}
  \item By Proposition \ref{MEL};
  \item We show this property through two aspects:
  \begin{enumerate}
    \item The left and right distributivity of $\cdot$ to $\parallel$, and $\parallel$ to $\cdot$, do not hold modulo any truly concurrent bisimulation equivalence.

        Left distributivity of $\cdot$ to $\parallel$: $(e_1\cdot e_2) \parallel (e_1\cdot e_3)\xrightarrow{\{e_1,e_1\}}e_2\parallel e_3$, while $e_1\cdot (e_2\parallel e_3)\nrightarrow^{\{e_1,e_1\}}$.

        Right distributivity of $\cdot$ to $\parallel$: $(e_1\cdot e_3) \parallel (e_2\cdot e_3)\xrightarrow{\{e_1,e_2\}}e_3\parallel e_3\xrightarrow{\{e_3,e_3\}}\surd$, while $(e_1\parallel e_2)\cdot e_3\xrightarrow{\{e_1,e_2\}}e_3\nrightarrow^{\{e_3,e_3\}}$.

        Left distributivity of $\parallel$ to $\cdot$: $(e_1\parallel e_2)\cdot(e_1\parallel e_3)\xrightarrow{\{e_1,e_2\}}e_1\parallel e_3\xrightarrow{\{e_1,e_3\}}\surd$, while $e_1\parallel (e_2\cdot e_3)\xrightarrow{\{e_1,e_2\}}e_3\nrightarrow^{\{e_1,e_3\}}$.

        Right distributivity of $\parallel$ to $\cdot$: $(e_1\parallel e_3)\cdot (e_2\parallel e_3)\xrightarrow{\{e_1,e_3\}}e_2\parallel e_3\xrightarrow{\{e_2,e_3\}}\surd$, while $(e_1\cdot e_2)\parallel e_3\xrightarrow{\{e_1,e_3\}}e_2\nrightarrow^{\{e_2,e_3\}}$.

        This means that there are not normal forms for the closed basic $APTC$ terms.

    \item There are causality relations among different parallel branches can not be expressed by closed basic $APTC$ terms.

        \begin{figure}
            \centering
            \includegraphics{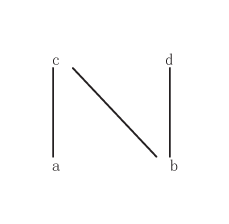}
            \caption{}
            \label{causality}
        \end{figure}

        We consider the graph as Fig. \ref{causality} illustrates. There are four events labeled $a,b,c,d$, and there are three causality relations: $c$ after $a$, $d$ after $b$, and $c$ after $b$. This graph can not be expressed by basic $APTC$ terms. $a$ and $b$ are in parallel, $c$ after $a$, so $c$ and $a$ are in the same parallel branch; $d$ after $b$, so $d$ and $b$ are in the same parallel branch; so $c$ and $d$ are in different parallel branches. But, $c$ after $b$ means that $c$ and $d$ are in the same parallel branch. This causes contradictions, it means that the graph in Fig. \ref{causality} can not be expressed by closed basic $APTC$ terms.
  \end{enumerate}
\end{enumerate}
\end{proof}

Until now, we see that parallelism acts as a fundamental computational pattern, and any elimination theorem does not hold any more. In nature, an event structure $\mathcal{E}$ (see Definition \ref{PES}) is a graph defined by causality and conflict relations among events, while concurrency and consistency are implicitly defined by causality and conflict. The above conclusions say that an event structure $\mathcal{E}$ cannot be fully structured, the explicit parallel operator $\parallel$ in a fully structured event structure combined by $\cdot$, $+$ and $\parallel$ can not be replaced by $\cdot$ and $+$, and a fully structured event structure combined by $\cdot$, $+$ and $\parallel$ has no a normal form.

The above propositions mean that a perfectly sound and complete axiomatization of parallelism for truly concurrent bisimulation equivalence (like $ACP$ \cite{ACP} for bisimulation equivalence) \emph{cannot} be established. Then, what can we do for $APTC$?

\subsection{Axiom System of Parallelism}

Though a fully sound and complete axiomatization for $APTC$ seems impossible, we must and can do something, we believe. We also believe that the future is fully implied by the history, let us reconsider parallelism in interleaving bisimulation. In $ACP$ \cite{ACP}, the full parallelism is captured by an auxiliary left merge and communication merge, left merge captures the interleaving bisimulation semantics, while communication merge expresses the communications among parallel branches. In true concurrency, if we try to define parallelism explicitly like $APTC$, the left merge captured Milner's expansion law does not hold any more, while communications among different parallel branches captured by communication merge still stand there. So, it is reasonable to assume that causality relations among different parallel branches are all communications among them. The communication between two parallel branches is defined as a communicating function between two communicating events $e_1, e_2\in \mathbb{E}$, $\gamma(e_1, e_2): \mathbb{E}\times \mathbb{E}\rightarrow \mathbb{E}$.

The communications among parallel branches are still defined by the communication operator $\mid$, which is expressed by four transition rules in Table \ref{TRForCommunication}. The whole parallelism semantics is modeled by the parallel operator $\parallel$ and communication operator $\mid$, we denote the whole parallel operator as $\between$ (for the transition rules of $\between$, we omit them).

\begin{center}
    \begin{table}
        $$\frac{x\xrightarrow{e_1}\surd\quad y\xrightarrow{e_2}\surd}{x\mid y\xrightarrow{\gamma(e_1,e_2)}\surd} \quad\frac{x\xrightarrow{e_1}x'\quad y\xrightarrow{e_2}\surd}{x\mid y\xrightarrow{\gamma(e_1,e_2)}x'}$$
        $$\frac{x\xrightarrow{e_1}\surd\quad y\xrightarrow{e_2}y'}{x\mid y\xrightarrow{\gamma(e_1,e_2)}y'} \quad\frac{x\xrightarrow{e_1}x'\quad y\xrightarrow{e_2}y'}{x\mid y\xrightarrow{\gamma(e_1,e_2)}x'\between y'}$$
        \caption{Transition rules of communication operator $\mid$}
        \label{TRForCommunication}
    \end{table}
\end{center}

Note that the last transition rule for the parallel operator $\parallel$ in Table \ref{TRForParallel} should be modified to the following one.

$$\frac{x\xrightarrow{e_1}x'\quad y\xrightarrow{e_2}y'}{x\parallel y\xrightarrow{\{e_1,e_2\}}x'\between y'}$$

By communication operator $\mid$, the causality relation among different parallel branches are structured (we will show the algebra laws on communication operator in the following). Now, let us consider conflicts in parallelism. The conflicts exist within the same parallel branches can be captured by $+$ by a structured way, but, how to express conflicts among events in different parallel branches? The conflict relation is also a binary relation between two events $e_1, e_2\in \mathbb{E}$, $\sharp(e_1,e_2):\mathbb{E}\times \mathbb{E}\rightarrow \mathbb{E}$, and we know that $\sharp$ is irreflexive, symmetric and hereditary with respect to $\cdot$, that is, for all $e,e',e''\in \mathbb{E}$, if $e\sharp e'\cdot e''$, then $e\sharp e''$ (see Definition \ref{PES}).

These conflicts among different parallel branches must be eliminated to make the concurrent process structured. We are inspired by modeling of priority in $ACP$ \cite{ACP}, the conflict elimination is also captured by two auxiliary operators, the unary conflict elimination operator $\Theta$ and the binary unless operator $\triangleleft$. The transition rules for $\Theta$ and $\triangleleft$ are expressed by ten transition rules in Table \ref{TRForConflict}.

\begin{center}
    \begin{table}
        $$\frac{x\xrightarrow{e_1}\surd\quad (\sharp(e_1,e_2))}{\Theta(x)\xrightarrow{e_1}\surd} \quad\frac{x\xrightarrow{e_2}\surd\quad (\sharp(e_1,e_2))}{\Theta(x)\xrightarrow{e_2}\surd}$$
        $$\frac{x\xrightarrow{e_1}x'\quad (\sharp(e_1,e_2))}{\Theta(x)\xrightarrow{e_1}\Theta(x')} \quad\frac{x\xrightarrow{e_2}x'\quad (\sharp(e_1,e_2))}{\Theta(x)\xrightarrow{e_2}\Theta(x')}$$
        $$\frac{x\xrightarrow{e_1}\surd \quad y\nrightarrow^{e_2}\quad (\sharp(e_1,e_2))}{x\triangleleft y\xrightarrow{\tau}\surd}
        \quad\frac{x\xrightarrow{e_1}x' \quad y\nrightarrow^{e_2}\quad (\sharp(e_1,e_2))}{x\triangleleft y\xrightarrow{\tau}x'}$$
        $$\frac{x\xrightarrow{e_1}\surd \quad y\nrightarrow^{e_3}\quad (\sharp(e_1,e_2),e_2\leq e_3)}{x\triangleleft y\xrightarrow{e_1}\surd}
        \quad\frac{x\xrightarrow{e_1}x' \quad y\nrightarrow^{e_3}\quad (\sharp(e_1,e_2),e_2\leq e_3)}{x\triangleleft y\xrightarrow{e_1}x'}$$
        $$\frac{x\xrightarrow{e_3}\surd \quad y\nrightarrow^{e_2}\quad (\sharp(e_1,e_2),e_1\leq e_3)}{x\triangleleft y\xrightarrow{\tau}\surd}
        \quad\frac{x\xrightarrow{e_3}x' \quad y\nrightarrow^{e_2}\quad (\sharp(e_1,e_2),e_1\leq e_3)}{x\triangleleft y\xrightarrow{\tau}x'}$$
        \caption{Transition rules of conflict elimination}
        \label{TRForConflict}
    \end{table}
\end{center}

In four transition rules in Table \ref{TRForConflict}, there is a new constant $\tau$ called silent step (see section \ref{abs}), this makes the semantics of conflict elimination is really based on weakly true concurrency (see Definition \ref{WPSB} and Definition \ref{WHHPB}), and we should move it to section \ref{abs}. But the movement would make $APTC$ incomplete (conflicts among different parallel branches cannot be expressed), let us forget this regret and just remember that $\tau$ can be eliminated, without anything on weakly true concurrency.

Ok, causality relations and conflict relations among events in different parallel branches are structured. In the following, we prove the congruence theorem.

\begin{theorem}[Congruence theorem of $APTC$]
Truly concurrent bisimulation equivalences $\sim_{p}$, $\sim_s$, $\sim_{hp}$ and $\sim_{hhp}$ are all congruences with respect to $APTC$.
\end{theorem}

\begin{proof}
(1) Case pomset bisimulation equivalence $\sim_p$.

\begin{itemize}
  \item Case parallel operator $\parallel$. Let $x_1,x_2$ and $y_1,y_2$ be $APTC$ processes, and $x_1\sim_{p} y_1$, $x_2\sim_{p} y_2$, it is sufficient to prove that $x_1\parallel x_2\sim_{p} y_1\parallel y_2$.

      By the definition of pomset bisimulation $\sim_p$ (Definition \ref{PSB}), $x_1\sim_p y_1$ means that

      $$x_1\xrightarrow{X_1} x_1' \quad y_1\xrightarrow{Y_1} y_1'$$

      with $X_1\subseteq x_1$, $Y_1\subseteq y_1$, $X_1\sim Y_1$ and $x_1'\sim_p y_1'$. The meaning of $x_2\sim_p y_2$ is similar.

      By the pomset transition rules for parallel operator $\parallel$ in Table \ref{TRForParallel}, we can get

      $$x_1\parallel x_2\xrightarrow{\{X_1,X_2\}} \surd \quad y_1\parallel y_2\xrightarrow{\{Y_1,Y_2\}} \surd$$

      with $X_1\subseteq x_1$, $Y_1\subseteq y_1$, $X_2\subseteq x_2$, $Y_2\subseteq y_2$, $X_1\sim Y_1$ and $X_2\sim Y_2$, so, we get $x_1\parallel x_2\sim_p y_1\parallel y_2$, as desired.

      Or, we can get

      $$x_1\parallel x_2\xrightarrow{\{X_1,X_2\}} x_1' \quad y_1\parallel y_2\xrightarrow{\{Y_1,Y_2\}} y_1'$$

      with $X_1\subseteq x_1$, $Y_1\subseteq y_1$, $X_2\subseteq x_2$, $Y_2\subseteq y_2$, $X_1\sim Y_1$, $X_2\sim Y_2$, and $x_1'\sim_p y_1'$, so, we get $x_1\parallel x_2\sim_p y_1\parallel y_2$, as desired.

      Or, we can get

      $$x_1\parallel x_2\xrightarrow{\{X_1,X_2\}} x_2' \quad y_1\parallel y_2\xrightarrow{\{Y_1,Y_2\}} y_2'$$

      with $X_1\subseteq x_1$, $Y_1\subseteq y_1$, $X_2\subseteq x_2$, $Y_2\subseteq y_2$, $X_1\sim Y_1$, $X_2\sim Y_2$, and $x_2'\sim_p y_2'$, so, we get $x_1\parallel x_2\sim_p y_1\parallel y_2$, as desired.

      Or, we can get

      $$x_1\parallel x_2\xrightarrow{\{X_1,X_2\}} x_1'\between x_2' \quad y_1\parallel y_2\xrightarrow{\{Y_1,Y_2\}} y_1'\between y_2'$$

      with $X_1\subseteq x_1$, $Y_1\subseteq y_1$, $X_2\subseteq x_2$, $Y_2\subseteq y_2$, $X_1\sim Y_1$, $X_2\sim Y_2$, $x_1'\sim_p y_1'$ and $x_2'\sim_p y_2'$, and also the assumption $x_1'\between x_2'\sim_p y_1'\between y_2'$, so, we get $x_1\parallel x_2\sim_p y_1\parallel y_2$, as desired.

  \item Case communication operator $\mid$. It can be proved similarly to the case of parallel operator $\parallel$, we omit it. Note that, a communication is defined between two single communicating events.

  \item Case conflict elimination operator $\Theta$. It can be proved similarly to the above cases, we omit it. Note that the conflict elimination operator $\Theta$ is a unary operator.

  \item Case unless operator $\triangleleft$. It can be proved similarly to the case of parallel operator $\parallel$, we omit it. Note that, a conflict relation is defined between two single events.

\end{itemize}

(2) The cases of step bisimulation $\sim_s$, hp-bisimulation $\sim_{hp}$ and hhp-bisimulation $\sim_{hhp}$ can be proven similarly, we omit them.
\end{proof}

So, we design the axioms of parallelism in Table \ref{AxiomsForParallelism}, including algebraic laws for parallel operator $\parallel$, communication operator $\mid$, conflict elimination operator $\Theta$ and unless operator $\triangleleft$, and also the whole parallel operator $\between$. Since the communication between two communicating events in different parallel branches may cause deadlock (a state of inactivity), which is caused by mismatch of two communicating events or the imperfectness of the communication channel. We introduce a new constant $\delta$ to denote the deadlock, and let the atomic event $e\in \mathbb{E}\cup\{\delta\}$.

\begin{center}
    \begin{table}
        \begin{tabular}{@{}ll@{}}
            \hline No. &Axiom\\
            $A6$ & $x+ \delta = x$\\
            $A7$ & $\delta\cdot x =\delta$\\
            $P1$ & $x\between y = x\parallel y + x\mid y$\\
            $P2$ & $x\parallel y = y \parallel x$\\
            $P3$ & $(x\parallel y)\parallel z = x\parallel (y\parallel z)$\\
            $P4$ & $e_1\parallel (e_2\cdot y) = (e_1\parallel e_2)\cdot y$\\
            $P5$ & $(e_1\cdot x)\parallel e_2 = (e_1\parallel e_2)\cdot x$\\
            $P6$ & $(e_1\cdot x)\parallel (e_2\cdot y) = (e_1\parallel e_2)\cdot (x\between y)$\\
            $P7$ & $(x+ y)\parallel z = (x\parallel z)+ (y\parallel z)$\\
            $P8$ & $x\parallel (y+ z) = (x\parallel y)+ (x\parallel z)$\\
            $P9$ & $\delta\parallel x = \delta$\\
            $P10$ & $x\parallel \delta = \delta$\\
            $C11$ & $e_1\mid e_2 = \gamma(e_1,e_2)$\\
            $C12$ & $e_1\mid (e_2\cdot y) = \gamma(e_1,e_2)\cdot y$\\
            $C13$ & $(e_1\cdot x)\mid e_2 = \gamma(e_1,e_2)\cdot x$\\
            $C14$ & $(e_1\cdot x)\mid (e_2\cdot y) = \gamma(e_1,e_2)\cdot (x\between y)$\\
            $C15$ & $(x+ y)\mid z = (x\mid z) + (y\mid z)$\\
            $C16$ & $x\mid (y+ z) = (x\mid y)+ (x\mid z)$\\
            $C17$ & $\delta\mid x = \delta$\\
            $C18$ & $x\mid\delta = \delta$\\
            $CE19$ & $\Theta(e) = e$\\
            $CE20$ & $\Theta(\delta) = \delta$\\
            $CE21$ & $\Theta(x+ y) = \Theta(x)\triangleleft y + \Theta(y)\triangleleft x$\\
            $CE22$ & $\Theta(x\cdot y)=\Theta(x)\cdot\Theta(y)$\\
            $CE23$ & $\Theta(x\parallel y) = ((\Theta(x)\triangleleft y)\parallel y)+ ((\Theta(y)\triangleleft x)\parallel x)$\\
            $CE24$ & $\Theta(x\mid y) = ((\Theta(x)\triangleleft y)\mid y)+ ((\Theta(y)\triangleleft x)\mid x)$\\
            $U25$ & $(\sharp(e_1,e_2))\quad e_1\triangleleft e_2 = \tau$\\
            $U26$ & $(\sharp(e_1,e_2),e_2\leq e_3)\quad e_1\triangleleft e_3 = e_1$\\
            $U27$ & $(\sharp(e_1,e_2),e_2\leq e_3)\quad e3\triangleleft e_1 = \tau$\\
            $U28$ & $e\triangleleft \delta = e$\\
            $U29$ & $\delta \triangleleft e = \delta$\\
            $U30$ & $(x+ y)\triangleleft z = (x\triangleleft z)+ (y\triangleleft z)$\\
            $U31$ & $(x\cdot y)\triangleleft z = (x\triangleleft z)\cdot (y\triangleleft z)$\\
            $U32$ & $(x\parallel y)\triangleleft z = (x\triangleleft z)\parallel (y\triangleleft z)$\\
            $U33$ & $(x\mid y)\triangleleft z = (x\triangleleft z)\mid (y\triangleleft z)$\\
            $U34$ & $x\triangleleft (y+ z) = (x\triangleleft y)\triangleleft z$\\
            $U35$ & $x\triangleleft (y\cdot z)=(x\triangleleft y)\triangleleft z$\\
            $U36$ & $x\triangleleft (y\parallel z) = (x\triangleleft y)\triangleleft z$\\
            $U37$ & $x\triangleleft (y\mid z) = (x\triangleleft y)\triangleleft z$\\
        \end{tabular}
        \caption{Axioms of parallelism}
        \label{AxiomsForParallelism}
    \end{table}
\end{center}

We explain the intuitions of the axioms of parallelism in Table \ref{AxiomsForParallelism} in the following. The axiom $A6$ says that the deadlock $\delta$ is redundant in the process term $t+ \delta$. $A7$ says that the deadlock blocks all behaviors of the process term $\delta\cdot t$.

The axiom $P1$ is the definition of the whole parallelism $\between$, which says that $s\between t$ either is the form of $s\parallel t$ or $s\mid t$. $P2$ says that $\parallel$ satisfies commutative law, while $P3$ says that $\parallel$ satisfies associativity. $P4$, $P5$ and $P6$ are the defining axioms of $\parallel$, say the $s\parallel t$ executes $s$ and $t$ concurrently. $P7$ and $P8$ are the right and left distributivity of $\parallel$ to $+$. $P9$ and $P10$ say that both $\delta\parallel t$ and $t\parallel\delta$ all block any event.

$C11$, $C12$, $C13$ and $C14$ are the defining axioms of the communication operator $\mid$ which say that $s\mid t$ makes a communication between $s$ and $t$. $C15$ and $C16$ are the right and left distributivity of $\mid$ to $+$. $C17$ and $C18$ say that both $\delta\mid t$ and $t\mid\delta$ all block any event.

$CE19$ and $CE20$ say that the conflict elimination operator $\Theta$ leaves atomic events and the deadlock unchanged. $CE21-CE24$ are the functions of $\Theta$ acting on the operators $+$, $\cdot$, $\parallel$ and $\mid$. $U25$, $U26$ and $U27$ are the defining laws of the unless operator $\triangleleft$, in $U25$ and $U27$, there is a new constant $\tau$, the silent step, we will discuss $\tau$ in details in section \ref{abs}, in these two axioms, we just need to remember that $\tau$ really keeps silent. $U28$ says that the deadlock $\delta$ cannot block any event in the process term $e\triangleleft\delta$, while $U29$ says that $\delta\triangleleft e$ does not exhibit any behavior. $U30-U37$ are the disguised right and left distributivity of $\triangleleft$ to the operators $+$, $\cdot$, $\parallel$ and $\mid$.

\subsection{Properties of Parallelism}

Based on the definition of basic terms for $APTC$ (see Definition \ref{BTAPTC}) and axioms of parallelism (see Table \ref{AxiomsForParallelism}), we can prove the elimination theorem of parallelism.

\begin{theorem}[Elimination theorem of parallelism]\label{ETParallelism}
Let $p$ be a closed $APTC$ term. Then there is a basic $APTC$ term $q$ such that $APTC\vdash p=q$.
\end{theorem}

\begin{proof}
(1) Firstly, suppose that the following ordering on the signature of $APTC$ is defined: $\parallel > \cdot > +$ and the symbol $\parallel$ is given the lexicographical status for the first argument, then for each rewrite rule $p\rightarrow q$ in Table \ref{TRSForAPTC} relation $p>_{lpo} q$ can easily be proved. We obtain that the term rewrite system shown in Table \ref{TRSForAPTC} is strongly normalizing, for it has finitely many rewriting rules, and $>$ is a well-founded ordering on the signature of $APTC$, and if $s>_{lpo} t$, for each rewriting rule $s\rightarrow t$ is in Table \ref{TRSForAPTC} (see Theorem \ref{SN}).

\begin{center}
    \begin{table}
        \begin{tabular}{@{}ll@{}}
            \hline No. &Rewriting Rule\\
            $RA6$ & $x+ \delta \rightarrow x$\\
            $RA7$ & $\delta\cdot x \rightarrow\delta$\\
            $RP1$ & $x\between y \rightarrow x\parallel y + x\mid y$\\
            $RP2$ & $x\parallel y \rightarrow y \parallel x$\\
            $RP3$ & $(x\parallel y)\parallel z \rightarrow x\parallel (y\parallel z)$\\
            $RP4$ & $e_1\parallel (e_2\cdot y) \rightarrow (e_1\parallel e_2)\cdot y$\\
            $RP5$ & $(e_1\cdot x)\parallel e_2 \rightarrow (e_1\parallel e_2)\cdot x$\\
            $RP6$ & $(e_1\cdot x)\parallel (e_2\cdot y) \rightarrow (e_1\parallel e_2)\cdot (x\between y)$\\
            $RP7$ & $(x+ y)\parallel z \rightarrow (x\parallel z)+ (y\parallel z)$\\
            $RP8$ & $x\parallel (y+ z) \rightarrow (x\parallel y)+ (x\parallel z)$\\
            $RP9$ & $\delta\parallel x \rightarrow \delta$\\
            $RP10$ & $x\parallel \delta \rightarrow \delta$\\
            $RC11$ & $e_1\mid e_2 \rightarrow \gamma(e_1,e_2)$\\
            $RC12$ & $e_1\mid (e_2\cdot y) \rightarrow \gamma(e_1,e_2)\cdot y$\\
            $RC13$ & $(e_1\cdot x)\mid e_2 \rightarrow \gamma(e_1,e_2)\cdot x$\\
            $RC14$ & $(e_1\cdot x)\mid (e_2\cdot y) \rightarrow \gamma(e_1,e_2)\cdot (x\between y)$\\
            $RC15$ & $(x+ y)\mid z \rightarrow (x\mid z) + (y\mid z)$\\
            $RC16$ & $x\mid (y+ z) \rightarrow (x\mid y)+ (x\mid z)$\\
            $RC17$ & $\delta\mid x \rightarrow \delta$\\
            $RC18$ & $x\mid\delta \rightarrow \delta$\\
            $RCE19$ & $\Theta(e) \rightarrow e$\\
            $RCE20$ & $\Theta(\delta) \rightarrow \delta$\\
            $RCE21$ & $\Theta(x+ y) \rightarrow \Theta(x)\triangleleft y + \Theta(y)\triangleleft x$\\
            $RCE22$ & $\Theta(x\cdot y)\rightarrow\Theta(x)\cdot\Theta(y)$\\
            $RCE23$ & $\Theta(x\parallel y) \rightarrow ((\Theta(x)\triangleleft y)\parallel y)+ ((\Theta(y)\triangleleft x)\parallel x)$\\
            $RCE24$ & $\Theta(x\mid y) \rightarrow ((\Theta(x)\triangleleft y)\mid y)+ ((\Theta(y)\triangleleft x)\mid x)$\\
            $RU25$ & $(\sharp(e_1,e_2))\quad e_1\triangleleft e_2 \rightarrow \tau$\\
            $RU26$ & $(\sharp(e_1,e_2),e_2\cdot e_3)\quad e_1\triangleleft e_3 \rightarrow e_1$\\
            $RU27$ & $(\sharp(e_1,e_2),e_2\cdot e_3)\quad e3\triangleleft e_1 \rightarrow \tau$\\
            $RU28$ & $e\triangleleft \delta \rightarrow e$\\
            $RU29$ & $\delta \triangleleft e \rightarrow \delta$\\
            $RU30$ & $(x+ y)\triangleleft z \rightarrow (x\triangleleft z)+ (y\triangleleft z)$\\
            $RU31$ & $(x\cdot y)\triangleleft z \rightarrow (x\triangleleft z)\cdot (y\triangleleft z)$\\
            $RU32$ & $(x\parallel y)\triangleleft z \rightarrow (x\triangleleft z)\parallel (y\triangleleft z)$\\
            $RU33$ & $(x\mid y)\triangleleft z \rightarrow (x\triangleleft z)\mid (y\triangleleft z)$\\
            $RU34$ & $x\triangleleft (y+ z) \rightarrow (x\triangleleft y)\triangleleft z$\\
            $RU35$ & $x\triangleleft (y\cdot z)\rightarrow (x\triangleleft y)\triangleleft z$\\
            $RU36$ & $x\triangleleft (y\parallel z) \rightarrow (x\triangleleft y)\triangleleft z$\\
            $RU37$ & $x\triangleleft (y\mid z) \rightarrow (x\triangleleft y)\triangleleft z$\\
        \end{tabular}
        \caption{Term rewrite system of $APTC$}
        \label{TRSForAPTC}
    \end{table}
\end{center}

(2) Then we prove that the normal forms of closed $APTC$ terms are basic $APTC$ terms.

Suppose that $p$ is a normal form of some closed $APTC$ term and suppose that $p$ is not a basic $APTC$ term. Let $p'$ denote the smallest sub-term of $p$ which is not a basic $APTC$ term. It implies that each sub-term of $p'$ is a basic $APTC$ term. Then we prove that $p$ is not a term in normal form. It is sufficient to induct on the structure of $p'$:

\begin{itemize}
  \item Case $p'\equiv e, e\in \mathbb{E}$. $p'$ is a basic $APTC$ term, which contradicts the assumption that $p'$ is not a basic $APTC$ term, so this case should not occur.
  \item Case $p'\equiv p_1\cdot p_2$. By induction on the structure of the basic $APTC$ term $p_1$:
      \begin{itemize}
        \item Subcase $p_1\in \mathbb{E}$. $p'$ would be a basic $APTC$ term, which contradicts the assumption that $p'$ is not a basic $APTC$ term;
        \item Subcase $p_1\equiv e\cdot p_1'$. $RA5$ rewriting rule in Table \ref{TRSForBATC} can be applied. So $p$ is not a normal form;
        \item Subcase $p_1\equiv p_1'+ p_1''$. $RA4$ rewriting rule in Table \ref{TRSForBATC} can be applied. So $p$ is not a normal form;
        \item Subcase $p_1\equiv p_1'\parallel p_1''$. $p'$ would be a basic $APTC$ term, which contradicts the assumption that $p'$ is not a basic $APTC$ term;
        \item Subcase $p_1\equiv p_1'\mid p_1''$. $RC11$ rewrite rule in Table \ref{TRSForAPTC} can be applied. So $p$ is not a normal form;
        \item Subcase $p_1\equiv \Theta(p_1')$. $RCE19$ and $RCE20$ rewrite rules in Table \ref{TRSForAPTC} can be applied. So $p$ is not a normal form.
      \end{itemize}
  \item Case $p'\equiv p_1+ p_2$. By induction on the structure of the basic $APTC$ terms both $p_1$ and $p_2$, all subcases will lead to that $p'$ would be a basic $APTC$ term, which contradicts the assumption that $p'$ is not a basic $APTC$ term.
  \item Case $p'\equiv p_1\parallel p_2$. By induction on the structure of the basic $APTC$ terms both $p_1$ and $p_2$, all subcases will lead to that $p'$ would be a basic $APTC$ term, which contradicts the assumption that $p'$ is not a basic $APTC$ term.
  \item Case $p'\equiv p_1\mid p_2$. By induction on the structure of the basic $APTC$ terms both $p_1$ and $p_2$, all subcases will lead to that $p'$ would be a basic $APTC$ term, which contradicts the assumption that $p'$ is not a basic $APTC$ term.
  \item Case $p'\equiv \Theta(p_1)$. By induction on the structure of the basic $APTC$ term $p_1$, $RCE19-RCE24$ rewrite rules in Table \ref{TRSForAPTC} can be applied. So $p$ is not a normal form.
  \item Case $p'\equiv p_1\triangleleft p_2$. By induction on the structure of the basic $APTC$ terms both $p_1$ and $p_2$, all subcases will lead to that $p'$ would be a basic $APTC$ term, which contradicts the assumption that $p'$ is not a basic $APTC$ term.
\end{itemize}
\end{proof}

\subsection{Structured Operational Semantics of Parallelism}

It is quite a challenge to prove the algebraic laws in Table \ref{AxiomsForParallelism} is sound/complete or unsound/incomplete modulo truly concurrent behavioral equivalence (pomset bisimulation equivalence, step bisimulation equivalence, hp-bisimulation equivalence and hhp-bisimulation equivalence), in this subsection, we try to do these.

\begin{theorem}[Generalization of the algebra for parallelism with respect to $BATC$]
The algebra for parallelism is a generalization of $BATC$.
\end{theorem}

\begin{proof}
It follows from the following three facts.

\begin{enumerate}
  \item The transition rules of $BATC$ in section \ref{batc} are all source-dependent;
  \item The sources of the transition rules for the algebra for parallelism contain an occurrence of $\between$, or $\parallel$, or $\mid$, or $\Theta$, or $\triangleleft$;
  \item The transition rules of $APTC$ are all source-dependent.
\end{enumerate}

So, the algebra for parallelism is a generalization of $BATC$, that is, $BATC$ is an embedding of the algebra for parallelism, as desired.
\end{proof}

\begin{theorem}[Soundness of parallelism modulo step bisimulation equivalence]\label{SPSBE}
Let $x$ and $y$ be $APTC$ terms. If $APTC\vdash x=y$, then $x\sim_{s} y$.
\end{theorem}

\begin{proof}
Since step bisimulation $\sim_s$ is both an equivalent and a congruent relation with respect to the operators $\between$, $\parallel$, $\mid$, $\Theta$ and $\triangleleft$, we only need to check if each axiom in Table \ref{AxiomsForParallelism} is sound modulo step bisimulation equivalence.

Though transition rules in Table \ref{TRForParallel}, \ref{TRForCommunication}, and \ref{TRForConflict} are defined in the flavor of single event, they can be modified into a step (a set of events within which each event is pairwise concurrent), we omit them. If we treat a single event as a step containing just one event, the proof of this soundness theorem does not exist any problem, so we use this way and still use the transition rules in Table \ref{TRForParallel}, \ref{TRForCommunication}, and \ref{TRForConflict}.

We omit the defining axioms, including axioms $P1$, $C11$, $CE19$, $CE20$, $U25-U27$ (the soundness of $U25$ and $U27$ is remained to section \ref{abs}); we also omit the trivial axioms related to $\delta$, including axioms $A6$, $A7$, $P9$, $P10$, $C17$, $C18$, $U28$ and $U29$; in the following, we only prove the soundness of the non-trivial axioms, including axioms $P2-P8$, $C12-C16$, $CE21-CE24$ and $U30-U37$.

\begin{itemize}
  \item \textbf{Axiom $P2$}. Let $p,q$ be $APTC$ processes, and $p\parallel q=q\parallel p$, it is sufficient to prove that $p\parallel q\sim_s q\parallel p$. By the transition rules for operator $\parallel$ in Table \ref{TRForParallel}, we get

      $$\frac{p\xrightarrow{e_1}\surd\quad q\xrightarrow{e_2}\surd}{p\parallel q\xrightarrow{\{e_1,e_2\}}\surd}
      \quad\frac{p\xrightarrow{e_1}\surd\quad q\xrightarrow{e_2}\surd}{q\parallel p\xrightarrow{\{e_1,e_2\}}\surd}$$

      $$\frac{p\xrightarrow{e_1}p'\quad q\xrightarrow{e_2}\surd}{p\parallel q\xrightarrow{\{e_1,e_2\}}p'}
      \quad\frac{p\xrightarrow{e_1}p'\quad q\xrightarrow{e_2}\surd}{q\parallel p\xrightarrow{\{e_1,e_2\}}p'}$$

      $$\frac{p\xrightarrow{e_1}\surd\quad q\xrightarrow{e_2}q'}{p\parallel q\xrightarrow{\{e_1,e_2\}}q'}
      \quad\frac{p\xrightarrow{e_1}\surd\quad q\xrightarrow{e_2}q'}{q\parallel p\xrightarrow{\{e_1,e_2\}}q'}$$

      $$\frac{p\xrightarrow{e_1}p'\quad q\xrightarrow{e_2}q'}{p\parallel q\xrightarrow{\{e_1,e_2\}}p'\between q'}
      \quad\frac{p\xrightarrow{e_1}p'\quad q\xrightarrow{e_2}q'}{q\parallel p\xrightarrow{\{e_1,e_2\}}q'\between p'}$$

      So, with the assumption $p'\between q' = q'\between p'$, $p\parallel q\sim_s q\parallel p$, as desired.
  \item \textbf{Axiom $P3$}. Let $p,q,r$ be $APTC$ processes, and $(p\parallel q)\parallel r=p\parallel(q\parallel r)$, it is sufficient to prove that $(p\parallel q)\parallel r\sim_s p\parallel(q\parallel r)$. By the transition rules for operator $\parallel$ in Table \ref{TRForParallel}, we get

      $$\frac{p\xrightarrow{e_1}\surd\quad q\xrightarrow{e_2}\surd\quad r\xrightarrow{e_3}\surd}{(p\parallel q)\parallel r\xrightarrow{\{e_1,e_2,e_3\}}\surd}
      \quad\frac{p\xrightarrow{e_1}\surd\quad q\xrightarrow{e_2}\surd\quad r\xrightarrow{e_3}\surd}{p\parallel(q\parallel r)\xrightarrow{\{e_1,e_2,e_3\}}\surd}$$

      $$\quad\frac{p\xrightarrow{e_1}p'\quad q\xrightarrow{e_2}\surd\quad r\xrightarrow{e_3}\surd}{(p\parallel q)\parallel r\xrightarrow{\{e_1,e_2,e_3\}}p'}
      \quad\frac{p\xrightarrow{e_1}p'\quad q\xrightarrow{e_2}\surd\quad r\xrightarrow{e_3}\surd}{p\parallel(q\parallel r)\xrightarrow{\{e_1,e_2,e_3\}}p'}$$

      There are also two cases that two process terms successfully terminate, we omit them.

      $$\frac{p\xrightarrow{e_1}p'\quad q\xrightarrow{e_2}q'\quad r\xrightarrow{e_3}\surd}{(p\parallel q)\parallel r\xrightarrow{\{e_1,e_2,e_3\}}p'\between q'}
      \quad\frac{p\xrightarrow{e_1}p'\quad q\xrightarrow{e_2}q'\quad r\xrightarrow{e_3}\surd}{p\parallel(q\parallel r)\xrightarrow{\{e_1,e_2,e_3\}}p'\between q'}$$

      There are also other cases that just one process term successfully terminate, we also omit them.

      $$\frac{p\xrightarrow{e_1}p'\quad q\xrightarrow{e_2}q'\quad r\xrightarrow{e_3} r'}{(p\parallel q)\parallel r'\xrightarrow{\{e_1,e_2,e_3\}}(p'\between q')\between r'}
      \quad\frac{p\xrightarrow{e_1}p'\quad q\xrightarrow{e_2}q'\quad r\xrightarrow{e_3} r'}{p\parallel (q\parallel r)\xrightarrow{\{e_1,e_2,e_3\}}p'\between(q'\between r')}$$

      So, with the assumption $(p'\between q')\between r' = p'\between (q'\between r')$, $(p\parallel q)\parallel r\sim_s p\parallel (q\parallel r)$, as desired.
  \item \textbf{Axiom $P4$}. Let $q$ be an $APTC$ process, and $e_1\parallel (e_2\cdot q)=(e_1\parallel e_2)\cdot q$, it is sufficient to prove that $e_1\parallel(e_2\cdot q)\sim_s (e_1\parallel e_2)\cdot q$. By the transition rules for operator $\parallel$ in Table \ref{TRForParallel}, we get

      $$\frac{e_1\xrightarrow{e_1}\surd\quad e_2\cdot q\xrightarrow{e_2}q}{e_1\parallel (e_2\cdot q)\xrightarrow{\{e_1,e_2\}}q}$$

      $$\frac{e_1\xrightarrow{e_1}\surd\quad e_2\xrightarrow{e_2}\surd}{(e_1\parallel e_2)\cdot q\xrightarrow{\{e_1,e_2\}}q}$$

      So, $e_1\parallel (e_2\cdot q)\sim_s (e_1\parallel e_2)\cdot q$, as desired.
  \item \textbf{Axiom $P5$}. Let $p$ be an $APTC$ process, and $(e_1\cdot p)\parallel e_2=(e_1\parallel e_2)\cdot p$, it is sufficient to prove that $(e_1\cdot p)\parallel e_2\sim_s (e_1\parallel e_2)\cdot p$. By the transition rules for operator $\parallel$ in Table \ref{TRForParallel}, we get

      $$\frac{e_1\cdot p\xrightarrow{e_1}p\quad e_2\xrightarrow{e_2}\surd}{(e_1\cdot p)\parallel e_2\xrightarrow{\{e_1,e_2\}}p}$$

      $$\frac{e_1\xrightarrow{e_1}\surd\quad e_2\xrightarrow{e_2}\surd}{(e_1\parallel e_2)\cdot p\xrightarrow{\{e_1,e_2\}}p}$$

      So, $(e_1\cdot p)\parallel e_2\sim_s (e_1\parallel e_2)\cdot p$, as desired.
  \item \textbf{Axiom $P6$}. Let $p,q$ be $APTC$ processes, and $(e_1\cdot p)\parallel (e_2\cdot q)=(e_1\parallel e_2)\cdot (p\between q)$, it is sufficient to prove that $(e_1\cdot p)\parallel(e_2\cdot q)\sim_s (e_1\parallel e_2)\cdot (p\between q)$. By the transition rules for operator $\parallel$ in Table \ref{TRForParallel}, we get

      $$\frac{e_1\cdot p\xrightarrow{e_1}p\quad e_2\cdot q\xrightarrow{e_2}q}{(e_1\cdot p)\parallel (e_2\cdot q)\xrightarrow{\{e_1,e_2\}}p\between q}$$

      $$\frac{e_1\xrightarrow{e_1}\surd\quad e_2\xrightarrow{e_2}\surd}{(e_1\parallel e_2)\cdot (p\between q)\xrightarrow{\{e_1,e_2\}}p\between q}$$

      So, $(e_1\cdot p)\parallel (e_2\cdot q)\sim_s (e_1\parallel e_2)\cdot (p\between q)$, as desired.
  \item \textbf{Axiom $P7$}. Let $p,q,r$ be $APTC$ processes, and $(p+ q)\parallel r = (p\parallel r) + (q\parallel r)$, it is sufficient to prove that $(p+ q)\parallel r \sim_s (p\parallel r) + (q\parallel r)$. By the transition rules for operators $+$ and $\parallel$ in Table \ref{STRForBATC} and \ref{TRForParallel}, we get

      $$\frac{p\xrightarrow{e_1}\surd\quad r\xrightarrow{e_2}\surd}{(p+ q)\parallel r\xrightarrow{\{e_1,e_2\}}\surd}
      \quad \frac{p\xrightarrow{e_1}\surd\quad r\xrightarrow{e_2}\surd}{(p\parallel r)+ (q\parallel r)\xrightarrow{\{e_1,e_2\}}\surd}$$

      $$\frac{q\xrightarrow{e_1}\surd\quad r\xrightarrow{e_2}\surd}{(p+ q)\parallel r\xrightarrow{\{e_1,e_2\}}\surd}
      \quad \frac{q\xrightarrow{e_1}\surd\quad r\xrightarrow{e_2}\surd}{(p\parallel r)+ (q\parallel r)\xrightarrow{\{e_1,e_2\}}\surd}$$

      $$\frac{p\xrightarrow{e_1}p'\quad r\xrightarrow{e_2}\surd}{(p+ q)\parallel r\xrightarrow{\{e_1,e_2\}}p'}
      \quad \frac{p\xrightarrow{e_1}p'\quad r\xrightarrow{e_2}\surd}{(p\parallel r)+ (q\parallel r)\xrightarrow{\{e_1,e_2\}}p'}$$

      $$\frac{q\xrightarrow{e_1}q'\quad r\xrightarrow{e_2}\surd}{(p+ q)\parallel r\xrightarrow{\{e_1,e_2\}}q'}
      \quad \frac{q\xrightarrow{e_1}q'\quad r\xrightarrow{e_2}\surd}{(p\parallel r)+ (q\parallel r)\xrightarrow{\{e_1,e_2\}}q'}$$

      $$\frac{p\xrightarrow{e_1}\surd\quad r\xrightarrow{e_2}r'}{(p+ q)\parallel r\xrightarrow{\{e_1,e_2\}}r'}
      \quad \frac{p\xrightarrow{e_1}\surd\quad r\xrightarrow{e_2}r'}{(p\parallel r)+ (q\parallel r)\xrightarrow{\{e_1,e_2\}}r'}$$

      $$\frac{q\xrightarrow{e_1}\surd\quad r\xrightarrow{e_2}r'}{(p+ q)\parallel r\xrightarrow{\{e_1,e_2\}}r'}
      \quad \frac{q\xrightarrow{e_1}\surd\quad r\xrightarrow{e_2}r'}{(p\parallel r)+ (q\parallel r)\xrightarrow{\{e_1,e_2\}}r'}$$

      $$\frac{p\xrightarrow{e_1}p'\quad r\xrightarrow{e_2}r'}{(p+ q)\parallel r\xrightarrow{\{e_1,e_2\}}p'\between r'}
      \quad \frac{p\xrightarrow{e_1}p'\quad r\xrightarrow{e_2}r'}{(p\parallel r)+ (q\parallel r)\xrightarrow{\{e_1,e_2\}}p'\between r'}$$

      $$\frac{q\xrightarrow{e_1}q'\quad r\xrightarrow{e_2}r'}{(p+ q)\parallel r\xrightarrow{\{e_1,e_2\}}q'\between r'}
      \quad \frac{q\xrightarrow{e_1}q'\quad r\xrightarrow{e_2}r'}{(p\parallel r)+ (q\parallel r)\xrightarrow{\{e_1,e_2\}}q'\between r'}$$

      So, $(p+ q)\parallel r\sim_s (p\parallel r)+ (q\parallel r)$, as desired.
  \item \textbf{Axiom $P8$}. Let $p,q,r$ be $APTC$ processes, and $p\parallel(q+ r) = (p\parallel q) + (p\parallel r)$, it is sufficient to prove that $p\parallel(q+ r) \sim_s (p\parallel q) + (p\parallel r)$. By the transition rules for operators $+$ and $\parallel$ in Table \ref{STRForBATC} and \ref{TRForParallel}, we get

      $$\frac{p\xrightarrow{e_1}\surd\quad q\xrightarrow{e_2}\surd}{p\parallel (q+ r)\xrightarrow{\{e_1,e_2\}}\surd}
      \quad \frac{p\xrightarrow{e_1}\surd\quad q\xrightarrow{e_2}\surd}{(p\parallel q)+ (p\parallel r)\xrightarrow{\{e_1,e_2\}}\surd}$$

      $$\frac{p\xrightarrow{e_1}\surd\quad r\xrightarrow{e_2}\surd}{p\parallel (q+ r)\xrightarrow{\{e_1,e_2\}}\surd}
      \quad \frac{p\xrightarrow{e_1}\surd\quad r\xrightarrow{e_2}\surd}{(p\parallel q)+ (p\parallel r)\xrightarrow{\{e_1,e_2\}}\surd}$$

      $$\frac{p\xrightarrow{e_1}p'\quad q\xrightarrow{e_2}\surd}{p\parallel (q+ r)\xrightarrow{\{e_1,e_2\}}p'}
      \quad \frac{p\xrightarrow{e_1}p'\quad q\xrightarrow{e_2}\surd}{(p\parallel q)+ (p\parallel r)\xrightarrow{\{e_1,e_2\}}p'}$$

      $$\frac{p\xrightarrow{e_1}p'\quad r\xrightarrow{e_2}\surd}{p\parallel (q+ r)\xrightarrow{\{e_1,e_2\}}p'}
      \quad \frac{p\xrightarrow{e_1}p'\quad r\xrightarrow{e_2}\surd}{(p\parallel q)+ (p\parallel r)\xrightarrow{\{e_1,e_2\}}p'}$$

      $$\frac{p\xrightarrow{e_1}\surd\quad q\xrightarrow{e_2}q'}{p\parallel (q+ r)\xrightarrow{\{e_1,e_2\}}q'}
      \quad \frac{p\xrightarrow{e_1}\surd\quad q\xrightarrow{e_2}q'}{(p\parallel q)+ (p\parallel r)\xrightarrow{\{e_1,e_2\}}q'}$$

      $$\frac{p\xrightarrow{e_1}\surd\quad r\xrightarrow{e_2}r'}{p\parallel (q+ r)\xrightarrow{\{e_1,e_2\}}r'}
      \quad \frac{p\xrightarrow{e_1}\surd\quad r\xrightarrow{e_2}r'}{(p\parallel q)+ (p\parallel r)\xrightarrow{\{e_1,e_2\}}r'}$$

      $$\frac{p\xrightarrow{e_1}p'\quad q\xrightarrow{e_2}q'}{p\parallel (q+ r)\xrightarrow{\{e_1,e_2\}}p'\between q'}
      \quad \frac{p\xrightarrow{e_1}p'\quad q\xrightarrow{e_2}q'}{(p\parallel q)+ (p\parallel r)\xrightarrow{\{e_1,e_2\}}p'\between q'}$$

      $$\frac{p\xrightarrow{e_1}p'\quad r\xrightarrow{e_2}r'}{p\parallel (q+ r)\xrightarrow{\{e_1,e_2\}}p'\between r'}
      \quad \frac{p\xrightarrow{e_1}p'\quad r\xrightarrow{e_2}r'}{(p\parallel q)+ (p\parallel r)\xrightarrow{\{e_1,e_2\}}p'\between r'}$$

      So, $p\parallel(q+ r) \sim_s (p\parallel q) + (p\parallel r)$, as desired.
  \item \textbf{Axiom $C12$}. Let $q$ be an $APTC$ process, and $e_1\mid (e_2\cdot q)=\gamma(e_1,e_2)\cdot q$, it is sufficient to prove that $e_1\mid(e_2\cdot q)\sim_s \gamma(e_1,e_2)\cdot q$. By the transition rules for operator $\mid$ in Table \ref{TRForCommunication}, we get

      $$\frac{e_1\xrightarrow{e_1}\surd\quad e_2\cdot q\xrightarrow{e_2}q}{e_1\mid (e_2\cdot q)\xrightarrow{\gamma(e_1,e_2)}q}$$

      $$\frac{e_1\xrightarrow{e_1}\surd\quad e_2\xrightarrow{e_2}\surd}{\gamma(e_1,e_2)\cdot q\xrightarrow{\gamma(e_1,e_2)}q}$$

      So, $e_1\mid (e_2\cdot q)\sim_s \gamma(e_1,e_2)\cdot q$, as desired.
  \item \textbf{Axiom $C13$}. Let $p$ be an $APTC$ process, and $(e_1\cdot p)\mid e_2=\gamma(e_1,e_2)\cdot p$, it is sufficient to prove that $(e_1\cdot p)\mid e_2\sim_s \gamma(e_1,e_2)\cdot p$. By the transition rules for operator $\mid$ in Table \ref{TRForCommunication}, we get

      $$\frac{e_1\cdot p\xrightarrow{e_1}p\quad e_2\xrightarrow{e_2}\surd}{(e_1\cdot p)\mid e_2\xrightarrow{\gamma(e_1,e_2)}p}$$

      $$\frac{e_1\xrightarrow{e_1}\surd\quad e_2\xrightarrow{e_2}\surd}{\gamma(e_1,e_2)\cdot p\xrightarrow{\gamma(e_1,e_2)}p}$$

      So, $(e_1\cdot p)\mid e_2\sim_s \gamma(e_1,e_2)\cdot p$, as desired.
  \item \textbf{Axiom $C14$}. Let $p,q$ be $APTC$ processes, and $(e_1\cdot p)\mid (e_2\cdot q)=\gamma(e_1,e_2)\cdot (p\between q)$, it is sufficient to prove that $(e_1\cdot p)\mid(e_2\cdot q)\sim_s \gamma(e_1,e_2)\cdot (p\between q)$. By the transition rules for operator $\mid$ in Table \ref{TRForCommunication}, we get

      $$\frac{e_1\cdot p\xrightarrow{e_1}p\quad e_2\cdot q\xrightarrow{e_2}q}{(e_1\cdot p)\mid (e_2\cdot q)\xrightarrow{\gamma(e_1,e_2)}p\between q}$$

      $$\frac{e_1\xrightarrow{e_1}\surd\quad e_2\xrightarrow{e_2}\surd}{\gamma(e_1,e_2)\cdot (p\between q)\xrightarrow{\gamma(e_1,e_2)}p\between q}$$

      So, $(e_1\cdot p)\mid (e_2\cdot q)\sim_s \gamma(e_1,e_2)\cdot (p\between q)$, as desired.
  \item \textbf{Axiom $C15$}. Let $p,q,r$ be $APTC$ processes, and $(p+ q)\mid r = (p\mid r) + (q\mid r)$, it is sufficient to prove that $(p+ q)\mid r \sim_s (p\mid r) + (q\mid r)$. By the transition rules for operators $+$ and $\mid$ in Table \ref{STRForBATC} and \ref{TRForCommunication}, we get

      $$\frac{p\xrightarrow{e_1}\surd\quad r\xrightarrow{e_2}\surd}{(p+ q)\mid r\xrightarrow{\gamma(e_1,e_2)}\surd}
      \quad \frac{p\xrightarrow{e_1}\surd\quad r\xrightarrow{e_2}\surd}{(p\mid r)+ (q\mid r)\xrightarrow{\gamma(e_1,e_2)}\surd}$$

      $$\frac{q\xrightarrow{e_1}\surd\quad r\xrightarrow{e_2}\surd}{(p+ q)\mid r\xrightarrow{\gamma(e_1,e_2)}\surd}
      \quad \frac{q\xrightarrow{e_1}\surd\quad r\xrightarrow{e_2}\surd}{(p\mid r)+ (q\mid r)\xrightarrow{\gamma(e_1,e_2)}\surd}$$

      $$\frac{p\xrightarrow{e_1}p'\quad r\xrightarrow{e_2}\surd}{(p+ q)\mid r\xrightarrow{\gamma(e_1,e_2)}p'}
      \quad \frac{p\xrightarrow{e_1}p'\quad r\xrightarrow{e_2}\surd}{(p\mid r)+ (q\mid r)\xrightarrow{\gamma(e_1,e_2)}p'}$$

      $$\frac{q\xrightarrow{e_1}q'\quad r\xrightarrow{e_2}\surd}{(p+ q)\mid r\xrightarrow{\gamma(e_1,e_2)}q'}
      \quad \frac{q\xrightarrow{e_1}q'\quad r\xrightarrow{e_2}\surd}{(p\mid r)+ (q\mid r)\xrightarrow{\gamma(e_1,e_2)}q'}$$

      $$\frac{p\xrightarrow{e_1}\surd\quad r\xrightarrow{e_2}r'}{(p+ q)\mid r\xrightarrow{\gamma(e_1,e_2)}r'}
      \quad \frac{p\xrightarrow{e_1}\surd\quad r\xrightarrow{e_2}r'}{(p\mid r)+ (q\mid r)\xrightarrow{\gamma(e_1,e_2)}r'}$$

      $$\frac{q\xrightarrow{e_1}\surd\quad r\xrightarrow{e_2}r'}{(p+ q)\mid r\xrightarrow{\gamma(e_1,e_2)}r'}
      \quad \frac{q\xrightarrow{e_1}\surd\quad r\xrightarrow{e_2}r'}{(p\mid r)+ (q\mid r)\xrightarrow{\gamma(e_1,e_2)}r'}$$

      $$\frac{p\xrightarrow{e_1}p'\quad r\xrightarrow{e_2}r'}{(p+ q)\mid r\xrightarrow{\gamma(e_1,e_2)}p'\between r'}
      \quad \frac{p\xrightarrow{e_1}p'\quad r\xrightarrow{e_2}r'}{(p\mid r)+ (q\mid r)\xrightarrow{\gamma(e_1,e_2)}p'\between r'}$$

      $$\frac{q\xrightarrow{e_1}q'\quad r\xrightarrow{e_2}r'}{(p+ q)\mid r\xrightarrow{\gamma(e_1,e_2)}q'\between r'}
      \quad \frac{q\xrightarrow{e_1}q'\quad r\xrightarrow{e_2}r'}{(p\mid r)+ (q\mid r)\xrightarrow{\gamma(e_1,e_2)}q'\between r'}$$

      So, $(p+ q)\mid r\sim_s (p\mid r)+ (q\mid r)$, as desired.
  \item \textbf{Axiom $C16$}. Let $p,q,r$ be $APTC$ processes, and $p\mid(q+ r) = (p\mid q) + (p\mid r)$, it is sufficient to prove that $p\mid(q+ r) \sim_s (p\mid q) + (p\mid r)$. By the transition rules for operators $+$ and $\mid$ in Table \ref{STRForBATC} and \ref{TRForCommunication}, we get

      $$\frac{p\xrightarrow{e_1}\surd\quad q\xrightarrow{e_2}\surd}{p\mid (q+ r)\xrightarrow{\gamma(e_1,e_2)}\surd}
      \quad \frac{p\xrightarrow{e_1}\surd\quad q\xrightarrow{e_2}\surd}{(p\mid q)+ (p\mid r)\xrightarrow{\gamma(e_1,e_2)}\surd}$$

      $$\frac{p\xrightarrow{e_1}\surd\quad r\xrightarrow{e_2}\surd}{p\mid (q+ r)\xrightarrow{\gamma(e_1,e_2)}\surd}
      \quad \frac{p\xrightarrow{e_1}\surd\quad r\xrightarrow{e_2}\surd}{(p\mid q)+ (p\mid r)\xrightarrow{\gamma(e_1,e_2)}\surd}$$

      $$\frac{p\xrightarrow{e_1}p'\quad q\xrightarrow{e_2}\surd}{p\mid (q+ r)\xrightarrow{\gamma(e_1,e_2)}p'}
      \quad \frac{p\xrightarrow{e_1}p'\quad q\xrightarrow{e_2}\surd}{(p\mid q)+ (p\mid r)\xrightarrow{\gamma(e_1,e_2)}p'}$$

      $$\frac{p\xrightarrow{e_1}p'\quad r\xrightarrow{e_2}\surd}{p\mid (q+ r)\xrightarrow{\gamma(e_1,e_2)}p'}
      \quad \frac{p\xrightarrow{e_1}p'\quad r\xrightarrow{e_2}\surd}{(p\mid q)+ (p\mid r)\xrightarrow{\gamma(e_1,e_2)}p'}$$

      $$\frac{p\xrightarrow{e_1}\surd\quad q\xrightarrow{e_2}q'}{p\mid (q+ r)\xrightarrow{\gamma(e_1,e_2)}q'}
      \quad \frac{p\xrightarrow{e_1}\surd\quad q\xrightarrow{e_2}q'}{(p\mid q)+ (p\mid r)\xrightarrow{\gamma(e_1,e_2)}q'}$$

      $$\frac{p\xrightarrow{e_1}\surd\quad r\xrightarrow{e_2}r'}{p\mid (q+ r)\xrightarrow{\gamma(e_1,e_2)}r'}
      \quad \frac{p\xrightarrow{e_1}\surd\quad r\xrightarrow{e_2}r'}{(p\mid q)+ (p\mid r)\xrightarrow{\gamma(e_1,e_2)}r'}$$

      $$\frac{p\xrightarrow{e_1}p'\quad q\xrightarrow{e_2}q'}{p\mid (q+ r)\xrightarrow{\gamma(e_1,e_2)}p'\between q'}
      \quad \frac{p\xrightarrow{e_1}p'\quad q\xrightarrow{e_2}q'}{(p\mid q)+ (p\mid r)\xrightarrow{\gamma(e_1,e_2)}p'\between q'}$$

      $$\frac{p\xrightarrow{e_1}p'\quad r\xrightarrow{e_2}r'}{p\mid (q+ r)\xrightarrow{\gamma(e_1,e_2)}p'\between r'}
      \quad \frac{p\xrightarrow{e_1}p'\quad r\xrightarrow{e_2}r'}{(p\mid q)+ (p\mid r)\xrightarrow{\gamma(e_1,e_2)}p'\between r'}$$

      So, $p\mid(q+ r) \sim_s (p\mid q) + (p\mid r)$, as desired.
  \item \textbf{Axiom $CE21$}. Let $p,q$ be $APTC$ processes, and $\Theta(p+ q)=\Theta(p)\triangleleft q + \Theta(q)\triangleleft p$, it is sufficient to prove that $\Theta(p+ q) \sim_s \Theta(p)\triangleleft q + \Theta(q)\triangleleft p$. By the transition rules for operators $+$ in Table \ref{STRForBATC}, and $\Theta$ and $\triangleleft$ in Table \ref{TRForConflict}, we get

      $$\frac{p\xrightarrow{e_1}\surd (\sharp(e_1,e_2))}{\Theta(p+ q)\xrightarrow{e_1}\surd}
      \quad\frac{p\xrightarrow{e_1}\surd (\sharp(e_1,e_2))}{\Theta(p)\triangleleft q + \Theta(q)\triangleleft p\xrightarrow{e_1}\surd}$$

      $$\frac{q\xrightarrow{e_2}\surd (\sharp(e_1,e_2))}{\Theta(p+ q)\xrightarrow{e_2}\surd}
      \quad\frac{q\xrightarrow{e_2}\surd (\sharp(e_1,e_2))}{\Theta(p)\triangleleft q + \Theta(q)\triangleleft p\xrightarrow{e_2}\surd}$$

      $$\frac{p\xrightarrow{e_1}p' (\sharp(e_1,e_2))}{\Theta(p+ q)\xrightarrow{e_1}\Theta(p')}
      \quad\frac{p\xrightarrow{e_1}p' (\sharp(e_1,e_2))}{\Theta(p)\triangleleft q + \Theta(q)\triangleleft p\xrightarrow{e_1}\Theta(p')}$$

      $$\frac{q\xrightarrow{e_2}q' (\sharp(e_1,e_2))}{\Theta(p+ q)\xrightarrow{e_2}\Theta(q')}
      \quad\frac{q\xrightarrow{e_2}q' (\sharp(e_1,e_2))}{\Theta(p)\triangleleft q + \Theta(q)\triangleleft p\xrightarrow{e_2}\Theta(q')}$$

      So, $\Theta(p+ q) \sim_s \Theta(p)\triangleleft q + \Theta(q)\triangleleft p$, as desired.
  \item \textbf{Axiom $CE22$}. Let $p,q$ be $APTC$ processes, and $\Theta(p\cdot q)=\Theta(p)\cdot \Theta(q)$, it is sufficient to prove that $\Theta(p\cdot q) \sim_s \Theta(p)\cdot \Theta(q)$. By the transition rules for operators $\cdot$ in Table \ref{STRForBATC}, and $\Theta$ in Table \ref{TRForConflict}, we get

      $$\frac{p\xrightarrow{e_1}\surd}{\Theta(p\cdot q)\xrightarrow{e_1}\Theta(q)}
      \quad\frac{p\xrightarrow{e_1}\surd}{\Theta(p)\cdot\Theta(q)\xrightarrow{e_1}\Theta(q)}$$

      $$\frac{p\xrightarrow{e_1}p'}{\Theta(p\cdot q)\xrightarrow{e_1}\Theta(p'\cdot q)}
      \quad\frac{p\xrightarrow{e_1}p'}{\Theta(p)\cdot\Theta(q)\xrightarrow{e_1}\Theta(p')\cdot\Theta(q)}$$

      So, with the assumption $\Theta(p'\cdot q)=\Theta(p')\cdot\Theta(q)$, $\Theta(p\cdot q) \sim_s \Theta(p)\cdot \Theta(q)$, as desired.
  \item \textbf{Axiom $CE23$}. Let $p,q$ be $APTC$ processes, and $\Theta(p\parallel q)=((\Theta(p)\triangleleft q)\parallel q) + ((\Theta(q)\triangleleft p)\parallel p)$, it is sufficient to prove that $\Theta(p\parallel q) \sim_s ((\Theta(p)\triangleleft q)\parallel q) + ((\Theta(q)\triangleleft p)\parallel p)$. By the transition rules for operators $+$ in Table \ref{STRForBATC}, and $\Theta$ and $\triangleleft$ in Table \ref{TRForConflict}, and $\parallel$ in Table \ref{TRForParallel} we get

      $$\frac{p\xrightarrow{e_1}\surd \quad q\xrightarrow{e_2}\surd}{\Theta(p\parallel q)\xrightarrow{\{e_1,e_2\}}\surd}$$
      $$\frac{p\xrightarrow{e_1}\surd \quad q\xrightarrow{e_2}\surd}{((\Theta(p)\triangleleft q)\parallel q) + ((\Theta(q)\triangleleft p)\parallel p)\xrightarrow{\{e_1,e_2\}}\surd}$$

      $$\frac{p\xrightarrow{e_1}p' \quad q\xrightarrow{e_2}\surd}{\Theta(p\parallel q)\xrightarrow{\{e_1,e_2\}}\Theta(p')}$$
      $$\frac{p\xrightarrow{e_1}p' \quad q\xrightarrow{e_2}\surd}{((\Theta(p)\triangleleft q)\parallel q) + ((\Theta(q)\triangleleft p)\parallel p)\xrightarrow{\{e_1,e_2\}}\Theta(p')}$$

      $$\frac{p\xrightarrow{e_1}\surd \quad q\xrightarrow{e_2}q'}{\Theta(p\parallel q)\xrightarrow{\{e_1,e_2\}}\Theta(q')}$$
      $$\frac{p\xrightarrow{e_1}\surd \quad q\xrightarrow{e_2}q'}{((\Theta(p)\triangleleft q)\parallel q) + ((\Theta(q)\triangleleft p)\parallel p)\xrightarrow{\{e_1,e_2\}}\Theta(q')}$$

      $$\frac{p\xrightarrow{e_1}p' \quad q\xrightarrow{e_2}q'}{\Theta(p\parallel q)\xrightarrow{\{e_1,e_2\}}\Theta(p'\between q')}$$
      $$\frac{p\xrightarrow{e_1}p' \quad q\xrightarrow{e_2}q'}{((\Theta(p)\triangleleft q)\parallel q) + ((\Theta(q)\triangleleft p)\parallel p)\xrightarrow{\{e_1,e_2\}}((\Theta(p')\triangleleft q')\between q') + ((\Theta(q')\triangleleft p')\between p')}$$

      So, with the assumption $\Theta(p'\between q')=((\Theta(p')\triangleleft q')\between q') + ((\Theta(q')\triangleleft p')\between p')$, $\Theta(p\parallel q) \sim_s ((\Theta(p)\triangleleft q)\parallel q) + ((\Theta(q)\triangleleft p)\parallel p)$, as desired.
  \item \textbf{Axiom $CE24$}. Let $p,q$ be $APTC$ processes, and $\Theta(p\mid q)=((\Theta(p)\triangleleft q)\mid q) + ((\Theta(q)\triangleleft p)\mid p)$, it is sufficient to prove that $\Theta(p\mid q) \sim_s ((\Theta(p)\triangleleft q)\mid q) + ((\Theta(q)\triangleleft p)\mid p)$. By the transition rules for operators $+$ in Table \ref{STRForBATC}, and $\Theta$ and $\triangleleft$ in Table \ref{TRForConflict}, and $\mid$ in Table \ref{TRForCommunication} we get

      $$\frac{p\xrightarrow{e_1}\surd \quad q\xrightarrow{e_2}\surd}{\Theta(p\mid q)\xrightarrow{\gamma(e_1,e_2)}\surd}$$
      $$\frac{p\xrightarrow{e_1}\surd \quad q\xrightarrow{e_2}\surd}{((\Theta(p)\triangleleft q)\mid q) + ((\Theta(q)\triangleleft p)\mid p)\xrightarrow{\gamma(e_1,e_2)}\surd}$$

      $$\frac{p\xrightarrow{e_1}p' \quad q\xrightarrow{e_2}\surd}{\Theta(p\mid q)\xrightarrow{\gamma(e_1,e_2)}\Theta(p')}$$
      $$\frac{p\xrightarrow{e_1}p' \quad q\xrightarrow{e_2}\surd}{((\Theta(p)\triangleleft q)\mid q) + ((\Theta(q)\triangleleft p)\mid p)\xrightarrow{\gamma(e_1,e_2)}\Theta(p')}$$

      $$\frac{p\xrightarrow{e_1}\surd \quad q\xrightarrow{e_2}q'}{\Theta(p\mid q)\xrightarrow{\gamma(e_1,e_2)}\Theta(q')}$$
      $$\frac{p\xrightarrow{e_1}\surd \quad q\xrightarrow{e_2}q'}{((\Theta(p)\triangleleft q)\mid q) + ((\Theta(q)\triangleleft p)\mid p)\xrightarrow{\gamma(e_1,e_2)}\Theta(q')}$$

      $$\frac{p\xrightarrow{e_1}p' \quad q\xrightarrow{e_2}q'}{\Theta(p\mid q)\xrightarrow{\gamma(e_1,e_2)}\Theta(p'\between q')}$$
      $$\frac{p\xrightarrow{e_1}p' \quad q\xrightarrow{e_2}q'}{((\Theta(p)\triangleleft q)\mid q) + ((\Theta(q)\triangleleft p)\mid p)\xrightarrow{\gamma(e_1,e_2)}((\Theta(p')\triangleleft q')\between q') + ((\Theta(q')\triangleleft p')\between p')}$$

      So, with the assumption $\Theta(p'\between q')=((\Theta(p')\triangleleft q')\between q') + ((\Theta(q')\triangleleft p')\between p')$, $\Theta(p\mid q) \sim_s ((\Theta(p)\triangleleft q)\mid q) + ((\Theta(q)\triangleleft p)\mid p)$, as desired.
  \item \textbf{Axiom $U30$}. Let $p,q,r$ be $APTC$ processes, and $(p+ q)\triangleleft r = (p\triangleleft r) + (q\triangleleft r)$, it is sufficient to prove that $(p+ q)\triangleleft r \sim_s (p\triangleleft r) + (q\triangleleft r)$. By the transition rules for operators $+$ and $\triangleleft$ in Table \ref{STRForBATC} and \ref{TRForConflict}, we get

      $$\frac{p\xrightarrow{e_1}\surd}{(p+ q)\triangleleft r\xrightarrow{e_1}\surd}
      \quad \frac{p\xrightarrow{e_1}\surd}{(p\triangleleft r)+ (q\triangleleft r)\xrightarrow{e_1}\surd}$$

      $$\frac{q\xrightarrow{e_2}\surd}{(p+ q)\triangleleft r\xrightarrow{e_2}\surd}
      \quad \frac{q\xrightarrow{e_2}\surd}{(p\triangleleft r)+ (q\triangleleft r)\xrightarrow{e_2}\surd}$$

      $$\frac{p\xrightarrow{e_1}p'}{(p+ q)\triangleleft r\xrightarrow{e_1}p'\triangleleft r}
      \quad \frac{p\xrightarrow{e_1}p'}{(p\triangleleft r)+ (q\triangleleft r)\xrightarrow{e_1}p'\triangleleft r}$$

      $$\frac{q\xrightarrow{e_2}q'}{(p+ q)\triangleleft r\xrightarrow{e_2}q'\triangleleft r}
      \quad \frac{q\xrightarrow{e_2}q'}{(p\triangleleft r)+ (q\triangleleft r)\xrightarrow{e_2}q'\triangleleft r}$$

      Let us forget anything about $\tau$. So, $(p+ q)\triangleleft r\sim_s (p\triangleleft r)+ (q\triangleleft r)$, as desired.
  \item \textbf{Axiom $U31$}. Let $p,q,r$ be $APTC$ processes, and $(p\cdot q)\triangleleft r = (p\triangleleft r) \cdot (q\triangleleft r)$, it is sufficient to prove that $(p\cdot q)\triangleleft r \sim_s (p\triangleleft r) \cdot (q\triangleleft r)$. By the transition rules for operators $\cdot$ and $\triangleleft$ in Table \ref{STRForBATC} and \ref{TRForConflict}, we get

      $$\frac{p\xrightarrow{e_1}\surd}{(p\cdot q)\triangleleft r\xrightarrow{e_1}q\triangleleft r}
      \quad \frac{p\xrightarrow{e_1}\surd}{(p\triangleleft r)\cdot (q\triangleleft r)\xrightarrow{e_1}q\triangleleft r}$$

      $$\frac{p\xrightarrow{e_1}p'}{(p\cdot q)\triangleleft r\xrightarrow{e_1}(p'\cdot q)\triangleleft r}
      \quad \frac{p\xrightarrow{e_1}p'}{(p\triangleleft r)\cdot (q\triangleleft r)\xrightarrow{e_1}(p'\triangleleft r)\cdot (q\triangleleft r)}$$

      Let us forget anything about $\tau$. With the assumption $(p'\cdot q)\triangleleft r = (p'\triangleleft r) \cdot (q\triangleleft r)$, so, $(p\cdot q)\triangleleft r\sim_s (p\triangleleft r)\cdot (q\triangleleft r)$, as desired.
  \item \textbf{Axiom $U32$}. Let $p,q,r$ be $APTC$ processes, and $(p\parallel q)\triangleleft r = (p\triangleleft r) \parallel (q\triangleleft r)$, it is sufficient to prove that $(p\parallel q)\triangleleft r \sim_s (p\triangleleft r) \parallel (q\triangleleft r)$. By the transition rules for operators $\parallel$ and $\triangleleft$ in Table \ref{TRForParallel} and \ref{TRForConflict}, we get

      $$\frac{p\xrightarrow{e_1}\surd\quad q\xrightarrow{e_2}\surd}{(p\parallel q)\triangleleft r\xrightarrow{\{e_1,e_2\}}\surd}
      \quad \frac{p\xrightarrow{e_1}\surd\quad q\xrightarrow{e_2}\surd}{(p\triangleleft r)\parallel (q\triangleleft r)\xrightarrow{\{e_1,e_2\}}\surd}$$

      $$\frac{p\xrightarrow{e_1}p'\quad q\xrightarrow{e_2}\surd}{(p\parallel q)\triangleleft r\xrightarrow{\{e_1,e_2\}}p'\triangleleft r}
      \quad \frac{p\xrightarrow{e_1}p'\quad q\xrightarrow{e_2}\surd}{(p\triangleleft r)\parallel (q\triangleleft r)\xrightarrow{\{e_1,e_2\}}p'\triangleleft r}$$

      $$\frac{p\xrightarrow{e_1}\surd\quad q\xrightarrow{e_2}q'}{(p\parallel q)\triangleleft r\xrightarrow{\{e_1,e_2\}}q'\triangleleft r}
      \quad \frac{p\xrightarrow{e_1}\surd\quad q\xrightarrow{e_2}q'}{(p\triangleleft r)\parallel (q\triangleleft r)\xrightarrow{\{e_1,e_2\}}q'\triangleleft r}$$

      $$\frac{p\xrightarrow{e_1}p'\quad q\xrightarrow{e_2}q'}{(p\parallel q)\triangleleft r\xrightarrow{\{e_1,e_2\}}(p'\between q')\triangleleft r}
      \quad \frac{p\xrightarrow{e_1}p'\quad q\xrightarrow{e_2}q'}{(p\triangleleft r)\parallel (q\triangleleft r)\xrightarrow{\{e_1,e_2\}}(p'\triangleleft r)\between (q'\triangleleft r)}$$

      Let us forget anything about $\tau$. With the assumption $(p'\between q')\triangleleft r = (p'\triangleleft r) \between (q'\triangleleft r)$, so, $(p\parallel q)\triangleleft r\sim_s (p\triangleleft r)\parallel (q\triangleleft r)$, as desired.
  \item \textbf{Axiom $U33$}. Let $p,q,r$ be $APTC$ processes, and $(p\mid q)\triangleleft r = (p\triangleleft r) \mid (q\triangleleft r)$, it is sufficient to prove that $(p\mid q)\triangleleft r \sim_s (p\triangleleft r) \mid (q\triangleleft r)$. By the transition rules for operators $\mid$ and $\triangleleft$ in Table \ref{TRForCommunication} and \ref{TRForConflict}, we get

      $$\frac{p\xrightarrow{e_1}\surd\quad q\xrightarrow{e_2}\surd}{(p\mid q)\triangleleft r\xrightarrow{\gamma(e_1,e_2)}\surd}
      \quad \frac{p\xrightarrow{e_1}\surd\quad q\xrightarrow{e_2}\surd}{(p\triangleleft r)\mid (q\triangleleft r)\xrightarrow{\gamma(e_1,e_2)}\surd}$$

      $$\frac{p\xrightarrow{e_1}p'\quad q\xrightarrow{e_2}\surd}{(p\mid q)\triangleleft r\xrightarrow{\gamma(e_1,e_2)}p'\triangleleft r}
      \quad \frac{p\xrightarrow{e_1}p'\quad q\xrightarrow{e_2}\surd}{(p\triangleleft r)\mid (q\triangleleft r)\xrightarrow{\gamma(e_1,e_2)}p'\triangleleft r}$$

      $$\frac{p\xrightarrow{e_1}\surd\quad q\xrightarrow{e_2}q'}{(p\mid q)\triangleleft r\xrightarrow{\gamma(e_1,e_2)}q'\triangleleft r}
      \quad \frac{p\xrightarrow{e_1}\surd\quad q\xrightarrow{e_2}q'}{(p\triangleleft r)\mid (q\triangleleft r)\xrightarrow{\gamma(e_1,e_2)}q'\triangleleft r}$$

      $$\frac{p\xrightarrow{e_1}p'\quad q\xrightarrow{e_2}q'}{(p\mid q)\triangleleft r\xrightarrow{\gamma(e_1,e_2)}(p'\between q')\triangleleft r}
      \quad \frac{p\xrightarrow{e_1}p'\quad q\xrightarrow{e_2}q'}{(p\triangleleft r)\mid (q\triangleleft r)\xrightarrow{\gamma(e_1,e_2)}(p'\triangleleft r)\between (q'\triangleleft r)}$$

      Let us forget anything about $\tau$. With the assumption $(p'\between q')\triangleleft r = (p'\triangleleft r) \between (q'\triangleleft r)$, so, $(p\mid q)\triangleleft r\sim_s (p\triangleleft r)\mid (q\triangleleft r)$, as desired.
  \item \textbf{Axiom $U34$}. Let $p,q,r$ be $APTC$ processes, and $p\triangleleft (q+ r) = (p\triangleleft q)\triangleleft r$, it is sufficient to prove that $p\triangleleft (q+ r) \sim_s (p\triangleleft q)\triangleleft r$. By the transition rules for operators $+$ and $\triangleleft$ in Table \ref{STRForBATC} and \ref{TRForConflict}, we get

      $$\frac{p\xrightarrow{e_1}\surd}{p\triangleleft (q+ r)\xrightarrow{e_1}\surd}
      \quad \frac{p\xrightarrow{e_1}\surd}{(p\triangleleft q)\triangleleft r\xrightarrow{e_1}\surd}$$

      $$\frac{p\xrightarrow{e_1}p'}{p\triangleleft (q+ r)\xrightarrow{e_1}p'\triangleleft (q+ r)}
      \quad \frac{p\xrightarrow{e_1}p'}{(p\triangleleft q)\triangleleft r\xrightarrow{e_1}(p'\triangleleft q)\triangleleft r}$$

      Let us forget anything about $\tau$. With the assumption $p'\triangleleft (q+ r) = (p'\triangleleft q)\triangleleft r$, so, $p\triangleleft (q+ r) \sim_s (p\triangleleft q)\triangleleft r$, as desired.
  \item \textbf{Axiom $U35$}. Let $p,q,r$ be $APTC$ processes, and $p\triangleleft (q\cdot r) = (p\triangleleft q)\triangleleft r$, it is sufficient to prove that $p\triangleleft (q\cdot r) \sim_s (p\triangleleft q)\triangleleft r$. By the transition rules for operators $\cdot$ and $\triangleleft$ in Table \ref{STRForBATC} and \ref{TRForConflict}, we get

      $$\frac{p\xrightarrow{e_1}\surd}{p\triangleleft (q\cdot r)\xrightarrow{e_1}\surd}
      \quad \frac{p\xrightarrow{e_1}\surd}{(p\triangleleft q)\triangleleft r\xrightarrow{e_1}\surd}$$

      $$\frac{p\xrightarrow{e_1}p'}{p\triangleleft (q\cdot r)\xrightarrow{e_1}p'\triangleleft (q\cdot r)}
      \quad \frac{p\xrightarrow{e_1}p'}{(p\triangleleft q)\triangleleft r\xrightarrow{e_1}(p'\triangleleft q)\triangleleft r}$$

      Let us forget anything about $\tau$. With the assumption $p'\triangleleft (q\cdot r) = (p'\triangleleft q)\triangleleft r$, so, $p\triangleleft (q\cdot r) \sim_s (p\triangleleft q)\triangleleft r$, as desired.
  \item \textbf{Axiom $U36$}. Let $p,q,r$ be $APTC$ processes, and $p\triangleleft (q\parallel r) = (p\triangleleft q)\triangleleft r$, it is sufficient to prove that $p\triangleleft (q\parallel r) \sim_s (p\triangleleft q)\triangleleft r$. By the transition rules for operators $\parallel$ and $\triangleleft$ in Table \ref{TRForParallel} and \ref{TRForConflict}, we get

      $$\frac{p\xrightarrow{e_1}\surd}{p\triangleleft (q\parallel r)\xrightarrow{e_1}\surd}
      \quad \frac{p\xrightarrow{e_1}\surd}{(p\triangleleft q)\triangleleft r\xrightarrow{e_1}\surd}$$

      $$\frac{p\xrightarrow{e_1}p'}{p\triangleleft (q\parallel r)\xrightarrow{e_1}p'\triangleleft (q\parallel r)}
      \quad \frac{p\xrightarrow{e_1}p'}{(p\triangleleft q)\triangleleft r\xrightarrow{e_1}(p'\triangleleft q)\triangleleft r}$$

      Let us forget anything about $\tau$. With the assumption $p'\triangleleft (q\parallel r) = (p'\triangleleft q)\triangleleft r$, so, $p\triangleleft (q\parallel r) \sim_s (p\triangleleft q)\triangleleft r$, as desired.
  \item \textbf{Axiom $U37$}. Let $p,q,r$ be $APTC$ processes, and $p\triangleleft (q\mid r) = (p\triangleleft q)\triangleleft r$, it is sufficient to prove that $p\triangleleft (q\mid r) \sim_s (p\triangleleft q)\triangleleft r$. By the transition rules for operators $\mid$ and $\triangleleft$ in Table \ref{TRForCommunication} and \ref{TRForConflict}, we get

      $$\frac{p\xrightarrow{e_1}\surd}{p\triangleleft (q\mid r)\xrightarrow{e_1}\surd}
      \quad \frac{p\xrightarrow{e_1}\surd}{(p\triangleleft q)\triangleleft r\xrightarrow{e_1}\surd}$$

      $$\frac{p\xrightarrow{e_1}p'}{p\triangleleft (q\mid r)\xrightarrow{e_1}p'\triangleleft (q\mid r)}
      \quad \frac{p\xrightarrow{e_1}p'}{(p\triangleleft q)\triangleleft r\xrightarrow{e_1}(p'\triangleleft q)\triangleleft r}$$

      Let us forget anything about $\tau$. With the assumption $p'\triangleleft (q\mid r) = (p'\triangleleft q)\triangleleft r$, so, $p\triangleleft (q\mid r) \sim_s (p\triangleleft q)\triangleleft r$, as desired.
\end{itemize}
\end{proof}

\begin{theorem}[Completeness of parallelism modulo step bisimulation equivalence]\label{CPSBE}
Let $p$ and $q$ be closed $APTC$ terms, if $p\sim_{s} q$ then $p=q$.
\end{theorem}

\begin{proof}
Firstly, by the elimination theorem of $APTC$ (see Theorem \ref{ETParallelism}), we know that for each closed $APTC$ term $p$, there exists a closed basic $APTC$ term $p'$, such that $APTC\vdash p=p'$, so, we only need to consider closed basic $APTC$ terms.

The basic terms (see Definition \ref{BTAPTC}) modulo associativity and commutativity (AC) of conflict $+$ (defined by axioms $A1$ and $A2$ in Table \ref{AxiomsForBATC}) and associativity and commutativity (AC) of parallel $\parallel$ (defined by axioms $P2$ and $P3$ in Table \ref{AxiomsForParallelism}), and these equivalences is denoted by $=_{AC}$. Then, each equivalence class $s$ modulo AC of $+$ and $\parallel$ has the following normal form

$$s_1+\cdots+ s_k$$

with each $s_i$ either an atomic event or of the form

$$t_1\cdot\cdots\cdot t_m$$

with each $t_j$ either an atomic event or of the form

$$u_1\parallel\cdots\parallel u_l$$

with each $u_l$ an atomic event, and each $s_i$ is called the summand of $s$.

Now, we prove that for normal forms $n$ and $n'$, if $n\sim_{s} n'$ then $n=_{AC}n'$. It is sufficient to induct on the sizes of $n$ and $n'$.

\begin{itemize}
  \item Consider a summand $e$ of $n$. Then $n\xrightarrow{e}\surd$, so $n\sim_s n'$ implies $n'\xrightarrow{e}\surd$, meaning that $n'$ also contains the summand $e$.
  \item Consider a summand $t_1\cdot t_2$ of $n$,
  \begin{itemize}
    \item if $t_1\equiv e'$, then $n\xrightarrow{e'}t_2$, so $n\sim_s n'$ implies $n'\xrightarrow{e'}t_2'$ with $t_2\sim_s t_2'$, meaning that $n'$ contains a summand $e'\cdot t_2'$. Since $t_2$ and $t_2'$ are normal forms and have sizes smaller than $n$ and $n'$, by the induction hypotheses if $t_2\sim_s t_2'$ then $t_2=_{AC} t_2'$;
    \item if $t_1\equiv e_1\parallel\cdots\parallel e_l$, then $n\xrightarrow{\{e_1,\cdots,e_l\}}t_2$, so $n\sim_s n'$ implies $n'\xrightarrow{\{e_1,\cdots,e_l\}}t_2'$ with $t_2\sim_s t_2'$, meaning that $n'$ contains a summand $(e_1\parallel\cdots\parallel e_l)\cdot t_2'$. Since $t_2$ and $t_2'$ are normal forms and have sizes smaller than $n$ and $n'$, by the induction hypotheses if $t_2\sim_s t_2'$ then $t_2=_{AC} t_2'$.
  \end{itemize}
\end{itemize}

So, we get $n=_{AC} n'$.

Finally, let $s$ and $t$ be basic $APTC$ terms, and $s\sim_s t$, there are normal forms $n$ and $n'$, such that $s=n$ and $t=n'$. The soundness theorem of parallelism modulo step bisimulation equivalence (see Theorem \ref{SPSBE}) yields $s\sim_s n$ and $t\sim_s n'$, so $n\sim_s s\sim_s t\sim_s n'$. Since if $n\sim_s n'$ then $n=_{AC}n'$, $s=n=_{AC}n'=t$, as desired.
\end{proof}

\begin{theorem}[Soundness of parallelism modulo pomset bisimulation equivalence]\label{SPPBE}
Let $x$ and $y$ be $APTC$ terms. If $APTC\vdash x=y$, then $x\sim_{p} y$.
\end{theorem}

\begin{proof}
Since pomset bisimulation $\sim_{p}$ is both an equivalent and a congruent relation with respect to the operators $\between$, $\parallel$, $\mid$, $\Theta$ and $\triangleleft$, we only need to check if each axiom in Table \ref{AxiomsForParallelism} is sound modulo pomset bisimulation equivalence.

From the definition of pomset bisimulation (see Definition \ref{PSB}), we know that pomset bisimulation is defined by pomset transitions, which are labeled by pomsets. In a pomset transition, the events in the pomset are either within causality relations (defined by $\cdot$) or in concurrency (implicitly defined by $\cdot$ and $+$, and explicitly defined by $\between$), of course, they are pairwise consistent (without conflicts). In Theorem \ref{SPSBE}, we have already proven the case that all events are pairwise concurrent, so, we only need to prove the case of events in causality. Without loss of generality, we take a pomset of $P=\{e_1,e_2:e_1\cdot e_2\}$. Then the pomset transition labeled by the above $P$ is just composed of one single event transition labeled by $e_1$ succeeded by another single event transition labeled by $e_2$, that is, $\xrightarrow{P}=\xrightarrow{e_1}\xrightarrow{e_2}$.

Similarly to the proof of soundness of parallelism modulo step bisimulation equivalence (see Theorem \ref{SPSBE}), we can prove that each axiom in Table \ref{AxiomsForParallelism} is sound modulo pomset bisimulation equivalence, we omit them.
\end{proof}

\begin{theorem}[Completeness of parallelism modulo pomset bisimulation equivalence]\label{CPPBE}
Let $p$ and $q$ be closed $APTC$ terms, if $p\sim_{p} q$ then $p=q$.
\end{theorem}

\begin{proof}
Firstly, by the elimination theorem of $APTC$ (see Theorem \ref{ETParallelism}), we know that for each closed $APTC$ term $p$, there exists a closed basic $APTC$ term $p'$, such that $APTC\vdash p=p'$, so, we only need to consider closed basic $APTC$ terms.

The basic terms (see Definition \ref{BTAPTC}) modulo associativity and commutativity (AC) of conflict $+$ (defined by axioms $A1$ and $A2$ in Table \ref{AxiomsForBATC}) and associativity and commutativity (AC) of parallel $\parallel$ (defined by axioms $P2$ and $P3$ in Table \ref{AxiomsForParallelism}), and these equivalences is denoted by $=_{AC}$. Then, each equivalence class $s$ modulo AC of $+$ and $\parallel$ has the following normal form

$$s_1+\cdots+ s_k$$

with each $s_i$ either an atomic event or of the form

$$t_1\cdot\cdots\cdot t_m$$

with each $t_j$ either an atomic event or of the form

$$u_1\parallel\cdots\parallel u_l$$

with each $u_l$ an atomic event, and each $s_i$ is called the summand of $s$.

Now, we prove that for normal forms $n$ and $n'$, if $n\sim_{p} n'$ then $n=_{AC}n'$. It is sufficient to induct on the sizes of $n$ and $n'$.

\begin{itemize}
  \item Consider a summand $e$ of $n$. Then $n\xrightarrow{e}\surd$, so $n\sim_p n'$ implies $n'\xrightarrow{e}\surd$, meaning that $n'$ also contains the summand $e$.
  \item Consider a summand $t_1\cdot t_2$ of $n$,
  \begin{itemize}
    \item if $t_1\equiv e'$, then $n\xrightarrow{e'}t_2$, so $n\sim_p n'$ implies $n'\xrightarrow{e'}t_2'$ with $t_2\sim_p t_2'$, meaning that $n'$ contains a summand $e'\cdot t_2'$. Since $t_2$ and $t_2'$ are normal forms and have sizes smaller than $n$ and $n'$, by the induction hypotheses if $t_2\sim_p t_2'$ then $t_2=_{AC} t_2'$;
    \item if $t_1\equiv e_1\parallel\cdots\parallel e_l$, then $n\xrightarrow{\{e_1,\cdots,e_l\}}t_2$, so $n\sim_p n'$ implies $n'\xrightarrow{\{e_1,\cdots,e_l\}}t_2'$ with $t_2\sim_p t_2'$, meaning that $n'$ contains a summand $(e_1\parallel\cdots\parallel e_l)\cdot t_2'$. Since $t_2$ and $t_2'$ are normal forms and have sizes smaller than $n$ and $n'$, by the induction hypotheses if $t_2\sim_p t_2'$ then $t_2=_{AC} t_2'$.
  \end{itemize}
\end{itemize}

So, we get $n=_{AC} n'$.

Finally, let $s$ and $t$ be basic $APTC$ terms, and $s\sim_p t$, there are normal forms $n$ and $n'$, such that $s=n$ and $t=n'$. The soundness theorem of parallelism modulo pomset bisimulation equivalence (see Theorem \ref{SPPBE}) yields $s\sim_p n$ and $t\sim_p n'$, so $n\sim_p s\sim_p t\sim_p n'$. Since if $n\sim_p n'$ then $n=_{AC}n'$, $s=n=_{AC}n'=t$, as desired.
\end{proof}

\begin{theorem}[Soundness of parallelism modulo hp-bisimulation equivalence]\label{SPHPBE}
Let $x$ and $y$ be $APTC$ terms. If $APTC\vdash x=y$, then $x\sim_{hp} y$.
\end{theorem}

\begin{proof}
Since hp-bisimulation $\sim_{hp}$ is both an equivalent and a congruent relation with respect to the operators $\between$, $\parallel$, $\mid$, $\Theta$ and $\triangleleft$, we only need to check if each axiom in Table \ref{AxiomsForParallelism} is sound modulo hp-bisimulation equivalence.

From the definition of hp-bisimulation (see Definition \ref{HHPB}), we know that hp-bisimulation is defined on the posetal product $(C_1,f,C_2),f:C_1\rightarrow C_2\textrm{ isomorphism}$. Two process terms $s$ related to $C_1$ and $t$ related to $C_2$, and $f:C_1\rightarrow C_2\textrm{ isomorphism}$. Initially, $(C_1,f,C_2)=(\emptyset,\emptyset,\emptyset)$, and $(\emptyset,\emptyset,\emptyset)\in\sim_{hp}$. When $s\xrightarrow{e}s'$ ($C_1\xrightarrow{e}C_1'$), there will be $t\xrightarrow{e}t'$ ($C_2\xrightarrow{e}C_2'$), and we define $f'=f[e\mapsto e]$. Then, if $(C_1,f,C_2)\in\sim_{hp}$, then $(C_1',f',C_2')\in\sim_{hp}$.

Similarly to the proof of soundness of parallelism modulo pomset bisimulation equivalence (see Theorem \ref{SPPBE}), we can prove that each axiom in Table \ref{AxiomsForParallelism} is sound modulo hp-bisimulation equivalence, we just need additionally to check the above conditions on hp-bisimulation, we omit them.
\end{proof}

\begin{theorem}[Completeness of parallelism modulo hp-bisimulation equivalence]\label{CPHPBE}
Let $p$ and $q$ be closed $APTC$ terms, if $p\sim_{hp} q$ then $p=q$.
\end{theorem}

\begin{proof}
Firstly, by the elimination theorem of $APTC$ (see Theorem \ref{ETParallelism}), we know that for each closed $APTC$ term $p$, there exists a closed basic $APTC$ term $p'$, such that $APTC\vdash p=p'$, so, we only need to consider closed basic $APTC$ terms.

The basic terms (see Definition \ref{BTAPTC}) modulo associativity and commutativity (AC) of conflict $+$ (defined by axioms $A1$ and $A2$ in Table \ref{AxiomsForBATC}) and associativity and commutativity (AC) of parallel $\parallel$ (defined by axioms $P2$ and $P3$ in Table \ref{AxiomsForParallelism}), and these equivalences is denoted by $=_{AC}$. Then, each equivalence class $s$ modulo AC of $+$ and $\parallel$ has the following normal form

$$s_1+\cdots+ s_k$$

with each $s_i$ either an atomic event or of the form

$$t_1\cdot\cdots\cdot t_m$$

with each $t_j$ either an atomic event or of the form

$$u_1\parallel\cdots\parallel u_l$$

with each $u_l$ an atomic event, and each $s_i$ is called the summand of $s$.

Now, we prove that for normal forms $n$ and $n'$, if $n\sim_{hp} n'$ then $n=_{AC}n'$. It is sufficient to induct on the sizes of $n$ and $n'$.

\begin{itemize}
  \item Consider a summand $e$ of $n$. Then $n\xrightarrow{e}\surd$, so $n\sim_{hp} n'$ implies $n'\xrightarrow{e}\surd$, meaning that $n'$ also contains the summand $e$.
  \item Consider a summand $t_1\cdot t_2$ of $n$,
  \begin{itemize}
    \item if $t_1\equiv e'$, then $n\xrightarrow{e'}t_2$, so $n\sim_{hp} n'$ implies $n'\xrightarrow{e'}t_2'$ with $t_2\sim_{hp} t_2'$, meaning that $n'$ contains a summand $e'\cdot t_2'$. Since $t_2$ and $t_2'$ are normal forms and have sizes smaller than $n$ and $n'$, by the induction hypotheses if $t_2\sim_{hp} t_2'$ then $t_2=_{AC} t_2'$;
    \item if $t_1\equiv e_1\parallel\cdots\parallel e_l$, then $n\xrightarrow{\{e_1,\cdots,e_l\}}t_2$, so $n\sim_{hp} n'$ implies $n'\xrightarrow{\{e_1,\cdots,e_l\}}t_2'$ with $t_2\sim_{hp} t_2'$, meaning that $n'$ contains a summand $(e_1\parallel\cdots\parallel e_l)\cdot t_2'$. Since $t_2$ and $t_2'$ are normal forms and have sizes smaller than $n$ and $n'$, by the induction hypotheses if $t_2\sim_{hp} t_2'$ then $t_2=_{AC} t_2'$.
  \end{itemize}
\end{itemize}

So, we get $n=_{AC} n'$.

Finally, let $s$ and $t$ be basic $APTC$ terms, and $s\sim_{hp} t$, there are normal forms $n$ and $n'$, such that $s=n$ and $t=n'$. The soundness theorem of parallelism modulo hp-bisimulation equivalence (see Theorem \ref{SPPBE}) yields $s\sim_{hp} n$ and $t\sim_{hp} n'$, so $n\sim_{hp} s\sim_{hp} t\sim_{hp} n'$. Since if $n\sim_{hp} n'$ then $n=_{AC}n'$, $s=n=_{AC}n'=t$, as desired.
\end{proof}

\begin{proposition}[About Soundness and Completeness of parallelism modulo hhp-bisimulation equivalence]\label{SCPHHPBE}
1. Let $x$ and $y$ be $APTC$ terms. If $APTC\vdash x=y\nRightarrow x\sim_{hhp} y$;

2. If $p$ and $q$ are closed $APTC$ terms, then $p\sim_{hhp} q\nRightarrow p=q$.
\end{proposition}

\begin{proof}
Imperfectly, the algebraic laws in Table \ref{AxiomsForParallelism} are not sound and complete modulo hhp-bisimulation equivalence, we just need enumerate several key axioms in Table \ref{AxiomsForParallelism} are not sound modulo hhp-bisimulation equivalence.

From the definition of hhp-bisimulation (see Definition \ref{HHPB}), we know that an hhp-bisimulation is a downward closed hp-bisimulation. That is, for any posetal products $(C_1,f,C_2)$ and $(C_1',f,C_2')$, if $(C_1,f,C_2)\subseteq (C_1',f',C_2')$ pointwise and $(C_1',f',C_2')\in \sim_{hhp}$, then $(C_1,f,C_2)\in \sim_{hhp}$.

Now, let us consider the axioms $P7$ and $P8$ (the right and left distributivity of $\parallel$ to $+$). Let $s_1=(a+ b)\parallel c$, $t_1=(a\parallel c)+ (b\parallel c)$, and $s_2=a\parallel (b+ c)$, $t_2=(a\parallel b)+ (a\parallel c)$. We know that $s_1\sim_{hp} t_1$ and $s_2\sim_{hp} t_2$ (by Theorem \ref{SPHPBE}), we prove that $s_1\nsim_{hhp} t_1$ and $s_2\nsim_{hhp} t_2$. Let $(C(s_1),f_1,C(t_1))$ and $(C(s_2),f_2,C(t_2))$ are the corresponding posetal products.

\begin{itemize}
  \item \textbf{Axiom $P7$}. $s_1\xrightarrow{\{a,c\}}\surd(s_1')$ ($C(s_1)\xrightarrow{\{a,c\}}C(s_1')$), then $t_1\xrightarrow{\{a,c\}}\surd(t_1')$ ($C(t_1)\xrightarrow{\{a,c\}}C(t_1')$), we define $f_1'=f_1[a\mapsto a, c\mapsto c]$, obviously, $(C(s_1),f_1,C(t_1))\in \sim_{hp}$ and $(C(s_1'),f_1',C(t_1'))\in \sim_{hp}$. But, $(C(s_1),f_1,C(t_1))\in \sim_{hhp}$ and $(C(s_1'),f_1',C(t_1'))\in \nsim_{hhp}$, just because they are not downward closed. Let $(C(s_1''),f_1'',C(t_1''))$, and $f_1''=f_1[c\mapsto c]$, $s_1\xrightarrow{c}s_1''$ ($C(s_1)\xrightarrow{c}C(s_1'')$), $t_1\xrightarrow{c}t_1''$ ($C(t_1)\xrightarrow{c}C(t_1'')$), it is easy to see that $(C(s_1''),f_1'',C(t_1''))\subseteq (C(s_1'),f_1',C(t_1'))$ pointwise, while $(C(s_1''),f_1'',C(t_1''))\notin \sim_{hp}$, because $s_1''$ and $C(s_1'')$ exist, but $t_1''$ and $C(t_1'')$ do not exist.
  \item \textbf{Axiom $P8$}. $s_2\xrightarrow{\{a,c\}}\surd(s_2')$ ($C(s_2)\xrightarrow{\{a,c\}}C(s_2')$), then $t_2\xrightarrow{\{a,c\}}\surd(t_2')$ ($C(t_2)\xrightarrow{\{a,c\}}C(t_2')$), we define $f_2'=f_2[a\mapsto a, c\mapsto c]$, obviously, $(C(s_2),f_2,C(t_2))\in \sim_{hp}$ and $(C(s_2'),f_2',C(t_2'))\in \sim_{hp}$. But, $(C(s_2),f_2,C(t_2))\in \sim_{hhp}$ and $(C(s_2'),f_2',C(t_2'))\in \nsim_{hhp}$, just because they are not downward closed. Let $(C(s_2''),f_2'',C(t_2''))$, and $f_2''=f_2[a\mapsto a]$, $s_2\xrightarrow{a}s_2''$ ($C(s_2)\xrightarrow{a}C(s_2'')$), $t_2\xrightarrow{a}t_2''$ ($C(t_2)\xrightarrow{a}C(t_2'')$), it is easy to see that $(C(s_2''),f_2'',C(t_2''))\subseteq (C(s_2'),f_2',C(t_2'))$ pointwise, while $(C(s_2''),f_2'',C(t_2''))\notin \sim_{hp}$, because $s_2''$ and $C(s_2'')$ exist, but $t_2''$ and $C(t_2'')$ do not exist.
\end{itemize}

The unsoundness of parallelism modulo hhp-bisimulation equivalence makes the completeness of parallelism modulo hhp-bisimulation equivalence meaningless. Further more, unsoundness of $P7$ and $P8$ lead to the elimination theorem of $APTC$ (see Theorem \ref{ETParallelism}) failing, so, the non-existence of normal form also makes the completeness impossible.
\end{proof}

Actually, a finite sound and complete axiomatization for parallel composition $\parallel$ modulo hhp-bisimulation equivalence does not exist, about the axiomatization for hhp-bisimilarity, please refer to section \ref{ahhpb} for details.

In following sections, we will discuss nothing about hhp-bisimulation, because the following encapsulation, recursion and abstraction are based on the algebraic laws in this section.

Finally, let us explain the so-called absorption law \cite{AL} in a straightforward way. Process term $P=a\parallel(b+ c)+ a\parallel b+ b\parallel (a+ c)$, and process term $Q=a\parallel(b+ c)+ b\parallel(a+ c)$, equated by the absorption law.

Modulo $\sim_s$, $\sim_p$, and $\sim_{hp}$, by use of the axioms of $BATC$ and $APTC$, we have the following deductions:

\begin{eqnarray}
P&=&a\parallel(b+ c)+ a\parallel b+ b\parallel (a+ c)\nonumber\\
&\overset{\text{P8}}{=}&a\parallel b+ a\parallel c+ a\parallel b+ b\parallel a+ b\parallel c\nonumber\\
&\overset{\text{P2}}{=}&a\parallel b+ a\parallel c+ a\parallel b+ a\parallel b+ b\parallel c\nonumber\\
&\overset{\text{A3}}{=}&a\parallel b+ a\parallel c+ b\parallel c\nonumber
\end{eqnarray}

\begin{eqnarray}
Q&=&a\parallel(b+ c)+ b\parallel (a+ c)\nonumber\\
&\overset{\text{P8}}{=}&a\parallel b+ a\parallel c+ b\parallel a+ b\parallel c\nonumber\\
&\overset{\text{P2}}{=}&a\parallel b+ a\parallel c+ a\parallel b+ b\parallel c\nonumber\\
&\overset{\text{A3}}{=}&a\parallel b+ a\parallel c+ b\parallel c\nonumber
\end{eqnarray}

It means that $P=Q$ modulo $\sim_s$, $\sim_p$, and $\sim_{hp}$, that is, $P\sim_s Q$, $P\sim_p Q$ and $P\sim_{hp} Q$. But, $P\neq Q$ modulo $\sim_{hhp}$, which means that $P\nsim_{hhp} Q$.

\subsection{Encapsulation}

The mismatch of two communicating events in different parallel branches can cause deadlock, so the deadlocks in the concurrent processes should be eliminated. Like $ACP$ \cite{ACP}, we also introduce the unary encapsulation operator $\partial_H$ for set $H$ of atomic events, which renames all atomic events in $H$ into $\delta$. The whole algebra including parallelism for true concurrency in the above subsections, deadlock $\delta$ and encapsulation operator $\partial_H$, is called Algebra for Parallelism in True Concurrency, abbreviated $APTC$.

The transition rules of encapsulation operator $\partial_H$ are shown in Table \ref{TRForEncapsulation}.

\begin{center}
    \begin{table}
        $$\frac{x\xrightarrow{e}\surd}{\partial_H(x)\xrightarrow{e}\surd}\quad (e\notin H)\quad\quad\frac{x\xrightarrow{e}x'}{\partial_H(x)\xrightarrow{e}\partial_H(x')}\quad(e\notin H)$$
        \caption{Transition rules of encapsulation operator $\partial_H$}
        \label{TRForEncapsulation}
    \end{table}
\end{center}

Based on the transition rules for encapsulation operator $\partial_H$ in Table \ref{TRForEncapsulation}, we design the axioms as Table \ref{AxiomsForEncapsulation} shows.

\begin{center}
    \begin{table}
        \begin{tabular}{@{}ll@{}}
            \hline No. &Axiom\\
            $D1$ & $e\notin H\quad\partial_H(e) = e$\\
            $D2$ & $e\in H\quad \partial_H(e) = \delta$\\
            $D3$ & $\partial_H(\delta) = \delta$\\
            $D4$ & $\partial_H(x+ y) = \partial_H(x)+\partial_H(y)$\\
            $D5$ & $\partial_H(x\cdot y) = \partial_H(x)\cdot\partial_H(y)$\\
            $D6$ & $\partial_H(x\parallel y) = \partial_H(x)\parallel\partial_H(y)$\\
        \end{tabular}
        \caption{Axioms of encapsulation operator}
        \label{AxiomsForEncapsulation}
    \end{table}
\end{center}

The axioms $D1-D3$ are the defining laws for the encapsulation operator $\partial_H$, $D1$ leaves atomic events outside $H$ unchanged, $D2$ renames atomic events in $H$ into $\delta$, and $D3$ says that it leaves $\delta$ unchanged. $D4-D6$ say that in term $\partial_H(t)$, all transitions of $t$ labeled with atomic events in $H$ are blocked.

\begin{theorem}[Conservativity of $APTC$ with respect to the algebra for parallelism]
$APTC$ is a conservative extension of the algebra for parallelism.
\end{theorem}

\begin{proof}
It follows from the following two facts (see Theorem \ref{TCE}).

\begin{enumerate}
  \item The transition rules of the algebra for parallelism in the above subsections are all source-dependent;
  \item The sources of the transition rules for the encapsulation operator contain an occurrence of $\partial_H$.
\end{enumerate}

So, $APTC$ is a conservative extension of the algebra for parallelism, as desired.
\end{proof}

\begin{theorem}[Congruence theorem of encapsulation operator $\partial_H$]
Truly concurrent bisimulation equivalences $\sim_{p}$, $\sim_s$, $\sim_{hp}$ and $\sim_{hhp}$ are all congruences with respect to encapsulation operator $\partial_H$.
\end{theorem}

\begin{proof}
(1) Case pomset bisimulation equivalence $\sim_p$.

Let $x$ and $y$ be $APTC$ processes, and $x\sim_{p} y$, it is sufficient to prove that $\partial_H(x)\sim_{p} \partial_H(y)$.

By the definition of pomset bisimulation $\sim_p$ (Definition \ref{PSB}), $x\sim_p y$ means that

$$x\xrightarrow{X} x' \quad y\xrightarrow{Y} y'$$

with $X\subseteq x$, $Y\subseteq y$, $X\sim Y$ and $x'\sim_p y'$.

By the pomset transition rules for encapsulation operator $\partial_H$ in Table \ref{TRForEncapsulation}, we can get

$$\partial_H(x)\xrightarrow{X} \surd (X\nsubseteq H) \quad \partial_H(y)\xrightarrow{Y} \surd (Y\nsubseteq H)$$

with $X\subseteq x$, $Y\subseteq y$, and $X\sim Y$, so, we get $\partial_H(x)\sim_p \partial_H(y)$, as desired.

Or, we can get

$$\partial_H(x)\xrightarrow{X} \partial_H(x') (X\nsubseteq H) \quad \partial_H(y)\xrightarrow{Y} \partial_H(y') (Y\nsubseteq H)$$

with $X\subseteq x$, $Y\subseteq y$, $X\sim Y$, $x'\sim_p y'$ and the assumption $\partial_H(x')\sim_p\partial_H(y')$, so, we get $\partial_H(x)\sim_p \partial_H(y)$, as desired.

(2) The cases of step bisimulation $\sim_s$, hp-bisimulation $\sim_{hp}$ and hhp-bisimulation $\sim_{hhp}$ can be proven similarly, we omit them.
\end{proof}

\begin{theorem}[Elimination theorem of $APTC$]\label{ETEncapsulation}
Let $p$ be a closed $APTC$ term including the encapsulation operator $\partial_H$. Then there is a basic $APTC$ term $q$ such that $APTC\vdash p=q$.
\end{theorem}

\begin{proof}
(1) Firstly, suppose that the following ordering on the signature of $APTC$ is defined: $\parallel > \cdot > +$ and the symbol $\parallel$ is given the lexicographical status for the first argument, then for each rewrite rule $p\rightarrow q$ in Table \ref{TRSForEncapsulation} relation $p>_{lpo} q$ can easily be proved. We obtain that the term rewrite system shown in Table \ref{TRSForEncapsulation} is strongly normalizing, for it has finitely many rewriting rules, and $>$ is a well-founded ordering on the signature of $APTC$, and if $s>_{lpo} t$, for each rewriting rule $s\rightarrow t$ is in Table \ref{TRSForEncapsulation} (see Theorem \ref{SN}).

\begin{center}
    \begin{table}
        \begin{tabular}{@{}ll@{}}
            \hline No. &Rewriting Rule\\
            $RD1$ & $e\notin H\quad\partial_H(e) \rightarrow e$\\
            $RD2$ & $e\in H\quad \partial_H(e) \rightarrow \delta$\\
            $RD3$ & $\partial_H(\delta) \rightarrow \delta$\\
            $RD4$ & $\partial_H(x+ y) \rightarrow \partial_H(x)+\partial_H(y)$\\
            $RD5$ & $\partial_H(x\cdot y) \rightarrow \partial_H(x)\cdot\partial_H(y)$\\
            $RD6$ & $\partial_H(x\parallel y) \rightarrow \partial_H(x)\parallel\partial_H(y)$\\
        \end{tabular}
        \caption{Term rewrite system of encapsulation operator $\partial_H$}
        \label{TRSForEncapsulation}
    \end{table}
\end{center}

(2) Then we prove that the normal forms of closed $APTC$ terms including encapsulation operator $\partial_H$ are basic $APTC$ terms.

Suppose that $p$ is a normal form of some closed $APTC$ term and suppose that $p$ is not a basic $APTC$ term. Let $p'$ denote the smallest sub-term of $p$ which is not a basic $APTC$ term. It implies that each sub-term of $p'$ is a basic $APTC$ term. Then we prove that $p$ is not a term in normal form. It is sufficient to induct on the structure of $p'$, following from Theorem \ref{ETAPTC}, we only prove the new case $p'\equiv \partial_H(p_1)$:

\begin{itemize}
  \item Case $p_1\equiv e$. The transition rules $RD1$ or $RD2$ can be applied, so $p$ is not a normal form;
  \item Case $p_1\equiv \delta$. The transition rules $RD3$ can be applied, so $p$ is not a normal form;
  \item Case $p_1\equiv p_1'+ p_1''$. The transition rules $RD4$ can be applied, so $p$ is not a normal form;
  \item Case $p_1\equiv p_1'\cdot p_1''$. The transition rules $RD5$ can be applied, so $p$ is not a normal form;
  \item Case $p_1\equiv p_1'\parallel p_1''$. The transition rules $RD6$ can be applied, so $p$ is not a normal form.
\end{itemize}
\end{proof}

\begin{theorem}[Soundness of $APTC$ modulo step bisimulation equivalence]\label{SAPTCSBE}
Let $x$ and $y$ be $APTC$ terms including encapsulation operator $\partial_H$. If $APTC\vdash x=y$, then $x\sim_{s} y$.
\end{theorem}

\begin{proof}
Since step bisimulation $\sim_s$ is both an equivalent and a congruent relation with respect to the operator $\partial_H$, we only need to check if each axiom in Table \ref{AxiomsForEncapsulation} is sound modulo step bisimulation equivalence.

Though transition rules in Table \ref{TRForEncapsulation} are defined in the flavor of single event, they can be modified into a step (a set of events within which each event is pairwise concurrent), we omit them. If we treat a single event as a step containing just one event, the proof of this soundness theorem does not exist any problem, so we use this way and still use the transition rules in Table \ref{TRForEncapsulation}.

We omit the defining axioms, including axioms $D1-D3$, and we only prove the soundness of the non-trivial axioms, including axioms $D4-D6$.

\begin{itemize}
  \item \textbf{Axiom $D4$}. Let $p,q$ be $APTC$ processes, and $\partial_H(p+ q)=\partial_H(p)+\partial_H(q)$, it is sufficient to prove that $\partial_H(p+ q)\sim_s \partial_H(p)+\partial_H(q)$. By the transition rules for operator $+$ in Table \ref{STRForBATC} and $\partial_H$ in Table \ref{TRForEncapsulation}, we get

      $$\frac{p\xrightarrow{e_1}\surd\quad(e_1\notin H)}{\partial_H(p+ q)\xrightarrow{e_1}\surd}
      \quad\frac{p\xrightarrow{e_1}\surd\quad(e_1\notin H)}{\partial_H(p)+ \partial_H(q)\xrightarrow{e_1}\surd}$$

      $$\frac{q\xrightarrow{e_2}\surd\quad(e_2\notin H)}{\partial_H(p+ q)\xrightarrow{e_2}\surd}
      \quad\frac{q\xrightarrow{e_2}\surd\quad(e_2\notin H)}{\partial_H(p)+ \partial_H(q)\xrightarrow{e_2}\surd}$$

      $$\frac{p\xrightarrow{e_1}p'\quad(e_1\notin H)}{\partial_H(p+ q)\xrightarrow{e_1}\partial_H(p')}
      \quad\frac{p\xrightarrow{e_1}p'\quad(e_1\notin H)}{\partial_H(p)+ \partial_H(q)\xrightarrow{e_1}\partial_H(p')}$$

      $$\frac{q\xrightarrow{e_2}q'\quad(e_2\notin H)}{\partial_H(p+ q)\xrightarrow{e_2}\partial_H(q')}
      \quad\frac{q\xrightarrow{e_2}q'\quad(e_2\notin H)}{\partial_H(p)+ \partial_H(q)\xrightarrow{e_2}\partial_H(q')}$$

      So, $\partial_H(p+ q)\sim_s \partial_H(p)+\partial_H(q)$, as desired.
  \item \textbf{Axiom $D5$}. Let $p,q$ be $APTC$ processes, and $\partial_H(p\cdot q)=\partial_H(p)\cdot\partial_H(q)$, it is sufficient to prove that $\partial_H(p\cdot q)\sim_s \partial_H(p)\cdot\partial_H(q)$. By the transition rules for operator $\cdot$ in Table \ref{STRForBATC} and $\partial_H$ in Table \ref{TRForEncapsulation}, we get

      $$\frac{p\xrightarrow{e_1}\surd\quad(e_1\notin H)}{\partial_H(p\cdot q)\xrightarrow{e_1}\partial_H(q)}
      \quad\frac{p\xrightarrow{e_1}\surd\quad(e_1\notin H)}{\partial_H(p)\cdot \partial_H(q)\xrightarrow{e_1}\partial_H(q)}$$

      $$\frac{p\xrightarrow{e_1}p'\quad(e_1\notin H)}{\partial_H(p\cdot q)\xrightarrow{e_1}\partial_H(p'\cdot q)}
      \quad\frac{p\xrightarrow{e_1}p'\quad(e_1\notin H)}{\partial_H(p)\cdot \partial_H(q)\xrightarrow{e_1}\partial_H(p')\cdot\partial_H(q)}$$

      So, with the assumption $\partial_H(p'\cdot q)=\partial_H(p')\cdot\partial_H(q)$, $\partial_H(p\cdot q)\sim_s\partial_H(p)\cdot\partial_H(q)$, as desired.
  \item \textbf{Axiom $D6$}. Let $p,q$ be $APTC$ processes, and $\partial_H(p\parallel q)=\partial_H(p)\parallel\partial_H(q)$, it is sufficient to prove that $\partial_H(p\parallel q)\sim_s \partial_H(p)\parallel\partial_H(q)$. By the transition rules for operator $\parallel$ in Table \ref{TRForParallel} and $\partial_H$ in Table \ref{TRForEncapsulation}, we get

      $$\frac{p\xrightarrow{e_1}\surd\quad q\xrightarrow{e_2}\surd\quad(e_1,e_2\notin H)}{\partial_H(p\parallel q)\xrightarrow{\{e_1,e_2\}}\surd}
      \quad\frac{p\xrightarrow{e_1}\surd\quad q\xrightarrow{e_2}\surd\quad(e_1,e_2\notin H)}{\partial_H(p)\parallel \partial_H(q)\xrightarrow{\{e_1,e_2\}}\surd}$$

      $$\frac{p\xrightarrow{e_1}p'\quad q\xrightarrow{e_2}\surd\quad(e_1,e_2\notin H)}{\partial_H(p\parallel q)\xrightarrow{\{e_1,e_2\}}\partial_H(p')}
      \quad\frac{p\xrightarrow{e_1}p'\quad q\xrightarrow{e_2}\surd\quad(e_1,e_2\notin H)}{\partial_H(p)\parallel \partial_H(q)\xrightarrow{\{e_1,e_2\}}\partial_H(p')}$$

      $$\frac{p\xrightarrow{e_1}\surd\quad q\xrightarrow{e_2}q'\quad(e_1,e_2\notin H)}{\partial_H(p\parallel q)\xrightarrow{\{e_1,e_2\}}\partial_H(q')}
      \quad\frac{p\xrightarrow{e_1}\surd\quad q\xrightarrow{e_2}q'\quad(e_1,e_2\notin H)}{\partial_H(p)\parallel \partial_H(q)\xrightarrow{\{e_1,e_2\}}\partial_H(q')}$$

      $$\frac{p\xrightarrow{e_1}p'\quad q\xrightarrow{e_2}q'\quad(e_1,e_2\notin H)}{\partial_H(p\parallel q)\xrightarrow{\{e_1,e_2\}}\partial_H(p'\between q')}
      \quad\frac{p\xrightarrow{e_1}p'\quad q\xrightarrow{e_2}q'\quad(e_1,e_2\notin H)}{\partial_H(p)\parallel \partial_H(q)\xrightarrow{\{e_1,e_2\}}\partial_H(p')\between\partial_H(q')}$$

      So, with the assumption $\partial_H(p'\between q')=\partial_H(p')\between\partial_H(q')$, $\partial_H(p\parallel q)\sim_s \partial_H(p)\parallel\partial_H(q)$, as desired.
\end{itemize}
\end{proof}

\begin{theorem}[Completeness of $APTC$ modulo step bisimulation equivalence]\label{CAPTCSBE}
Let $p$ and $q$ be closed $APTC$ terms including encapsulation operator $\partial_H$, if $p\sim_{s} q$ then $p=q$.
\end{theorem}

\begin{proof}
Firstly, by the elimination theorem of $APTC$ (see Theorem \ref{ETEncapsulation}), we know that the normal form of $APTC$ does not contain $\partial_H$, and for each closed $APTC$ term $p$, there exists a closed basic $APTC$ term $p'$, such that $APTC\vdash p=p'$, so, we only need to consider closed basic $APTC$ terms.

Similarly to Theorem \ref{CPSBE}, we can prove that for normal forms $n$ and $n'$, if $n\sim_{s} n'$ then $n=_{AC}n'$.

Finally, let $s$ and $t$ be basic $APTC$ terms, and $s\sim_s t$, there are normal forms $n$ and $n'$, such that $s=n$ and $t=n'$. The soundness theorem of $APTC$ modulo step bisimulation equivalence (see Theorem \ref{SAPTCSBE}) yields $s\sim_s n$ and $t\sim_s n'$, so $n\sim_s s\sim_s t\sim_s n'$. Since if $n\sim_s n'$ then $n=_{AC}n'$, $s=n=_{AC}n'=t$, as desired.
\end{proof}

\begin{theorem}[Soundness of $APTC$ modulo pomset bisimulation equivalence]\label{SAPTCPBE}
Let $x$ and $y$ be $APTC$ terms including encapsulation operator $\partial_H$. If $APTC\vdash x=y$, then $x\sim_{p} y$.
\end{theorem}

\begin{proof}
Since pomset bisimulation $\sim_{p}$ is both an equivalent and a congruent relation with respect to the operator $\partial_H$, we only need to check if each axiom in Table \ref{AxiomsForEncapsulation} is sound modulo pomset bisimulation equivalence.

From the definition of pomset bisimulation (see Definition \ref{PSB}), we know that pomset bisimulation is defined by pomset transitions, which are labeled by pomsets. In a pomset transition, the events in the pomset are either within causality relations (defined by $\cdot$) or in concurrency (implicitly defined by $\cdot$ and $+$, and explicitly defined by $\between$), of course, they are pairwise consistent (without conflicts). In Theorem \ref{SAPTCSBE}, we have already proven the case that all events are pairwise concurrent, so, we only need to prove the case of events in causality. Without loss of generality, we take a pomset of $P=\{e_1,e_2:e_1\cdot e_2\}$. Then the pomset transition labeled by the above $P$ is just composed of one single event transition labeled by $e_1$ succeeded by another single event transition labeled by $e_2$, that is, $\xrightarrow{P}=\xrightarrow{e_1}\xrightarrow{e_2}$.

Similarly to the proof of soundness of $APTC$ modulo step bisimulation equivalence (see Theorem \ref{SAPTCSBE}), we can prove that each axiom in Table \ref{AxiomsForEncapsulation} is sound modulo pomset bisimulation equivalence, we omit them.
\end{proof}

\begin{theorem}[Completeness of $APTC$ modulo pomset bisimulation equivalence]\label{CAPTCPBE}
Let $p$ and $q$ be closed $APTC$ terms including encapsulation operator $\partial_H$, if $p\sim_{p} q$ then $p=q$.
\end{theorem}

\begin{proof}
Firstly, by the elimination theorem of $APTC$ (see Theorem \ref{ETEncapsulation}), we know that the normal form of $APTC$ does not contain $\partial_H$, and for each closed $APTC$ term $p$, there exists a closed basic $APTC$ term $p'$, such that $APTC\vdash p=p'$, so, we only need to consider closed basic $APTC$ terms.

Similarly to Theorem \ref{CAPTCSBE}, we can prove that for normal forms $n$ and $n'$, if $n\sim_{p} n'$ then $n=_{AC}n'$.

Finally, let $s$ and $t$ be basic $APTC$ terms, and $s\sim_p t$, there are normal forms $n$ and $n'$, such that $s=n$ and $t=n'$. The soundness theorem of $APTC$ modulo pomset bisimulation equivalence (see Theorem \ref{SAPTCPBE}) yields $s\sim_p n$ and $t\sim_p n'$, so $n\sim_p s\sim_p t\sim_p n'$. Since if $n\sim_p n'$ then $n=_{AC}n'$, $s=n=_{AC}n'=t$, as desired.
\end{proof}

\begin{theorem}[Soundness of $APTC$ modulo hp-bisimulation equivalence]\label{SAPTCHPBE}
Let $x$ and $y$ be $APTC$ terms including encapsulation operator $\partial_H$. If $APTC\vdash x=y$, then $x\sim_{hp} y$.
\end{theorem}

\begin{proof}
Since hp-bisimulation $\sim_{hp}$ is both an equivalent and a congruent relation with respect to the operator $\partial_H$, we only need to check if each axiom in Table \ref{AxiomsForEncapsulation} is sound modulo hp-bisimulation equivalence.

From the definition of hp-bisimulation (see Definition \ref{HHPB}), we know that hp-bisimulation is defined on the posetal product $(C_1,f,C_2),f:C_1\rightarrow C_2\textrm{ isomorphism}$. Two process terms $s$ related to $C_1$ and $t$ related to $C_2$, and $f:C_1\rightarrow C_2\textrm{ isomorphism}$. Initially, $(C_1,f,C_2)=(\emptyset,\emptyset,\emptyset)$, and $(\emptyset,\emptyset,\emptyset)\in\sim_{hp}$. When $s\xrightarrow{e}s'$ ($C_1\xrightarrow{e}C_1'$), there will be $t\xrightarrow{e}t'$ ($C_2\xrightarrow{e}C_2'$), and we define $f'=f[e\mapsto e]$. Then, if $(C_1,f,C_2)\in\sim_{hp}$, then $(C_1',f',C_2')\in\sim_{hp}$.

Similarly to the proof of soundness of $APTC$ modulo pomset bisimulation equivalence (see Theorem \ref{SAPTCPBE}), we can prove that each axiom in Table \ref{AxiomsForEncapsulation} is sound modulo hp-bisimulation equivalence, we just need additionally to check the above conditions on hp-bisimulation, we omit them.
\end{proof}

\begin{theorem}[Completeness of $APTC$ modulo hp-bisimulation equivalence]\label{CAPTCHPBE}
Let $p$ and $q$ be closed $APTC$ terms including encapsulation operator $\partial_H$, if $p\sim_{hp} q$ then $p=q$.
\end{theorem}

\begin{proof}
Firstly, by the elimination theorem of $APTC$ (see Theorem \ref{ETEncapsulation}), we know that the normal form of $APTC$ does not contain $\partial_H$, and for each closed $APTC$ term $p$, there exists a closed basic $APTC$ term $p'$, such that $APTC\vdash p=p'$, so, we only need to consider closed basic $APTC$ terms.

Similarly to Theorem \ref{CAPTCPBE}, we can prove that for normal forms $n$ and $n'$, if $n\sim_{hp} n'$ then $n=_{AC}n'$.

Finally, let $s$ and $t$ be basic $APTC$ terms, and $s\sim_{hp} t$, there are normal forms $n$ and $n'$, such that $s=n$ and $t=n'$. The soundness theorem of $APTC$ modulo hp-bisimulation equivalence (see Theorem \ref{SAPTCHPBE}) yields $s\sim_{hp} n$ and $t\sim_{hp} n'$, so $n\sim_{hp} s\sim_{hp} t\sim_{hp} n'$. Since if $n\sim_{hp} n'$ then $n=_{AC}n'$, $s=n=_{AC}n'=t$, as desired.
\end{proof}

\section{Recursion}\label{rec}

In this section, we introduce recursion to capture infinite processes based on $APTC$. Since in $APTC$, there are three basic operators $\cdot$, $+$ and $\parallel$, the recursion must be adapted this situation to include $\parallel$.

In the following, $E,F,G$ are recursion specifications, $X,Y,Z$ are recursive variables.

\subsection{Guarded Recursive Specifications}

\begin{definition}[Recursive specification]
A recursive specification is a finite set of recursive equations

$$X_1=t_1(X_1,\cdots,X_n)$$
$$\cdots$$
$$X_n=t_n(X_1,\cdots,X_n)$$

where the left-hand sides of $X_i$ are called recursion variables, and the right-hand sides $t_i(X_1,\cdots,X_n)$ are process terms in $APTC$ with possible occurrences of the recursion variables $X_1,\cdots,X_n$.
\end{definition}

\begin{definition}[Solution]
Processes $p_1,\cdots,p_n$ are a solution for a recursive specification $\{X_i=t_i(X_1,\cdots,X_n)|i\in\{1,\cdots,n\}\}$ (with respect to truly concurrent bisimulation equivalences $\sim_s$($\sim_p$, $\sim_{hp}$)) if $p_i\sim_s (\sim_p, \sim_{hp})t_i(p_1,\cdots,p_n)$ for $i\in\{1,\cdots,n\}$.
\end{definition}

\begin{definition}[Guarded recursive specification]
A recursive specification

$$X_1=t_1(X_1,\cdots,X_n)$$
$$...$$
$$X_n=t_n(X_1,\cdots,X_n)$$

is guarded if the right-hand sides of its recursive equations can be adapted to the form by applications of the axioms in $APTC$ and replacing recursion variables by the right-hand sides of their recursive equations,

$$(a_{11}\parallel\cdots\parallel a_{1i_1})\cdot s_1(X_1,\cdots,X_n)+\cdots+(a_{k1}\parallel\cdots\parallel a_{ki_k})\cdot s_k(X_1,\cdots,X_n)+(b_{11}\parallel\cdots\parallel b_{1j_1})+\cdots+(b_{1j_1}\parallel\cdots\parallel b_{lj_l})$$

where $a_{11},\cdots,a_{1i_1},a_{k1},\cdots,a_{ki_k},b_{11},\cdots,b_{1j_1},b_{1j_1},\cdots,b_{lj_l}\in \mathbb{E}$, and the sum above is allowed to be empty, in which case it represents the deadlock $\delta$.
\end{definition}

\begin{definition}[Linear recursive specification]\label{LRS}
A recursive specification is linear if its recursive equations are of the form

$$(a_{11}\parallel\cdots\parallel a_{1i_1})X_1+\cdots+(a_{k1}\parallel\cdots\parallel a_{ki_k})X_k+(b_{11}\parallel\cdots\parallel b_{1j_1})+\cdots+(b_{1j_1}\parallel\cdots\parallel b_{lj_l})$$

where $a_{11},\cdots,a_{1i_1},a_{k1},\cdots,a_{ki_k},b_{11},\cdots,b_{1j_1},b_{1j_1},\cdots,b_{lj_l}\in \mathbb{E}$, and the sum above is allowed to be empty, in which case it represents the deadlock $\delta$.
\end{definition}

For a guarded recursive specifications $E$ with the form

$$X_1=t_1(X_1,\cdots,X_n)$$
$$\cdots$$
$$X_n=t_n(X_1,\cdots,X_n)$$

the behavior of the solution $\langle X_i|E\rangle$ for the recursion variable $X_i$ in $E$, where $i\in\{1,\cdots,n\}$, is exactly the behavior of their right-hand sides $t_i(X_1,\cdots,X_n)$, which is captured by the two transition rules in Table \ref{TRForGR}.

\begin{center}
    \begin{table}
        $$\frac{t_i(\langle X_1|E\rangle,\cdots,\langle X_n|E\rangle)\xrightarrow{\{e_1,\cdots,e_k\}}\surd}{\langle X_i|E\rangle\xrightarrow{\{e_1,\cdots,e_k\}}\surd}$$
        $$\frac{t_i(\langle X_1|E\rangle,\cdots,\langle X_n|E\rangle)\xrightarrow{\{e_1,\cdots,e_k\}} y}{\langle X_i|E\rangle\xrightarrow{\{e_1,\cdots,e_k\}} y}$$
        \caption{Transition rules of guarded recursion}
        \label{TRForGR}
    \end{table}
\end{center}

\begin{theorem}[Conservitivity of $APTC$ with guarded recursion]
$APTC$ with guarded recursion is a conservative extension of $APTC$.
\end{theorem}

\begin{proof}
Since the transition rules of $APTC$ are source-dependent, and the transition rules for guarded recursion in Table \ref{TRForGR} contain only a fresh constant in their source, so the transition rules of $APTC$ with guarded recursion are a conservative extension of those of $APTC$.
\end{proof}

\begin{theorem}[Congruence theorem of $APTC$ with guarded recursion]
Truly concurrent bisimulation equivalences $\sim_{p}$, $\sim_s$ and $\sim_{hp}$ are all congruences with respect to $APTC$ with guarded recursion.
\end{theorem}

\begin{proof}
It follows the following two facts:
\begin{enumerate}
  \item in a guarded recursive specification, right-hand sides of its recursive equations can be adapted to the form by applications of the axioms in $APTC$ and replacing recursion variables by the right-hand sides of their recursive equations;
  \item truly concurrent bisimulation equivalences $\sim_{p}$, $\sim_s$ and $\sim_{hp}$ are all congruences with respect to all operators of $APTC$.
\end{enumerate}
\end{proof}

\subsection{Recursive Definition and Specification Principles}

The $RDP$ (Recursive Definition Principle) and the $RSP$ (Recursive Specification Principle) are shown in Table \ref{RDPRSP}.

\begin{center}
\begin{table}
  \begin{tabular}{@{}ll@{}}
\hline No. &Axiom\\
  $RDP$ & $\langle X_i|E\rangle = t_i(\langle X_1|E\rangle,\cdots,\langle X_n|E\rangle)\quad (i\in\{1,\cdots,n\})$\\
  $RSP$ & if $y_i=t_i(y_1,\cdots,y_n)$ for $i\in\{1,\cdots,n\}$, then $y_i=\langle X_i|E\rangle \quad(i\in\{1,\cdots,n\})$\\
\end{tabular}
\caption{Recursive definition and specification principle}
\label{RDPRSP}
\end{table}
\end{center}

$RDP$ follows immediately from the two transition rules for guarded recursion, which express that $\langle X_i|E\rangle$ and $t_i(\langle X_1|E\rangle,\cdots,\langle X_n|E\rangle)$ have the same initial transitions for $i\in\{1,\cdots,n\}$. $RSP$ follows from the fact that guarded recursive specifications have only one solution.

\begin{theorem}[Elimination theorem of $APTC$ with linear recursion]\label{ETRecursion}
Each process term in $APTC$ with linear recursion is equal to a process term $\langle X_1|E\rangle$ with $E$ a linear recursive specification.
\end{theorem}

\begin{proof}
By applying structural induction with respect to term size, each process term $t_1$ in $APTC$ with linear recursion generates a process can be expressed in the form of equations

$$t_i=(a_{i11}\parallel\cdots\parallel a_{i1i_1})t_{i1}+\cdots+(a_{ik_i1}\parallel\cdots\parallel a_{ik_ii_k})t_{ik_i}+(b_{i11}\parallel\cdots\parallel b_{i1i_1})+\cdots+(b_{il_i1}\parallel\cdots\parallel b_{il_ii_l})$$

for $i\in\{1,\cdots,n\}$. Let the linear recursive specification $E$ consist of the recursive equations

$$X_i=(a_{i11}\parallel\cdots\parallel a_{i1i_1})X_{i1}+\cdots+(a_{ik_i1}\parallel\cdots\parallel a_{ik_ii_k})X_{ik_i}+(b_{i11}\parallel\cdots\parallel b_{i1i_1})+\cdots+(b_{il_i1}\parallel\cdots\parallel b_{il_ii_l})$$

for $i\in\{1,\cdots,n\}$. Replacing $X_i$ by $t_i$ for $i\in\{1,\cdots,n\}$ is a solution for $E$, $RSP$ yields $t_1=\langle X_1|E\rangle$.
\end{proof}

\begin{theorem}[Soundness of $APTC$ with guarded recursion]\label{SAPTCR}
Let $x$ and $y$ be $APTC$ with guarded recursion terms. If $APTC\textrm{ with guarded recursion}\vdash x=y$, then
\begin{enumerate}
  \item $x\sim_{s} y$;
  \item $x\sim_{p} y$;
  \item $x\sim_{hp} y$.
\end{enumerate}
\end{theorem}

\begin{proof}
(1) Soundness of $APTC$ with guarded recursion with respect to step bisimulation $\sim_s$.

Since step bisimulation $\sim_s$ is both an equivalent and a congruent relation with respect to $APTC$ with guarded recursion, we only need to check if each axiom in Table \ref{RDPRSP} is sound modulo step bisimulation equivalence.

Though transition rules in Table \ref{TRForGR} are defined in the flavor of single event, they can be modified into a step (a set of events within which each event is pairwise concurrent), we omit them. If we treat a single event as a step containing just one event, the proof of this soundness theorem does not exist any problem, so we use this way and still use the transition rules in Table \ref{TRForGR}.

\begin{itemize}
  \item \textbf{$RDP$}. $\langle X_i|E\rangle = t_i(\langle X_1|E\rangle,\cdots,\langle X_n|E\rangle)\quad (i\in\{1,\cdots,n\})$, it is sufficient to prove that $\langle X_i|E\rangle \sim_s t_i(\langle X_1|E,\cdots,X_n|E\rangle)\quad (i\in\{1,\cdots,n\})$. By the transition rules for guarded recursion in Table \ref{TRForGR}, we get

      $$\frac{t_i(\langle X_1|E\rangle,\cdots,\langle X_n|E\rangle)\xrightarrow{\{e_1,\cdots,e_k\}}\surd}{\langle X_i|E\rangle\xrightarrow{\{e_1,\cdots,e_k\}}\surd}$$

      $$\frac{t_i(\langle X_1|E\rangle,\cdots,\langle X_n|E\rangle)\xrightarrow{\{e_1,\cdots,e_k\}} y}{\langle X_i|E\rangle\xrightarrow{\{e_1,\cdots,e_k\}} y}$$

      So, $\langle X_i|E\rangle \sim_s t_i(\langle X_1|E\rangle,\cdots,\langle X_n|E\rangle)\quad (i\in\{1,\cdots,n\})$, as desired.
  \item \textbf{$RSP$}. if $y_i=t_i(y_1,\cdots,y_n)$ for $i\in\{1,\cdots,n\}$, then $y_i=\langle X_i|E\rangle \quad(i\in\{1,\cdots,n\})$, it is sufficient to prove that if $y_i=t_i(y_1,\cdots,y_n)$ for $i\in\{1,\cdots,n\}$, then $y_i\sim_s \langle X_i|E\rangle \quad(i\in\{1,\cdots,n\})$. By the transition rules for guarded recursion in Table \ref{TRForGR}, we get

      $$\frac{t_i(\langle X_1|E\rangle,\cdots,\langle X_n|E\rangle)\xrightarrow{\{e_1,\cdots,e_k\}}\surd}{\langle X_i|E\rangle\xrightarrow{\{e_1,\cdots,e_k\}}\surd}$$

      $$\frac{t_i(\langle X_1|E\rangle,\cdots,\langle X_n|E\rangle)\xrightarrow{\{e_1,\cdots,e_k\}}\surd}{y_i\xrightarrow{\{e_1,\cdots,e_k\}}\surd}$$

      $$\frac{t_i(\langle X_1|E\rangle,\cdots,\langle X_n|E\rangle)\xrightarrow{\{e_1,\cdots,e_k\}} y}{\langle X_i|E\rangle\xrightarrow{\{e_1,\cdots,e_k\}} y}$$

      $$\frac{t_i(\langle X_1|E\rangle,\cdots,\langle X_n|E\rangle)\xrightarrow{\{e_1,\cdots,e_k\}} y}{y_i\xrightarrow{\{e_1,\cdots,e_k\}} y}$$

      So, if $y_i=t_i(y_1,\cdots,y_n)$ for $i\in\{1,\cdots,n\}$, then $y_i\sim_s \langle X_i|E\rangle \quad(i\in\{1,\cdots,n\})$, as desired.
\end{itemize}

(2) Soundness of $APTC$ with guarded recursion with respect to pomset bisimulation $\sim_p$.

Since pomset bisimulation $\sim_{p}$ is both an equivalent and a congruent relation with respect to the guarded recursion, we only need to check if each axiom in Table \ref{RDPRSP} is sound modulo pomset bisimulation equivalence.

From the definition of pomset bisimulation (see Definition \ref{PSB}), we know that pomset bisimulation is defined by pomset transitions, which are labeled by pomsets. In a pomset transition, the events in the pomset are either within causality relations (defined by $\cdot$) or in concurrency (implicitly defined by $\cdot$ and $+$, and explicitly defined by $\between$), of course, they are pairwise consistent (without conflicts). In (1), we have already proven the case that all events are pairwise concurrent, so, we only need to prove the case of events in causality. Without loss of generality, we take a pomset of $P=\{e_1,e_2:e_1\cdot e_2\}$. Then the pomset transition labeled by the above $P$ is just composed of one single event transition labeled by $e_1$ succeeded by another single event transition labeled by $e_2$, that is, $\xrightarrow{P}=\xrightarrow{e_1}\xrightarrow{e_2}$.

Similarly to the proof of soundness of $APTC$ with guarded recursion modulo step bisimulation equivalence (1), we can prove that each axiom in Table \ref{RDPRSP} is sound modulo pomset bisimulation equivalence, we omit them.

(3) Soundness of $APTC$ with guarded recursion with respect to hp-bisimulation $\sim_{hp}$.

Since hp-bisimulation $\sim_{hp}$ is both an equivalent and a congruent relation with respect to guarded recursion, we only need to check if each axiom in Table \ref{RDPRSP} is sound modulo hp-bisimulation equivalence.

From the definition of hp-bisimulation (see Definition \ref{HHPB}), we know that hp-bisimulation is defined on the posetal product $(C_1,f,C_2),f:C_1\rightarrow C_2\textrm{ isomorphism}$. Two process terms $s$ related to $C_1$ and $t$ related to $C_2$, and $f:C_1\rightarrow C_2\textrm{ isomorphism}$. Initially, $(C_1,f,C_2)=(\emptyset,\emptyset,\emptyset)$, and $(\emptyset,\emptyset,\emptyset)\in\sim_{hp}$. When $s\xrightarrow{e}s'$ ($C_1\xrightarrow{e}C_1'$), there will be $t\xrightarrow{e}t'$ ($C_2\xrightarrow{e}C_2'$), and we define $f'=f[e\mapsto e]$. Then, if $(C_1,f,C_2)\in\sim_{hp}$, then $(C_1',f',C_2')\in\sim_{hp}$.

Similarly to the proof of soundness of $APTC$ with guarded recursion modulo pomset bisimulation equivalence (2), we can prove that each axiom in Table \ref{RDPRSP} is sound modulo hp-bisimulation equivalence, we just need additionally to check the above conditions on hp-bisimulation, we omit them.
\end{proof}

\begin{theorem}[Completeness of $APTC$ with linear recursion]\label{CAPTCR}
Let $p$ and $q$ be closed $APTC$ with linear recursion terms, then,
\begin{enumerate}
  \item if $p\sim_{s} q$ then $p=q$;
  \item if $p\sim_{p} q$ then $p=q$;
  \item if $p\sim_{hp} q$ then $p=q$.
\end{enumerate}
\end{theorem}

\begin{proof}
Firstly, by the elimination theorem of $APTC$ with guarded recursion (see Theorem \ref{ETRecursion}), we know that each process term in $APTC$ with linear recursion is equal to a process term $\langle X_1|E\rangle$ with $E$ a linear recursive specification.

It remains to prove the following cases.

(1) If $\langle X_1|E_1\rangle \sim_s \langle Y_1|E_2\rangle$ for linear recursive specification $E_1$ and $E_2$, then $\langle X_1|E_1\rangle = \langle Y_1|E_2\rangle$.

Let $E_1$ consist of recursive equations $X=t_X$ for $X\in \mathcal{X}$ and $E_2$
consists of recursion equations $Y=t_Y$ for $Y\in\mathcal{Y}$. Let the linear recursive specification $E$ consist of recursion equations $Z_{XY}=t_{XY}$, and $\langle X|E_1\rangle\sim_s\langle Y|E_2\rangle$, and $t_{XY}$ consists of the following summands:

\begin{enumerate}
  \item $t_{XY}$ contains a summand $(a_1\parallel\cdots\parallel a_m)Z_{X'Y'}$ iff $t_X$ contains the summand $(a_1\parallel\cdots\parallel a_m)X'$ and $t_Y$ contains the summand $(a_1\parallel\cdots\parallel a_m)Y'$ such that $\langle X'|E_1\rangle\sim_s\langle Y'|E_2\rangle$;
  \item $t_{XY}$ contains a summand $b_1\parallel\cdots\parallel b_n$ iff $t_X$ contains the summand $b_1\parallel\cdots\parallel b_n$ and $t_Y$ contains the summand $b_1\parallel\cdots\parallel b_n$.
\end{enumerate}

Let $\sigma$ map recursion variable $X$ in $E_1$ to $\langle X|E_1\rangle$, and let $\psi$ map recursion variable $Z_{XY}$ in $E$ to $\langle X|E_1\rangle$. So, $\sigma((a_1\parallel\cdots\parallel a_m)X')\equiv(a_1\parallel\cdots\parallel a_m)\langle X'|E_1\rangle\equiv\psi((a_1\parallel\cdots\parallel a_m)Z_{X'Y'})$, so by $RDP$, we get $\langle X|E_1\rangle=\sigma(t_X)=\psi(t_{XY})$. Then by $RSP$, $\langle X|E_1\rangle=\langle Z_{XY}|E\rangle$, particularly, $\langle X_1|E_1\rangle=\langle Z_{X_1Y_1}|E\rangle$. Similarly, we can obtain $\langle Y_1|E_2\rangle=\langle Z_{X_1Y_1}|E\rangle$. Finally, $\langle X_1|E_1\rangle=\langle Z_{X_1Y_1}|E\rangle=\langle Y_1|E_2\rangle$, as desired.

(2) If $\langle X_1|E_1\rangle \sim_p \langle Y_1|E_2\rangle$ for linear recursive specification $E_1$ and $E_2$, then $\langle X_1|E_1\rangle = \langle Y_1|E_2\rangle$.

It can be proven similarly to (1), we omit it.

(3) If $\langle X_1|E_1\rangle \sim_{hp} \langle Y_1|E_2\rangle$ for linear recursive specification $E_1$ and $E_2$, then $\langle X_1|E_1\rangle = \langle Y_1|E_2\rangle$.

It can be proven similarly to (1), we omit it.
\end{proof}

\subsection{Approximation Induction Principle}

In this subsection, we introduce approximation induction principle ($AIP$) and try to explain that $AIP$ is still valid in true concurrency. $AIP$ can be used to try and equate truly concurrent bisimilar guarded recursive specifications. $AIP$ says that if two process terms are truly concurrent bisimilar up to any finite depth, then they are truly concurrent bisimilar.

Also, we need the auxiliary unary projection operator $\pi_n$ for $n\in\mathbb{N}$ and $\mathbb{N}\triangleq\{0,1,2,\cdots\}$. The transition rules of $\pi_n$ are expressed in Table \ref{TRForProjection}.

\begin{center}
    \begin{table}
        $$\frac{x\xrightarrow{\{e_1,\cdots,e_k\}}\surd}{\pi_{n+1}(x)\xrightarrow{\{e_1,\cdots,e_k\}}\surd}
        \quad\frac{x\xrightarrow{\{e_1,\cdots,e_k\}}x'}{\pi_{n+1}(x)\xrightarrow{\{e_1,\cdots,e_k\}}\pi_n(x')}$$
        \caption{Transition rules of encapsulation operator $\partial_H$}
        \label{TRForProjection}
    \end{table}
\end{center}

Based on the transition rules for projection operator $\pi_n$ in Table \ref{TRForProjection}, we design the axioms as Table \ref{AxiomsForProjection} shows.

\begin{center}
    \begin{table}
        \begin{tabular}{@{}ll@{}}
            \hline No. &Axiom\\
            $PR1$ & $\pi_n(x+y)=\pi_n(x)+\pi_n(y)$\\
            $PR2$ & $\pi_n(x\parallel y)=\pi_n(x)\parallel \pi_n(y)$\\
            $PR3$ & $\pi_{n+1}(e_1\parallel\cdots\parallel e_k)=e_1\parallel\cdots\parallel e_k$\\
            $PR4$ & $\pi_{n+1}((e_1\parallel\cdots\parallel e_k)\cdot x)=(e_1\parallel\cdots\parallel e_k)\cdot\pi_n(x)$\\
            $PR5$ & $\pi_0(x)=\delta$\\
            $PR6$ & $\pi_n(\delta)=\delta$\\
        \end{tabular}
        \caption{Axioms of projection operator}
        \label{AxiomsForProjection}
    \end{table}
\end{center}

The axioms $PR1-PR2$ say that $\pi_n(s+t)$ and $\pi_n(s\parallel t)$ can execute transitions of $s$ and $t$ up to depth $n$. $PR3$ says that $\pi_{n+1}(e_1\parallel\cdots\parallel e_k)$ executes $\{e_1,\cdots,e_{k}\}$ and terminates successfully. $PR4$ says that $\pi_{n+1}((e_1\parallel\cdots\parallel e_k)\cdot t)$ executes $\{e_1,\cdots,e_{k}\}$ and then executes transitions of $t$ up to depth $n$. $PR5$ and $PR6$ say that $\pi_0(t)$ and $\pi_n(\delta)$ exhibit no actions.

\begin{theorem}[Conservativity of $APTC$ with projection operator and guarded recursion]
$APTC$  with projection operator and guarded recursion is a conservative extension of $APTC$ with guarded recursion.
\end{theorem}

\begin{proof}
It follows from the following two facts (see Theorem \ref{TCE}).

\begin{enumerate}
  \item The transition rules of $APTC$ with guarded recursion are all source-dependent;
  \item The sources of the transition rules for the projection operator contain an occurrence of $\pi_n$.
\end{enumerate}
\end{proof}

\begin{theorem}[Congruence theorem of projection operator $\pi_n$]
Truly concurrent bisimulation equivalences $\sim_{p}$, $\sim_s$, $\sim_{hp}$ and $\sim_{hhp}$ are all congruences with respect to projection operator $\pi_n$.
\end{theorem}

\begin{proof}
(1) Case pomset bisimulation equivalence $\sim_p$.

Let $x$ and $y$ be $APTC$ with projection operator and guarded recursion processes, and $x\sim_{p} y$, it is sufficient to prove that $\pi_{n+1}(x)\sim_{p} \pi_{n+1}(y)$.

By the definition of pomset bisimulation $\sim_p$ (Definition \ref{PSB}), $x\sim_p y$ means that

$$x\xrightarrow{X} x' \quad y\xrightarrow{Y} y'$$

with $X\subseteq x$, $Y\subseteq y$, $X\sim Y$ and $x'\sim_p y'$.

By the pomset transition rules for projection operator $\pi_n$ in Table \ref{TRForProjection}, we can get

$$\pi_{n+1}(x)\xrightarrow{X} \surd \quad \pi_{n+1}(y)\xrightarrow{Y} \surd$$

with $X\subseteq x$, $Y\subseteq y$, and $X\sim Y$, so, we get $\pi_{n+1}(x)\sim_p \pi_{n+1}(y)$, as desired.

Or, we can get

$$\pi_{n+1}(x)\xrightarrow{X} \pi_n(x')\quad \pi_{n+1}(y)\xrightarrow{Y} \pi_n(y')$$

with $X\subseteq x$, $Y\subseteq y$, $X\sim Y$, $x'\sim_p y'$ and the assumption $\pi_n(x')\sim_p\pi_n(y')$, so, we get $\pi_{n+1}(x)\sim_p \pi_{n+1}(y)$, as desired.

(2) The cases of step bisimulation $\sim_s$, hp-bisimulation $\sim_{hp}$ and hhp-bisimulation $\sim_{hhp}$ can be proven similarly, we omit them.
\end{proof}

\begin{theorem}[Elimination theorem of $APTC$ with linear recursion and projection operator]\label{ETProjection}
Each process term in $APTC$ with linear recursion and projection operator is equal to a process term $\langle X_1|E\rangle$ with $E$ a linear recursive specification.
\end{theorem}

\begin{proof}
By applying structural induction with respect to term size, each process term $t_1$ in $APTC$ with linear recursion and projection operator $\pi_n$ generates a process can be expressed in the form of equations

$$t_i=(a_{i11}\parallel\cdots\parallel a_{i1i_1})t_{i1}+\cdots+(a_{ik_i1}\parallel\cdots\parallel a_{ik_ii_k})t_{ik_i}+(b_{i11}\parallel\cdots\parallel b_{i1i_1})+\cdots+(b_{il_i1}\parallel\cdots\parallel b_{il_ii_l})$$

for $i\in\{1,\cdots,n\}$. Let the linear recursive specification $E$ consist of the recursive equations

$$X_i=(a_{i11}\parallel\cdots\parallel a_{i1i_1})X_{i1}+\cdots+(a_{ik_i1}\parallel\cdots\parallel a_{ik_ii_k})X_{ik_i}+(b_{i11}\parallel\cdots\parallel b_{i1i_1})+\cdots+(b_{il_i1}\parallel\cdots\parallel b_{il_ii_l})$$

for $i\in\{1,\cdots,n\}$. Replacing $X_i$ by $t_i$ for $i\in\{1,\cdots,n\}$ is a solution for $E$, $RSP$ yields $t_1=\langle X_1|E\rangle$.

That is, in $E$, there is not the occurrence of projection operator $\pi_n$.
\end{proof}

\begin{theorem}[Soundness of $APTC$ with projection operator and guarded recursion]\label{SAPTCR}
Let $x$ and $y$ be $APTC$ with projection operator and guarded recursion terms. If $APTC$ with projection operator and guarded recursion $\vdash x=y$, then
\begin{enumerate}
  \item $x\sim_{s} y$;
  \item $x\sim_{p} y$;
  \item $x\sim_{hp} y$.
\end{enumerate}
\end{theorem}

\begin{proof}
(1) Soundness of $APTC$ with projection operator and guarded recursion with respect to step bisimulation $\sim_s$.

Since step bisimulation $\sim_s$ is both an equivalent and a congruent relation with respect to $APTC$ with projection operator and guarded recursion, we only need to check if each axiom in Table \ref{AxiomsForProjection} is sound modulo step bisimulation equivalence.

Though transition rules in Table \ref{TRForProjection} are defined in the flavor of single event, they can be modified into a step (a set of events within which each event is pairwise concurrent), we omit them. If we treat a single event as a step containing just one event, the proof of this soundness theorem does not exist any problem, so we use this way and still use the transition rules in Table \ref{TRForProjection}.

We only prove the soundness of the non-trivial axioms $PR1$, $PR2$ and $PR4$.

\begin{itemize}
  \item \textbf{Axiom $PR1$}. Let $p$ and $q$ be $APTC$ with guarded recursion and projection operator processes. $\pi_n(p+q)=\pi_n(p)+\pi_n(q)$, it is sufficient to prove that $\pi_n(p+q)\sim_s\pi_n(p)+\pi_n(q)$. By the transition rules for projection operator $\pi_n$ in Table \ref{TRForProjection} and $+$ in Table \ref{STRForBATC}, we get

      $$\frac{p\xrightarrow{\{e_1,\cdots,e_k\}}\surd}{\pi_{n+1}(p+ q)\xrightarrow{\{e_1,\cdots,e_k\}}\surd}
      \quad\frac{p\xrightarrow{\{e_1,\cdots,e_k\}}\surd}{\pi_{n+1}(p)+ \pi_{n+1}(q)\xrightarrow{\{e_1,\cdots,e_k\}}\surd}$$

      $$\frac{q\xrightarrow{\{e_1',\cdots,e_k'\}}\surd}{\pi_{n+1}(p+ q)\xrightarrow{\{e_1',\cdots,e_k'\}}\surd}
      \quad\frac{q\xrightarrow{\{e_1',\cdots,e_k'\}}\surd}{\pi_{n+1}(p)+ \pi_{n+1}(q)\xrightarrow{\{e_1',\cdots,e_k'\}}\surd}$$

      $$\frac{p\xrightarrow{\{e_1,\cdots,e_k\}}p'}{\pi_{n+1}(p+ q)\xrightarrow{\{e_1,\cdots,e_k\}}\pi_n(p')}
      \quad\frac{p\xrightarrow{\{e_1,\cdots,e_k\}}p'}{\pi_{n+1}(p)+ \pi_{n+1}(q)\xrightarrow{\{e_1,\cdots,e_k\}}\pi_n(p')}$$

      $$\frac{q\xrightarrow{\{e_1',\cdots,e_k'\}}q'}{\pi_{n+1}(p+ q)\xrightarrow{\{e_1',\cdots,e_k'\}}\pi_{n}(q')}
      \quad\frac{q\xrightarrow{\{e_1',\cdots,e_k'\}}q'}{\pi_{n+1}(p)+ \pi_{n+1}(q)\xrightarrow{\{e_1',\cdots,e_k'\}}\pi_{n}(q')}$$, we get

      So, $\pi_n(p+q)\sim_s\pi_n(p)+\pi_n(q)$, as desired.
  \item \textbf{Axiom $PR2$}. Let $p,q$ be $APTC$  with guarded recursion and projection operator processes, and $\pi_{n}(p\parallel q)=\pi_{n}(p)\parallel\pi_{n}(q)$, it is sufficient to prove that $\pi_{n}(p\parallel q)\sim_s \pi_{n}(p)\parallel\pi_{n}(q)$. By the transition rules for operator $\parallel$ in Table \ref{TRForParallel} and $\pi_{n}$ in Table \ref{TRForProjection}, we get

      $$\frac{p\xrightarrow{e_1}\surd\quad q\xrightarrow{e_2}\surd}{\pi_{n+1}(p\parallel q)\xrightarrow{\{e_1,e_2\}}\surd}
      \quad\frac{p\xrightarrow{e_1}\surd\quad q\xrightarrow{e_2}\surd}{\pi_{n+1}(p)\parallel \pi_{n+1}(q)\xrightarrow{\{e_1,e_2\}}\surd}$$

      $$\frac{p\xrightarrow{e_1}p'\quad q\xrightarrow{e_2}\surd}{\pi_{n+1}(p\parallel q)\xrightarrow{\{e_1,e_2\}}\pi_{n}(p')}
      \quad\frac{p\xrightarrow{e_1}p'\quad q\xrightarrow{e_2}\surd}{\pi_{n+1}(p)\parallel \pi_{n+1}(q)\xrightarrow{\{e_1,e_2\}}\pi_{n}(p')}$$

      $$\frac{p\xrightarrow{e_1}\surd\quad q\xrightarrow{e_2}q'}{\pi_{n+1}(p\parallel q)\xrightarrow{\{e_1,e_2\}}\pi_{n}(q')}
      \quad\frac{p\xrightarrow{e_1}\surd\quad q\xrightarrow{e_2}q'}{\pi_{n+1}(p)\parallel \pi_{n+1}(q)\xrightarrow{\{e_1,e_2\}}\pi_{n}(q')}$$

      $$\frac{p\xrightarrow{e_1}p'\quad q\xrightarrow{e_2}q'}{\pi_{n+1}(p\parallel q)\xrightarrow{\{e_1,e_2\}}\pi_{n}(p'\between q')}
      \quad\frac{p\xrightarrow{e_1}p'\quad q\xrightarrow{e_2}q'}{\pi_{n+1}(p)\parallel \pi_{n+1}(q)\xrightarrow{\{e_1,e_2\}}\pi_{n}(p')\between\pi_{n}(q')}$$

      So, with the assumption $\pi_{n}(p'\between q')=\pi_{n}(p')\between\pi_{n}(q')$, $\pi_{n}(p\parallel q)\sim_s \pi_{n}(p)\parallel\pi_{n}(q)$, as desired.
  \item \textbf{Axiom $PR4$}. Let $p$ be an $APTC$ with guarded recursion and projection operator process, and $\pi_{n+1}(e\cdot p)=e\cdot\pi_{n}(p)$, it is sufficient to prove that $\pi_{n+1}(e\cdot p)\sim_s e\cdot\pi_{n}(p)$. By the transition rules for operator $\cdot$ in Table \ref{STRForBATC} and $\pi_n$ in Table \ref{TRForProjection}, we get

      $$\frac{e_1\parallel\cdots\parallel e_k\xrightarrow{\{e_1,\cdots,e_k\}}\surd}{\pi_{n+1}((e_1\parallel\cdots\parallel e_k)\cdot p)\xrightarrow{\{e_1,\cdots,e_k\}}\pi_n(p)}
      \quad\frac{e_1\parallel\cdots\parallel e_k\xrightarrow{\{e_1,\cdots,e_k\}}\surd}{(e_1\parallel\cdots\parallel e_k)\cdot\pi_n(p)\xrightarrow{\{e_1,\cdots,e_k\}}\pi_n(p)}$$

      So, $\pi_{n+1}(e\cdot p)\sim_s e\cdot\pi_{n}(p)$, as desired.
\end{itemize}
(2) Soundness of $APTC$ with guarded recursion and projection operator with respect to pomset bisimulation $\sim_p$.

Since pomset bisimulation $\sim_{p}$ is both an equivalent and a congruent relation with respect to $APTC$ with guarded recursion and projection operator, we only need to check if each axiom in Table \ref{AxiomsForProjection} is sound modulo pomset bisimulation equivalence.

From the definition of pomset bisimulation (see Definition \ref{PSB}), we know that pomset bisimulation is defined by pomset transitions, which are labeled by pomsets. In a pomset transition, the events in the pomset are either within causality relations (defined by $\cdot$) or in concurrency (implicitly defined by $\cdot$ and $+$, and explicitly defined by $\between$), of course, they are pairwise consistent (without conflicts). In (1), we have already proven the case that all events are pairwise concurrent, so, we only need to prove the case of events in causality. Without loss of generality, we take a pomset of $P=\{e_1,e_2:e_1\cdot e_2\}$. Then the pomset transition labeled by the above $P$ is just composed of one single event transition labeled by $e_1$ succeeded by another single event transition labeled by $e_2$, that is, $\xrightarrow{P}=\xrightarrow{e_1}\xrightarrow{e_2}$.

Similarly to the proof of soundness of $APTC$ with guarded recursion and projection operator modulo step bisimulation equivalence (1), we can prove that each axiom in Table \ref{AxiomsForProjection} is sound modulo pomset bisimulation equivalence, we omit them.

(3) Soundness of $APTC$ with guarded recursion and projection operator with respect to hp-bisimulation $\sim_{hp}$.

Since hp-bisimulation $\sim_{hp}$ is both an equivalent and a congruent relation with respect to $APTC$ with guarded recursion and projection operator, we only need to check if each axiom in Table \ref{AxiomsForProjection} is sound modulo hp-bisimulation equivalence.

From the definition of hp-bisimulation (see Definition \ref{HHPB}), we know that hp-bisimulation is defined on the posetal product $(C_1,f,C_2),f:C_1\rightarrow C_2\textrm{ isomorphism}$. Two process terms $s$ related to $C_1$ and $t$ related to $C_2$, and $f:C_1\rightarrow C_2\textrm{ isomorphism}$. Initially, $(C_1,f,C_2)=(\emptyset,\emptyset,\emptyset)$, and $(\emptyset,\emptyset,\emptyset)\in\sim_{hp}$. When $s\xrightarrow{e}s'$ ($C_1\xrightarrow{e}C_1'$), there will be $t\xrightarrow{e}t'$ ($C_2\xrightarrow{e}C_2'$), and we define $f'=f[e\mapsto e]$. Then, if $(C_1,f,C_2)\in\sim_{hp}$, then $(C_1',f',C_2')\in\sim_{hp}$.

Similarly to the proof of soundness of $APTC$ with guarded recursion and projection operator modulo pomset bisimulation equivalence (2), we can prove that each axiom in Table \ref{AxiomsForProjection} is sound modulo hp-bisimulation equivalence, we just need additionally to check the above conditions on hp-bisimulation, we omit them.
\end{proof}

Then $AIP$ is given in Table \ref{AIP}.

\begin{center}
    \begin{table}
        \begin{tabular}{@{}ll@{}}
            \hline No. &Axiom\\
            $AIP$ & if $\pi_n(x)=\pi_n(y)$ for $n\in\mathbb{N}$, then $x=y$\\
        \end{tabular}
        \caption{$AIP$}
        \label{AIP}
    \end{table}
\end{center}

\begin{theorem}[Soundness of $AIP$]\label{SAIP}
Let $x$ and $y$ be $APTC$ with projection operator and guarded recursion terms.

\begin{enumerate}
  \item If $\pi_n(x)\sim_s\pi_n(y)$ for $n\in\mathbb{N}$, then $x\sim_s y$;
  \item If $\pi_n(x)\sim_p\pi_n(y)$ for $n\in\mathbb{N}$, then $x\sim_p y$;
  \item If $\pi_n(x)\sim_{hp}\pi_n(y)$ for $n\in\mathbb{N}$, then $x\sim_{hp} y$.
\end{enumerate}
\end{theorem}

\begin{proof}
(1) If $\pi_n(x)\sim_s\pi_n(y)$ for $n\in\mathbb{N}$, then $x\sim_s y$.

Since step bisimulation $\sim_{s}$ is both an equivalent and a congruent relation with respect to $APTC$ with guarded recursion and projection operator, we only need to check if $AIP$ in Table \ref{AIP} is sound modulo step bisimulation equivalence.

Let $p,p_0$ and $q,q_0$ be closed $APTC$ with projection operator and guarded recursion terms such that $\pi_n(p_0)\sim_s\pi_n(q_0)$ for $n\in\mathbb{N}$. We define a relation $R$ such that $p R q$ iff $\pi_n(p)\sim_s \pi_n(q)$. Obviously, $p_0 R q_0$, next, we prove that $R\in\sim_s$.

Let $p R q$ and $p\xrightarrow{\{e_1,\cdots,e_k\}}\surd$, then $\pi_1(p)\xrightarrow{\{e_1,\cdots,e_k\}}\surd$, $\pi_1(p)\sim_s\pi_1(q)$ yields $\pi_1(q)\xrightarrow{\{e_1,\cdots,e_k\}}\surd$. Similarly, $q\xrightarrow{\{e_1,\cdots,e_k\}}\surd$ implies $p\xrightarrow{\{e_1,\cdots,e_k\}}\surd$.

Let $p R q$ and $p\xrightarrow{\{e_1,\cdots,e_k\}}p'$. We define the set of process terms

$$S_n\triangleq\{q'|q\xrightarrow{\{e_1,\cdots,e_k\}}q'\textrm{ and }\pi_n(p')\sim_s\pi_n(q')\}$$

\begin{enumerate}
  \item Since $\pi_{n+1}(p)\sim_s\pi_{n+1}(q)$ and $\pi_{n+1}(p)\xrightarrow{\{e_1,\cdots,e_k\}}\pi_n(p')$, there exist $q'$ such that $\pi_{n+1}(q)\xrightarrow{\{e_1,\cdots,e_k\}}\pi_n(q')$ and $\pi_{n}(p')\sim_s\pi_{n}(q')$. So, $S_n$ is not empty.
  \item There are only finitely many $q'$ such that $q\xrightarrow{\{e_1,\cdots,e_k\}}q'$, so, $S_n$ is finite.
  \item $\pi_{n+1}(p)\sim_s\pi_{n+1}(q)$ implies $\pi_{n}(p')\sim_s\pi_{n}(q')$, so $S_n\supseteq S_{n+1}$.
\end{enumerate}

So, $S_n$ has a non-empty intersection, and let $q'$ be in this intersection, then $q\xrightarrow{\{e_1,\cdots,e_k\}}q'$ and $\pi_n(p')\sim_s\pi_n(q')$, so $p' R q'$. Similarly, let $p\mathcal{q}q$, we can obtain $q\xrightarrow{\{e_1,\cdots,e_k\}}q'$ implies $p\xrightarrow{\{e_1,\cdots,e_k\}}p'$ such that $p' R q'$.

Finally, $R\in\sim_s$ and $p_0\sim_s q_0$, as desired.

(2) If $\pi_n(x)\sim_p\pi_n(y)$ for $n\in\mathbb{N}$, then $x\sim_p y$.

Since pomset bisimulation $\sim_{p}$ is both an equivalent and a congruent relation with respect to $APTC$ with guarded recursion and projection operator, we only need to check if $AIP$ in Table \ref{AIP} is sound modulo pomset bisimulation equivalence.

From the definition of pomset bisimulation (see Definition \ref{PSB}), we know that pomset bisimulation is defined by pomset transitions, which are labeled by pomsets. In a pomset transition, the events in the pomset are either within causality relations (defined by $\cdot$) or in concurrency (implicitly defined by $\cdot$ and $+$, and explicitly defined by $\between$), of course, they are pairwise consistent (without conflicts). In (1), we have already proven the case that all events are pairwise concurrent, so, we only need to prove the case of events in causality. Without loss of generality, we take a pomset of $P=\{e_1,e_2:e_1\cdot e_2\}$. Then the pomset transition labeled by the above $P$ is just composed of one single event transition labeled by $e_1$ succeeded by another single event transition labeled by $e_2$, that is, $\xrightarrow{P}=\xrightarrow{e_1}\xrightarrow{e_2}$.

Similarly to the proof of soundness of $AIP$ modulo step bisimulation equivalence (1), we can prove that $AIP$ in Table \ref{AIP} is sound modulo pomset bisimulation equivalence, we omit them.

(3) If $\pi_n(x)\sim_{hp}\pi_n(y)$ for $n\in\mathbb{N}$, then $x\sim_{hp} y$.

Since hp-bisimulation $\sim_{hp}$ is both an equivalent and a congruent relation with respect to $APTC$ with guarded recursion and projection operator, we only need to check if $AIP$ in Table \ref{AIP} is sound modulo hp-bisimulation equivalence.

From the definition of hp-bisimulation (see Definition \ref{HHPB}), we know that hp-bisimulation is defined on the posetal product $(C_1,f,C_2),f:C_1\rightarrow C_2\textrm{ isomorphism}$. Two process terms $s$ related to $C_1$ and $t$ related to $C_2$, and $f:C_1\rightarrow C_2\textrm{ isomorphism}$. Initially, $(C_1,f,C_2)=(\emptyset,\emptyset,\emptyset)$, and $(\emptyset,\emptyset,\emptyset)\in\sim_{hp}$. When $s\xrightarrow{e}s'$ ($C_1\xrightarrow{e}C_1'$), there will be $t\xrightarrow{e}t'$ ($C_2\xrightarrow{e}C_2'$), and we define $f'=f[e\mapsto e]$. Then, if $(C_1,f,C_2)\in\sim_{hp}$, then $(C_1',f',C_2')\in\sim_{hp}$.

Similarly to the proof of soundness of $AIP$ modulo pomset bisimulation equivalence (2), we can prove that $AIP$ in Table \ref{AIP} is sound modulo hp-bisimulation equivalence, we just need additionally to check the above conditions on hp-bisimulation, we omit them.
\end{proof}

\begin{theorem}[Completeness of $AIP$]\label{CAIP}
Let $p$ and $q$ be closed $APTC$ with linear recursion and projection operator terms, then,
\begin{enumerate}
  \item if $p\sim_{s} q$ then $\pi_n(p)=\pi_n(q)$;
  \item if $p\sim_{p} q$ then $\pi_n(p)=\pi_n(q)$;
  \item if $p\sim_{hp} q$ then $\pi_n(p)=\pi_n(q)$.
\end{enumerate}
\end{theorem}

\begin{proof}
Firstly, by the elimination theorem of $APTC$ with guarded recursion and projection operator (see Theorem \ref{ETProjection}), we know that each process term in $APTC$ with linear recursion and projection operator is equal to a process term $\langle X_1|E\rangle$ with $E$ a linear recursive specification:

$$X_i=(a_{i11}\parallel\cdots\parallel a_{i1i_1})X_{i1}+\cdots+(a_{ik_i1}\parallel\cdots\parallel a_{ik_ii_k})X_{ik_i}+(b_{i11}\parallel\cdots\parallel b_{i1i_1})+\cdots+(b_{il_i1}\parallel\cdots\parallel b_{il_ii_l})$$

for $i\in\{1,\cdots,n\}$.

It remains to prove the following cases.

(1) if $p\sim_{s} q$ then $\pi_n(p)=\pi_n(q)$.

Let $p\sim_s q$, and fix an $n\in\mathbb{N}$, there are $p',q'$ in basic $APTC$ terms such that $p'=\pi_n(p)$ and $q'=\pi_n(q)$. Since $\sim_s$ is a congruence with respect to $APTC$, if $p\sim_s q$ then $\pi_n(p)\sim_s\pi_n(q)$. The soundness theorem yields $p'\sim_s\pi_n(p)\sim_s\pi_n(q)\sim_s q'$. Finally, the completeness of $APTC$ modulo $\sim_s$ (see Theorem \ref{SAPTCSBE}) ensures $p'=q'$, and $\pi_n(p)=p'=q'=\pi_n(q)$, as desired.

(2) if $p\sim_{p} q$ then $\pi_n(p)=\pi_n(q)$.

Let $p\sim_p q$, and fix an $n\in\mathbb{N}$, there are $p',q'$ in basic $APTC$ terms such that $p'=\pi_n(p)$ and $q'=\pi_n(q)$. Since $\sim_p$ is a congruence with respect to $APTC$, if $p\sim_p q$ then $\pi_n(p)\sim_p\pi_n(q)$. The soundness theorem yields $p'\sim_p\pi_n(p)\sim_p\pi_n(q)\sim_p q'$. Finally, the completeness of $APTC$ modulo $\sim_p$ (see Theorem \ref{SAPTCPBE}) ensures $p'=q'$, and $\pi_n(p)=p'=q'=\pi_n(q)$, as desired.

(3) if $p\sim_{hp} q$ then $\pi_n(p)=\pi_n(q)$.

Let $p\sim_{hp} q$, and fix an $n\in\mathbb{N}$, there are $p',q'$ in basic $APTC$ terms such that $p'=\pi_n(p)$ and $q'=\pi_n(q)$. Since $\sim_{hp}$ is a congruence with respect to $APTC$, if $p\sim_{hp} q$ then $\pi_n(p)\sim_{hp}\pi_n(q)$. The soundness theorem yields $p'\sim_{hp}\pi_n(p)\sim_{hp}\pi_n(q)\sim_{hp} q'$. Finally, the completeness of $APTC$ modulo $\sim_{hp}$ (see Theorem \ref{SAPTCHPBE}) ensures $p'=q'$, and $\pi_n(p)=p'=q'=\pi_n(q)$, as desired.
\end{proof}

\section{Abstraction}\label{abs}

To abstract away from the internal implementations of a program, and verify that the program exhibits the desired external behaviors, the silent step $\tau$ and abstraction operator $\tau_I$ are introduced, where $I\subseteq \mathbb{E}$ denotes the internal events. The silent step $\tau$ represents the internal events, when we consider the external behaviors of a process, $\tau$ events can be removed, that is, $\tau$ events must keep silent. The transition rule of $\tau$ is shown in Table \ref{TRForTau}. In the following, let the atomic event $e$ range over $\mathbb{E}\cup\{\delta\}\cup\{\tau\}$, and let the communication function $\gamma:\mathbb{E}\cup\{\tau\}\times \mathbb{E}\cup\{\tau\}\rightarrow \mathbb{E}\cup\{\delta\}$, with each communication involved $\tau$ resulting into $\delta$.

\begin{center}
    \begin{table}
        $$\frac{}{\tau\xrightarrow{\tau}\surd}$$
        \caption{Transition rule of the silent step}
        \label{TRForTau}
    \end{table}
\end{center}

The silent step $\tau$ was firstly introduced by Milner in his CCS \cite{CCS}, the algebraic laws about $\tau$ were introduced in \cite{ALNC}, and finally the axiomatization of $\tau$ and $\tau_I$ formed a part of $ACP$ \cite{ACP}. Though $\tau$ has been discussed in the interleaving bisimulation background, several years ago, we introduced $\tau$ into true concurrency, called weakly true concurrency \cite{WTC}, and also designed its logic based on a uniform logic for true concurrency \cite{LTC1} \cite{LTC2}.

In this section, we try to find the algebraic laws of $\tau$ and $\tau_I$ in true concurrency, or, exactly, to what extent the laws of $\tau$ and $\tau_I$ in interleaving bisimulation fit the situation of true concurrency.

\subsection{Rooted Branching Truly Concurrent Bisimulation Equivalence}

In section \ref{tc}, we introduce $\tau$ into event structure, and also give the concept of weakly true concurrency. In this subsection, we give the concepts of rooted branching truly concurrent bisimulation equivalences, based on these concepts, we can design the axiom system of the silent step $\tau$ and the abstraction operator $\tau_I$. Similarly to rooted branching bisimulation equivalence, rooted branching truly concurrent bisimulation equivalences are following.

\begin{definition}[Branching pomset, step bisimulation]\label{BPSB}
Assume a special termination predicate $\downarrow$, and let $\surd$ represent a state with $\surd\downarrow$. Let $\mathcal{E}_1$, $\mathcal{E}_2$ be PESs. A branching pomset bisimulation is a relation $R\subseteq\mathcal{C}(\mathcal{E}_1)\times\mathcal{C}(\mathcal{E}_2)$, such that:
 \begin{enumerate}
   \item if $(C_1,C_2)\in R$, and $C_1\xrightarrow{X}C_1'$ then
   \begin{itemize}
     \item either $X\equiv \tau^*$, and $(C_1',C_2)\in R$;
     \item or there is a sequence of (zero or more) $\tau$-transitions $C_2\xrightarrow{\tau^*} C_2^0$, such that $(C_1,C_2^0)\in R$ and $C_2^0\xRightarrow{X}C_2'$ with $(C_1',C_2')\in R$;
   \end{itemize}
   \item if $(C_1,C_2)\in R$, and $C_2\xrightarrow{X}C_2'$ then
   \begin{itemize}
     \item either $X\equiv \tau^*$, and $(C_1,C_2')\in R$;
     \item or there is a sequence of (zero or more) $\tau$-transitions $C_1\xrightarrow{\tau^*} C_1^0$, such that $(C_1^0,C_2)\in R$ and $C_1^0\xRightarrow{X}C_1'$ with $(C_1',C_2')\in R$;
   \end{itemize}
   \item if $(C_1,C_2)\in R$ and $C_1\downarrow$, then there is a sequence of (zero or more) $\tau$-transitions $C_2\xrightarrow{\tau^*}C_2^0$ such that $(C_1,C_2^0)\in R$ and $C_2^0\downarrow$;
   \item if $(C_1,C_2)\in R$ and $C_2\downarrow$, then there is a sequence of (zero or more) $\tau$-transitions $C_1\xrightarrow{\tau^*}C_1^0$ such that $(C_1^0,C_2)\in R$ and $C_1^0\downarrow$.
 \end{enumerate}

We say that $\mathcal{E}_1$, $\mathcal{E}_2$ are branching pomset bisimilar, written $\mathcal{E}_1\approx_{bp}\mathcal{E}_2$, if there exists a branching pomset bisimulation $R$, such that $(\emptyset,\emptyset)\in R$.

By replacing pomset transitions with steps, we can get the definition of branching step bisimulation. When PESs $\mathcal{E}_1$ and $\mathcal{E}_2$ are branching step bisimilar, we write $\mathcal{E}_1\approx_{bs}\mathcal{E}_2$.
\end{definition}

\begin{definition}[Rooted branching pomset, step bisimulation]\label{RBPSB}
Assume a special termination predicate $\downarrow$, and let $\surd$ represent a state with $\surd\downarrow$. Let $\mathcal{E}_1$, $\mathcal{E}_2$ be PESs. A rooted branching pomset bisimulation is a relation $R\subseteq\mathcal{C}(\mathcal{E}_1)\times\mathcal{C}(\mathcal{E}_2)$, such that:
 \begin{enumerate}
   \item if $(C_1,C_2)\in R$, and $C_1\xrightarrow{X}C_1'$ then $C_2\xrightarrow{X}C_2'$ with $C_1'\approx_{bp}C_2'$;
   \item if $(C_1,C_2)\in R$, and $C_2\xrightarrow{X}C_2'$ then $C_1\xrightarrow{X}C_1'$ with $C_1'\approx_{bp}C_2'$;
   \item if $(C_1,C_2)\in R$ and $C_1\downarrow$, then $C_2\downarrow$;
   \item if $(C_1,C_2)\in R$ and $C_2\downarrow$, then $C_1\downarrow$.
 \end{enumerate}

We say that $\mathcal{E}_1$, $\mathcal{E}_2$ are rooted branching pomset bisimilar, written $\mathcal{E}_1\approx_{rbp}\mathcal{E}_2$, if there exists a rooted branching pomset bisimulation $R$, such that $(\emptyset,\emptyset)\in R$.

By replacing pomset transitions with steps, we can get the definition of rooted branching step bisimulation. When PESs $\mathcal{E}_1$ and $\mathcal{E}_2$ are rooted branching step bisimilar, we write $\mathcal{E}_1\approx_{rbs}\mathcal{E}_2$.
\end{definition}

\begin{definition}[Branching (hereditary) history-preserving bisimulation]\label{BHHPB}
Assume a special termination predicate $\downarrow$, and let $\surd$ represent a state with $\surd\downarrow$. A branching history-preserving (hp-) bisimulation is a weakly posetal relation $R\subseteq\mathcal{C}(\mathcal{E}_1)\overline{\times}\mathcal{C}(\mathcal{E}_2)$ such that:

 \begin{enumerate}
   \item if $(C_1,f,C_2)\in R$, and $C_1\xrightarrow{e_1}C_1'$ then
   \begin{itemize}
     \item either $e_1\equiv \tau$, and $(C_1',f[e_1\mapsto \tau],C_2)\in R$;
     \item or there is a sequence of (zero or more) $\tau$-transitions $C_2\xrightarrow{\tau^*} C_2^0$, such that $(C_1,f,C_2^0)\in R$ and $C_2^0\xrightarrow{e_2}C_2'$ with $(C_1',f[e_1\mapsto e_2],C_2')\in R$;
   \end{itemize}
   \item if $(C_1,f,C_2)\in R$, and $C_2\xrightarrow{e_2}C_2'$ then
   \begin{itemize}
     \item either $e_2\equiv \tau$, and $(C_1,f[e_2\mapsto \tau],C_2')\in R$;
     \item or there is a sequence of (zero or more) $\tau$-transitions $C_1\xrightarrow{\tau^*} C_1^0$, such that $(C_1^0,f,C_2)\in R$ and $C_1^0\xrightarrow{e_1}C_1'$ with $(C_1',f[e_2\mapsto e_1],C_2')\in R$;
   \end{itemize}
   \item if $(C_1,f,C_2)\in R$ and $C_1\downarrow$, then there is a sequence of (zero or more) $\tau$-transitions $C_2\xrightarrow{\tau^*}C_2^0$ such that $(C_1,f,C_2^0)\in R$ and $C_2^0\downarrow$;
   \item if $(C_1,f,C_2)\in R$ and $C_2\downarrow$, then there is a sequence of (zero or more) $\tau$-transitions $C_1\xrightarrow{\tau^*}C_1^0$ such that $(C_1^0,f,C_2)\in R$ and $C_1^0\downarrow$.
 \end{enumerate}

$\mathcal{E}_1,\mathcal{E}_2$ are branching history-preserving (hp-)bisimilar and are written $\mathcal{E}_1\approx_{bhp}\mathcal{E}_2$ if there exists a branching hp-bisimulation $R$ such that $(\emptyset,\emptyset,\emptyset)\in R$.

A branching hereditary history-preserving (hhp-)bisimulation is a downward closed branching hp-bisimulation. $\mathcal{E}_1,\mathcal{E}_2$ are branching hereditary history-preserving (hhp-)bisimilar and are written $\mathcal{E}_1\approx_{bhhp}\mathcal{E}_2$.
\end{definition}

\begin{definition}[Rooted branching (hereditary) history-preserving bisimulation]\label{RBHHPB}
Assume a special termination predicate $\downarrow$, and let $\surd$ represent a state with $\surd\downarrow$. A rooted branching history-preserving (hp-) bisimulation is a weakly posetal relation $R\subseteq\mathcal{C}(\mathcal{E}_1)\overline{\times}\mathcal{C}(\mathcal{E}_2)$ such that:

 \begin{enumerate}
   \item if $(C_1,f,C_2)\in R$, and $C_1\xrightarrow{e_1}C_1'$, then $C_2\xrightarrow{e_2}C_2'$ with $C_1'\approx_{bhp}C_2'$;
   \item if $(C_1,f,C_2)\in R$, and $C_2\xrightarrow{e_2}C_2'$, then $C_1\xrightarrow{e_1}C_1'$ with $C_1'\approx_{bhp}C_2'$;
   \item if $(C_1,f,C_2)\in R$ and $C_1\downarrow$, then $C_2\downarrow$;
   \item if $(C_1,f,C_2)\in R$ and $C_2\downarrow$, then $C_1\downarrow$.
 \end{enumerate}

$\mathcal{E}_1,\mathcal{E}_2$ are rooted branching history-preserving (hp-)bisimilar and are written $\mathcal{E}_1\approx_{rbhp}\mathcal{E}_2$ if there exists a rooted branching hp-bisimulation $R$ such that $(\emptyset,\emptyset,\emptyset)\in R$.

A rooted branching hereditary history-preserving (hhp-)bisimulation is a downward closed rooted branching hp-bisimulation. $\mathcal{E}_1,\mathcal{E}_2$ are rooted branching hereditary history-preserving (hhp-)bisimilar and are written $\mathcal{E}_1\approx_{rbhhp}\mathcal{E}_2$.
\end{definition}

\subsection{Guarded Linear Recursion}

The silent step $\tau$ as an atomic event, is introduced into $E$. Considering the recursive specification $X=\tau X$, $\tau s$, $\tau\tau s$, and $\tau\cdots s$ are all its solutions, that is, the solutions make the existence of $\tau$-loops which cause unfairness. To prevent $\tau$-loops, we extend the definition of linear recursive specification (Definition \ref{LRS}) to the guarded one.

\begin{definition}[Guarded linear recursive specification]\label{GLRS}
A recursive specification is linear if its recursive equations are of the form

$$(a_{11}\parallel\cdots\parallel a_{1i_1})X_1+\cdots+(a_{k1}\parallel\cdots\parallel a_{ki_k})X_k+(b_{11}\parallel\cdots\parallel b_{1j_1})+\cdots+(b_{1j_1}\parallel\cdots\parallel b_{lj_l})$$

where $a_{11},\cdots,a_{1i_1},a_{k1},\cdots,a_{ki_k},b_{11},\cdots,b_{1j_1},b_{1j_1},\cdots,b_{lj_l}\in \mathbb{E}\cup\{\tau\}$, and the sum above is allowed to be empty, in which case it represents the deadlock $\delta$.

A linear recursive specification $E$ is guarded if there does not exist an infinite sequence of $\tau$-transitions $\langle X|E\rangle\xrightarrow{\tau}\langle X'|E\rangle\xrightarrow{\tau}\langle X''|E\rangle\xrightarrow{\tau}\cdots$.
\end{definition}

\begin{theorem}[Conservitivity of $APTC$ with silent step and guarded linear recursion]
$APTC$ with silent step and guarded linear recursion is a conservative extension of $APTC$ with linear recursion.
\end{theorem}

\begin{proof}
Since the transition rules of $APTC$ with linear recursion are source-dependent, and the transition rules for silent step in Table \ref{TRForTau} contain only a fresh constant $\tau$ in their source, so the transition rules of $APTC$ with silent step and guarded linear recursion is a conservative extension of those of $APTC$ with linear recursion.
\end{proof}

\begin{theorem}[Congruence theorem of $APTC$ with silent step and guarded linear recursion]
Rooted branching truly concurrent bisimulation equivalences $\approx_{rbp}$, $\approx_{rbs}$ and $\approx_{rbhp}$ are all congruences with respect to $APTC$ with silent step and guarded linear recursion.
\end{theorem}

\begin{proof}
It follows the following three facts:
\begin{enumerate}
  \item in a guarded linear recursive specification, right-hand sides of its recursive equations can be adapted to the form by applications of the axioms in $APTC$ and replacing recursion variables by the right-hand sides of their recursive equations;
  \item truly concurrent bisimulation equivalences $\sim_{p}$, $\sim_s$ and $\sim_{hp}$ are all congruences with respect to all operators of $APTC$, while truly concurrent bisimulation equivalences $\sim_{p}$, $\sim_s$ and $\sim_{hp}$ imply the corresponding rooted branching truly concurrent bisimulations $\approx_{rbp}$, $\approx_{rbs}$ and $\approx_{rbhp}$ (see Proposition \ref{WSCBE}), so rooted branching truly concurrent bisimulations $\approx_{rbp}$, $\approx_{rbs}$ and $\approx_{rbhp}$ are all congruences with respect to all operators of $APTC$;
  \item While $\mathbb{E}$ is extended to $\mathbb{E}\cup\{\tau\}$, it can be proved that rooted branching truly concurrent bisimulations $\approx_{rbp}$, $\approx_{rbs}$ and $\approx_{rbhp}$ are all congruences with respect to all operators of $APTC$, we omit it.
\end{enumerate}
\end{proof}

\subsection{Algebraic Laws for the Silent Step}

We design the axioms for the silent step $\tau$ in Table \ref{AxiomsForTau}.

\begin{center}
\begin{table}
  \begin{tabular}{@{}ll@{}}
\hline No. &Axiom\\
  $B1$ & $e\cdot\tau=e$\\
  $B2$ & $e\cdot(\tau\cdot(x+y)+x)=e\cdot(x+y)$\\
  $B3$ & $x\parallel\tau=x$\\
\end{tabular}
\caption{Axioms of silent step}
\label{AxiomsForTau}
\end{table}
\end{center}

The axioms $B1$, $B2$ and $B3$ are the conditions in which $\tau$ really keeps silent to act with the operators $\cdot$, $+$ and $\parallel$.

\begin{theorem}[Elimination theorem of $APTC$ with silent step and guarded linear recursion]\label{ETTau}
Each process term in $APTC$ with silent step and guarded linear recursion is equal to a process term $\langle X_1|E\rangle$ with $E$ a guarded linear recursive specification.
\end{theorem}

\begin{proof}
By applying structural induction with respect to term size, each process term $t_1$ in $APTC$ with silent step and guarded linear recursion generates a process can be expressed in the form of equations

$$t_i=(a_{i11}\parallel\cdots\parallel a_{i1i_1})t_{i1}+\cdots+(a_{ik_i1}\parallel\cdots\parallel a_{ik_ii_k})t_{ik_i}+(b_{i11}\parallel\cdots\parallel b_{i1i_1})+\cdots+(b_{il_i1}\parallel\cdots\parallel b_{il_ii_l})$$

for $i\in\{1,\cdots,n\}$. Let the linear recursive specification $E$ consist of the recursive equations

$$X_i=(a_{i11}\parallel\cdots\parallel a_{i1i_1})X_{i1}+\cdots+(a_{ik_i1}\parallel\cdots\parallel a_{ik_ii_k})X_{ik_i}+(b_{i11}\parallel\cdots\parallel b_{i1i_1})+\cdots+(b_{il_i1}\parallel\cdots\parallel b_{il_ii_l})$$

for $i\in\{1,\cdots,n\}$. Replacing $X_i$ by $t_i$ for $i\in\{1,\cdots,n\}$ is a solution for $E$, $RSP$ yields $t_1=\langle X_1|E\rangle$.
\end{proof}

\begin{theorem}[Soundness of $APTC$ with silent step and guarded linear recursion]\label{SAPTCTAU}
Let $x$ and $y$ be $APTC$ with silent step and guarded linear recursion terms. If $APTC$ with silent step and guarded linear recursion $\vdash x=y$, then
\begin{enumerate}
  \item $x\approx_{rbs} y$;
  \item $x\approx_{rbp} y$;
  \item $x\approx_{rbhp} y$.
\end{enumerate}
\end{theorem}

\begin{proof}
(1) Soundness of $APTC$ with silent step and guarded linear recursion with respect to rooted branching step bisimulation $\approx_{rbs}$.

Since rooted branching step bisimulation $\approx_{rbs}$ is both an equivalent and a congruent relation with respect to $APTC$ with silent step and guarded linear recursion, we only need to check if each axiom in Table \ref{AxiomsForTau} is sound modulo rooted branching step bisimulation equivalence.

Though transition rules in Table \ref{TRForTau} are defined in the flavor of single event, they can be modified into a step (a set of events within which each event is pairwise concurrent), we omit them. If we treat a single event as a step containing just one event, the proof of this soundness theorem does not exist any problem, so we use this way and still use the transition rules in Table \ref{TRForTau}.

\begin{itemize}
  \item \textbf{Axiom $B1$}. Assume that $e\cdot\tau=e$, it is sufficient to prove that $e\cdot\tau\approx_{rbs}e$. By the transition rules for operator $\cdot$ in Table \ref{STRForBATC} and $\tau$ in Table \ref{TRForTau}, we get

      $$\frac{e\xrightarrow{e}\surd}{e\cdot\tau\xrightarrow{e}\xrightarrow{\tau}\surd}$$

      $$\frac{e\xrightarrow{e}\surd}{e\xrightarrow{e}\surd}$$

      So, $e\cdot\tau\approx_{rbs}e$, as desired.

  \item \textbf{Axiom $B2$}. Let $p$ and $q$ be $APTC$ with silent step and guarded linear recursion processes, and assume that $e\cdot(\tau\cdot(p+q)+p)=e\cdot(p+q)$, it is sufficient to prove that $e\cdot(\tau\cdot(p+q)+p)\approx_{rbs}e\cdot(p+q)$. By the transition rules for operators $\cdot$ and $+$ in Table \ref{STRForBATC} and $\tau$ in Table \ref{TRForTau}, we get

      $$\frac{e\xrightarrow{e}\surd\quad p\xrightarrow{e_1}\surd}{e\cdot(\tau\cdot(p+q)+p)\xrightarrow{e}\xrightarrow{e_1}\surd}$$

      $$\frac{e\xrightarrow{e}\surd\quad p\xrightarrow{e_1}\surd}{e\cdot(p+q)\xrightarrow{e}\xrightarrow{e_1}\surd}$$

      $$\frac{e\xrightarrow{e}\surd\quad p\xrightarrow{e_1}p'}{e\cdot(\tau\cdot(p+q)+p)\xrightarrow{e}\xrightarrow{e_1}p'}$$

      $$\frac{e\xrightarrow{e}\surd\quad p\xrightarrow{e_1}p'}{e\cdot(p+q)\xrightarrow{e}\xrightarrow{e_1}p'}$$

      $$\frac{e\xrightarrow{e}\surd\quad q\xrightarrow{e_2}\surd}{e\cdot(\tau\cdot(p+q)+p)\xrightarrow{e}\xrightarrow{\tau}\xrightarrow{e_2}\surd}$$

      $$\frac{e\xrightarrow{e}\surd\quad q\xrightarrow{e_2}\surd}{e\cdot(p+q)\xrightarrow{e}\xrightarrow{e_2}\surd}$$

      $$\frac{e\xrightarrow{e}\surd\quad q\xrightarrow{e_2}q'}{e\cdot(\tau\cdot(p+q)+p)\xrightarrow{e}\xrightarrow{\tau}\xrightarrow{e_2}q'}$$

      $$\frac{e\xrightarrow{e}\surd\quad q\xrightarrow{e_2}q'}{e\cdot(p+q)\xrightarrow{e}\xrightarrow{e_2}q'}$$

      So, $e\cdot(\tau\cdot(p+q)+p)\approx_{rbs}e\cdot(p+q)$, as desired.

  \item \textbf{Axiom $B3$}. Let $p$ be an $APTC$ with silent step and guarded linear recursion process, and assume that $p\parallel\tau=p$, it is sufficient to prove that $p\parallel\tau\approx_{rbs}p$. By the transition rules for operator $\parallel$ in Table \ref{TRForParallel} and $\tau$ in Table \ref{TRForTau}, we get

      $$\frac{p\xrightarrow{e}\surd}{p\parallel\tau\xRightarrow{e}\surd}$$

      $$\frac{p\xrightarrow{e}p'}{p\parallel\tau\xRightarrow{e}p'}$$

      So, $p\parallel\tau\approx_{rbs}p$, as desired.
\end{itemize}

(2) Soundness of $APTC$ with silent step and guarded linear recursion with respect to rooted branching pomset bisimulation $\approx_{rbp}$.

Since rooted branching pomset bisimulation $\approx_{rbp}$ is both an equivalent and a congruent relation with respect to $APTC$ with silent step and guarded linear recursion, we only need to check if each axiom in Table \ref{AxiomsForTau} is sound modulo rooted branching pomset bisimulation $\approx_{rbp}$.

From the definition of rooted branching pomset bisimulation $\approx_{rbp}$ (see Definition \ref{RBPSB}), we know that rooted branching pomset bisimulation $\approx_{rbp}$ is defined by weak pomset transitions, which are labeled by pomsets with $\tau$. In a weak pomset transition, the events in the pomset are either within causality relations (defined by $\cdot$) or in concurrency (implicitly defined by $\cdot$ and $+$, and explicitly defined by $\between$), of course, they are pairwise consistent (without conflicts). In (1), we have already proven the case that all events are pairwise concurrent, so, we only need to prove the case of events in causality. Without loss of generality, we take a pomset of $P=\{e_1,e_2:e_1\cdot e_2\}$. Then the weak pomset transition labeled by the above $P$ is just composed of one single event transition labeled by $e_1$ succeeded by another single event transition labeled by $e_2$, that is, $\xRightarrow{P}=\xRightarrow{e_1}\xRightarrow{e_2}$.

Similarly to the proof of soundness of $APTC$ with silent step and guarded linear recursion modulo rooted branching step bisimulation $\approx_{rbs}$ (1), we can prove that each axiom in Table \ref{AxiomsForTau} is sound modulo rooted branching pomset bisimulation $\approx_{rbp}$, we omit them.

(3) Soundness of $APTC$ with silent step and guarded linear recursion with respect to rooted branching hp-bisimulation $\approx_{rbhp}$.

Since rooted branching hp-bisimulation $\approx_{rbhp}$ is both an equivalent and a congruent relation with respect to $APTC$ with silent step and guarded linear recursion, we only need to check if each axiom in Table \ref{AxiomsForTau} is sound modulo rooted branching hp-bisimulation $\approx_{rbhp}$.

From the definition of rooted branching hp-bisimulation $\approx_{rbhp}$ (see Definition \ref{RBHHPB}), we know that rooted branching hp-bisimulation $\approx_{rbhp}$ is defined on the weakly posetal product $(C_1,f,C_2),f:\hat{C_1}\rightarrow \hat{C_2}\textrm{ isomorphism}$. Two process terms $s$ related to $C_1$ and $t$ related to $C_2$, and $f:\hat{C_1}\rightarrow \hat{C_2}\textrm{ isomorphism}$. Initially, $(C_1,f,C_2)=(\emptyset,\emptyset,\emptyset)$, and $(\emptyset,\emptyset,\emptyset)\in\approx_{rbhp}$. When $s\xrightarrow{e}s'$ ($C_1\xrightarrow{e}C_1'$), there will be $t\xRightarrow{e}t'$ ($C_2\xRightarrow{e}C_2'$), and we define $f'=f[e\mapsto e]$. Then, if $(C_1,f,C_2)\in\approx_{rbhp}$, then $(C_1',f',C_2')\in\approx_{rbhp}$.

Similarly to the proof of soundness of $APTC$ with silent step and guarded linear recursion modulo rooted branching pomset bisimulation equivalence (2), we can prove that each axiom in Table \ref{AxiomsForTau} is sound modulo rooted branching hp-bisimulation equivalence, we just need additionally to check the above conditions on rooted branching hp-bisimulation, we omit them.
\end{proof}

\begin{theorem}[Completeness of $APTC$ with silent step and guarded linear recursion]\label{CAPTCTAU}
Let $p$ and $q$ be closed $APTC$ with silent step and guarded linear recursion terms, then,
\begin{enumerate}
  \item if $p\approx_{rbs} q$ then $p=q$;
  \item if $p\approx_{rbp} q$ then $p=q$;
  \item if $p\approx_{rbhp} q$ then $p=q$.
\end{enumerate}
\end{theorem}

\begin{proof}
Firstly, by the elimination theorem of $APTC$ with silent step and guarded linear recursion (see Theorem \ref{ETTau}), we know that each process term in $APTC$ with silent step and guarded linear recursion is equal to a process term $\langle X_1|E\rangle$ with $E$ a guarded linear recursive specification.

It remains to prove the following cases.

(1) If $\langle X_1|E_1\rangle \approx_{rbs} \langle Y_1|E_2\rangle$ for guarded linear recursive specification $E_1$ and $E_2$, then $\langle X_1|E_1\rangle = \langle Y_1|E_2\rangle$.

Firstly, the recursive equation $W=\tau+\cdots+\tau$ with $W\nequiv X_1$ in $E_1$ and $E_2$, can be removed, and the corresponding summands $aW$ are replaced by $a$, to get $E_1'$ and $E_2'$, by use of the axioms $RDP$, $A3$ and $B1$, and $\langle X|E_1\rangle = \langle X|E_1'\rangle$, $\langle Y|E_2\rangle = \langle Y|E_2'\rangle$.

Let $E_1$ consists of recursive equations $X=t_X$ for $X\in \mathcal{X}$ and $E_2$
consists of recursion equations $Y=t_Y$ for $Y\in\mathcal{Y}$, and are not the form $\tau+\cdots+\tau$. Let the guarded linear recursive specification $E$ consists of recursion equations $Z_{XY}=t_{XY}$, and $\langle X|E_1\rangle\approx_{rbs}\langle Y|E_2\rangle$, and $t_{XY}$ consists of the following summands:

\begin{enumerate}
  \item $t_{XY}$ contains a summand $(a_1\parallel\cdots\parallel a_m)Z_{X'Y'}$ iff $t_X$ contains the summand $(a_1\parallel\cdots\parallel a_m)X'$ and $t_Y$ contains the summand $(a_1\parallel\cdots\parallel a_m)Y'$ such that $\langle X'|E_1\rangle\approx_{rbs}\langle Y'|E_2\rangle$;
  \item $t_{XY}$ contains a summand $b_1\parallel\cdots\parallel b_n$ iff $t_X$ contains the summand $b_1\parallel\cdots\parallel b_n$ and $t_Y$ contains the summand $b_1\parallel\cdots\parallel b_n$;
  \item $t_{XY}$ contains a summand $\tau Z_{X'Y}$ iff $XY\nequiv X_1Y_1$, $t_X$ contains the summand $\tau X'$, and $\langle X'|E_1\rangle\approx_{rbs}\langle Y|E_2\rangle$;
  \item $t_{XY}$ contains a summand $\tau Z_{XY'}$ iff $XY\nequiv X_1Y_1$, $t_Y$ contains the summand $\tau Y'$, and $\langle X|E_1\rangle\approx_{rbs}\langle Y'|E_2\rangle$.
\end{enumerate}

Since $E_1$ and $E_2$ are guarded, $E$ is guarded. Constructing the process term $u_{XY}$ consist of the following summands:

\begin{enumerate}
  \item $u_{XY}$ contains a summand $(a_1\parallel\cdots\parallel a_m)\langle X'|E_1\rangle$ iff $t_X$ contains the summand $(a_1\parallel\cdots\parallel a_m)X'$ and $t_Y$ contains the summand $(a_1\parallel\cdots\parallel a_m)Y'$ such that $\langle X'|E_1\rangle\approx_{rbs}\langle Y'|E_2\rangle$;
  \item $u_{XY}$ contains a summand $b_1\parallel\cdots\parallel b_n$ iff $t_X$ contains the summand $b_1\parallel\cdots\parallel b_n$ and $t_Y$ contains the summand $b_1\parallel\cdots\parallel b_n$;
  \item $u_{XY}$ contains a summand $\tau \langle X'|E_1\rangle$ iff $XY\nequiv X_1Y_1$, $t_X$ contains the summand $\tau X'$, and $\langle X'|E_1\rangle\approx_{rbs}\langle Y|E_2\rangle$.
\end{enumerate}

Let the process term $s_{XY}$ be defined as follows:

\begin{enumerate}
  \item $s_{XY}\triangleq\tau\langle X|E_1\rangle + u_{XY}$ iff $XY\nequiv X_1Y_1$, $t_Y$ contains the summand $\tau Y'$, and $\langle X|E_1\rangle\approx_{rbs}\langle Y'|E_2\rangle$;
  \item $s_{XY}\triangleq\langle X|E_1\rangle$, otherwise.
\end{enumerate}

So, $\langle X|E_1\rangle=\langle X|E_1\rangle+u_{XY}$, and $(a_1\parallel\cdots\parallel a_m)(\tau\langle X|E_1\rangle+u_{XY})=(a_1\parallel\cdots\parallel a_m)((\tau\langle X|E_1\rangle+u_{XY})+u_{XY})=(a_1\parallel\cdots\parallel a_m)(\langle X|E_1\rangle+u_{XY})=(a_1\parallel\cdots\parallel a_m)\langle X|E_1\rangle$, hence, $(a_1\parallel\cdots\parallel a_m)s_{XY}=(a_1\parallel\cdots\parallel a_m)\langle X|E_1\rangle$.

Let $\sigma$ map recursion variable $X$ in $E_1$ to $\langle X|E_1\rangle$, and let $\psi$ map recursion variable $Z_{XY}$ in $E$ to $s_{XY}$. It is sufficient to prove $s_{XY}=\psi(t_{XY})$ for recursion variables $Z_{XY}$ in $E$. Either $XY\equiv X_1Y_1$ or $XY\nequiv X_1Y_1$, we all can get $s_{XY}=\psi(t_{XY})$. So, $s_{XY}=\langle Z_{XY}|E\rangle$ for recursive variables $Z_{XY}$ in $E$ is a solution for $E$. Then by $RSP$, particularly, $\langle X_1|E_1\rangle=\langle Z_{X_1Y_1}|E\rangle$. Similarly, we can obtain $\langle Y_1|E_2\rangle=\langle Z_{X_1Y_1}|E\rangle$. Finally, $\langle X_1|E_1\rangle=\langle Z_{X_1Y_1}|E\rangle=\langle Y_1|E_2\rangle$, as desired.

(2) If $\langle X_1|E_1\rangle \approx_{rbp} \langle Y_1|E_2\rangle$ for guarded linear recursive specification $E_1$ and $E_2$, then $\langle X_1|E_1\rangle = \langle Y_1|E_2\rangle$.

It can be proven similarly to (1), we omit it.

(3) If $\langle X_1|E_1\rangle \approx_{rbhb} \langle Y_1|E_2\rangle$ for guarded linear recursive specification $E_1$ and $E_2$, then $\langle X_1|E_1\rangle = \langle Y_1|E_2\rangle$.

It can be proven similarly to (1), we omit it.
\end{proof}

\subsection{Abstraction}

The unary abstraction operator $\tau_I$ ($I\subseteq \mathbb{E}$) renames all atomic events in $I$ into $\tau$. $APTC$ with silent step and abstraction operator is called $APTC_{\tau}$. The transition rules of operator $\tau_I$ are shown in Table \ref{TRForAbstraction}.

\begin{center}
    \begin{table}
        $$\frac{x\xrightarrow{e}\surd}{\tau_I(x)\xrightarrow{e}\surd}\quad e\notin I
        \quad\quad\frac{x\xrightarrow{e}x'}{\tau_I(x)\xrightarrow{e}\tau_I(x')}\quad e\notin I$$

        $$\frac{x\xrightarrow{e}\surd}{\tau_I(x)\xrightarrow{\tau}\surd}\quad e\in I
        \quad\quad\frac{x\xrightarrow{e}x'}{\tau_I(x)\xrightarrow{\tau}\tau_I(x')}\quad e\in I$$
        \caption{Transition rule of the abstraction operator}
        \label{TRForAbstraction}
    \end{table}
\end{center}

\begin{theorem}[Conservitivity of $APTC_{\tau}$ with guarded linear recursion]
$APTC_{\tau}$ with guarded linear recursion is a conservative extension of $APTC$ with silent step and guarded linear recursion.
\end{theorem}

\begin{proof}
Since the transition rules of $APTC$ with silent step and guarded linear recursion are source-dependent, and the transition rules for abstraction operator in Table \ref{TRForAbstraction} contain only a fresh operator $\tau_I$ in their source, so the transition rules of $APTC_{\tau}$ with guarded linear recursion is a conservative extension of those of $APTC$ with silent step and guarded linear recursion.
\end{proof}

\begin{theorem}[Congruence theorem of $APTC_{\tau}$ with guarded linear recursion]
Rooted branching truly concurrent bisimulation equivalences $\approx_{rbp}$, $\approx_{rbs}$ and $\approx_{rbhp}$ are all congruences with respect to $APTC_{\tau}$ with guarded linear recursion.
\end{theorem}

\begin{proof}

(1) Case rooted branching pomset bisimulation equivalence $\approx_{rbp}$.

Let $x$ and $y$ be $APTC_{\tau}$ with guarded linear recursion processes, and $x\approx_{rbp} y$, it is sufficient to prove that $\tau_I(x)\approx_{rbp} \tau_I(y)$.

By the transition rules for operator $\tau_I$ in Table \ref{TRForAbstraction}, we can get

$$\tau_I(x)\xrightarrow{X} \surd (X\nsubseteq I) \quad \tau_I(y)\xrightarrow{Y} \surd (Y\nsubseteq I)$$

with $X\subseteq x$, $Y\subseteq y$, and $X\sim Y$.

Or, we can get

$$\tau_I(x)\xrightarrow{X} \tau_I(x') (X\nsubseteq I) \quad \tau_I(y)\xrightarrow{Y} \tau_I(y') (Y\nsubseteq I)$$

with $X\subseteq x$, $Y\subseteq y$, and $X\sim Y$ and the hypothesis $\tau_I(x')\approx_{rbp}\tau_I(y')$.

Or, we can get

$$\tau_I(x)\xrightarrow{\tau^*} \surd (X\subseteq I) \quad \tau_I(y)\xrightarrow{\tau^*} \surd (Y\subseteq I)$$

with $X\subseteq x$, $Y\subseteq y$, and $X\sim Y$.

Or, we can get

$$\tau_I(x)\xrightarrow{\tau^*} \tau_I(x') (X\subseteq I) \quad \tau_I(y)\xrightarrow{\tau^*} \tau_I(y') (Y\subseteq I)$$

with $X\subseteq x$, $Y\subseteq y$, and $X\sim Y$ and the hypothesis $\tau_I(x')\approx_{rbp}\tau_I(y')$.

So, we get $\tau_I(x)\approx_{rbp} \tau_I(y)$, as desired

(2) The cases of rooted branching step bisimulation $\approx_{rbs}$, rooted branching hp-bisimulation $\approx_{rbhp}$ can be proven similarly, we omit them.
\end{proof}

We design the axioms for the abstraction operator $\tau_I$ in Table \ref{AxiomsForAbstraction}.

\begin{center}
\begin{table}
  \begin{tabular}{@{}ll@{}}
\hline No. &Axiom\\
  $TI1$ & $e\notin I\quad \tau_I(e)=e$\\
  $TI2$ & $e\in I\quad \tau_I(e)=\tau$\\
  $TI3$ & $\tau_I(\delta)=\delta$\\
  $TI4$ & $\tau_I(x+y)=\tau_I(x)+\tau_I(y)$\\
  $TI5$ & $\tau_I(x\cdot y)=\tau_I(x)\cdot\tau_I(y)$\\
  $TI6$ & $\tau_I(x\parallel y)=\tau_I(x)\parallel\tau_I(y)$\\
\end{tabular}
\caption{Axioms of abstraction operator}
\label{AxiomsForAbstraction}
\end{table}
\end{center}

The axioms $TI1-TI3$ are the defining laws for the abstraction operator $\tau_I$; $TI4-TI6$ say that in process term $\tau_I(t)$, all transitions of $t$ labeled with atomic events from $I$ are renamed into $\tau$.

\begin{theorem}[Soundness of $APTC_{\tau}$ with guarded linear recursion]\label{SAPTCABS}
Let $x$ and $y$ be $APTC_{\tau}$ with guarded linear recursion terms. If $APTC_{\tau}$ with guarded linear recursion $\vdash x=y$, then
\begin{enumerate}
  \item $x\approx_{rbs} y$;
  \item $x\approx_{rbp} y$;
  \item $x\approx_{rbhp} y$.
\end{enumerate}
\end{theorem}

\begin{proof}
(1) Soundness of $APTC_{\tau}$ with guarded linear recursion with respect to rooted branching step bisimulation $\approx_{rbs}$.

Since rooted branching step bisimulation $\approx_{rbs}$ is both an equivalent and a congruent relation with respect to $APTC_{\tau}$ with guarded linear recursion, we only need to check if each axiom in Table \ref{AxiomsForAbstraction} is sound modulo rooted branching step bisimulation equivalence.

Though transition rules in Table \ref{TRForAbstraction} are defined in the flavor of single event, they can be modified into a step (a set of events within which each event is pairwise concurrent), we omit them. If we treat a single event as a step containing just one event, the proof of this soundness theorem does not exist any problem, so we use this way and still use the transition rules in Table \ref{AxiomsForAbstraction}.

We only prove soundness of the non-trivial axioms $TI4-TI6$, and omit the defining axioms $TI1-TI3$.

\begin{itemize}
  \item \textbf{Axiom $TI4$}. Let $p,q$ be $APTC_{\tau}$ with guarded linear recursion processes, and $\tau_I(p+ q)=\tau_I(p)+\tau_I(q)$, it is sufficient to prove that $\tau_I(p+ q)\approx_{rbs} \tau_I(p)+\tau_I(q)$. By the transition rules for operator $+$ in Table \ref{STRForBATC} and $\tau_I$ in Table \ref{TRForAbstraction}, we get

      $$\frac{p\xrightarrow{e_1}\surd\quad(e_1\notin I)}{\tau_I(p+ q)\xrightarrow{e_1}\surd}
      \quad\frac{p\xrightarrow{e_1}\surd\quad(e_1\notin I)}{\tau_I(p)+ \tau_I(q)\xrightarrow{e_1}\surd}$$

      $$\frac{q\xrightarrow{e_2}\surd\quad(e_2\notin I)}{\tau_I(p+ q)\xrightarrow{e_2}\surd}
      \quad\frac{q\xrightarrow{e_2}\surd\quad(e_2\notin I)}{\tau_I(p)+ \tau_I(q)\xrightarrow{e_2}\surd}$$

      $$\frac{p\xrightarrow{e_1}p'\quad(e_1\notin I)}{\tau_I(p+ q)\xrightarrow{e_1}\tau_I(p')}
      \quad\frac{p\xrightarrow{e_1}p'\quad(e_1\notin I)}{\tau_I(p)+ \tau_I(q)\xrightarrow{e_1}\tau_I(p')}$$

      $$\frac{q\xrightarrow{e_2}q'\quad(e_2\notin I)}{\tau_I(p+ q)\xrightarrow{e_2}\tau_I(q')}
      \quad\frac{q\xrightarrow{e_2}q'\quad(e_2\notin I)}{\tau_I(p)+ \tau_I(q)\xrightarrow{e_2}\tau_I(q')}$$

      $$\frac{p\xrightarrow{e_1}\surd\quad(e_1\in I)}{\tau_I(p+ q)\xrightarrow{\tau}\surd}
      \quad\frac{p\xrightarrow{e_1}\surd\quad(e_1\in I)}{\tau_I(p)+ \tau_I(q)\xrightarrow{\tau}\surd}$$

      $$\frac{q\xrightarrow{e_2}\surd\quad(e_2\in I)}{\tau_I(p+ q)\xrightarrow{\tau}\surd}
      \quad\frac{q\xrightarrow{e_2}\surd\quad(e_2\in I)}{\tau_I(p)+ \tau_I(q)\xrightarrow{\tau}\surd}$$

      $$\frac{p\xrightarrow{e_1}p'\quad(e_1\in I)}{\tau_I(p+ q)\xrightarrow{\tau}\tau_I(p')}
      \quad\frac{p\xrightarrow{e_1}p'\quad(e_1\in I)}{\tau_I(p)+ \tau_I(q)\xrightarrow{\tau}\tau_I(p')}$$

      $$\frac{q\xrightarrow{e_2}q'\quad(e_2\in I)}{\tau_I(p+ q)\xrightarrow{\tau}\tau_I(q')}
      \quad\frac{q\xrightarrow{e_2}q'\quad(e_2\in I)}{\tau_I(p)+ \tau_I(q)\xrightarrow{\tau}\tau_I(q')}$$

      So, $\tau_I(p+ q)\approx_{rbs} \tau_I(p)+\tau_I(q)$, as desired.
  \item \textbf{Axiom $TI5$}. Let $p,q$ be $APTC_{\tau}$ with guarded linear recursion processes, and $\tau_I(p\cdot q)=\tau_I(p)\cdot\tau_I(q)$, it is sufficient to prove that $\tau_I(p\cdot q)\approx_{rbs} \tau_I(p)\cdot\tau_I(q)$. By the transition rules for operator $\cdot$ in Table \ref{STRForBATC} and $\tau_I$ in Table \ref{TRForAbstraction}, we get

      $$\frac{p\xrightarrow{e_1}\surd\quad(e_1\notin I)}{\tau_I(p\cdot q)\xrightarrow{e_1}\tau_I(q)}
      \quad\frac{p\xrightarrow{e_1}\surd\quad(e_1\notin I)}{\tau_I(p)\cdot \tau_I(q)\xrightarrow{e_1}\tau_I(q)}$$

      $$\frac{p\xrightarrow{e_1}p'\quad(e_1\notin I)}{\tau_I(p\cdot q)\xrightarrow{e_1}\tau_I(p'\cdot q)}
      \quad\frac{p\xrightarrow{e_1}p'\quad(e_1\notin I)}{\tau_I(p)\cdot \tau_I(q)\xrightarrow{e_1}\tau_I(p')\cdot\tau_I(q)}$$

      $$\frac{p\xrightarrow{e_1}\surd\quad(e_1\in I)}{\tau_I(p\cdot q)\xrightarrow{\tau}\tau_I(q)}
      \quad\frac{p\xrightarrow{e_1}\surd\quad(e_1\in I)}{\tau_I(p)\cdot \tau_I(q)\xrightarrow{\tau}\tau_I(q)}$$

      $$\frac{p\xrightarrow{e_1}p'\quad(e_1\in I)}{\tau_I(p\cdot q)\xrightarrow{\tau}\tau_I(p'\cdot q)}
      \quad\frac{p\xrightarrow{e_1}p'\quad(e_1\in I)}{\tau_I(p)\cdot \tau_I(q)\xrightarrow{\tau}\tau_I(p')\cdot\tau_I(q)}$$

      So, with the assumption $\tau_I(p'\cdot q)=\tau_I(p')\cdot\tau_I(q)$, $\tau_I(p\cdot q)\approx_{rbs}\tau_I(p)\cdot\tau_I(q)$, as desired.
  \item \textbf{Axiom $TI6$}. Let $p,q$ be $APTC_{\tau}$ with guarded linear recursion processes, and $\tau_I(p\parallel q)=\tau_I(p)\parallel\tau_I(q)$, it is sufficient to prove that $\tau_I(p\parallel q)\approx_{rbs} \tau_I(p)\parallel\tau_I(q)$. By the transition rules for operator $\parallel$ in Table \ref{TRForParallel} and $\tau_I$ in Table \ref{TRForAbstraction}, we get

      $$\frac{p\xrightarrow{e_1}\surd\quad q\xrightarrow{e_2}\surd\quad(e_1,e_2\notin I)}{\tau_I(p\parallel q)\xrightarrow{\{e_1,e_2\}}\surd}
      \quad\frac{p\xrightarrow{e_1}\surd\quad q\xrightarrow{e_2}\surd\quad(e_1,e_2\notin I)}{\tau_I(p)\parallel \tau_I(q)\xrightarrow{\{e_1,e_2\}}\surd}$$

      $$\frac{p\xrightarrow{e_1}p'\quad q\xrightarrow{e_2}\surd\quad(e_1,e_2\notin I)}{\tau_I(p\parallel q)\xrightarrow{\{e_1,e_2\}}\tau_I(p')}
      \quad\frac{p\xrightarrow{e_1}p'\quad q\xrightarrow{e_2}\surd\quad(e_1,e_2\notin I)}{\tau_I(p)\parallel \tau_I(q)\xrightarrow{\{e_1,e_2\}}\tau_I(p')}$$

      $$\frac{p\xrightarrow{e_1}\surd\quad q\xrightarrow{e_2}q'\quad(e_1,e_2\notin I)}{\tau_I(p\parallel q)\xrightarrow{\{e_1,e_2\}}\tau_I(q')}
      \quad\frac{p\xrightarrow{e_1}\surd\quad q\xrightarrow{e_2}q'\quad(e_1,e_2\notin I)}{\tau_I(p)\parallel \tau_I(q)\xrightarrow{\{e_1,e_2\}}\tau_I(q')}$$

      $$\frac{p\xrightarrow{e_1}p'\quad q\xrightarrow{e_2}q'\quad(e_1,e_2\notin I)}{\tau_I(p\parallel q)\xrightarrow{\{e_1,e_2\}}\tau_I(p'\between q')}
      \quad\frac{p\xrightarrow{e_1}p'\quad q\xrightarrow{e_2}q'\quad(e_1,e_2\notin I)}{\tau_I(p)\parallel \tau_I(q)\xrightarrow{\{e_1,e_2\}}\tau_I(p')\between\tau_I(q')}$$

      $$\frac{p\xrightarrow{e_1}\surd\quad q\xrightarrow{e_2}\surd\quad(e_1\notin I,e_2\in I)}{\tau_I(p\parallel q)\xRightarrow{e_1}\surd}
      \quad\frac{p\xrightarrow{e_1}\surd\quad q\xrightarrow{e_2}\surd\quad(e_1\notin I,e_2\in I)}{\tau_I(p)\parallel \tau_I(q)\xRightarrow{e_1}\surd}$$

      $$\frac{p\xrightarrow{e_1}p'\quad q\xrightarrow{e_2}\surd\quad(e_1\notin I,e_2\in I)}{\tau_I(p\parallel q)\xRightarrow{e_1}\tau_I(p')}
      \quad\frac{p\xrightarrow{e_1}p'\quad q\xrightarrow{e_2}\surd\quad(e_1\notin I,e_2\in I)}{\tau_I(p)\parallel \tau_I(q)\xRightarrow{e_1}\tau_I(p')}$$

      $$\frac{p\xrightarrow{e_1}\surd\quad q\xrightarrow{e_2}q'\quad(e_1\notin I,e_2\in I)}{\tau_I(p\parallel q)\xRightarrow{e_1}\tau_I(q')}
      \quad\frac{p\xrightarrow{e_1}\surd\quad q\xrightarrow{e_2}q'\quad(e_1\notin I,e_2\in I)}{\tau_I(p)\parallel \tau_I(q)\xRightarrow{e_1}\tau_I(q')}$$

      $$\frac{p\xrightarrow{e_1}p'\quad q\xrightarrow{e_2}q'\quad(e_1\notin I,e_2\in I)}{\tau_I(p\parallel q)\xRightarrow{e_1}\tau_I(p'\between q')}
      \quad\frac{p\xrightarrow{e_1}p'\quad q\xrightarrow{e_2}q'\quad(e_1\notin I,e_2\in I)}{\tau_I(p)\parallel \tau_I(q)\xRightarrow{e_1}\tau_I(p')\between\tau_I(q')}$$

      $$\frac{p\xrightarrow{e_1}\surd\quad q\xrightarrow{e_2}\surd\quad(e_1\in I,e_2\notin I)}{\tau_I(p\parallel q)\xRightarrow{e_2}\surd}
      \quad\frac{p\xrightarrow{e_1}\surd\quad q\xrightarrow{e_2}\surd\quad(e_1\in I,e_2\notin I)}{\tau_I(p)\parallel \tau_I(q)\xRightarrow{e_2}\surd}$$

      $$\frac{p\xrightarrow{e_1}p'\quad q\xrightarrow{e_2}\surd\quad(e_1\in I,e_2\notin I)}{\tau_I(p\parallel q)\xRightarrow{e_2}\tau_I(p')}
      \quad\frac{p\xrightarrow{e_1}p'\quad q\xrightarrow{e_2}\surd\quad(e_1\in I,e_2\notin I)}{\tau_I(p)\parallel \tau_I(q)\xRightarrow{e_2}\tau_I(p')}$$

      $$\frac{p\xrightarrow{e_1}\surd\quad q\xrightarrow{e_2}q'\quad(e_1\in I,e_2\notin I)}{\tau_I(p\parallel q)\xRightarrow{e_2}\tau_I(q')}
      \quad\frac{p\xrightarrow{e_1}\surd\quad q\xrightarrow{e_2}q'\quad(e_1\in I,e_2\notin I)}{\tau_I(p)\parallel \tau_I(q)\xRightarrow{e_2}\tau_I(q')}$$

      $$\frac{p\xrightarrow{e_1}p'\quad q\xrightarrow{e_2}q'\quad(e_1\in I,e_2\notin I)}{\tau_I(p\parallel q)\xRightarrow{e_2}\tau_I(p'\between q')}
      \quad\frac{p\xrightarrow{e_1}p'\quad q\xrightarrow{e_2}q'\quad(e_1\in I,e_2\notin I)}{\tau_I(p)\parallel \tau_I(q)\xRightarrow{e_2}\tau_I(p')\between\tau_I(q')}$$

      $$\frac{p\xrightarrow{e_1}\surd\quad q\xrightarrow{e_2}\surd\quad(e_1,e_2\in I)}{\tau_I(p\parallel q)\xrightarrow{\tau^*}\surd}
      \quad\frac{p\xrightarrow{e_1}\surd\quad q\xrightarrow{e_2}\surd\quad(e_1,e_2\in I)}{\tau_I(p)\parallel \tau_I(q)\xrightarrow{\tau^*}\surd}$$

      $$\frac{p\xrightarrow{e_1}p'\quad q\xrightarrow{e_2}\surd\quad(e_1,e_2\in I)}{\tau_I(p\parallel q)\xrightarrow{\tau^*}\tau_I(p')}
      \quad\frac{p\xrightarrow{e_1}p'\quad q\xrightarrow{e_2}\surd\quad(e_1,e_2\in I)}{\tau_I(p)\parallel \tau_I(q)\xrightarrow{\tau^*}\tau_I(p')}$$

      $$\frac{p\xrightarrow{e_1}\surd\quad q\xrightarrow{e_2}q'\quad(e_1,e_2\in I)}{\tau_I(p\parallel q)\xrightarrow{\tau^*}\tau_I(q')}
      \quad\frac{p\xrightarrow{e_1}\surd\quad q\xrightarrow{e_2}q'\quad(e_1,e_2\in I)}{\tau_I(p)\parallel \tau_I(q)\xrightarrow{\tau^*}\tau_I(q')}$$

      $$\frac{p\xrightarrow{e_1}p'\quad q\xrightarrow{e_2}q'\quad(e_1,e_2\in I)}{\tau_I(p\parallel q)\xrightarrow{\tau^*}\tau_I(p'\between q')}
      \quad\frac{p\xrightarrow{e_1}p'\quad q\xrightarrow{e_2}q'\quad(e_1,e_2\in I)}{\tau_I(p)\parallel \tau_I(q)\xrightarrow{\tau^*}\tau_I(p')\between\tau_I(q')}$$

      So, with the assumption $\tau_I(p'\between q')=\tau_I(p')\between\tau_I(q')$, $\tau_I(p\parallel q)\approx_{rbs} \tau_I(p)\parallel\tau_I(q)$, as desired.
\end{itemize}

(2) Soundness of $APTC_{\tau}$ with guarded linear recursion with respect to rooted branching pomset bisimulation $\approx_{rbp}$.

Since rooted branching pomset bisimulation $\approx_{rbp}$ is both an equivalent and a congruent relation with respect to $APTC_{\tau}$ with guarded linear recursion, we only need to check if each axiom in Table \ref{AxiomsForAbstraction} is sound modulo rooted branching pomset bisimulation $\approx_{rbp}$.

From the definition of rooted branching pomset bisimulation $\approx_{rbp}$ (see Definition \ref{RBPSB}), we know that rooted branching pomset bisimulation $\approx_{rbp}$ is defined by weak pomset transitions, which are labeled by pomsets with $\tau$. In a weak pomset transition, the events in the pomset are either within causality relations (defined by $\cdot$) or in concurrency (implicitly defined by $\cdot$ and $+$, and explicitly defined by $\between$), of course, they are pairwise consistent (without conflicts). In (1), we have already proven the case that all events are pairwise concurrent, so, we only need to prove the case of events in causality. Without loss of generality, we take a pomset of $P=\{e_1,e_2:e_1\cdot e_2\}$. Then the weak pomset transition labeled by the above $P$ is just composed of one single event transition labeled by $e_1$ succeeded by another single event transition labeled by $e_2$, that is, $\xRightarrow{P}=\xRightarrow{e_1}\xRightarrow{e_2}$.

Similarly to the proof of soundness of $APTC_{\tau}$ with guarded linear recursion modulo rooted branching step bisimulation $\approx_{rbs}$ (1), we can prove that each axiom in Table \ref{AxiomsForAbstraction} is sound modulo rooted branching pomset bisimulation $\approx_{rbp}$, we omit them.

(3) Soundness of $APTC_{\tau}$ with guarded linear recursion with respect to rooted branching hp-bisimulation $\approx_{rbhp}$.

Since rooted branching hp-bisimulation $\approx_{rbhp}$ is both an equivalent and a congruent relation with respect to $APTC_{\tau}$ with guarded linear recursion, we only need to check if each axiom in Table \ref{AxiomsForAbstraction} is sound modulo rooted branching hp-bisimulation $\approx_{rbhp}$.

From the definition of rooted branching hp-bisimulation $\approx_{rbhp}$ (see Definition \ref{RBHHPB}), we know that rooted branching hp-bisimulation $\approx_{rbhp}$ is defined on the weakly posetal product $(C_1,f,C_2),f:\hat{C_1}\rightarrow \hat{C_2}\textrm{ isomorphism}$. Two process terms $s$ related to $C_1$ and $t$ related to $C_2$, and $f:\hat{C_1}\rightarrow \hat{C_2}\textrm{ isomorphism}$. Initially, $(C_1,f,C_2)=(\emptyset,\emptyset,\emptyset)$, and $(\emptyset,\emptyset,\emptyset)\in\approx_{rbhp}$. When $s\xrightarrow{e}s'$ ($C_1\xrightarrow{e}C_1'$), there will be $t\xRightarrow{e}t'$ ($C_2\xRightarrow{e}C_2'$), and we define $f'=f[e\mapsto e]$. Then, if $(C_1,f,C_2)\in\approx_{rbhp}$, then $(C_1',f',C_2')\in\approx_{rbhp}$.

Similarly to the proof of soundness of $APTC_{\tau}$ with guarded linear recursion modulo rooted branching pomset bisimulation equivalence (2), we can prove that each axiom in Table \ref{AxiomsForAbstraction} is sound modulo rooted branching hp-bisimulation equivalence, we just need additionally to check the above conditions on rooted branching hp-bisimulation, we omit them.
\end{proof}

Though $\tau$-loops are prohibited in guarded linear recursive specifications (see Definition \ref{GLRS}) in a specifiable way, they can be constructed using the abstraction operator, for example, there exist $\tau$-loops in the process term $\tau_{\{a\}}(\langle X|X=aX\rangle)$. To avoid $\tau$-loops caused by $\tau_I$ and ensure fairness, the concept of cluster and $CFAR$ (Cluster Fair Abstraction Rule) \cite{CFAR} are still valid in true concurrency, we introduce them below.

\begin{definition}[Cluster]\label{CLUSTER}
Let $E$ be a guarded linear recursive specification, and $I\subseteq \mathbb{E}$. Two recursion variable $X$ and $Y$ in $E$ are in the same cluster for $I$ iff there exist sequences of transitions $\langle X|E\rangle\xrightarrow{\{b_{11},\cdots, b_{1i}\}}\cdots\xrightarrow{\{b_{m1},\cdots, b_{mi}\}}\langle Y|E\rangle$ and $\langle Y|E\rangle\xrightarrow{\{c_{11},\cdots, c_{1j}\}}\cdots\xrightarrow{\{c_{n1},\cdots, c_{nj}\}}\langle X|E\rangle$, where $b_{11},\cdots,b_{mi},c_{11},\cdots,c_{nj}\in I\cup\{\tau\}$.

$a_1\parallel\cdots\parallel a_k$ or $(a_1\parallel\cdots\parallel a_k) X$ is an exit for the cluster $C$ iff: (1) $a_1\parallel\cdots\parallel a_k$ or $(a_1\parallel\cdots\parallel a_k) X$ is a summand at the right-hand side of the recursive equation for a recursion variable in $C$, and (2) in the case of $(a_1\parallel\cdots\parallel a_k) X$, either $a_l\notin I\cup\{\tau\}(l\in\{1,2,\cdots,k\})$ or $X\notin C$.
\end{definition}

\begin{center}
\begin{table}
  \begin{tabular}{@{}ll@{}}
\hline No. &Axiom\\
  $CFAR$ & If $X$ is in a cluster for $I$ with exits \\
           & $\{(a_{11}\parallel\cdots\parallel a_{1i})Y_1,\cdots,(a_{m1}\parallel\cdots\parallel a_{mi})Y_m, b_{11}\parallel\cdots\parallel b_{1j},\cdots,b_{n1}\parallel\cdots\parallel b_{nj}\}$, \\
           & then $\tau\cdot\tau_I(\langle X|E\rangle)=$\\
           & $\tau\cdot\tau_I((a_{11}\parallel\cdots\parallel a_{1i})\langle Y_1|E\rangle+\cdots+(a_{m1}\parallel\cdots\parallel a_{mi})\langle Y_m|E\rangle+b_{11}\parallel\cdots\parallel b_{1j}+\cdots+b_{n1}\parallel\cdots\parallel b_{nj})$\\
\end{tabular}
\caption{Cluster fair abstraction rule}
\label{CFAR}
\end{table}
\end{center}

\begin{theorem}[Soundness of $CFAR$]\label{SCFAR}
$CFAR$ is sound modulo rooted branching truly concurrent bisimulation equivalences $\approx_{rbs}$, $\approx_{rbp}$ and $\approx_{rbhp}$.
\end{theorem}

\begin{proof}
(1) Soundness of $CFAR$ with respect to rooted branching step bisimulation $\approx_{rbs}$.

Let $X$ be in a cluster for $I$ with exits $\{(a_{11}\parallel\cdots\parallel a_{1i})Y_1,\cdots,(a_{m1}\parallel\cdots\parallel a_{mi})Y_m,b_{11}\parallel\cdots\parallel b_{1j},\cdots,b_{n1}\parallel\cdots\parallel b_{nj}\}$. Then $\langle X|E\rangle$ can execute a string of atomic events from $I\cup\{\tau\}$ inside the cluster of $X$, followed by an exit $(a_{i'1}\parallel\cdots\parallel a_{i'i})Y_{i'}$ for $i'\in\{1,\cdots,m\}$ or $b_{j'1}\parallel\cdots\parallel b_{j'j}$ for $j'\in\{1,\cdots,n\}$. Hence, $\tau_I(\langle X|E\rangle)$ can execute a string of $\tau^*$ inside the cluster of $X$, followed by an exit $\tau_I((a_{i'1}\parallel\cdots\parallel a_{i'i})\langle Y_{i'}|E\rangle)$ for $i'\in\{1,\cdots,m\}$ or $\tau_I(b_{j'1}\parallel\cdots\parallel b_{j'j})$ for $j'\in\{1,\cdots,n\}$. And these $\tau^*$ are non-initial in $\tau\tau_I(\langle X|E\rangle)$, so they are truly silent by the axiom $B1$, we obtain $\tau\tau_I(\langle X|E\rangle)\approx_{rbs}\tau\cdot\tau_I((a_{11}\parallel\cdots\parallel a_{1i})\langle Y_1|E\rangle+\cdots+(a_{m1}\parallel\cdots\parallel a_{mi})\langle Y_m|E\rangle+b_{11}\parallel\cdots\parallel b_{1j}+\cdots+b_{n1}\parallel\cdots\parallel b_{nj})$, as desired.

(2) Soundness of $CFAR$ with respect to rooted branching pomset bisimulation $\approx_{rbp}$.

From the definition of rooted branching pomset bisimulation $\approx_{rbp}$ (see Definition \ref{RBPSB}), we know that rooted branching pomset bisimulation $\approx_{rbp}$ is defined by weak pomset transitions, which are labeled by pomsets with $\tau$. In a weak pomset transition, the events in the pomset are either within causality relations (defined by $\cdot$) or in concurrency (implicitly defined by $\cdot$ and $+$, and explicitly defined by $\between$), of course, they are pairwise consistent (without conflicts). In (1), we have already proven the case that all events are pairwise concurrent, so, we only need to prove the case of events in causality. Without loss of generality, we take a pomset of $P=\{e_1,e_2:e_1\cdot e_2\}$. Then the weak pomset transition labeled by the above $P$ is just composed of one single event transition labeled by $e_1$ succeeded by another single event transition labeled by $e_2$, that is, $\xRightarrow{P}=\xRightarrow{e_1}\xRightarrow{e_2}$.

Similarly to the proof of soundness of $CFAR$ modulo rooted branching step bisimulation $\approx_{rbs}$ (1), we can prove that $CFAR$ in Table \ref{CFAR} is sound modulo rooted branching pomset bisimulation $\approx_{rbp}$, we omit them.

(3) Soundness of $CFAR$ with respect to rooted branching hp-bisimulation $\approx_{rbhp}$.

From the definition of rooted branching hp-bisimulation $\approx_{rbhp}$ (see Definition \ref{RBHHPB}), we know that rooted branching hp-bisimulation $\approx_{rbhp}$ is defined on the weakly posetal product $(C_1,f,C_2),f:\hat{C_1}\rightarrow \hat{C_2}\textrm{ isomorphism}$. Two process terms $s$ related to $C_1$ and $t$ related to $C_2$, and $f:\hat{C_1}\rightarrow \hat{C_2}\textrm{ isomorphism}$. Initially, $(C_1,f,C_2)=(\emptyset,\emptyset,\emptyset)$, and $(\emptyset,\emptyset,\emptyset)\in\approx_{rbhp}$. When $s\xrightarrow{e}s'$ ($C_1\xrightarrow{e}C_1'$), there will be $t\xRightarrow{e}t'$ ($C_2\xRightarrow{e}C_2'$), and we define $f'=f[e\mapsto e]$. Then, if $(C_1,f,C_2)\in\approx_{rbhp}$, then $(C_1',f',C_2')\in\approx_{rbhp}$.

Similarly to the proof of soundness of $CFAR$ modulo rooted branching pomset bisimulation equivalence (2), we can prove that $CFAR$ in Table \ref{CFAR} is sound modulo rooted branching hp-bisimulation equivalence, we just need additionally to check the above conditions on rooted branching hp-bisimulation, we omit them.
\end{proof}

\begin{theorem}[Completeness of $APTC_{\tau}$ with guarded linear recursion and $CFAR$]\label{CCFAR}
Let $p$ and $q$ be closed $APTC_{\tau}$ with guarded linear recursion and $CFAR$ terms, then,
\begin{enumerate}
  \item if $p\approx_{rbs} q$ then $p=q$;
  \item if $p\approx_{rbp} q$ then $p=q$;
  \item if $p\approx_{rbhp} q$ then $p=q$.
\end{enumerate}
\end{theorem}

\begin{proof}
(1) For the case of rooted branching step bisimulation, the proof is following.

Firstly, in the proof the Theorem \ref{CAPTCTAU}, we know that each process term $p$ in $APTC$ with silent step and guarded linear recursion is equal to a process term $\langle X_1|E\rangle$ with $E$ a guarded linear recursive specification. And we prove if $\langle X_1|E_1\rangle\approx_{rbs}\langle Y_1|E_2\rangle$, then $\langle X_1|E_1\rangle=\langle Y_1|E_2\rangle$

The only new case is $p\equiv\tau_I(q)$. Let $q=\langle X|E\rangle$ with $E$ a guarded linear recursive specification, so $p=\tau_I(\langle X|E\rangle)$. Then the collection of recursive variables in $E$ can be divided into its clusters $C_1,\cdots,C_N$ for $I$. Let

$$(a_{1i1}\parallel\cdots\parallel a_{k_{i1}i1}) Y_{i1}+\cdots+(a_{1im_i}\parallel\cdots\parallel a_{k_{im_i}im_i}) Y_{im_i}+b_{1i1}\parallel\cdots\parallel b_{l_{i1}i1}+\cdots+b_{1im_i}\parallel\cdots\parallel b_{l_{im_i}im_i}$$

be the conflict composition of exits for the cluster $C_i$, with $i\in\{1,\cdots,N\}$.

For $Z\in C_i$ with $i\in\{1,\cdots,N\}$, we define

$$s_Z\triangleq (\hat{a_{1i1}}\parallel\cdots\parallel \hat{a_{k_{i1}i1}}) \tau_I(\langle Y_{i1}|E\rangle)+\cdots+(\hat{a_{1im_i}}\parallel\cdots\parallel \hat{a_{k_{im_i}im_i}}) \tau_I(\langle Y_{im_i}|E\rangle)+\hat{b_{1i1}}\parallel\cdots\parallel \hat{b_{l_{i1}i1}}+\cdots+\hat{b_{1im_i}}\parallel\cdots\parallel \hat{b_{l_{im_i}im_i}}$$

For $Z\in C_i$ and $a_1,\cdots,a_j\in \mathbb{E}\cup\{\tau\}$ with $j\in\mathbb{N}$, we have

$(a_1\parallel\cdots\parallel a_j)\tau_I(\langle Z|E\rangle)$

$=(a_1\parallel\cdots\parallel a_j)\tau_I((a_{1i1}\parallel\cdots\parallel a_{k_{i1}i1}) \langle Y_{i1}|E\rangle+\cdots+(a_{1im_i}\parallel\cdots\parallel a_{k_{im_i}im_i}) \langle Y_{im_i}|E\rangle+b_{1i1}\parallel\cdots\parallel b_{l_{i1}i1}+\cdots+b_{1im_i}\parallel\cdots\parallel b_{l_{im_i}im_i})$

$=(a_1\parallel\cdots\parallel a_j)s_Z$

Let the linear recursive specification $F$ contain the same recursive variables as $E$, for $Z\in C_i$, $F$ contains the following recursive equation

$$Z=(\hat{a_{1i1}}\parallel\cdots\parallel \hat{a_{k_{i1}i1}}) Y_{i1}+\cdots+(\hat{a_{1im_i}}\parallel\cdots\parallel \hat{a_{k_{im_i}im_i}})  Y_{im_i}+\hat{b_{1i1}}\parallel\cdots\parallel \hat{b_{l_{i1}i1}}+\cdots+\hat{b_{1im_i}}\parallel\cdots\parallel \hat{b_{l_{im_i}im_i}}$$

It is easy to see that there is no sequence of one or more $\tau$-transitions from $\langle Z|F\rangle$ to itself, so $F$ is guarded.

For

$$s_Z=(\hat{a_{1i1}}\parallel\cdots\parallel \hat{a_{k_{i1}i1}}) Y_{i1}+\cdots+(\hat{a_{1im_i}}\parallel\cdots\parallel \hat{a_{k_{im_i}im_i}}) Y_{im_i}+\hat{b_{1i1}}\parallel\cdots\parallel \hat{b_{l_{i1}i1}}+\cdots+\hat{b_{1im_i}}\parallel\cdots\parallel \hat{b_{l_{im_i}im_i}}$$

is a solution for $F$. So, $(a_1\parallel\cdots\parallel a_j)\tau_I(\langle Z|E\rangle)=(a_1\parallel\cdots\parallel a_j)s_Z=(a_1\parallel\cdots\parallel a_j)\langle Z|F\rangle$.

So,

$$\langle Z|F\rangle=(\hat{a_{1i1}}\parallel\cdots\parallel \hat{a_{k_{i1}i1}}) \langle Y_{i1}|F\rangle+\cdots+(\hat{a_{1im_i}}\parallel\cdots\parallel \hat{a_{k_{im_i}im_i}}) \langle Y_{im_i}|F\rangle+\hat{b_{1i1}}\parallel\cdots\parallel \hat{b_{l_{i1}i1}}+\cdots+\hat{b_{1im_i}}\parallel\cdots\parallel \hat{b_{l_{im_i}im_i}}$$

Hence, $\tau_I(\langle X|E\rangle=\langle Z|F\rangle)$, as desired.

(2) For the case of rooted branching pomset bisimulation, it can be proven similarly to (1), we omit it.

(3) For the case of rooted branching hp-bisimulation, it can be proven similarly to (1), we omit it.
\end{proof}

Finally, in section \ref{aptc}, during conflict elimination, the axioms $U25$ and $U27$ are $(\sharp(e_1,e_2))\quad e_1\triangleleft e_2 = \tau$ and $(\sharp(e_1,e_2),e_2\leq e_3)\quad e3\triangleleft e_1 = \tau$. Their functions are like abstraction operator $\tau_I$, their rigorous soundness can be proven similarly to Theorem \ref{SAPTCTAU} and Theorem \ref{SAPTCABS}, really, they are based on weakly true concurrency. We just illustrate their intuition through an example.

\begin{eqnarray}
P&=&\Theta((a\cdot b\cdot c)\parallel (d\cdot e\cdot f)\quad(\sharp(b,e)))\nonumber\\
&\overset{\text{CE23}}{=}&(\Theta(a\cdot b\cdot c)\triangleleft (d\cdot e\cdot f))\parallel (d\cdot e\cdot f)+(\Theta(d\cdot e\cdot f)\triangleleft (a\cdot b\cdot c))\parallel (a\cdot b\cdot c)\quad(\sharp(b,e))\nonumber\\
&\overset{\text{CE22}}{=}&(a\cdot b\cdot c)\triangleleft (d\cdot e\cdot f))\parallel (d\cdot e\cdot f)+((d\cdot e\cdot f)\triangleleft (a\cdot b\cdot c))\parallel (a\cdot b\cdot c)\quad(\sharp(b,e))\nonumber\\
&\overset{\text{U31,U35}}{=}&(a\cdot\tau\cdot\tau)\parallel (d\cdot e\cdot f)+(d\cdot\tau\cdot\tau)\parallel (a\cdot b\cdot c)\nonumber\\
&\overset{\text{B1}}{=}&a\parallel (d\cdot e\cdot f)+d\parallel (a\cdot b\cdot c)\nonumber
\end{eqnarray}

We see that the conflict relation $\sharp(b,e)$ is eliminated.

\section{Applications}\label{app}

$APTC$ provides a formal framework based on truly concurrent behavioral semantics, which can be used to verify the correctness of system behaviors. In this section, we tend to choose one protocol verified by $ACP$ \cite{ACP} -- alternating bit protocol (ABP) \cite{ABP}.

The ABP protocol is used to ensure successful transmission of data through a corrupted channel. This success is based on the assumption that data can be resent an unlimited number of times, which is illustrated in Fig.\ref{ABP}, we alter it into the true concurrency situation.

\begin{enumerate}
  \item Data elements $d_1,d_2,d_3,\cdots$ from a finite set $\Delta$ are communicated between a Sender and a Receiver.
  \item If the Sender reads a datum from channel $A_1$, then this datum is sent to the Receiver in parallel through channel $A_2$.
  \item The Sender processes the data in $\Delta$, formes new data, and sends them to the Receiver through channel $B$.
  \item And the Receiver sends the datum into channel $C_2$.
  \item If channel $B$ is corrupted, the message communicated through $B$ can be turn into an error message $\bot$.
  \item Every time the Receiver receives a message via channel $B$, it sends an acknowledgement to the Sender via channel $D$, which is also corrupted.
  \item Finally, then Sender and the Receiver send out their outputs in parallel through channels $C_1$ and $C_2$.
\end{enumerate}

\begin{figure}
    \centering
    \includegraphics{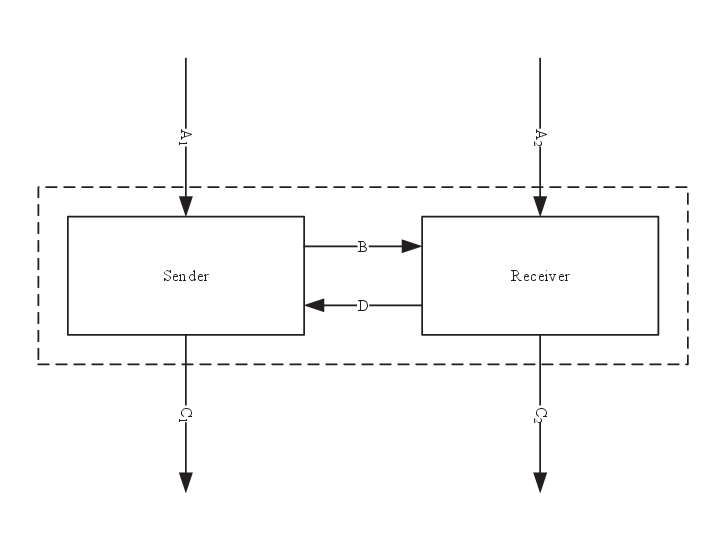}
    \caption{Alternating bit protocol}
    \label{ABP}
\end{figure}

In the truly concurrent ABP, the Sender sends its data to the Receiver; and the Receiver can also send its data to the Sender, for simplicity and without loss of generality, we assume that only the Sender sends its data and the Receiver only receives the data from the Sender. The Sender attaches a bit 0 to data elements $d_{2k-1}$ and a bit 1 to data elements $d_{2k}$, when they are sent into channel $B$. When the Receiver reads a datum, it sends back the attached bit via channel $D$. If the Receiver receives a corrupted message, then it sends back the previous acknowledgement to the Sender.

Then the state transition of the Sender can be described by $APTC$ as follows.

\begin{eqnarray}
&&S_b=\sum_{d\in\Delta}r_{A_1}(d)\cdot T_{db}\nonumber\\
&&T_{db}=(\sum_{d'\in\Delta}(s_B(d',b)\cdot s_{C_1}(d'))+s_B(\bot))\cdot U_{db}\nonumber\\
&&U_{db}=r_D(b)\cdot S_{1-b}+(r_D(1-b)+r_D(\bot))\cdot T_{db}\nonumber
\end{eqnarray}

where $s_B$ denotes sending data through channel $B$, $r_D$ denotes receiving data through channel $D$, similarly, $r_{A_1}$ means receiving data via channel $A_1$, $s_{C_1}$ denotes sending data via channel $C_1$, and $b\in\{0,1\}$.

And the state transition of the Receiver can be described by $APTC$ as follows.

\begin{eqnarray}
&&R_b=\sum_{d\in\Delta}r_{A_2}(d)\cdot R_b'\nonumber\\
&&R_b'=\sum_{d'\in\Delta}\{r_B(d',b)\cdot s_{C_2}(d')\cdot Q_b+r_B(d',1-b)\cdot Q_{1-b}\}+r_B(\bot)\cdot Q_{1-b}\nonumber\\
&&Q_b=(s_D(b)+s_D(\bot))\cdot R_{1-b}\nonumber
\end{eqnarray}

where $r_{A_2}$ denotes receiving data via channel $A_2$, $r_B$ denotes receiving data via channel $B$, $s_{C_2}$ denotes sending data via channel $C_2$, $s_D$ denotes sending data via channel $D$, and $b\in\{0,1\}$.

The send action and receive action of the same data through the same channel can communicate each other, otherwise, a deadlock $\delta$ will be caused. We define the following communication functions.

\begin{eqnarray}
&&\gamma(s_B(d',b),r_B(d',b))\triangleq c_B(d',b)\nonumber\\
&&\gamma(s_B(\bot),r_B(\bot))\triangleq c_B(\bot)\nonumber\\
&&\gamma(s_D(b),r_D(b))\triangleq c_D(b)\nonumber\\
&&\gamma(s_D(\bot),r_D(\bot))\triangleq c_D(\bot)\nonumber
\end{eqnarray}

Let $R_0$ and $S_0$ be in parallel, then the system $R_0S_0$ can be represented by the following process term.

$$\tau_I(\partial_H(\Theta(R_0\between S_0)))=\tau_I(\partial_H(R_0\between S_0))$$

where $H=\{s_B(d',b),r_B(d',b),s_D(b),r_D(b)|d'\in\Delta,b\in\{0,1\}\}\\
\{s_B(\bot),r_B(\bot),s_D(\bot),r_D(\bot)\}$

$I=\{c_B(d',b),c_D(b)|d'\in\Delta,b\in\{0,1\}\}\cup\{c_B(\bot),c_D(\bot)\}$.

Then we get the following conclusion.

\begin{theorem}[Correctness of the ABP protocol]
The ABP protocol $\tau_I(\partial_H(R_0\between S_0))$ exhibits desired external behaviors.
\end{theorem}

\begin{proof}

By use of the algebraic laws of $APTC$, we have the following expansions.

\begin{eqnarray}
R_0\between S_0&\overset{\text{P1}}{=}&R_0\parallel S_0+R_0\mid S_0\nonumber\\
&\overset{\text{RDP}}{=}&(\sum_{d\in\Delta}r_{A_2}(d)\cdot R_0')\parallel(\sum_{d\in\Delta}r_{A_1}(d)T_{d0})\nonumber\\
&&+(\sum_{d\in\Delta}r_{A_2}(d)\cdot R_0')\mid(\sum_{d\in\Delta}r_{A_1}(d)T_{d0})\nonumber\\
&\overset{\text{P6,C14}}{=}&\sum_{d\in\Delta}(r_{A_2}(d)\parallel r_{A_1}(d))R_0'\between T_{d0} + \delta\cdot R_0'\between T_{d0}\nonumber\\
&\overset{\text{A6,A7}}{=}&\sum_{d\in\Delta}(r_{A_2}(d)\parallel r_{A_1}(d))R_0'\between T_{d0}\nonumber
\end{eqnarray}

\begin{eqnarray}
\partial_H(R_0\between S_0)&=&\partial_H(\sum_{d\in\Delta}(r_{A_2}(d)\parallel r_{A_1}(d))R_0'\between T_{d0})\nonumber\\
&&=\sum_{d\in\Delta}(r_{A_2}(d)\parallel r_{A_1}(d))\partial_H(R_0'\between T_{d0})\nonumber
\end{eqnarray}

Similarly, we can get the following equations.

\begin{eqnarray}
\partial_H(R_0\between S_0)&=&\sum_{d\in\Delta}(r_{A_2}(d)\parallel r_{A_1}(d))\cdot\partial_H(T_{d0}\between R_0')\nonumber\\
\partial_H(T_{d0}\between R_0')&=&c_B(d',0)\cdot(s_{C_1}(d')\parallel s_{C_2}(d'))\cdot\partial_H(U_{d0}\between Q_0)+c_B(\bot)\cdot\partial_H(U_{d0}\between Q_1)\nonumber\\
\partial_H(U_{d0}\between Q_1)&=&(c_D(1)+c_D(\bot))\cdot\partial_H(T_{d0}\between R_0')\nonumber\\
\partial_H(Q_0\between U_{d0})&=&c_D(0)\cdot\partial_H(R_1\between S_1)+c_D(\bot)\cdot\partial_H(R_1'\between T_{d0})\nonumber\\
\partial_H(R_1'\between T_{d0})&=&(c_B(d',0)+c_B(\bot))\cdot\partial_H(Q_0\between U_{d0})\nonumber\\
\partial_H(R_1\between S_1)&=&\sum_{d\in\Delta}(r_{A_2}(d)\parallel r_{A_1}(d))\cdot\partial_H(T_{d1}\between R_1')\nonumber\\
\partial_H(T_{d1}\between R_1')&=&c_B(d',1)\cdot(s_{C_1}(d')\parallel s_{C_2}(d'))\cdot\partial_H(U_{d1}\between Q_1)+c_B(\bot)\cdot\partial_H(U_{d1}\between Q_0')\nonumber\\
\partial_H(U_{d1}\between Q_0')&=&(c_D(0)+c_D(\bot))\cdot\partial_H(T_{d1}\between R_1')\nonumber\\
\partial_H(Q_1\between U_{d1})&=&c_D(1)\cdot\partial_H(R_0\between S_0)+c_D(\bot)\cdot\partial_H(R_0'\between T_{d1})\nonumber\\
\partial_H(R_0'\between T_{d1})&=&(c_B(d',1)+c_B(\bot))\cdot\partial_H(Q_1\between U_{d1})\nonumber
\end{eqnarray}

Let $\partial_H(R_0\between S_0)=\langle X_1|E\rangle$, where E is the following guarded linear recursion specification:

\begin{eqnarray}
&&\{X_1=\sum_{d\in \Delta}(r_{A_2}(d)\parallel r_{A_1}(d))\cdot X_{2d},Y_1=\sum_{d\in\Delta}(r_{A_2}(d)\parallel r_{A_1}(d))\cdot Y_{2d},\nonumber\\
&&X_{2d}=c_B(d',0)\cdot X_{4d}+c_B(\bot)\cdot X_{3d}, Y_{2d}=c_B(d',1)\cdot Y_{4d}+c_B(\bot)\cdot Y_{3d},\nonumber\\
&&X_{3d}=(c_D(1)+c_D(\bot))\cdot X_{2d}, Y_{3d}=(c_D(0)+c_D(\bot))\cdot Y_{2d},\nonumber\\
&&X_{4d}=(s_{C_1}(d')\parallel s_{C_2}(d'))\cdot X_{5d}, Y_{4d}=(s_{C_1}(d')\parallel s_{C_2}(d'))\cdot Y_{5d},\nonumber\\
&&X_{5d}=c_D(0)\cdot Y_1+c_D(\bot)\cdot X_{6d}, Y_{5d}=c_D(1)\cdot X_1+c_D(\bot)\cdot Y_{6d},\nonumber\\
&&X_{6d}=(c_B(d,0)+c_B(\bot))\cdot X_{5d}, Y_{6d}=(c_B(d,1)+c_B(\bot))\cdot Y_{5d}\nonumber\\
&&|d,d'\in\Delta\}\nonumber
\end{eqnarray}

Then we apply abstraction operator $\tau_I$ into $\langle X_1|E\rangle$.

\begin{eqnarray}
\tau_I(\langle X_1|E\rangle)
&=&\sum_{d\in\Delta}(r_{A_1}(d)\parallel r_{A_2}(d))\cdot\tau_I(\langle X_{2d}|E\rangle)\nonumber\\
&=&\sum_{d\in\Delta}(r_{A_1}(d)\parallel r_{A_2}(d))\cdot\tau_I(\langle X_{4d}|E\rangle)\nonumber\\
&=&\sum_{d,d'\in\Delta}(r_{A_1}(d)\parallel r_{A_2}(d))\cdot (s_{C_1}(d')\parallel s_{C_2}(d'))\cdot\tau_I(\langle X_{5d}|E\rangle)\nonumber\\
&=&\sum_{d,d'\in\Delta}(r_{A_1}(d)\parallel r_{A_2}(d))\cdot (s_{C_1}(d')\parallel s_{C_2}(d'))\cdot\tau_I(\langle Y_1|E\rangle)\nonumber
\end{eqnarray}

Similarly, we can get $\tau_I(\langle Y_1|E\rangle)=\sum_{d,d'\in\Delta}(r_{A_1}(d)\parallel r_{A_2}(d))\cdot (s_{C_1}(d')\parallel s_{C_2}(d'))\cdot\tau_I(\langle X_1|E\rangle)
$.

We get $\tau_I(\partial_H(R_0\between S_0))=\sum_{d,d'\in \Delta}(r_{A_1}(d)\parallel r_{A_2}(d))\cdot (s_{C_1}(d')\parallel s_{C_2}(d'))\cdot \tau_I(\partial_H(R_0\between S_0))$. So, the ABP protocol $\tau_I(\partial_H(R_0\between S_0))$ exhibits desired external behaviors.
\end{proof}

\section{Extensions}{\label{ext}}

$APTC$ also has the modularity as $ACP$, so, $APTC$ can be extended easily. By introducing new operators or new constants, $APTC$ can have more properties, modularity provides $APTC$ an elegant fashion to express a new property. In this section, we take examples of placeholder which maybe capture the nature of true concurrency, renaming operator which is used to rename the atomic events and firstly introduced by Milner in his CCS \cite{CCS}, state operator which can explicitly define states, and a more complex extension called guards which can express conditionals.

\subsection{Placeholder}\label{ph}

Through verification of ABP protocol \cite{ABP} in section \ref{app}, we see that the verification is in a structural symmetric way. Let we see the following example.

\begin{eqnarray}
(a\cdot r_b)\between w_b&=&(a\parallel w_b)\cdot r_b+\gamma(a,w_b)\cdot r_b\nonumber\\
&=&\delta\cdot r_b+\delta\cdot r_b \nonumber\\
&=&\delta+\delta \nonumber\\
&=&\delta\nonumber
\end{eqnarray}

With $\gamma(r_b,w_b)\triangleq c_b$. The communication $c_b$ does not occur and the above equation should be able to be equal to $a\cdot c_b$. How to deal this situation?

It is caused that the two communicating actions are not at the same causal depth. That is, we must pad something in hole of $(a\cdot r_b)\between([-]\cdot w_b)$ to make $r_b$ and $w_b$ in the same causal depth.

Can we pad $\tau$ into that hole? No. Because $\tau\cdot w_b\neq w_b$. We must pad something new to that hole.

\subsubsection{Transition Rules of Shadow Constant}

We introduce a constant called shadow constant $\circledS$ to act for the placeholder that we ever used to deal entanglement in quantum process algebra. The transition rule of the shadow constant $\circledS$ is shown in Table \ref{TRForShadow}. The rule say that $\circledS$ can terminate successfully without executing any action.

\begin{center}
    \begin{table}
        $$\frac{}{\circledS\rightarrow\surd}$$
        \caption{Transition rule of the shadow constant}
        \label{TRForShadow}
    \end{table}
\end{center}

\begin{theorem}[Conservativity of $APTC$ with respect to the shadow constant]
$APTC_{\tau}$ with guarded linear recursion and shadow constant is a conservative extension of $APTC_{\tau}$ with guarded linear recursion.
\end{theorem}

\begin{proof}
It follows from the following two facts (see Theorem \ref{TCE}).

\begin{enumerate}
  \item The transition rules of $APTC_{\tau}$ with guarded linear recursion in section \ref{abs} are all source-dependent;
  \item The sources of the transition rules for the shadow constant contain an occurrence of $\circledS$.
\end{enumerate}
\end{proof}

\subsubsection{Axioms for Shadow Constant}

We design the axioms for the shadow constant $\circledS$ in Table \ref{AxiomsForShadow}. And for $\circledS^e_i$, we add superscript $e$ to denote $\circledS$ is belonging to $e$ and subscript $i$ to denote that it is the $i$-th shadow of $e$. And we extend the set $\mathbb{E}$ to the set $\mathbb{E}\cup\{\tau\}\cup\{\delta\}\cup\{\circledS^{e}_i\}$.

\begin{center}
\begin{table}
  \begin{tabular}{@{}ll@{}}
\hline No. &Axiom\\
  $SC1$ & $\circledS\cdot x = x$\\
  $SC2$ & $x\cdot \circledS = x$\\
  $SC3$ & $\circledS^{e}\parallel e=e$\\
  $SC4$ & $e\parallel(\circledS^{e}\cdot y) = e\cdot y$\\
  $SC5$ & $\circledS^{e}\parallel(e\cdot y) = e\cdot y$\\
  $SC6$ & $(e\cdot x)\parallel\circledS^{e} = e\cdot x$\\
  $SC7$ & $(\circledS^{e}\cdot x)\parallel e = e\cdot x$\\
  $SC8$ & $(e\cdot x)\parallel(\circledS^{e}\cdot y) = e\cdot (x\between y)$\\
  $SC9$ & $(\circledS^{e}\cdot x)\parallel(e\cdot y) = e\cdot (x\between y)$\\
\end{tabular}
\caption{Axioms of shadow constant}
\label{AxiomsForShadow}
\end{table}
\end{center}

The mismatch of action and its shadows in parallelism will cause deadlock, that is, $e\parallel \circledS^{e'}=\delta$ with $e\neq e'$. We must make all shadows $\circledS^e_i$ are distinct, to ensure $f$ in hp-bisimulation is an isomorphism.

\begin{theorem}[Soundness of the shadow constant]\label{SShadow}
Let $x$ and $y$ be $APTC_{\tau}$ with guarded linear recursion and the shadow constant terms. If $APTC_{\tau}$ with guarded linear recursion and the shadow constant $\vdash x=y$, then
\begin{enumerate}
  \item $x\approx_{rbs} y$;
  \item $x\approx_{rbp} y$;
  \item $x\approx_{rbhp} y$.
\end{enumerate}
\end{theorem}

\begin{proof}
(1) Soundness of $APTC_{\tau}$ with guarded linear recursion and the shadow constant with respect to rooted branching step bisimulation $\approx_{rbs}$.

Since rooted branching step bisimulation $\approx_{rbs}$ is both an equivalent and a congruent relation with respect to $APTC_{\tau}$ with guarded linear recursion and the shadow constant, we only need to check if each axiom in Table \ref{AxiomsForShadow} is sound modulo rooted branching step bisimulation equivalence.

Though transition rules in Table \ref{TRForShadow} are defined in the flavor of single event, they can be modified into a step (a set of events within which each event is pairwise concurrent), we omit them. If we treat a single event as a step containing just one event, the proof of this soundness theorem does not exist any problem, so we use this way and still use the transition rules in Table \ref{AxiomsForShadow}.

The proof of soundness of $SC1-SC9$ modulo rooted branching step bisimulation is trivial, and we omit it.

(2) Soundness of $APTC_{\tau}$ with guarded linear recursion and the shadow constant with respect to rooted branching pomset bisimulation $\approx_{rbp}$.

Since rooted branching pomset bisimulation $\approx_{rbp}$ is both an equivalent and a congruent relation with respect to $APTC_{\tau}$ with guarded linear recursion and the shadow constant, we only need to check if each axiom in Table \ref{AxiomsForShadow} is sound modulo rooted branching pomset bisimulation $\approx_{rbp}$.

From the definition of rooted branching pomset bisimulation $\approx_{rbp}$ (see Definition \ref{RBPSB}), we know that rooted branching pomset bisimulation $\approx_{rbp}$ is defined by weak pomset transitions, which are labeled by pomsets with $\tau$. In a weak pomset transition, the events in the pomset are either within causality relations (defined by $\cdot$) or in concurrency (implicitly defined by $\cdot$ and $+$, and explicitly defined by $\between$), of course, they are pairwise consistent (without conflicts). In (1), we have already proven the case that all events are pairwise concurrent, so, we only need to prove the case of events in causality. Without loss of generality, we take a pomset of $P=\{e_1,e_2:e_1\cdot e_2\}$. Then the weak pomset transition labeled by the above $P$ is just composed of one single event transition labeled by $e_1$ succeeded by another single event transition labeled by $e_2$, that is, $\xRightarrow{P}=\xRightarrow{e_1}\xRightarrow{e_2}$.

Similarly to the proof of soundness of $APTC_{\tau}$ with guarded linear recursion and the shadow constant modulo rooted branching step bisimulation $\approx_{rbs}$ (1), we can prove that each axiom in Table \ref{AxiomsForShadow} is sound modulo rooted branching pomset bisimulation $\approx_{rbp}$, we omit them.

(3) Soundness of $APTC_{\tau}$ with guarded linear recursion and the shadow constant with respect to rooted branching hp-bisimulation $\approx_{rbhp}$.

Since rooted branching hp-bisimulation $\approx_{rbhp}$ is both an equivalent and a congruent relation with respect to $APTC_{\tau}$ with guarded linear recursion and the shadow constant, we only need to check if each axiom in Table \ref{AxiomsForShadow} is sound modulo rooted branching hp-bisimulation $\approx_{rbhp}$.

From the definition of rooted branching hp-bisimulation $\approx_{rbhp}$ (see Definition \ref{RBHHPB}), we know that rooted branching hp-bisimulation $\approx_{rbhp}$ is defined on the weakly posetal product $(C_1,f,C_2),f:\hat{C_1}\rightarrow \hat{C_2}\textrm{ isomorphism}$. Two process terms $s$ related to $C_1$ and $t$ related to $C_2$, and $f:\hat{C_1}\rightarrow \hat{C_2}\textrm{ isomorphism}$. Initially, $(C_1,f,C_2)=(\emptyset,\emptyset,\emptyset)$, and $(\emptyset,\emptyset,\emptyset)\in\approx_{rbhp}$. When $s\xrightarrow{e}s'$ ($C_1\xrightarrow{e}C_1'$), there will be $t\xRightarrow{e}t'$ ($C_2\xRightarrow{e}C_2'$), and we define $f'=f[e\mapsto e]$. Then, if $(C_1,f,C_2)\in\approx_{rbhp}$, then $(C_1',f',C_2')\in\approx_{rbhp}$.

Similarly to the proof of soundness of $APTC_{\tau}$ with guarded linear recursion and the shadow constant modulo rooted branching pomset bisimulation equivalence (2), we can prove that each axiom in Table \ref{AxiomsForShadow} is sound modulo rooted branching hp-bisimulation equivalence, we just need additionally to check the above conditions on rooted branching hp-bisimulation, we omit them.
\end{proof}

\begin{theorem}[Completeness of the shadow constant]\label{CRenaming}
Let $p$ and $q$ be closed $APTC_{\tau}$ with guarded linear recursion and $CFAR$ and the shadow constant terms, then,
\begin{enumerate}
  \item if $p\approx_{rbs} q$ then $p=q$;
  \item if $p\approx_{rbp} q$ then $p=q$;
  \item if $p\approx_{rbhp} q$ then $p=q$.
\end{enumerate}
\end{theorem}

\begin{proof}
(1) For the case of rooted branching step bisimulation, the proof is following.

Firstly, in the proof the Theorem \ref{CCFAR}, we know that each process term $p$ in $APTC_{\tau}$ with guarded linear recursion is equal to a process term $\langle X_1|E\rangle$ with $E$ a guarded linear recursive specification. And we prove if $\langle X_1|E_1\rangle\approx_{rbs}\langle Y_1|E_2\rangle$, then $\langle X_1|E_1\rangle=\langle Y_1|E_2\rangle$

There are no necessary to induct with respect to the structure of process term $p$, because there are no new cases. The only new situation is that now the set $\mathbb{E}$ contains some new constants $\circledS^e_i$ for $e\in\mathbb{E}$ and $i\in\mathbb{N}$. Since $\circledS^e_i$ does not do anything, so, naturedly, if $\langle X_1|E_1\rangle\approx_{rbs}\langle Y_1|E_2\rangle$, the only thing is that we should prevent $\circledS$-loops in the recursion in a specific way as same as preventing $\tau$-loops, then $\langle X_1|E_1\rangle=\langle Y_1|E_2\rangle$.

(2) For the case of rooted branching pomset bisimulation, it can be proven similarly to (1), we omit it.

(3) For the case of rooted branching hp-bisimulation, it can be proven similarly to (1), we omit it.
\end{proof}

\subsubsection{Some Discussions on True Concurrency}

With the shadow constant, we have

\begin{eqnarray}
(a\cdot r_b)\between w_b&=&(a\cdot r_b) \between (\circledS^a_1\cdot w_b) \nonumber\\
&=&a\cdot c_b\nonumber
\end{eqnarray}

with $\gamma(r_b,w_b)\triangleq c_b$.

And we see the following example:

\begin{eqnarray}
a\between b&=&a\parallel b+a\mid b \nonumber\\
&=&a\parallel b + a\parallel b + a\parallel b +a\mid b \nonumber\\
&=&a\parallel (\circledS^a_1\cdot b) + (\circledS^b_1\cdot a)\parallel b+a\parallel b +a\mid b \nonumber\\
&=&(a\parallel\circledS^a_1)\cdot b + (\circledS^b_1\parallel b)\cdot a+a\parallel b +a\mid b \nonumber\\
&=&a\cdot b+b\cdot a+a\parallel b + a\mid b\nonumber
\end{eqnarray}

What do we see? Yes. The parallelism contains both interleaving and true concurrency. This may be why true concurrency is called \emph{\textbf{true} concurrency}.

\subsubsection{Verification of Traditional Alternating Bit Protocol}

With the help of shadow constant, now we can verify the traditional alternating bit protocol (ABP) \cite{ABP}.

The ABP protocol is used to ensure successful transmission of data through a corrupted channel. This success is based on the assumption that data can be resent an unlimited number of times, which is illustrated in Fig.\ref{ABP2}, we alter it into the true concurrency situation.

\begin{enumerate}
  \item Data elements $d_1,d_2,d_3,\cdots$ from a finite set $\Delta$ are communicated between a Sender and a Receiver.
  \item If the Sender reads a datum from channel $A$.
  \item The Sender processes the data in $\Delta$, formes new data, and sends them to the Receiver through channel $B$.
  \item And the Receiver sends the datum into channel $C$.
  \item If channel $B$ is corrupted, the message communicated through $B$ can be turn into an error message $\bot$.
  \item Every time the Receiver receives a message via channel $B$, it sends an acknowledgement to the Sender via channel $D$, which is also corrupted.
\end{enumerate}

\begin{figure}
    \centering
    \includegraphics{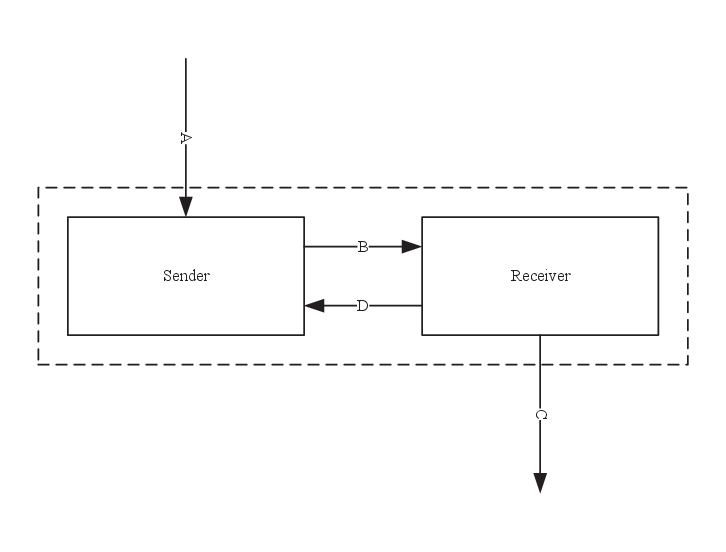}
    \caption{Alternating bit protocol}
    \label{ABP2}
\end{figure}

The Sender attaches a bit 0 to data elements $d_{2k-1}$ and a bit 1 to data elements $d_{2k}$, when they are sent into channel $B$. When the Receiver reads a datum, it sends back the attached bit via channel $D$. If the Receiver receives a corrupted message, then it sends back the previous acknowledgement to the Sender.

Then the state transition of the Sender can be described by $APTC$ as follows.

\begin{eqnarray}
&&S_b=\sum_{d\in\Delta}r_{A}(d)\cdot T_{db}\nonumber\\
&&T_{db}=(\sum_{d'\in\Delta}(s_B(d',b)\cdot \circledS^{s_{C}(d')})+s_B(\bot))\cdot U_{db}\nonumber\\
&&U_{db}=r_D(b)\cdot S_{1-b}+(r_D(1-b)+r_D(\bot))\cdot T_{db}\nonumber
\end{eqnarray}

where $s_B$ denotes sending data through channel $B$, $r_D$ denotes receiving data through channel $D$, similarly, $r_{A}$ means receiving data via channel $A$, $\circledS^{s_{C}(d')}$ denotes the shadow of $s_{C}(d')$.

And the state transition of the Receiver can be described by $APTC$ as follows.

\begin{eqnarray}
&&R_b=\sum_{d\in\Delta}\circledS^{r_{A}(d)}\cdot R_b'\nonumber\\
&&R_b'=\sum_{d'\in\Delta}\{r_B(d',b)\cdot s_{C}(d')\cdot Q_b+r_B(d',1-b)\cdot Q_{1-b}\}+r_B(\bot)\cdot Q_{1-b}\nonumber\\
&&Q_b=(s_D(b)+s_D(\bot))\cdot R_{1-b}\nonumber
\end{eqnarray}

where $\circledS^{r_{A}(d)}$ denotes the shadow of $r_{A}(d)$, $r_B$ denotes receiving data via channel $B$, $s_{C}$ denotes sending data via channel $C$, $s_D$ denotes sending data via channel $D$, and $b\in\{0,1\}$.

The send action and receive action of the same data through the same channel can communicate each other, otherwise, a deadlock $\delta$ will be caused. We define the following communication functions.

\begin{eqnarray}
&&\gamma(s_B(d',b),r_B(d',b))\triangleq c_B(d',b)\nonumber\\
&&\gamma(s_B(\bot),r_B(\bot))\triangleq c_B(\bot)\nonumber\\
&&\gamma(s_D(b),r_D(b))\triangleq c_D(b)\nonumber\\
&&\gamma(s_D(\bot),r_D(\bot))\triangleq c_D(\bot)\nonumber
\end{eqnarray}

Let $R_0$ and $S_0$ be in parallel, then the system $R_0S_0$ can be represented by the following process term.

$$\tau_I(\partial_H(\Theta(R_0\between S_0)))=\tau_I(\partial_H(R_0\between S_0))$$

where $H=\{s_B(d',b),r_B(d',b),s_D(b),r_D(b)|d'\in\Delta,b\in\{0,1\}\}\\
\{s_B(\bot),r_B(\bot),s_D(\bot),r_D(\bot)\}$

$I=\{c_B(d',b),c_D(b)|d'\in\Delta,b\in\{0,1\}\}\cup\{c_B(\bot),c_D(\bot)\}$.

Then we get the following conclusion.

\begin{theorem}[Correctness of the ABP protocol]
The ABP protocol $\tau_I(\partial_H(R_0\between S_0))$ exhibits desired external behaviors.
\end{theorem}

\begin{proof}

Similarly, we can get $\tau_I(\langle X_1|E\rangle)=\sum_{d,d'\in\Delta}r_{A}(d)\cdot s_{C}(d')\cdot\tau_I(\langle Y_1|E\rangle)$ and $\tau_I(\langle Y_1|E\rangle)=\sum_{d,d'\in\Delta}r_{A}(d)\cdot s_{C}(d')\cdot\tau_I(\langle X_1|E\rangle)$.

So, the ABP protocol $\tau_I(\partial_H(R_0\between S_0))$ exhibits desired external behaviors.
\end{proof}

\subsection{Renaming}

\subsubsection{Transition Rules of Renaming Operator}

Renaming operator $\rho_f(t)$ renames all actions in process term $t$, and assumes a renaming function $f:\mathbb{E}\cup\{\tau\}\rightarrow \mathbb{E}\cup\{\tau\}$ with $f(\tau)\triangleq\tau$, which is expressed by the following two transition rules in Table \ref{TRForRenaming}.

\begin{center}
    \begin{table}
        $$\frac{x\xrightarrow{e}\surd}{\rho_f(x)\xrightarrow{f(e)}\surd}
        \quad\frac{x\xrightarrow{e}x'}{\rho_f(x)\xrightarrow{f(e)}\rho_f(x')}$$
        \caption{Transition rule of the renaming operator}
        \label{TRForRenaming}
    \end{table}
\end{center}

\begin{theorem}[Conservativity of $APTC$ with respect to the renaming operator]
$APTC_{\tau}$ with guarded linear recursion and renaming operator is a conservative extension of $APTC_{\tau}$ with guarded linear recursion.
\end{theorem}

\begin{proof}
It follows from the following two facts (see Theorem \ref{TCE}).

\begin{enumerate}
  \item The transition rules of $APTC_{\tau}$ with guarded linear recursion in section \ref{abs} are all source-dependent;
  \item The sources of the transition rules for the renaming operator contain an occurrence of $\rho_f$.
\end{enumerate}
\end{proof}

\begin{theorem}[Congruence theorem of the renaming operator]
Rooted branching truly concurrent bisimulation equivalences $\approx_{rbp}$, $\approx_{rbs}$ and $\approx_{rbhp}$ are all congruences with respect to $APTC_{\tau}$ with guarded linear recursion and the renaming operator.
\end{theorem}

\begin{proof}

(1) Case rooted branching pomset bisimulation equivalence $\approx_{rbp}$.

Let $x$ and $y$ be $APTC_{\tau}$ with guarded linear recursion and the renaming operator processes, and $x\approx_{rbp} y$, it is sufficient to prove that $\rho_f(x)\approx_{rbp} \rho_f(y)$.

By the transition rules for operator $\rho_f$ in Table \ref{TRForRenaming}, we can get

$$\rho_f(x)\xrightarrow{f(X)} \surd \quad \rho_f(y)\xrightarrow{f(Y)} \surd$$

with $X\subseteq x$, $Y\subseteq y$, and $X\sim Y$.

Or, we can get

$$\rho_f(x)\xrightarrow{f(X)} \rho_f(x') \quad \rho_f(y)\xrightarrow{f(Y)} \rho_f(y')$$

with $X\subseteq x$, $Y\subseteq y$, and $X\sim Y$ and the hypothesis $\rho_f(x')\approx_{rbp}\rho_f(y')$.

So, we get $\rho_f(x)\approx_{rbp} \rho_f(y)$, as desired

(2) The cases of rooted branching step bisimulation $\approx_{rbs}$, rooted branching hp-bisimulation $\approx_{rbhp}$ can be proven similarly, we omit them.
\end{proof}

\subsubsection{Axioms for Renaming Operators}

We design the axioms for the renaming operator $\rho_f$ in Table \ref{AxiomsForRenaming}.

\begin{center}
\begin{table}
  \begin{tabular}{@{}ll@{}}
\hline No. &Axiom\\
  $RN1$ & $\rho_f(e)=f(e)$\\
  $RN2$ & $\rho_f(\delta)=\delta$\\
  $RN3$ & $\rho_f(x+y)=\rho_f(x)+\rho_f(y)$\\
  $RN4$ & $\rho_f(x\cdot y)=\rho_f(x)\cdot\rho_f(y)$\\
  $RN5$ & $\rho_f(x\parallel y)=\rho_f(x)\parallel\rho_f(y)$\\
\end{tabular}
\caption{Axioms of renaming operator}
\label{AxiomsForRenaming}
\end{table}
\end{center}

$RN1-RN2$ are the defining equations for the renaming operator $\rho_f$; $RN3-RN5$ say that in $\rho_f(t)$, the labels of all transitions of $t$ are renamed by means of the mapping $f$.

\begin{theorem}[Soundness of the renaming operator]\label{SRenaming}
Let $x$ and $y$ be $APTC_{\tau}$ with guarded linear recursion and the renaming operator terms. If $APTC_{\tau}$ with guarded linear recursion and the renaming operator $\vdash x=y$, then
\begin{enumerate}
  \item $x\approx_{rbs} y$;
  \item $x\approx_{rbp} y$;
  \item $x\approx_{rbhp} y$.
\end{enumerate}
\end{theorem}

\begin{proof}
(1) Soundness of $APTC_{\tau}$ with guarded linear recursion and the renaming operator with respect to rooted branching step bisimulation $\approx_{rbs}$.

Since rooted branching step bisimulation $\approx_{rbs}$ is both an equivalent and a congruent relation with respect to $APTC_{\tau}$ with guarded linear recursion and the renaming operator, we only need to check if each axiom in Table \ref{AxiomsForRenaming} is sound modulo rooted branching step bisimulation equivalence.

Though transition rules in Table \ref{TRForRenaming} are defined in the flavor of single event, they can be modified into a step (a set of events within which each event is pairwise concurrent), we omit them. If we treat a single event as a step containing just one event, the proof of this soundness theorem does not exist any problem, so we use this way and still use the transition rules in Table \ref{AxiomsForRenaming}.

We only prove soundness of the non-trivial axioms $RN3-RN5$, and omit the defining axioms $RN1-RN2$.

\begin{itemize}
  \item \textbf{Axiom $RN3$}. Let $p,q$ be $APTC_{\tau}$ with guarded linear recursion and the renaming operator processes, and $\rho_f(p+ q)=\rho_f(p)+\rho_f(q)$, it is sufficient to prove that $\rho_f(p+ q)\approx_{rbs} \rho_f(p)+\rho_f(q)$. By the transition rules for operator $+$ in Table \ref{STRForBATC} and $\rho_f$ in Table \ref{TRForRenaming}, we get

      $$\frac{p\xrightarrow{e_1}\surd}{\rho_f(p+ q)\xrightarrow{f(e_1)}\surd}
      \quad\frac{p\xrightarrow{e_1}\surd}{\rho_f(p)+ \rho_f(q)\xrightarrow{f(e_1)}\surd}$$

      $$\frac{q\xrightarrow{e_2}\surd}{\rho_f(p+ q)\xrightarrow{f(e_2)}\surd}
      \quad\frac{q\xrightarrow{e_2}\surd}{\rho_f(p)+ \rho_f(q)\xrightarrow{f(e_2)}\surd}$$

      $$\frac{p\xrightarrow{e_1}p'}{\rho_f(p+ q)\xrightarrow{f(e_1)}\rho_f(p')}
      \quad\frac{p\xrightarrow{e_1}p'}{\rho_f(p)+ \rho_f(q)\xrightarrow{f(e_1)}\rho_f(p')}$$

      $$\frac{q\xrightarrow{e_2}q'}{\rho_f(p+ q)\xrightarrow{f(e_2)}\rho_f(q')}
      \quad\frac{q\xrightarrow{e_2}q'}{\rho_f(p)+ \rho_f(q)\xrightarrow{f(e_2)}\rho_f(q')}$$

      So, $\rho_f(p+ q)\approx_{rbs} \rho_f(p)+\rho_f(q)$, as desired.
  \item \textbf{Axiom $RN4$}. Let $p,q$ be $APTC_{\tau}$ with guarded linear recursion and the renaming operator processes, and $\rho_f(p\cdot q)=\rho_f(p)\cdot\rho_f(q)$, it is sufficient to prove that $\rho_f(p\cdot q)\approx_{rbs} \rho_f(p)\cdot\rho_f(q)$. By the transition rules for operator $\cdot$ in Table \ref{STRForBATC} and $\rho_f$ in Table \ref{TRForRenaming}, we get

      $$\frac{p\xrightarrow{e_1}\surd}{\rho_f(p\cdot q)\xrightarrow{f(e_1)}\rho_f(q)}
      \quad\frac{p\xrightarrow{e_1}\surd}{\rho_f(p)\cdot \rho_f(q)\xrightarrow{f(e_1)}\rho_f(q)}$$

      $$\frac{p\xrightarrow{e_1}p'}{\rho_f(p\cdot q)\xrightarrow{f(e_1)}\rho_f(p'\cdot q)}
      \quad\frac{p\xrightarrow{e_1}p'}{\rho_f(p)\cdot \rho_f(q)\xrightarrow{f(e_1)}\rho_f(p')\cdot\rho_f(q)}$$

      So, with the assumption $\rho_f(p'\cdot q)=\rho_f(p')\cdot\rho_f(q)$, $\rho_f(p\cdot q)\approx_{rbs}\rho_f(p)\cdot\rho_f(q)$, as desired.
  \item \textbf{Axiom $RN5$}. Let $p,q$ be $APTC_{\tau}$ with guarded linear recursion and the renaming operator processes, and $\rho_f(p\parallel q)=\rho_f(p)\parallel\rho_f(q)$, it is sufficient to prove that $\rho_f(p\parallel q)\approx_{rbs} \rho_f(p)\parallel\rho_f(q)$. By the transition rules for operator $\parallel$ in Table \ref{TRForParallel} and $\rho_f$ in Table \ref{TRForRenaming}, we get

      $$\frac{p\xrightarrow{e_1}\surd\quad q\xrightarrow{e_2}\surd}{\rho_f(p\parallel q)\xrightarrow{\{f(e_1),f(e_2)\}}\surd}
      \quad\frac{p\xrightarrow{e_1}\surd\quad q\xrightarrow{e_2}\surd}{\rho_f(p)\parallel \rho_f(q)\xrightarrow{\{f(e_1),f(e_2)\}}\surd}$$

      $$\frac{p\xrightarrow{e_1}p'\quad q\xrightarrow{e_2}\surd}{\rho_f(p\parallel q)\xrightarrow{\{f(e_1),f(e_2)\}}\rho_f(p')}
      \quad\frac{p\xrightarrow{e_1}p'\quad q\xrightarrow{e_2}\surd}{\rho_f(p)\parallel \rho_f(q)\xrightarrow{\{f(e_1),f(e_2)\}}\rho_f(p')}$$

      $$\frac{p\xrightarrow{e_1}\surd\quad q\xrightarrow{e_2}q'}{\rho_f(p\parallel q)\xrightarrow{\{f(e_1),f(e_2)\}}\rho_f(q')}
      \quad\frac{p\xrightarrow{e_1}\surd\quad q\xrightarrow{e_2}q'}{\rho_f(p)\parallel \rho_f(q)\xrightarrow{\{f(e_1),f(e_2)\}}\rho_f(q')}$$

      $$\frac{p\xrightarrow{e_1}p'\quad q\xrightarrow{e_2}q'}{\rho_f(p\parallel q)\xrightarrow{\{f(e_1),f(e_2)\}}\rho_f(p'\between q')}
      \quad\frac{p\xrightarrow{e_1}p'\quad q\xrightarrow{e_2}q'}{\rho_f(p)\parallel \rho_f(q)\xrightarrow{\{f(e_1),f(e_2)\}}\rho_f(p')\between\rho_f(q')}$$

      So, with the assumption $\rho_f(p'\between q')=\rho_f(p')\between\rho_f(q')$, $\rho_f(p\parallel q)\approx_{rbs} \rho_f(p)\parallel\rho_f(q)$, as desired.
\end{itemize}

(2) Soundness of $APTC_{\tau}$ with guarded linear recursion and the renaming operator with respect to rooted branching pomset bisimulation $\approx_{rbp}$.

Since rooted branching pomset bisimulation $\approx_{rbp}$ is both an equivalent and a congruent relation with respect to $APTC_{\tau}$ with guarded linear recursion and the renaming operator, we only need to check if each axiom in Table \ref{AxiomsForRenaming} is sound modulo rooted branching pomset bisimulation $\approx_{rbp}$.

From the definition of rooted branching pomset bisimulation $\approx_{rbp}$ (see Definition \ref{RBPSB}), we know that rooted branching pomset bisimulation $\approx_{rbp}$ is defined by weak pomset transitions, which are labeled by pomsets with $\tau$. In a weak pomset transition, the events in the pomset are either within causality relations (defined by $\cdot$) or in concurrency (implicitly defined by $\cdot$ and $+$, and explicitly defined by $\between$), of course, they are pairwise consistent (without conflicts). In (1), we have already proven the case that all events are pairwise concurrent, so, we only need to prove the case of events in causality. Without loss of generality, we take a pomset of $P=\{e_1,e_2:e_1\cdot e_2\}$. Then the weak pomset transition labeled by the above $P$ is just composed of one single event transition labeled by $e_1$ succeeded by another single event transition labeled by $e_2$, that is, $\xRightarrow{P}=\xRightarrow{e_1}\xRightarrow{e_2}$.

Similarly to the proof of soundness of $APTC_{\tau}$ with guarded linear recursion and the renaming operator modulo rooted branching step bisimulation $\approx_{rbs}$ (1), we can prove that each axiom in Table \ref{AxiomsForRenaming} is sound modulo rooted branching pomset bisimulation $\approx_{rbp}$, we omit them.

(3) Soundness of $APTC_{\tau}$ with guarded linear recursion and the renaming operator with respect to rooted branching hp-bisimulation $\approx_{rbhp}$.

Since rooted branching hp-bisimulation $\approx_{rbhp}$ is both an equivalent and a congruent relation with respect to $APTC_{\tau}$ with guarded linear recursion and the renaming operator, we only need to check if each axiom in Table \ref{AxiomsForRenaming} is sound modulo rooted branching hp-bisimulation $\approx_{rbhp}$.

From the definition of rooted branching hp-bisimulation $\approx_{rbhp}$ (see Definition \ref{RBHHPB}), we know that rooted branching hp-bisimulation $\approx_{rbhp}$ is defined on the weakly posetal product $(C_1,f,C_2),f:\hat{C_1}\rightarrow \hat{C_2}\textrm{ isomorphism}$. Two process terms $s$ related to $C_1$ and $t$ related to $C_2$, and $f:\hat{C_1}\rightarrow \hat{C_2}\textrm{ isomorphism}$. Initially, $(C_1,f,C_2)=(\emptyset,\emptyset,\emptyset)$, and $(\emptyset,\emptyset,\emptyset)\in\approx_{rbhp}$. When $s\xrightarrow{e}s'$ ($C_1\xrightarrow{e}C_1'$), there will be $t\xRightarrow{e}t'$ ($C_2\xRightarrow{e}C_2'$), and we define $f'=f[e\mapsto e]$. Then, if $(C_1,f,C_2)\in\approx_{rbhp}$, then $(C_1',f',C_2')\in\approx_{rbhp}$.

Similarly to the proof of soundness of $APTC_{\tau}$ with guarded linear recursion and the renaming operator modulo rooted branching pomset bisimulation equivalence (2), we can prove that each axiom in Table \ref{AxiomsForRenaming} is sound modulo rooted branching hp-bisimulation equivalence, we just need additionally to check the above conditions on rooted branching hp-bisimulation, we omit them.
\end{proof}

\begin{theorem}[Completeness of the renaming operator]\label{CRenaming}
Let $p$ and $q$ be closed $APTC_{\tau}$ with guarded linear recursion and $CFAR$ and the renaming operator terms, then,
\begin{enumerate}
  \item if $p\approx_{rbs} q$ then $p=q$;
  \item if $p\approx_{rbp} q$ then $p=q$;
  \item if $p\approx_{rbhp} q$ then $p=q$.
\end{enumerate}
\end{theorem}

\begin{proof}
(1) For the case of rooted branching step bisimulation, the proof is following.

Firstly, in the proof the Theorem \ref{CCFAR}, we know that each process term $p$ in $APTC_{\tau}$ with guarded linear recursion is equal to a process term $\langle X_1|E\rangle$ with $E$ a guarded linear recursive specification. And we prove if $\langle X_1|E_1\rangle\approx_{rbs}\langle Y_1|E_2\rangle$, then $\langle X_1|E_1\rangle=\langle Y_1|E_2\rangle$

Structural induction with respect to process term $p$ can be applied. The only new case (where $RN1-RN5$ are needed) is $p \equiv \rho_f(q)$. First assuming $q=\langle X_1|E\rangle$ with a guarded linear recursive specification $E$, we prove the case of $p=\rho_f(\langle X_1|E\rangle)$. Let $E$ consist of guarded linear recursive equations

$$X_i=(a_{1i1}\parallel\cdots\parallel a_{k_{i1}i1})X_{i1}+...+(a_{1im_i}\parallel\cdots\parallel a_{k_{im_i}im_i})X_{im_i}+b_{1i1}\parallel\cdots\parallel b_{l_{i1}i1}+...+b_{1im_i}\parallel\cdots\parallel b_{l_{im_i}im_i}$$

for $i\in{1,...,n}$. Let $F$ consist of guarded linear recursive equations

$Y_i=(f(a_{1i1})\parallel\cdots\parallel f(a_{k_{i1}i1}))Y_{i1}+...+(f(a_{1im_i})\parallel\cdots\parallel f(a_{k_{im_i}im_i}))Y_{im_i}$

$+f(b_{1i1})\parallel\cdots\parallel f(b_{l_{i1}i1})+...+f(b_{1im_i})\parallel\cdots\parallel f(b_{l_{im_i}im_i})$

for $i\in{1,...,n}$.

\begin{eqnarray}
&&\rho_f(\langle X_i|E\rangle) \nonumber\\
&\overset{\text{RDP}}{=}&\rho_f((a_{1i1}\parallel\cdots\parallel a_{k_{i1}i1})X_{i1}+...+(a_{1im_i}\parallel\cdots\parallel a_{k_{im_i}im_i})X_{im_i}+b_{1i1}\parallel\cdots\parallel b_{l_{i1}i1}+...+b_{1im_i}\parallel\cdots\parallel b_{l_{im_i}im_i}) \nonumber\\
&\overset{\text{RN1-RN5}}{=}&(f(a_{1i1})\parallel\cdots\parallel f(a_{k_{i1}i1}))\rho_f(X_{i1})+...+(f(a_{1im_i})\parallel\cdots\parallel f(a_{k_{im_i}im_i}))\rho_f(X_{im_i})\nonumber\\
&&+f(b_{1i1})\parallel\cdots\parallel f(b_{l_{i1}i1})+...+f(b_{1im_i})\parallel\cdots\parallel f(b_{l_{im_i}im_i}) \nonumber
\end{eqnarray}

Replacing $Y_i$ by $\rho_f(\langle X_i|E\rangle)$ for $i\in\{1,...,n\}$ is a solution for $F$. So by $RSP$, $\rho_f(\langle X_1|E\rangle)=\langle Y_1|F\rangle$, as desired.

(2) For the case of rooted branching pomset bisimulation, it can be proven similarly to (1), we omit it.

(3) For the case of rooted branching hp-bisimulation, it can be proven similarly to (1), we omit it.
\end{proof}

\subsection{State Operator for APTC}{\label{so}}

\subsubsection{Transition Rules of State Operator}

State operator permits explicitly to describe states, where $S$ denotes a finite set of states, $action(s,e)$ denotes the visible behavior of $e$ in state $s$ with $action:S\times \mathbb{E}\rightarrow \mathbb{E}$, $effect(s,e)$ represents the state that results if $e$ is executed in $s$ with $effect:S\times \mathbb{E}\rightarrow S$.  State operator $\lambda_s(t)$ which denotes process term $t$ in $s$, is expressed by the following transition rules in Table \ref{TRForState}. Note that $action$ and $effect$ are extended to $\mathbb{E}\cup\{\tau\}$ by defining $action(s,\tau)\triangleq\tau$ and $effect(s,\tau)\triangleq s$. We use $e_1\%e_2$ to denote that $e_1$ and $e_2$ are in race condition.

\begin{center}
    \begin{table}
        $$\frac{x\xrightarrow{e}\surd}{\lambda_s(x)\xrightarrow{action(s,e)}\surd}
        \quad\frac{x\xrightarrow{e}x'}{\lambda_s(x)\xrightarrow{action(s,e)}\lambda_{effect(s,e)}(x')}$$

        $$\frac{x\xrightarrow{e_1}\surd\quad y\xnrightarrow{e_2}\quad(e_1\%e_2)}{\lambda_s(x\parallel y)\xrightarrow{action(s,e_1)}\lambda_{effect(s,e_1)}(y)} \quad\frac{x\xrightarrow{e_1}x'\quad y\xnrightarrow{e_2}\quad(e_1\%e_2)}{\lambda_s(x\parallel y)\xrightarrow{action(s,e_1)}\lambda_{effect(s,e_1)}(x'\between y)}$$

        $$\frac{x\xnrightarrow{e_1}\quad y\xrightarrow{e_2}\surd\quad(e_1\%e_2)}{\lambda_s(x\parallel y)\xrightarrow{action(s,e_2)}\lambda_{effect(s,e_2)}(x)} \quad\frac{x\xnrightarrow{e_1}\quad y\xrightarrow{e_2}y'\quad(e_1\%e_2)}{\lambda_s(x\parallel y)\xrightarrow{action(s,e_2)}\lambda_{effect(s,e_2)}(x\between y')}$$

        $$\frac{x\xrightarrow{e_1}\surd\quad y\xrightarrow{e_2}\surd}{\lambda_s(x\parallel y)\xrightarrow{\{action(s,e_1),action(s,e_2)\}}\surd}$$

        $$\frac{x\xrightarrow{e_1}x'\quad y\xrightarrow{e_2}\surd}{\lambda_s(x\parallel y)\xrightarrow{\{action(s,e_1),action(s,e_2)\}}\lambda_{effect(s,e_1)\cup effect(s,e_2)}(x')}$$

        $$\frac{x\xrightarrow{e_1}\surd\quad y\xrightarrow{e_2}y'}{\lambda_s(x\parallel y)\xrightarrow{\{action(s,e_1),action(s,e_2)\}}\lambda_{effect(s,e_1)\cup effect(s,e_2)}(y')}$$

        $$\frac{x\xrightarrow{e_1}x'\quad y\xrightarrow{e_2}y'}{\lambda_s(x\parallel y)\xrightarrow{\{action(s,e_1),action(s,e_2)\}}\lambda_{effect(s,e_1)\cup effect(s,e_2)}(x'\between y')}$$
        \caption{Transition rule of the state operator}
        \label{TRForState}
    \end{table}
\end{center}

\begin{theorem}[Conservativity of $APTC$ with respect to the state operator]
$APTC_{\tau}$ with guarded linear recursion and state operator is a conservative extension of $APTC_{\tau}$ with guarded linear recursion.
\end{theorem}

\begin{proof}
It follows from the following two facts.

\begin{enumerate}
  \item The transition rules of $APTC_{\tau}$ with guarded linear recursion are all source-dependent;
  \item The sources of the transition rules for the state operator contain an occurrence of $\lambda_s$.
\end{enumerate}
\end{proof}

\begin{theorem}[Congruence theorem of the state operator]
Rooted branching truly concurrent bisimulation equivalences $\approx_{rbp}$, $\approx_{rbs}$ and $\approx_{rbhp}$ are all congruences with respect to $APTC_{\tau}$ with guarded linear recursion and the state operator.
\end{theorem}

\begin{proof}

(1) Case rooted branching pomset bisimulation equivalence $\approx_{rbp}$.

Let $x$ and $y$ be $APTC_{\tau}$ with guarded linear recursion and the state operator processes, and $x\approx_{rbp} y$, it is sufficient to prove that $\lambda_s(x)\approx_{rbp} \lambda_s(y)$.

By the transition rules for operator $\lambda_s$ in Table \ref{TRForState}, we can get

$$\lambda_s(x)\xrightarrow{action(s,X)} \surd \quad \lambda_s(y)\xrightarrow{action(s,Y)} \surd$$

with $X\subseteq x$, $Y\subseteq y$, and $X\sim Y$.

Or, we can get

$$\lambda_s(x)\xrightarrow{action(s,X)} \lambda_{effect(s,X)}(x') \quad \lambda_s(y)\xrightarrow{action(s,Y)} \lambda_{effect(s,Y)}(y')$$

with $X\subseteq x$, $Y\subseteq y$, and $X\sim Y$ and the hypothesis $\lambda_{effect(s,X)}(x')\approx_{rbp}\lambda_{effect(s,Y)}(y')$.

So, we get $\lambda_s(x)\approx_{rbp} \lambda_s(y)$, as desired

(2) The cases of rooted branching step bisimulation $\approx_{rbs}$, rooted branching hp-bisimulation $\approx_{rbhp}$ can be proven similarly, we omit them.
\end{proof}

\subsubsection{Axioms for State Operators}

We design the axioms for the state operator $\lambda_s$ in Table \ref{AxiomsForState}.

\begin{center}
\begin{table}
  \begin{tabular}{@{}ll@{}}
\hline No. &Axiom\\
  $SO1$ & $\lambda_s(e)=action(s,e)$\\
  $SO2$ & $\lambda_s(\delta)=\delta$\\
  $SO3$ & $\lambda_s(x+y)=\lambda_s(x)+\lambda_s(y)$\\
  $SO4$ & $\lambda_s(e\cdot y)=action(s,e)\cdot\lambda_{effect(s,e)}(y)$\\
  $SO5$ & $\lambda_s(x\parallel y)=\lambda_s(x)\parallel\lambda_s(y)$\\
\end{tabular}
\caption{Axioms of state operator}
\label{AxiomsForState}
\end{table}
\end{center}

\begin{theorem}[Soundness of the state operator]\label{SState}
Let $x$ and $y$ be $APTC_{\tau}$ with guarded linear recursion and the state operator terms. If $APTC_{\tau}$ with guarded linear recursion and the state operator $\vdash x=y$, then
\begin{enumerate}
  \item $x\approx_{rbs} y$;
  \item $x\approx_{rbp} y$;
  \item $x\approx_{rbhp} y$.
\end{enumerate}
\end{theorem}

\begin{proof}
(1) Soundness of $APTC_{\tau}$ with guarded linear recursion and the state operator with respect to rooted branching step bisimulation $\approx_{rbs}$.

Since rooted branching step bisimulation $\approx_{rbs}$ is both an equivalent and a congruent relation with respect to $APTC_{\tau}$ with guarded linear recursion and the state operator, we only need to check if each axiom in Table \ref{AxiomsForState} is sound modulo rooted branching step bisimulation equivalence.

Though transition rules in Table \ref{TRForState} are defined in the flavor of single event, they can be modified into a step (a set of events within which each event is pairwise concurrent), we omit them. If we treat a single event as a step containing just one event, the proof of this soundness theorem does not exist any problem, so we use this way and still use the transition rules in Table \ref{AxiomsForState}.

We only prove soundness of the non-trivial axioms $SO3-SO5$, and omit the defining axioms $SO1-SO2$.

\begin{itemize}
  \item \textbf{Axiom $SO3$}. Let $p,q$ be $APTC_{\tau}$ with guarded linear recursion and the state operator processes, and $\lambda_s(p+ q)=\lambda_s(p)+\lambda_s(q)$, it is sufficient to prove that $\lambda_s(p+ q)\approx_{rbs} \lambda_s(p)+\lambda_s(q)$. By the transition rules for operator $+$ and $\lambda_s$ in Table \ref{TRForState}, we get

      $$\frac{p\xrightarrow{e_1}\surd}{\lambda_s(p+ q)\xrightarrow{action(s,e_1)}\surd}
      \quad\frac{p\xrightarrow{e_1}\surd}{\lambda_s(p)+ \lambda_s(q)\xrightarrow{action(s,e_1)}\surd}$$

      $$\frac{q\xrightarrow{e_2}\surd}{\lambda_s(p+ q)\xrightarrow{action(s,e_2)}\surd}
      \quad\frac{q\xrightarrow{e_2}\surd}{\lambda_s(p)+ \lambda_s(q)\xrightarrow{action(s,e_2)}\surd}$$

      $$\frac{p\xrightarrow{e_1}p'}{\lambda_s(p+ q)\xrightarrow{action(s,e_1)}\lambda_{effect(s,e_1)}(p')}
      \quad\frac{p\xrightarrow{e_1}p'}{\lambda_s(p)+ \lambda_s(q)\xrightarrow{action(s,e_1)}\lambda_{effect(s,e_1)}(p')}$$

      $$\frac{q\xrightarrow{e_2}q'}{\lambda_s(p+ q)\xrightarrow{action(s,e_2)}\lambda_{effect(s,e_2)}(q')}
      \quad\frac{q\xrightarrow{e_2}q'}{\lambda_s(p)+ \lambda_s(q)\xrightarrow{action(s,e_2)}\lambda_{effect(s,e_2)}(q')}$$

      So, $\lambda_s(p+ q)\approx_{rbs} \lambda_s(p)+\lambda_s(q)$, as desired.
  \item \textbf{Axiom $SO4$}. Let $q$ be $APTC_{\tau}$ with guarded linear recursion and the state operator processes, and $\lambda_s(e\cdot q)=action(s,e)\cdot\lambda_{effect(s,e)}(q)$, it is sufficient to prove that $\lambda_s(e\cdot q)\approx_{rbs}action(s,e)\cdot\lambda_{effect(s,e)}(q)$. By the transition rules for operator $\cdot$ and $\lambda_s$ in Table \ref{TRForState}, we get

      $$\frac{e\xrightarrow{e}\surd}{\lambda_s(e\cdot q)\xrightarrow{action(s,e)}\lambda_{effect(s,e)}(q)}
      \quad\frac{action(s,e)\xrightarrow{action(s,e)}\surd}{action(s,e)\cdot \lambda_{effect(s,e)}(q)\xrightarrow{action(s,e)}\lambda_{effect(s,e)}(q)}$$

      So, $\lambda_s(e\cdot q)\approx_{rbs}action(s,e)\cdot\lambda_{effect(s,e)}(q)$, as desired.
  \item \textbf{Axiom $SO5$}. Let $p,q$ be $APTC_{\tau}$ with guarded linear recursion and the state operator processes, and $\lambda_s(p\parallel q)=\lambda_s(p)\parallel\lambda_s(q)$, it is sufficient to prove that $\lambda_s(p\parallel q)\approx_{rbs} \lambda_s(p)\parallel\lambda_s(q)$. By the transition rules for operator $\parallel$ and $\lambda_s$ in Table \ref{TRForState}, we get for the case $\neg(e_1\%e_2)$

      $$\frac{p\xrightarrow{e_1}\surd\quad q\xrightarrow{e_2}\surd}{\lambda_s(p\parallel q)\xrightarrow{\{action(s,e_1),action(s,e_2)\}}\surd}$$
      $$\frac{p\xrightarrow{e_1}\surd\quad q\xrightarrow{e_2}\surd}{\lambda_s(p)\parallel \lambda_s(q)\xrightarrow{\{action(s,e_1),action(s,e_2)\}}\surd}$$

      $$\frac{p\xrightarrow{e_1}p'\quad q\xrightarrow{e_2}\surd}{\lambda_s(p\parallel q)\xrightarrow{\{action(s,e_1),action(s,e_2)\}}\lambda_{effect(s,e_1)\cup effect(s,e_2)}(p')}$$
      $$\frac{p\xrightarrow{e_1}p'\quad q\xrightarrow{e_2}\surd}{\lambda_s(p)\parallel \lambda_s(q)\xrightarrow{\{action(s,e_1),action(s,e_2)\}}\lambda_{effect(s,e_1)\cup effect(s,e_2)}(p')}$$

      $$\frac{p\xrightarrow{e_1}\surd\quad q\xrightarrow{e_2}q'}{\lambda_s(p\parallel q)\xrightarrow{\{action(s,e_1),action(s,e_2)\}}\lambda_{effect(s,e_1)\cup effect(s,e_2)}(q')}$$
      $$\frac{p\xrightarrow{e_1}\surd\quad q\xrightarrow{e_2}q'}{\lambda_s(p)\parallel \lambda_s(q)\xrightarrow{\{action(s,e_1),action(s,e_2)\}}\lambda_{effect(s,e_1)\cup effect(s,e_2)}(q')}$$

      $$\frac{p\xrightarrow{e_1}p'\quad q\xrightarrow{e_2}q'}{\lambda_s(p\parallel q)\xrightarrow{\{action(s,e_1),action(s,e_2)\}}\lambda_{effect(s,e_1)\cup effect(s,e_2)}(p'\between q')}$$
      $$\frac{p\xrightarrow{e_1}p'\quad q\xrightarrow{e_2}q'}{\lambda_s(p)\parallel \lambda_s(q)\xrightarrow{\{action(s,e_1),action(s,e_2)\}}\lambda_{effect(s,e_1)\cup effect(s,e_2)}(p')\between\lambda_{effect(s,e_1)\cup effect(s,e_2)}(q')}$$

      So, with the assumption $\lambda_{effect(s,e_1)\cup effect(s,e_2)}(p'\between q')=\lambda_{effect(s,e_1)\cup effect(s,e_2)}(p')\between\lambda_{effect(s,e_1)\cup effect(s,e_2)}(q')$, $\lambda_s(p\parallel q)\approx_{rbs} \lambda_s(p)\parallel\lambda_s(q)$, as desired.
      For the case $e_1\%e_2$, we get

      $$\frac{p\xrightarrow{e_1}\surd\quad q\xnrightarrow{e_2}}{\lambda_s(p\parallel q)\xrightarrow{action(s,e_1)}\lambda_{effect(s,e_1)}(q)}$$
      $$\frac{p\xrightarrow{e_1}\surd\quad q\xnrightarrow{e_2}}{\lambda_s(p)\parallel \lambda_s(q)\xrightarrow{action(s,e_1)}\lambda_{effect(s,e_1)}(q)}$$

      $$\frac{p\xrightarrow{e_1}p'\quad q\xnrightarrow{e_2}}{\lambda_s(p\parallel q)\xrightarrow{action(s,e_1)}\lambda_{effect(s,e_1)}(p'\between q)}$$
      $$\frac{p\xrightarrow{e_1}p'\quad q\xnrightarrow{e_2}}{\lambda_s(p)\parallel \lambda_s(q)\xrightarrow{action(s,e_1)}\lambda_{effect(s,e_1)}(p')\between\lambda_{effect(s,e_1)}(q)}$$

      $$\frac{p\xnrightarrow{e_1}\quad q\xrightarrow{e_2}\surd}{\lambda_s(p\parallel q)\xrightarrow{action(s,e_2)}\lambda_{effect(s,e_2)}(p)}$$
      $$\frac{p\xnrightarrow{e_1}\quad q\xrightarrow{e_2}\surd}{\lambda_s(p)\parallel \lambda_s(q)\xrightarrow{action(s,e_2)}\lambda_{effect(s,e_2)}(p)}$$

      $$\frac{p\xnrightarrow{e_1}\quad q\xrightarrow{e_2}q'}{\lambda_s(p\parallel q)\xrightarrow{action(s,e_2)}\lambda_{effect(s,e_2)}(p\between q')}$$
      $$\frac{p\xnrightarrow{e_1}\quad q\xrightarrow{e_2}q'}{\lambda_s(p)\parallel \lambda_s(q)\xrightarrow{action(s,e_2)}\lambda_{effect(s,e_2)}(p)\between\lambda_{effect(s,e_2)}(q')}$$

      So, with the assumption $\lambda_{effect(s,e_1)}(p'\between q)=\lambda_{effect(s,e_1)}(p')\between\lambda_{effect(s,e_1)}(q)$ and $\lambda_{effect(s,e_2)}(p\between q')=\lambda_{effect(s,e_2)}(p)\between\lambda_{effect(s,e_2)}(q')$, $\lambda_s(p\parallel q)\approx_{rbs} \lambda_s(p)\parallel\lambda_s(q)$, as desired.
\end{itemize}

(2) Soundness of $APTC_{\tau}$ with guarded linear recursion and the state operator with respect to rooted branching pomset bisimulation $\approx_{rbp}$.

Since rooted branching pomset bisimulation $\approx_{rbp}$ is both an equivalent and a congruent relation with respect to $APTC_{\tau}$ with guarded linear recursion and the state operator, we only need to check if each axiom in Table \ref{AxiomsForState} is sound modulo rooted branching pomset bisimulation $\approx_{rbp}$.

From the definition of rooted branching pomset bisimulation $\approx_{rbp}$ (see Definition \ref{RBPSB}), we know that rooted branching pomset bisimulation $\approx_{rbp}$ is defined by weak pomset transitions, which are labeled by pomsets with $\tau$. In a weak pomset transition, the events in the pomset are either within causality relations (defined by $\cdot$) or in concurrency (implicitly defined by $\cdot$ and $+$, and explicitly defined by $\between$), of course, they are pairwise consistent (without conflicts). In (1), we have already proven the case that all events are pairwise concurrent, so, we only need to prove the case of events in causality. Without loss of generality, we take a pomset of $P=\{e_1,e_2:e_1\cdot e_2\}$. Then the weak pomset transition labeled by the above $P$ is just composed of one single event transition labeled by $e_1$ succeeded by another single event transition labeled by $e_2$, that is, $\xRightarrow{P}=\xRightarrow{e_1}\xRightarrow{e_2}$.

Similarly to the proof of soundness of $APTC_{\tau}$ with guarded linear recursion and the state operator modulo rooted branching step bisimulation $\approx_{rbs}$ (1), we can prove that each axiom in Table \ref{AxiomsForState} is sound modulo rooted branching pomset bisimulation $\approx_{rbp}$, we omit them.

(3) Soundness of $APTC_{\tau}$ with guarded linear recursion and the state operator with respect to rooted branching hp-bisimulation $\approx_{rbhp}$.

Since rooted branching hp-bisimulation $\approx_{rbhp}$ is both an equivalent and a congruent relation with respect to $APTC_{\tau}$ with guarded linear recursion and the state operator, we only need to check if each axiom in Table \ref{AxiomsForState} is sound modulo rooted branching hp-bisimulation $\approx_{rbhp}$.

From the definition of rooted branching hp-bisimulation $\approx_{rbhp}$ (see Definition \ref{RBHHPB}), we know that rooted branching hp-bisimulation $\approx_{rbhp}$ is defined on the weakly posetal product $(C_1,f,C_2),f:\hat{C_1}\rightarrow \hat{C_2}\textrm{ isomorphism}$. Two process terms $s$ related to $C_1$ and $t$ related to $C_2$, and $f:\hat{C_1}\rightarrow \hat{C_2}\textrm{ isomorphism}$. Initially, $(C_1,f,C_2)=(\emptyset,\emptyset,\emptyset)$, and $(\emptyset,\emptyset,\emptyset)\in\approx_{rbhp}$. When $s\xrightarrow{e}s'$ ($C_1\xrightarrow{e}C_1'$), there will be $t\xRightarrow{e}t'$ ($C_2\xRightarrow{e}C_2'$), and we define $f'=f[e\mapsto e]$. Then, if $(C_1,f,C_2)\in\approx_{rbhp}$, then $(C_1',f',C_2')\in\approx_{rbhp}$.

Similarly to the proof of soundness of $APTC_{\tau}$ with guarded linear recursion and the state operator modulo rooted branching pomset bisimulation equivalence (2), we can prove that each axiom in Table \ref{AxiomsForState} is sound modulo rooted branching hp-bisimulation equivalence, we just need additionally to check the above conditions on rooted branching hp-bisimulation, we omit them.
\end{proof}

\begin{theorem}[Completeness of the state operator]\label{CState}
Let $p$ and $q$ be closed $APTC_{\tau}$ with guarded linear recursion and $CFAR$ and the state operator terms, then,
\begin{enumerate}
  \item if $p\approx_{rbs} q$ then $p=q$;
  \item if $p\approx_{rbp} q$ then $p=q$;
  \item if $p\approx_{rbhp} q$ then $p=q$.
\end{enumerate}
\end{theorem}

\begin{proof}
(1) For the case of rooted branching step bisimulation, the proof is following.

Firstly, we know that each process term $p$ in $APTC_{\tau}$ with guarded linear recursion is equal to a process term $\langle X_1|E\rangle$ with $E$ a guarded linear recursive specification. And we prove if $\langle X_1|E_1\rangle\approx_{rbs}\langle Y_1|E_2\rangle$, then $\langle X_1|E_1\rangle=\langle Y_1|E_2\rangle$

Structural induction with respect to process term $p$ can be applied. The only new case (where $SO1-SO5$ are needed) is $p \equiv \lambda_{s_0}(q)$. First assuming $q=\langle X_1|E\rangle$ with a guarded linear recursive specification $E$, we prove the case of $p=\lambda_{s_0}(\langle X_1|E\rangle)$. Let $E$ consist of guarded linear recursive equations

$$X_i=(a_{1i1}\parallel\cdots\parallel a_{k_{i1}i1})X_{i1}+...+(a_{1im_i}\parallel\cdots\parallel a_{k_{im_i}im_i})X_{im_i}+b_{1i1}\parallel\cdots\parallel b_{l_{i1}i1}+...+b_{1im_i}\parallel\cdots\parallel b_{l_{im_i}im_i}$$

for $i\in{1,...,n}$. Let $F$ consist of guarded linear recursive equations

$Y_i(s)=(action(s,a_{1i1})\parallel\cdots\parallel action(s,a_{k_{i1}i1}))Y_{i1}(effect(s,a_{1i1})\cup\cdots\cup effect(s,a_{k_{i1}i1}))$

$+...+(action(s,a_{1im_i})\parallel\cdots\parallel action(s,a_{k_{im_i}im_i}))Y_{im_i}(effect(s,a_{1im_i})\cup\cdots\cup effect(s,a_{k_{im_i}im_i}))$

$+action(s,b_{1i1})\parallel\cdots\parallel action(s,b_{l_{i1}i1})+...+action(s,b_{1im_i})\parallel\cdots\parallel action(s,b_{l_{im_i}im_i})$

for $i\in{1,...,n}$.

\begin{eqnarray}
&&\lambda_s(\langle X_i|E\rangle) \nonumber\\
&\overset{\text{RDP}}{=}&\lambda_s((a_{1i1}\parallel\cdots\parallel a_{k_{i1}i1})X_{i1}+...+(a_{1im_i}\parallel\cdots\parallel a_{k_{im_i}im_i})X_{im_i}+b_{1i1}\parallel\cdots\parallel b_{l_{i1}i1}+...+b_{1im_i}\parallel\cdots\parallel b_{l_{im_i}im_i}) \nonumber\\
&\overset{\text{SO1-SO5}}{=}&(action(s,a_{1i1})\parallel\cdots\parallel action(s,a_{k_{i1}i1}))\lambda_{effect(s,a_{1i1})\cup\cdots\cup effect(s,a_{k_{i1}i1}))}(X_{i1})\nonumber\\
&&+...+(action(s,a_{1im_i})\parallel\cdots\parallel action(s,a_{k_{im_i}im_i}))\lambda_{effect(s,a_{1im_i})\cup\cdots\cup effect(s,a_{k_{im_i}im_i}))}(X_{im_i})\nonumber\\
&&+action(s,b_{1i1})\parallel\cdots\parallel action(s,b_{l_{i1}i1})+...+action(s,b_{1im_i})\parallel\cdots\parallel action(s,b_{l_{im_i}im_i}) \nonumber
\end{eqnarray}

Replacing $Y_i(s)$ by $\lambda_s(\langle X_i|E\rangle)$ for $i\in\{1,...,n\}$ is a solution for $F$. So by $RSP$, $\lambda_{s_0}(\langle X_1|E\rangle)=\langle Y_1(s_0)|F\rangle$, as desired.

(2) For the case of rooted branching pomset bisimulation, it can be proven similarly to (1), we omit it.

(3) For the case of rooted branching hp-bisimulation, it can be proven similarly to (1), we omit it.
\end{proof}

\subsection{Guards}

Traditional process algebras, such as CCS \cite{CC} \cite{CCS} and ACP \cite{ACP}, are well-known for capturing concurrency based on the interleaving bisimulation semantics. These algebras do not involve anything about data and conditionals, because data are hidden behind of actions. There are some efforts to extend these algebras with conditionals, for example, \cite{HLPA} extended ACP with guards. But these work did not form a perfect solution, just because it is difficult to deal the interleaving concurrency with conditionals.

Based on $APTC$, we add guards into it by adopting the solution of \cite{HLPA} in two ways: (1) the same solution to guards to form a Boolean Algebra; (2) the similar operational semantics based on configuration, which has two parts: the processes and the data states. Finally, we get some further results as follows:

\begin{enumerate}
  \item we design a sound and complete theory of concurrency and parallelism with guards;
  \item we design a sound and complete theory of recursion including concurrency with guards;
  \item we design a sound and complete theory of abstraction with guards;
  \item we design a sound Hoare logic \cite{HL} including concurrency and parallelism, recursion, and abstraction.
\end{enumerate}

\subsubsection{Operational Semantics}{\label{os2}}

In this section, we extend truly concurrent bisimilarities to the ones containing data states.

\begin{definition}[Prime event structure with silent event and empty event]\label{PESG}
Let $\Lambda$ be a fixed set of labels, ranged over $a,b,c,\cdots$ and $\tau,\epsilon$. A ($\Lambda$-labelled) prime event structure with silent event $\tau$ and empty event $\epsilon$ is a tuple $\mathcal{E}=\langle \mathbb{E}, \leq, \sharp, \lambda\rangle$, where $\mathbb{E}$ is a denumerable set of events, including the silent event $\tau$ and empty event $\epsilon$. Let $\hat{\mathbb{E}}=\mathbb{E}\backslash\{\tau,\epsilon\}$, exactly excluding $\tau$ and $\epsilon$, it is obvious that $\hat{\tau^*}=\epsilon$. Let $\lambda:\mathbb{E}\rightarrow\Lambda$ be a labelling function and let $\lambda(\tau)=\tau$ and $\lambda(\epsilon)=\epsilon$. And $\leq$, $\sharp$ are binary relations on $\mathbb{E}$, called causality and conflict respectively, such that:

\begin{enumerate}
  \item $\leq$ is a partial order and $\lceil e \rceil = \{e'\in \mathbb{E}|e'\leq e\}$ is finite for all $e\in \mathbb{E}$. It is easy to see that $e\leq\tau^*\leq e'=e\leq\tau\leq\cdots\leq\tau\leq e'$, then $e\leq e'$.
  \item $\sharp$ is irreflexive, symmetric and hereditary with respect to $\leq$, that is, for all $e,e',e''\in \mathbb{E}$, if $e\sharp e'\leq e''$, then $e\sharp e''$.
\end{enumerate}

Then, the concepts of consistency and concurrency can be drawn from the above definition:

\begin{enumerate}
  \item $e,e'\in \mathbb{E}$ are consistent, denoted as $e\frown e'$, if $\neg(e\sharp e')$. A subset $X\subseteq \mathbb{E}$ is called consistent, if $e\frown e'$ for all $e,e'\in X$.
  \item $e,e'\in \mathbb{E}$ are concurrent, denoted as $e\parallel e'$, if $\neg(e\leq e')$, $\neg(e'\leq e)$, and $\neg(e\sharp e')$.
\end{enumerate}
\end{definition}

\begin{definition}[Configuration]
Let $\mathcal{E}$ be a PES. A (finite) configuration in $\mathcal{E}$ is a (finite) consistent subset of events $C\subseteq \mathcal{E}$, closed with respect to causality (i.e. $\lceil C\rceil=C$), and a data state $s\in S$ with $S$ the set of all data states, denoted $\langle C, s\rangle$. The set of finite configurations of $\mathcal{E}$ is denoted by $\langle\mathcal{C}(\mathcal{E}), S\rangle$. We let $\hat{C}=C\backslash\{\tau\}\cup\{\epsilon\}$.
\end{definition}

A consistent subset of $X\subseteq \mathbb{E}$ of events can be seen as a pomset. Given $X, Y\subseteq \mathbb{E}$, $\hat{X}\sim \hat{Y}$ if $\hat{X}$ and $\hat{Y}$ are isomorphic as pomsets. In the following of the paper, we say $C_1\sim C_2$, we mean $\hat{C_1}\sim\hat{C_2}$.

\begin{definition}[Pomset transitions and step]
Let $\mathcal{E}$ be a PES and let $C\in\mathcal{C}(\mathcal{E})$, and $\emptyset\neq X\subseteq \mathbb{E}$, if $C\cap X=\emptyset$ and $C'=C\cup X\in\mathcal{C}(\mathcal{E})$, then $\langle C,s\rangle\xrightarrow{X} \langle C',s'\rangle$ is called a pomset transition from $\langle C,s\rangle$ to $\langle C',s'\rangle$. When the events in $X$ are pairwise concurrent, we say that $\langle C,s\rangle\xrightarrow{X}\langle C',s'\rangle$ is a step. It is obvious that $\rightarrow^*\xrightarrow{X}\rightarrow^*=\xrightarrow{X}$ and $\rightarrow^*\xrightarrow{e}\rightarrow^*=\xrightarrow{X}$ for any $e\in\mathbb{E}$ and $X\subseteq\mathbb{E}$.
\end{definition}

\begin{definition}[Weak pomset transitions and weak step]
Let $\mathcal{E}$ be a PES and let $C\in\mathcal{C}(\mathcal{E})$, and $\emptyset\neq X\subseteq \hat{\mathbb{E}}$, if $C\cap X=\emptyset$ and $\hat{C'}=\hat{C}\cup X\in\mathcal{C}(\mathcal{E})$, then $\langle C,s\rangle\xRightarrow{X} \langle C',s'\rangle$ is called a weak pomset transition from $\langle C,s\rangle$ to $\langle C',s'\rangle$, where we define $\xRightarrow{e}\triangleq\xrightarrow{\tau^*}\xrightarrow{e}\xrightarrow{\tau^*}$. And $\xRightarrow{X}\triangleq\xrightarrow{\tau^*}\xrightarrow{e}\xrightarrow{\tau^*}$, for every $e\in X$. When the events in $X$ are pairwise concurrent, we say that $\langle C,s\rangle\xRightarrow{X}\langle C',s'\rangle$ is a weak step.
\end{definition}

We will also suppose that all the PESs in this paper are image finite, that is, for any PES $\mathcal{E}$ and $C\in \mathcal{C}(\mathcal{E})$ and $a\in \Lambda$, $\{e\in \mathbb{E}|\langle C,s\rangle\xrightarrow{e} \langle C',s'\rangle\wedge \lambda(e)=a\}$ and $\{e\in\hat{\mathbb{E}}|\langle C,s\rangle\xRightarrow{e} \langle C',s'\rangle\wedge \lambda(e)=a\}$ is finite.

\begin{definition}[Pomset, step bisimulation]\label{PSBG}
Let $\mathcal{E}_1$, $\mathcal{E}_2$ be PESs. A pomset bisimulation is a relation $R\subseteq\langle\mathcal{C}(\mathcal{E}_1),S\rangle\times\langle\mathcal{C}(\mathcal{E}_2),S\rangle$, such that if $(\langle C_1,s\rangle,\langle C_2,s\rangle)\in R$, and $\langle C_1,s\rangle\xrightarrow{X_1}\langle C_1',s'\rangle$ then $\langle C_2,s\rangle\xrightarrow{X_2}\langle C_2',s'\rangle$, with $X_1\subseteq \mathbb{E}_1$, $X_2\subseteq \mathbb{E}_2$, $X_1\sim X_2$ and $(\langle C_1',s'\rangle,\langle C_2',s'\rangle)\in R$ for all $s\in S$, and vice-versa. We say that $\mathcal{E}_1$, $\mathcal{E}_2$ are pomset bisimilar, written $\mathcal{E}_1\sim_p\mathcal{E}_2$, if there exists a pomset bisimulation $R$, such that $(\langle\emptyset,\emptyset\rangle,\langle\emptyset,\emptyset\rangle)\in R$. By replacing pomset transitions with steps, we can get the definition of step bisimulation. When PESs $\mathcal{E}_1$ and $\mathcal{E}_2$ are step bisimilar, we write $\mathcal{E}_1\sim_s\mathcal{E}_2$.
\end{definition}

\begin{definition}[Weak pomset, step bisimulation]\label{WPSBG}
Let $\mathcal{E}_1$, $\mathcal{E}_2$ be PESs. A weak pomset bisimulation is a relation $R\subseteq\langle\mathcal{C}(\mathcal{E}_1),S\rangle\times\langle\mathcal{C}(\mathcal{E}_2),S\rangle$, such that if $(\langle C_1,s\rangle,\langle C_2,s\rangle)\in R$, and $\langle C_1,s\rangle\xRightarrow{X_1}\langle C_1',s'\rangle$ then $\langle C_2,s\rangle\xRightarrow{X_2}\langle C_2',s'\rangle$, with $X_1\subseteq \hat{\mathbb{E}_1}$, $X_2\subseteq \hat{\mathbb{E}_2}$, $X_1\sim X_2$ and $(\langle C_1',s'\rangle,\langle C_2',s'\rangle)\in R$ for all $s\in S$, and vice-versa. We say that $\mathcal{E}_1$, $\mathcal{E}_2$ are weak pomset bisimilar, written $\mathcal{E}_1\approx_p\mathcal{E}_2$, if there exists a weak pomset bisimulation $R$, such that $(\langle\emptyset,\emptyset\rangle,\langle\emptyset,\emptyset\rangle)\in R$. By replacing weak pomset transitions with weak steps, we can get the definition of weak step bisimulation. When PESs $\mathcal{E}_1$ and $\mathcal{E}_2$ are weak step bisimilar, we write $\mathcal{E}_1\approx_s\mathcal{E}_2$.
\end{definition}

\begin{definition}[Posetal product]
Given two PESs $\mathcal{E}_1$, $\mathcal{E}_2$, the posetal product of their configurations, denoted $\langle\mathcal{C}(\mathcal{E}_1),S\rangle\overline{\times}\langle\mathcal{C}(\mathcal{E}_2),S\rangle$, is defined as

$$\{(\langle C_1,s\rangle,f,\langle C_2,s\rangle)|C_1\in\mathcal{C}(\mathcal{E}_1),C_2\in\mathcal{C}(\mathcal{E}_2),f:C_1\rightarrow C_2 \textrm{ isomorphism}\}.$$

A subset $R\subseteq\langle\mathcal{C}(\mathcal{E}_1),S\rangle\overline{\times}\langle\mathcal{C}(\mathcal{E}_2),S\rangle$ is called a posetal relation. We say that $R$ is downward closed when for any $(\langle C_1,s\rangle,f,\langle C_2,s\rangle),(\langle C_1',s'\rangle,f',\langle C_2',s'\rangle)\in \langle\mathcal{C}(\mathcal{E}_1),S\rangle\overline{\times}\langle\mathcal{C}(\mathcal{E}_2),S\rangle$, if $(\langle C_1,s\rangle,f,\langle C_2,s\rangle)\subseteq (\langle C_1',s'\rangle,f',\langle C_2',s'\rangle)$ pointwise and $(\langle C_1',s'\rangle,f',\langle C_2',s'\rangle)\in R$, then $(\langle C_1,s\rangle,f,\langle C_2,s\rangle)\in R$.

For $f:X_1\rightarrow X_2$, we define $f[x_1\mapsto x_2]:X_1\cup\{x_1\}\rightarrow X_2\cup\{x_2\}$, $z\in X_1\cup\{x_1\}$,(1)$f[x_1\mapsto x_2](z)=
x_2$,if $z=x_1$;(2)$f[x_1\mapsto x_2](z)=f(z)$, otherwise. Where $X_1\subseteq \mathbb{E}_1$, $X_2\subseteq \mathbb{E}_2$, $x_1\in \mathbb{E}_1$, $x_2\in \mathbb{E}_2$.
\end{definition}

\begin{definition}[Weakly posetal product]
Given two PESs $\mathcal{E}_1$, $\mathcal{E}_2$, the weakly posetal product of their configurations, denoted $\langle\mathcal{C}(\mathcal{E}_1),S\rangle\overline{\times}\langle\mathcal{C}(\mathcal{E}_2),S\rangle$, is defined as

$$\{(\langle C_1,s\rangle,f,\langle C_2,s\rangle)|C_1\in\mathcal{C}(\mathcal{E}_1),C_2\in\mathcal{C}(\mathcal{E}_2),f:\hat{C_1}\rightarrow \hat{C_2} \textrm{ isomorphism}\}.$$

A subset $R\subseteq\langle\mathcal{C}(\mathcal{E}_1),S\rangle\overline{\times}\langle\mathcal{C}(\mathcal{E}_2),S\rangle$ is called a weakly posetal relation. We say that $R$ is downward closed when for any $(\langle C_1,s\rangle,f,\langle C_2,s\rangle),(\langle C_1',s'\rangle,f,\langle C_2',s'\rangle)\in \langle\mathcal{C}(\mathcal{E}_1),S\rangle\overline{\times}\langle\mathcal{C}(\mathcal{E}_2),S\rangle$, if $(\langle C_1,s\rangle,f,\langle C_2,s\rangle)\subseteq (\langle C_1',s'\rangle,f',\langle C_2',s'\rangle)$ pointwise and $(\langle C_1',s'\rangle,f',\langle C_2',s'\rangle)\in R$, then $(\langle C_1,s\rangle,f,\langle C_2,s\rangle)\in R$.

For $f:X_1\rightarrow X_2$, we define $f[x_1\mapsto x_2]:X_1\cup\{x_1\}\rightarrow X_2\cup\{x_2\}$, $z\in X_1\cup\{x_1\}$,(1)$f[x_1\mapsto x_2](z)=
x_2$,if $z=x_1$;(2)$f[x_1\mapsto x_2](z)=f(z)$, otherwise. Where $X_1\subseteq \hat{\mathbb{E}_1}$, $X_2\subseteq \hat{\mathbb{E}_2}$, $x_1\in \hat{\mathbb{E}}_1$, $x_2\in \hat{\mathbb{E}}_2$. Also, we define $f(\tau^*)=f(\tau^*)$.
\end{definition}

\begin{definition}[(Hereditary) history-preserving bisimulation]\label{HHPBG}
A history-preserving (hp-) bisimulation is a posetal relation $R\subseteq\langle\mathcal{C}(\mathcal{E}_1),S\rangle\overline{\times}\langle\mathcal{C}(\mathcal{E}_2),S\rangle$ such that if $(\langle C_1,s\rangle,f,\langle C_2,s\rangle)\in R$, and $\langle C_1,s\rangle\xrightarrow{e_1} \langle C_1',s'\rangle$, then $\langle C_2,s\rangle\xrightarrow{e_2} \langle C_2',s'\rangle$, with $(\langle C_1',s'\rangle,f[e_1\mapsto e_2],\langle C_2',s'\rangle)\in R$ for all $s\in S$, and vice-versa. $\mathcal{E}_1,\mathcal{E}_2$ are history-preserving (hp-)bisimilar and are written $\mathcal{E}_1\sim_{hp}\mathcal{E}_2$ if there exists a hp-bisimulation $R$ such that $(\langle\emptyset,\emptyset\rangle,\emptyset,\langle\emptyset,\emptyset\rangle)\in R$.

A hereditary history-preserving (hhp-)bisimulation is a downward closed hp-bisimulation. $\mathcal{E}_1,\mathcal{E}_2$ are hereditary history-preserving (hhp-)bisimilar and are written $\mathcal{E}_1\sim_{hhp}\mathcal{E}_2$.
\end{definition}

\begin{definition}[Weak (hereditary) history-preserving bisimulation]\label{WHHPBG}
A weak history-preserving (hp-) bisimulation is a weakly posetal relation $R\subseteq\langle\mathcal{C}(\mathcal{E}_1),S\rangle\overline{\times}\langle\mathcal{C}(\mathcal{E}_2),S\rangle$ such that if $(\langle C_1,s\rangle,f,\langle C_2,s\rangle)\in R$, and $\langle C_1,s\rangle\xRightarrow{e_1} \langle C_1',s'\rangle$, then $\langle C_2,s\rangle\xRightarrow{e_2} \langle C_2',s'\rangle$, with $(\langle C_1',s'\rangle,f[e_1\mapsto e_2],\langle C_2',s'\rangle)\in R$ for all $s\in S$, and vice-versa. $\mathcal{E}_1,\mathcal{E}_2$ are weak history-preserving (hp-)bisimilar and are written $\mathcal{E}_1\approx_{hp}\mathcal{E}_2$ if there exists a weak hp-bisimulation $R$ such that $(\langle\emptyset,\emptyset\rangle,\emptyset,\langle\emptyset,\emptyset\rangle)\in R$.

A weakly hereditary history-preserving (hhp-)bisimulation is a downward closed weak hp-bisimulation. $\mathcal{E}_1,\mathcal{E}_2$ are weakly hereditary history-preserving (hhp-)bisimilar and are written $\mathcal{E}_1\approx_{hhp}\mathcal{E}_2$.
\end{definition}

\subsubsection{$BATC$ with Guards}{\label{batcg}}

In this subsection, we will discuss the guards for $BATC$, which is denoted as $BATC_G$. Let $\mathbb{E}$ be the set of atomic events (actions), $G_{at}$ be the set of atomic guards, $\delta$ be the deadlock constant, and $\epsilon$ be the empty event. We extend $G_{at}$ to the set of basic guards $G$ with element $\phi,\psi,\cdots$, which is generated by the following formation rules:

$$\phi::=\delta|\epsilon|\neg\phi|\psi\in G_{at}|\phi+\psi|\phi\cdot\psi$$

In the following, let $e_1, e_2, e_1', e_2'\in \mathbb{E}$, $\phi,\psi\in G$ and let variables $x,y,z$ range over the set of terms for true concurrency, $p,q,s$ range over the set of closed terms. The predicate $test(\phi,s)$ represents that $\phi$ holds in the state $s$, and $test(\epsilon,s)$ holds and $test(\delta,s)$ does not hold. $effect(e,s)\in S$ denotes $s'$ in $s\xrightarrow{e}s'$. The predicate weakest precondition $wp(e,\phi)$ denotes that $\forall s\in S, test(\phi,effect(e,s))$ holds.

The set of axioms of $BATC_G$ consists of the laws given in Table \ref{AxiomsForBATCG}.

\begin{center}
    \begin{table}
        \begin{tabular}{@{}ll@{}}
            \hline No. &Axiom\\
            $A1$ & $x+ y = y+ x$\\
            $A2$ & $(x+ y)+ z = x+ (y+ z)$\\
            $A3$ & $x+ x = x$\\
            $A4$ & $(x+ y)\cdot z = x\cdot z + y\cdot z$\\
            $A5$ & $(x\cdot y)\cdot z = x\cdot(y\cdot z)$\\
            $A6$ & $x+\delta = x$\\
            $A7$ & $\delta\cdot x = \delta$\\
            $A8$ & $\epsilon\cdot x = x$\\
            $A9$ & $x\cdot\epsilon = x$\\
            $G1$ & $\phi\cdot\neg\phi = \delta$\\
            $G2$ & $\phi+\neg\phi = \epsilon$\\
            $G3$ & $\phi\delta = \delta$\\
            $G4$ & $\phi(x+y)=\phi x+\phi y$\\
            $G5$ & $\phi(x\cdot y)= \phi x\cdot y$\\
            $G6$ & $(\phi+\psi)x = \phi x + \psi x$\\
            $G7$ & $(\phi\cdot \psi)\cdot x = \phi\cdot(\psi\cdot x)$\\
            $G8$ & $\phi_0\cdot\cdots\cdot\phi_n = \delta$ if $\forall s\in S,\exists i\leq n.test(\neg\phi_i,s)$\\
            $G9$ & $wp(e,\phi)e\phi=wp(e,\phi)e$\\
            $G10$ & $\neg wp(e,\phi)e\neg\phi=\neg wp(e,\phi)e$\\
        \end{tabular}
        \caption{Axioms of $BATC_G$}
        \label{AxiomsForBATCG}
    \end{table}
\end{center}

Note that, by eliminating atomic event from the process terms, the axioms in Table \ref{AxiomsForBATCG} will lead to a Boolean Algebra. And $G9$ is a precondition of $e$ and $\phi$, $G10$ is the weakest precondition of $e$ and $\phi$. A data environment with $effct$ function is sufficiently deterministic, and it is obvious that if the weakest precondition is expressible and $G9$, $G10$ are sound, then the related data environment is sufficiently deterministic.

\begin{definition}[Basic terms of $BATC_G$]\label{BTBATCG}
The set of basic terms of $BATC_G$, $\mathcal{B}(BATC_G)$, is inductively defined as follows:
\begin{enumerate}
  \item $\mathbb{E}\subset\mathcal{B}(BATC_G)$;
  \item $G\subset\mathcal{B}(BATC_G)$;
  \item if $e\in \mathbb{E}, t\in\mathcal{B}(BATC_G)$ then $e\cdot t\in\mathcal{B}(BATC_G)$;
  \item if $\phi\in G, t\in\mathcal{B}(BATC_G)$ then $\phi\cdot t\in\mathcal{B}(BATC_G)$;
  \item if $t,s\in\mathcal{B}(BATC_G)$ then $t+ s\in\mathcal{B}(BATC_G)$.
\end{enumerate}
\end{definition}

\begin{theorem}[Elimination theorem of $BATC_G$]\label{ETBATCG}
Let $p$ be a closed $BATC_G$ term. Then there is a basic $BATC_G$ term $q$ such that $BATC_G\vdash p=q$.
\end{theorem}

\begin{proof}
(1) Firstly, suppose that the following ordering on the signature of $BATC_G$ is defined: $\cdot > +$ and the symbol $\cdot$ is given the lexicographical status for the first argument, then for each rewrite rule $p\rightarrow q$ in Table \ref{TRSForBATCG} relation $p>_{lpo} q$ can easily be proved. We obtain that the term rewrite system shown in Table \ref{TRSForBATCG} is strongly normalizing, for it has finitely many rewriting rules, and $>$ is a well-founded ordering on the signature of $BATC_G$, and if $s>_{lpo} t$, for each rewriting rule $s\rightarrow t$ is in Table \ref{TRSForBATCG} (see Theorem \ref{SN}).

\begin{center}
    \begin{table}
        \begin{tabular}{@{}ll@{}}
            \hline No. &Rewriting Rule\\
            $RA3$ & $x+ x \rightarrow x$\\
            $RA4$ & $(x+ y)\cdot z \rightarrow x\cdot z + y\cdot z$\\
            $RA5$ & $(x\cdot y)\cdot z \rightarrow x\cdot(y\cdot z)$\\
            $RA6$ & $x+\delta \rightarrow x$\\
            $RA7$ & $\delta\cdot x \rightarrow \delta$\\
            $RA8$ & $\epsilon\cdot x \rightarrow x$\\
            $RA9$ & $x\cdot\epsilon \rightarrow x$\\
            $RG1$ & $\phi\cdot\neg\phi \rightarrow \delta$\\
            $RG2$ & $\phi+\neg\phi \rightarrow \epsilon$\\
            $RG3$ & $\phi\delta \rightarrow \delta$\\
            $RG4$ & $\phi(x+y)\rightarrow\phi x+\phi y$\\
            $RG5$ & $\phi(x\cdot y)\rightarrow \phi x\cdot y$\\
            $RG6$ & $(\phi+\psi)x \rightarrow \phi x + \psi x$\\
            $RG7$ & $(\phi\cdot \psi)\cdot x \rightarrow \phi\cdot(\psi\cdot x)$\\
            $RG8$ & $\phi_0\cdot\cdots\cdot\phi_n \rightarrow \delta$ if $\forall s\in S,\exists i\leq n.test(\neg\phi_i,s)$\\
            $RG9$ & $wp(e,\phi)e\phi\rightarrow wp(e,\phi)e$\\
            $RG10$ & $\neg wp(e,\phi)e\neg\phi\rightarrow\neg wp(e,\phi)e$\\
        \end{tabular}
        \caption{Term rewrite system of $BATC_G$}
        \label{TRSForBATCG}
    \end{table}
\end{center}

(2) Then we prove that the normal forms of closed $BATC_G$ terms are basic $BATC_G$ terms.

Suppose that $p$ is a normal form of some closed $BATC_G$ term and suppose that $p$ is not a basic term. Let $p'$ denote the smallest sub-term of $p$ which is not a basic term. It implies that each sub-term of $p'$ is a basic term. Then we prove that $p$ is not a term in normal form. It is sufficient to induct on the structure of $p'$:

\begin{itemize}
  \item Case $p'\equiv e, e\in \mathbb{E}$. $p'$ is a basic term, which contradicts the assumption that $p'$ is not a basic term, so this case should not occur.
  \item Case $p'\equiv \phi, \phi\in G$. $p'$ is a basic term, which contradicts the assumption that $p'$ is not a basic term, so this case should not occur.
  \item Case $p'\equiv p_1\cdot p_2$. By induction on the structure of the basic term $p_1$:
      \begin{itemize}
        \item Subcase $p_1\in \mathbb{E}$. $p'$ would be a basic term, which contradicts the assumption that $p'$ is not a basic term;
        \item Subcase $p_1\in G$. $p'$ would be a basic term, which contradicts the assumption that $p'$ is not a basic term;
        \item Subcase $p_1\equiv e\cdot p_1'$. $RA5$ or $RA9$ rewriting rule can be applied. So $p$ is not a normal form;
        \item Subcase $p_1\equiv \phi\cdot p_1'$. $RG1$, $RG3$, $RG4$, $RG5$, $RG7$, or $RG8$ rewriting rules can be applied. So $p$ is not a normal form;
        \item Subcase $p_1\equiv p_1'+ p_1''$. $RA4$, $RA6$, $RG2$, or $RG6$ rewriting rules can be applied. So $p$ is not a normal form.
      \end{itemize}
  \item Case $p'\equiv p_1+ p_2$. By induction on the structure of the basic terms both $p_1$ and $p_2$, all subcases will lead to that $p'$ would be a basic term, which contradicts the assumption that $p'$ is not a basic term.
\end{itemize}
\end{proof}

We will define a term-deduction system which gives the operational semantics of $BATC_G$. We give the operational transition rules for $\epsilon$, atomic guard $\phi\in G_{at}$, atomic event $e\in\mathbb{E}$, operators $\cdot$ and $+$ as Table \ref{SETRForBATCG} shows. And the predicate $\xrightarrow{e}\surd$ represents successful termination after execution of the event $e$.

\begin{center}
    \begin{table}
        $$\frac{}{\langle\epsilon,s\rangle\rightarrow\langle\surd,s\rangle}$$
        $$\frac{}{\langle e,s\rangle\xrightarrow{e}\langle\surd,s'\rangle}\textrm{ if }s'\in effect(e,s)$$
        $$\frac{}{\langle\phi,s\rangle\rightarrow\langle\surd,s\rangle}\textrm{ if }test(\phi,s)$$
        $$\frac{\langle x,s\rangle\xrightarrow{e}\langle\surd,s'\rangle}{\langle x+ y,s\rangle\xrightarrow{e}\langle\surd,s'\rangle} \quad\frac{\langle x,s\rangle\xrightarrow{e}\langle x',s'\rangle}{\langle x+ y,s\rangle\xrightarrow{e}\langle x',s'\rangle}$$
        $$\frac{\langle y,s\rangle\xrightarrow{e}\langle\surd,s'\rangle}{\langle x+ y,s\rangle\xrightarrow{e}\langle\surd,s'\rangle} \quad\frac{\langle y,s\rangle\xrightarrow{e}\langle y',s'\rangle}{\langle x+ y,s\rangle\xrightarrow{e}\langle y',s'\rangle}$$
        $$\frac{\langle x,s\rangle\xrightarrow{e}\langle\surd,s'\rangle}{\langle x\cdot y,s\rangle\xrightarrow{e} \langle y,s'\rangle} \quad\frac{\langle x,s\rangle\xrightarrow{e}\langle x',s'\rangle}{\langle x\cdot y,s\rangle\xrightarrow{e}\langle x'\cdot y,s'\rangle}$$
        \caption{Single event transition rules of $BATC_G$}
        \label{SETRForBATCG}
    \end{table}
\end{center}

Note that, we replace the single atomic event $e\in\mathbb{E}$ by $X\subseteq\mathbb{E}$, we can obtain the pomset transition rules of $BATC_G$, and omit them.

\begin{theorem}[Congruence of $BATC_G$ with respect to truly concurrent bisimulation equivalences]\label{CBATCG}
(1) Pomset bisimulation equivalence $\sim_{p}$ is a congruence with respect to $BATC_G$.

(2) Step bisimulation equivalence $\sim_{s}$ is a congruence with respect to $BATC_G$.

(3) Hp-bisimulation equivalence $\sim_{hp}$ is a congruence with respect to $BATC_G$.

(4) Hhp-bisimulation equivalence $\sim_{hhp}$ is a congruence with respect to $BATC_G$.
\end{theorem}

\begin{proof}
(1) It is easy to see that pomset bisimulation is an equivalent relation on $BATC_G$ terms, we only need to prove that $\sim_{p}$ is preserved by the operators $\cdot$ and $+$. It is trivial and we leave the proof as an exercise for the readers.

(2) It is easy to see that step bisimulation is an equivalent relation on $BATC_G$ terms, we only need to prove that $\sim_{s}$ is preserved by the operators $\cdot$ and $+$. It is trivial and we leave the proof as an exercise for the readers.

(3) It is easy to see that hp-bisimulation is an equivalent relation on $BATC_G$ terms, we only need to prove that $\sim_{hp}$ is preserved by the operators $\cdot$ and $+$. It is trivial and we leave the proof as an exercise for the readers.

(4) It is easy to see that hhp-bisimulation is an equivalent relation on $BATC_G$ terms, we only need to prove that $\sim_{hhp}$ is preserved by the operators $\cdot$ and $+$. It is trivial and we leave the proof as an exercise for the readers.
\end{proof}

\begin{theorem}[Soundness of $BATC_G$ modulo truly concurrent bisimulation equivalences]\label{SBATCG}
(1) Let $x$ and $y$ be $BATC_G$ terms. If $BATC\vdash x=y$, then $x\sim_{p} y$.

(2) Let $x$ and $y$ be $BATC_G$ terms. If $BATC\vdash x=y$, then $x\sim_{s} y$.

(3) Let $x$ and $y$ be $BATC_G$ terms. If $BATC\vdash x=y$, then $x\sim_{hp} y$.

(4) Let $x$ and $y$ be $BATC_G$ terms. If $BATC\vdash x=y$, then $x\sim_{hhp} y$.
\end{theorem}

\begin{proof}
(1) Since pomset bisimulation $\sim_p$ is both an equivalent and a congruent relation, we only need to check if each axiom in Table \ref{AxiomsForBATCG} is sound modulo pomset bisimulation equivalence. We leave the proof as an exercise for the readers.

(2) Since step bisimulation $\sim_s$ is both an equivalent and a congruent relation, we only need to check if each axiom in Table \ref{AxiomsForBATCG} is sound modulo step bisimulation equivalence. We leave the proof as an exercise for the readers.

(3) Since hp-bisimulation $\sim_{hp}$ is both an equivalent and a congruent relation, we only need to check if each axiom in Table \ref{AxiomsForBATCG} is sound modulo hp-bisimulation equivalence. We leave the proof as an exercise for the readers.

(4) Since hhp-bisimulation $\sim_{hhp}$ is both an equivalent and a congruent relation, we only need to check if each axiom in Table \ref{AxiomsForBATCG} is sound modulo hhp-bisimulation equivalence. We leave the proof as an exercise for the readers.
\end{proof}

\begin{theorem}[Completeness of $BATC_G$ modulo truly concurrent bisimulation equivalences]\label{CBATCG}
(1) Let $p$ and $q$ be closed $BATC_G$ terms, if $p\sim_{p} q$ then $p=q$.

(2) Let $p$ and $q$ be closed $BATC_G$ terms, if $p\sim_{s} q$ then $p=q$.

(3) Let $p$ and $q$ be closed $BATC_G$ terms, if $p\sim_{hp} q$ then $p=q$.

(4) Let $p$ and $q$ be closed $BATC_G$ terms, if $p\sim_{hhp} q$ then $p=q$.
\end{theorem}

\begin{proof}
(1) Firstly, by the elimination theorem of $BATC_G$, we know that for each closed $BATC_G$ term $p$, there exists a closed basic $BATC_G$ term $p'$, such that $BATC_G\vdash p=p'$, so, we only need to consider closed basic $BATC_G$ terms.

The basic terms (see Definition \ref{BTBATCG}) modulo associativity and commutativity (AC) of conflict $+$ (defined by axioms $A1$ and $A2$ in Table \ref{AxiomsForBATCG}), and this equivalence is denoted by $=_{AC}$. Then, each equivalence class $s$ modulo AC of $+$ has the following normal form

$$s_1+\cdots+ s_k$$

with each $s_i$ either an atomic event, or an atomic guard, or of the form $t_1\cdot t_2$, and each $s_i$ is called the summand of $s$.

Now, we prove that for normal forms $n$ and $n'$, if $n\sim_{p} n'$ then $n=_{AC}n'$. It is sufficient to induct on the sizes of $n$ and $n'$.

\begin{itemize}
  \item Consider a summand $e$ of $n$. Then $\langle n,s\rangle\xrightarrow{e}\langle \surd,s'\rangle$, so $n\sim_p n'$ implies $\langle n',s\rangle\xrightarrow{e}\langle \surd,s\rangle$, meaning that $n'$ also contains the summand $e$.
  \item Consider a summand $\phi$ of $n$. Then $\langle n,s\rangle\rightarrow\langle \surd,s\rangle$, if $test(\phi,s)$ holds, so $n\sim_p n'$ implies $\langle n',s\rangle\rightarrow\langle \surd,s\rangle$, if $test(\phi,s)$ holds, meaning that $n'$ also contains the summand $\phi$.
  \item Consider a summand $t_1\cdot t_2$ of $n$. Then $\langle n,s\rangle\xrightarrow{t_1}\langle t_2,s'\rangle$, so $n\sim_p n'$ implies $\langle n',s\rangle\xrightarrow{t_1}\langle t_2',s'\rangle$ with $t_2\sim_p t_2'$, meaning that $n'$ contains a summand $t_1\cdot t_2'$. Since $t_2$ and $t_2'$ are normal forms and have sizes smaller than $n$ and $n'$, by the induction hypotheses $t_2\sim_p t_2'$ implies $t_2=_{AC} t_2'$.
\end{itemize}

So, we get $n=_{AC} n'$.

Finally, let $s$ and $t$ be basic terms, and $s\sim_p t$, there are normal forms $n$ and $n'$, such that $s=n$ and $t=n'$. The soundness theorem of $BATC_G$ modulo pomset bisimulation equivalence (see Theorem \ref{SBATCG}) yields $s\sim_p n$ and $t\sim_p n'$, so $n\sim_p s\sim_p t\sim_p n'$. Since if $n\sim_p n'$ then $n=_{AC}n'$, $s=n=_{AC}n'=t$, as desired.

(2) It can be proven similarly as (1).

(3) It can be proven similarly as (1).

(4) It can be proven similarly as (1).
\end{proof}

\begin{theorem}[Sufficient determinacy]
All related data environments with respect to $BATC_G$ can be sufficiently deterministic.
\end{theorem}

\begin{proof}
It only needs to check $effect(t,s)$ function is deterministic, and is sufficient to induct on the structure of term $t$. The only matter is the case $t=t_1+t_2$, with the help of guards, we can make $t_1=\phi_1\cdot t_1'$ and $t_2=\phi_2\cdot t_2'$, and $effct(t)$ is sufficiently deterministic.
\end{proof}

\subsubsection{$APTC$ with Guards}{\label{aptcg}}

In this subsection, we will extend $APTC$ with guards, which is abbreviated $APTC_G$. The set of basic guards $G$ with element $\phi,\psi,\cdots$, which is extended by the following formation rules:

$$\phi::=\delta|\epsilon|\neg\phi|\psi\in G_{at}|\phi+\psi|\phi\cdot\psi|\phi\parallel\psi$$

The set of axioms of $APTC_G$ including axioms of $BATC_G$ in Table \ref{AxiomsForBATCG} and the axioms are shown in Table \ref{AxiomsForAPTCG}.

\begin{center}
    \begin{table}
        \begin{tabular}{@{}ll@{}}
            \hline No. &Axiom\\
            $P1$ & $x\between y = x\parallel y + x\mid y$\\
            $P2$ & $e_1\parallel (e_2\cdot y) = (e_1\parallel e_2)\cdot y$\\
            $P3$ & $(e_1\cdot x)\parallel e_2 = (e_1\parallel e_2)\cdot x$\\
            $P4$ & $(e_1\cdot x)\parallel (e_2\cdot y) = (e_1\parallel e_2)\cdot (x\between y)$\\
            $P5$ & $(x+ y)\parallel z = (x\parallel z)+ (y\parallel z)$\\
            $P6$ & $x\parallel (y+ z) = (x\parallel y)+ (x\parallel z)$\\
            $P7$ & $\delta\parallel x = \delta$\\
            $P8$ & $x\parallel \delta = \delta$\\
            $P9$ & $\epsilon\parallel x = x$\\
            $P10$ & $x\parallel \epsilon = x$\\
            $C1$ & $e_1\mid e_2 = \gamma(e_1,e_2)$\\
            $C2$ & $e_1\mid (e_2\cdot y) = \gamma(e_1,e_2)\cdot y$\\
            $C3$ & $(e_1\cdot x)\mid e_2 = \gamma(e_1,e_2)\cdot x$\\
            $C4$ & $(e_1\cdot x)\mid (e_2\cdot y) = \gamma(e_1,e_2)\cdot (x\between y)$\\
            $C5$ & $(x+ y)\mid z = (x\mid z) + (y\mid z)$\\
            $C6$ & $x\mid (y+ z) = (x\mid y)+ (x\mid z)$\\
            $C7$ & $\delta\mid x = \delta$\\
            $C8$ & $x\mid\delta = \delta$\\
            $C9$ & $\epsilon\mid x = \delta$\\
            $C10$ & $x\mid\epsilon = \delta$\\
            $CE1$ & $\Theta(e) = e$\\
            $CE2$ & $\Theta(\delta) = \delta$\\
            $CE3$ & $\Theta(\epsilon) = \epsilon$\\
            $CE4$ & $\Theta(x+ y) = \Theta(x)\triangleleft y + \Theta(y)\triangleleft x$\\
            $CE5$ & $\Theta(x\cdot y)=\Theta(x)\cdot\Theta(y)$\\
            $CE6$ & $\Theta(x\parallel y) = ((\Theta(x)\triangleleft y)\parallel y)+ ((\Theta(y)\triangleleft x)\parallel x)$\\
            $CE7$ & $\Theta(x\mid y) = ((\Theta(x)\triangleleft y)\mid y)+ ((\Theta(y)\triangleleft x)\mid x)$\\
            $U1$ & $(\sharp(e_1,e_2))\quad e_1\triangleleft e_2 = \tau$\\
            $U2$ & $(\sharp(e_1,e_2),e_2\leq e_3)\quad e_1\triangleleft e_3 = e_1$\\
            $U3$ & $(\sharp(e_1,e_2),e_2\leq e_3)\quad e3\triangleleft e_1 = \tau$\\
            $U4$ & $e\triangleleft \delta = e$\\
            $U5$ & $\delta \triangleleft e = \delta$\\
            $U6$ & $e\triangleleft \epsilon = e$\\
            $U7$ & $\epsilon \triangleleft e = e$\\
            $U8$ & $(x+ y)\triangleleft z = (x\triangleleft z)+ (y\triangleleft z)$\\
            $U9$ & $(x\cdot y)\triangleleft z = (x\triangleleft z)\cdot (y\triangleleft z)$\\
            $U10$ & $(x\parallel y)\triangleleft z = (x\triangleleft z)\parallel (y\triangleleft z)$\\
            $U11$ & $(x\mid y)\triangleleft z = (x\triangleleft z)\mid (y\triangleleft z)$\\
            $U12$ & $x\triangleleft (y+ z) = (x\triangleleft y)\triangleleft z$\\
            $U13$ & $x\triangleleft (y\cdot z)=(x\triangleleft y)\triangleleft z$\\
            $U14$ & $x\triangleleft (y\parallel z) = (x\triangleleft y)\triangleleft z$\\
            $U15$ & $x\triangleleft (y\mid z) = (x\triangleleft y)\triangleleft z$\\
            $D1$ & $e\notin H\quad\partial_H(e) = e$\\
            $D2$ & $e\in H\quad \partial_H(e) = \delta$\\
            $D3$ & $\partial_H(\delta) = \delta$\\
            $D4$ & $\partial_H(x+ y) = \partial_H(x)+\partial_H(y)$\\
            $D5$ & $\partial_H(x\cdot y) = \partial_H(x)\cdot\partial_H(y)$\\
            $D6$ & $\partial_H(x\parallel y) = \partial_H(x)\parallel\partial_H(y)$\\
            $G11$ & $\phi(x\parallel y) =\phi x\parallel \phi y$\\
            $G12$ & $\phi(x\mid y) =\phi x\mid \phi y$\\
            $G13$ & $\phi\parallel \delta = \delta$\\
            $G14$ & $\delta\parallel \phi = \delta$\\
            $G15$ & $\phi\mid \delta = \delta$\\
            $G16$ & $\delta\mid \phi = \delta$\\
            $G17$ & $\phi\parallel \epsilon = \phi$\\
            $G18$ & $\epsilon\parallel \phi = \phi$\\
            $G19$ & $\phi\mid \epsilon = \delta$\\
            $G20$ & $\epsilon\mid \phi = \delta$\\
            $G21$ & $\phi\parallel\neg\phi = \delta$\\
            $G22$ & $\Theta(\phi) = \phi$\\
            $G23$ & $\partial_H(\phi) = \phi$\\
            $G24$ & $\phi_0\parallel\cdots\parallel\phi_n = \delta$ if $\forall s_0,\cdots,s_n\in S,\exists i\leq n.test(\neg\phi_i,s_0\cup\cdots\cup s_n)$\\
        \end{tabular}
        \caption{Axioms of $APTC_G$}
        \label{AxiomsForAPTCG}
    \end{table}
\end{center}

\begin{definition}[Basic terms of $APTC_G$]\label{BTAPTCG}
The set of basic terms of $APTC_G$, $\mathcal{B}(APTC_G)$, is inductively defined as follows:
\begin{enumerate}
    \item $\mathbb{E}\subset\mathcal{B}(APTC_G)$;
    \item $G\subset\mathcal{B}(APTC_G)$;
    \item if $e\in \mathbb{E}, t\in\mathcal{B}(APTC_G)$ then $e\cdot t\in\mathcal{B}(APTC_G)$;
    \item if $\phi\in G, t\in\mathcal{B}(APTC_G)$ then $\phi\cdot t\in\mathcal{B}(APTC_G)$;
    \item if $t,s\in\mathcal{B}(APTC_G)$ then $t+ s\in\mathcal{B}(APTC_G)$.
    \item if $t,s\in\mathcal{B}(APTC_G)$ then $t\parallel s\in\mathcal{B}(APTC_G)$.
\end{enumerate}
\end{definition}

Based on the definition of basic terms for $APTC_G$ (see Definition \ref{BTAPTCG}) and axioms of $APTC_G$, we can prove the elimination theorem of $APTC_G$.

\begin{theorem}[Elimination theorem of $APTC_G$]\label{ETAPTCG}
Let $p$ be a closed $APTC_G$ term. Then there is a basic $APTC_G$ term $q$ such that $APTC_G\vdash p=q$.
\end{theorem}

\begin{proof}
(1) Firstly, suppose that the following ordering on the signature of $APTC_G$ is defined: $\parallel > \cdot > +$ and the symbol $\parallel$ is given the lexicographical status for the first argument, then for each rewrite rule $p\rightarrow q$ in Table \ref{TRSForAPTCG} relation $p>_{lpo} q$ can easily be proved. We obtain that the term rewrite system shown in Table \ref{TRSForAPTCG} is strongly normalizing, for it has finitely many rewriting rules, and $>$ is a well-founded ordering on the signature of $APTC_G$, and if $s>_{lpo} t$, for each rewriting rule $s\rightarrow t$ is in Table \ref{TRSForAPTCG} (see Theorem \ref{SN}).

\begin{center}
    \begin{table}
        \begin{tabular}{@{}ll@{}}
            \hline No. &Rewriting Rule\\
            $RP1$ & $x\between y \rightarrow x\parallel y + x\mid y$\\
            $RP2$ & $e_1\parallel (e_2\cdot y) \rightarrow (e_1\parallel e_2)\cdot y$\\
            $RP3$ & $(e_1\cdot x)\parallel e_2 \rightarrow (e_1\parallel e_2)\cdot x$\\
            $RP4$ & $(e_1\cdot x)\parallel (e_2\cdot y) \rightarrow (e_1\parallel e_2)\cdot (x\between y)$\\
            $RP5$ & $(x+ y)\parallel z \rightarrow (x\parallel z)+ (y\parallel z)$\\
            $RP6$ & $x\parallel (y+ z) \rightarrow (x\parallel y)+ (x\parallel z)$\\
            $RP7$ & $\delta\parallel x \rightarrow \delta$\\
            $RP8$ & $x\parallel \delta \rightarrow \delta$\\
            $RP9$ & $\epsilon\parallel x \rightarrow x$\\
            $RP10$ & $x\parallel \epsilon \rightarrow x$\\
            $RC1$ & $e_1\mid e_2 \rightarrow \gamma(e_1,e_2)$\\
            $RC2$ & $e_1\mid (e_2\cdot y) \rightarrow \gamma(e_1,e_2)\cdot y$\\
            $RC3$ & $(e_1\cdot x)\mid e_2 \rightarrow \gamma(e_1,e_2)\cdot x$\\
            $RC4$ & $(e_1\cdot x)\mid (e_2\cdot y) \rightarrow \gamma(e_1,e_2)\cdot (x\between y)$\\
            $RC5$ & $(x+ y)\mid z \rightarrow (x\mid z) + (y\mid z)$\\
            $RC6$ & $x\mid (y+ z) \rightarrow (x\mid y)+ (x\mid z)$\\
            $RC7$ & $\delta\mid x \rightarrow \delta$\\
            $RC8$ & $x\mid\delta \rightarrow \delta$\\
            $RC9$ & $\epsilon\mid x \rightarrow \delta$\\
            $RC10$ & $x\mid\epsilon \rightarrow \delta$\\
            $RCE1$ & $\Theta(e) \rightarrow e$\\
            $RCE2$ & $\Theta(\delta) \rightarrow \delta$\\
            $RCE3$ & $\Theta(\epsilon) \rightarrow \epsilon$\\
            $RCE4$ & $\Theta(x+ y) \rightarrow \Theta(x)\triangleleft y + \Theta(y)\triangleleft x$\\
            $RCE5$ & $\Theta(x\cdot y)\rightarrow\Theta(x)\cdot\Theta(y)$\\
            $RCE6$ & $\Theta(x\parallel y) \rightarrow ((\Theta(x)\triangleleft y)\parallel y)+ ((\Theta(y)\triangleleft x)\parallel x)$\\
            $RCE7$ & $\Theta(x\mid y) \rightarrow ((\Theta(x)\triangleleft y)\mid y)+ ((\Theta(y)\triangleleft x)\mid x)$\\
            $RU1$ & $(\sharp(e_1,e_2))\quad e_1\triangleleft e_2 \rightarrow \tau$\\
            $RU2$ & $(\sharp(e_1,e_2),e_2\leq e_3)\quad e_1\triangleleft e_3 \rightarrow e_1$\\
            $RU3$ & $(\sharp(e_1,e_2),e_2\leq e_3)\quad e3\triangleleft e_1 \rightarrow \tau$\\
            $RU4$ & $e\triangleleft \delta \rightarrow e$\\
            $RU5$ & $\delta \triangleleft e \rightarrow \delta$\\
            $RU6$ & $e\triangleleft \epsilon \rightarrow e$\\
            $RU7$ & $\epsilon \triangleleft e \rightarrow e$\\
            $RU8$ & $(x+ y)\triangleleft z \rightarrow (x\triangleleft z)+ (y\triangleleft z)$\\
            $RU9$ & $(x\cdot y)\triangleleft z \rightarrow (x\triangleleft z)\cdot (y\triangleleft z)$\\
            $RU10$ & $(x\parallel y)\triangleleft z \rightarrow (x\triangleleft z)\parallel (y\triangleleft z)$\\
            $RU11$ & $(x\mid y)\triangleleft z \rightarrow (x\triangleleft z)\mid (y\triangleleft z)$\\
            $RU12$ & $x\triangleleft (y+ z) \rightarrow (x\triangleleft y)\triangleleft z$\\
            $RU13$ & $x\triangleleft (y\cdot z)\rightarrow(x\triangleleft y)\triangleleft z$\\
            $RU14$ & $x\triangleleft (y\parallel z) \rightarrow (x\triangleleft y)\triangleleft z$\\
            $RU15$ & $x\triangleleft (y\mid z) \rightarrow (x\triangleleft y)\triangleleft z$\\
            $RD1$ & $e\notin H\quad\partial_H(e) \rightarrow e$\\
            $RD2$ & $e\in H\quad \partial_H(e) \rightarrow \delta$\\
            $RD3$ & $\partial_H(\delta) \rightarrow \delta$\\
            $RD4$ & $\partial_H(x+ y) \rightarrow \partial_H(x)+\partial_H(y)$\\
            $RD5$ & $\partial_H(x\cdot y) \rightarrow \partial_H(x)\cdot\partial_H(y)$\\
            $RD6$ & $\partial_H(x\parallel y) \rightarrow \partial_H(x)\parallel\partial_H(y)$\\
            $RG11$ & $\phi(x\parallel y) \rightarrow\phi x\parallel \phi y$\\
            $RG12$ & $\phi(x\mid y) \rightarrow\phi x\mid \phi y$\\
            $RG13$ & $\phi\parallel \delta \rightarrow \delta$\\
            $RG14$ & $\delta\parallel \phi \rightarrow \delta$\\
            $RG15$ & $\phi\mid \delta \rightarrow \delta$\\
            $RG16$ & $\delta\mid \phi \rightarrow \delta$\\
            $RG17$ & $\phi\parallel \epsilon \rightarrow \phi$\\
            $RG18$ & $\epsilon\parallel \phi \rightarrow \phi$\\
            $RG19$ & $\phi\mid \epsilon \rightarrow \delta$\\
            $RG20$ & $\epsilon\mid \phi \rightarrow \delta$\\
            $RG21$ & $\phi\parallel\neg\phi \rightarrow \delta$\\
            $RG22$ & $\Theta(\phi) \rightarrow \phi$\\
            $RG23$ & $\partial_H(\phi) \rightarrow \phi$\\
            $RG24$ & $\phi_0\parallel\cdots\parallel\phi_n \rightarrow \delta$ if $\forall s_0,\cdots,s_n\in S,\exists i\leq n.test(\neg\phi_i,s_0\cup\cdots\cup s_n)$\\
        \end{tabular}
        \caption{Term rewrite system of $APTC_G$}
        \label{TRSForAPTCG}
    \end{table}
\end{center}

(2) Then we prove that the normal forms of closed $APTC_G$ terms are basic $APTC_G$ terms.

Suppose that $p$ is a normal form of some closed $APTC_G$ term and suppose that $p$ is not a basic $APTC_G$ term. Let $p'$ denote the smallest sub-term of $p$ which is not a basic $APTC_G$ term. It implies that each sub-term of $p'$ is a basic $APTC_G$ term. Then we prove that $p$ is not a term in normal form. It is sufficient to induct on the structure of $p'$:

\begin{itemize}
  \item Case $p'\equiv e, e\in \mathbb{E}$. $p'$ is a basic $APTC_G$ term, which contradicts the assumption that $p'$ is not a basic $APTC_G$ term, so this case should not occur.
  \item Case $p'\equiv \phi, \phi\in G$. $p'$ is a basic term, which contradicts the assumption that $p'$ is not a basic term, so this case should not occur.
  \item Case $p'\equiv p_1\cdot p_2$. By induction on the structure of the basic $APTC_G$ term $p_1$:
      \begin{itemize}
        \item Subcase $p_1\in \mathbb{E}$. $p'$ would be a basic $APTC_G$ term, which contradicts the assumption that $p'$ is not a basic $APTC_G$ term;
        \item Subcase $p_1\in G$. $p'$ would be a basic term, which contradicts the assumption that $p'$ is not a basic term;
        \item Subcase $p_1\equiv e\cdot p_1'$. $RA5$ or $RA9$ rewriting rules in Table \ref{TRSForBATCG} can be applied. So $p$ is not a normal form;
        \item Subcase $p_1\equiv \phi\cdot p_1'$. $RG1$, $RG3$, $RG4$, $RG5$, $RG7$, or $RG8$ rewriting rules can be applied. So $p$ is not a normal form;
        \item Subcase $p_1\equiv p_1'+ p_1''$. $RA4$, $RA6$, $RG2$, or $RG6$ rewriting rules in Table \ref{TRSForBATCG} can be applied. So $p$ is not a normal form;
        \item Subcase $p_1\equiv p_1'\parallel p_1''$. $RP2$-$RP10$ rewrite rules in Table \ref{TRSForAPTCG} can be applied. So $p$ is not a normal form;
        \item Subcase $p_1\equiv p_1'\mid p_1''$. $RC1$-$RC10$ rewrite rules in Table \ref{TRSForAPTCG} can be applied. So $p$ is not a normal form;
        \item Subcase $p_1\equiv \Theta(p_1')$. $RCE1$-$RCE7$ rewrite rules in Table \ref{TRSForAPTCG} can be applied. So $p$ is not a normal form;
        \item Subcase $p_1\equiv \partial_H(p_1')$. $RD1$-$RD6$ rewrite rules in Table \ref{TRSForAPTCG} can be applied. So $p$ is not a normal form.
      \end{itemize}
  \item Case $p'\equiv p_1+ p_2$. By induction on the structure of the basic $APTC_G$ terms both $p_1$ and $p_2$, all subcases will lead to that $p'$ would be a basic $APTC_G$ term, which contradicts the assumption that $p'$ is not a basic $APTC_G$ term.
  \item Case $p'\equiv p_1\parallel p_2$. By induction on the structure of the basic $APTC_G$ terms both $p_1$ and $p_2$, all subcases will lead to that $p'$ would be a basic $APTC_G$ term, which contradicts the assumption that $p'$ is not a basic $APTC_G$ term.
  \item Case $p'\equiv p_1\mid p_2$. By induction on the structure of the basic $APTC_G$ terms both $p_1$ and $p_2$, all subcases will lead to that $p'$ would be a basic $APTC_G$ term, which contradicts the assumption that $p'$ is not a basic $APTC_G$ term.
  \item Case $p'\equiv \Theta(p_1)$. By induction on the structure of the basic $APTC_G$ term $p_1$, $RCE1-RCE7$ rewrite rules in Table \ref{TRSForAPTCG} can be applied. So $p$ is not a normal form.
  \item Case $p'\equiv p_1\triangleleft p_2$. By induction on the structure of the basic $APTC_G$ terms both $p_1$ and $p_2$, all subcases will lead to that $p'$ would be a basic $APTC_G$ term, which contradicts the assumption that $p'$ is not a basic $APTC_G$ term.
  \item Case $p'\equiv \partial_H(p_1)$. By induction on the structure of the basic $APTC_G$ terms of $p_1$, all subcases will lead to that $p'$ would be a basic $APTC_G$ term, which contradicts the assumption that $p'$ is not a basic $APTC_G$ term.
\end{itemize}
\end{proof}

We will define a term-deduction system which gives the operational semantics of $APTC_G$. Two atomic events $e_1$ and $e_2$ are in race condition, which are denoted $e_1\% e_2$.

\begin{center}
    \begin{table}
        $$\frac{}{\langle e_1\parallel\cdots \parallel e_n,s\rangle\xrightarrow{\{e_1,\cdots,e_n\}}\langle\surd,s'\rangle}\textrm{ if }s'\in effect(e_1,s)\cup\cdots\cup effect(e_n,s)$$
        $$\frac{}{\langle\phi_1\parallel\cdots\parallel \phi_n,s\rangle\rightarrow\langle\surd,s\rangle}\textrm{ if }test(\phi_1,s),\cdots,test(\phi_n,s)$$

        $$\frac{\langle x,s\rangle\xrightarrow{e_1}\langle\surd,s'\rangle\quad \langle y,s\rangle\xrightarrow{e_2}\langle\surd,s''\rangle}{\langle x\parallel y,s\rangle\xrightarrow{\{e_1,e_2\}}\langle\surd,s'\cup s''\rangle} \quad\frac{\langle x,s\rangle\xrightarrow{e_1}\langle x',s'\rangle\quad \langle y,s\rangle\xrightarrow{e_2}\langle\surd,s''\rangle}{\langle x\parallel y,s\rangle\xrightarrow{\{e_1,e_2\}}\langle x',s'\cup s''\rangle}$$

        $$\frac{\langle x,s\rangle\xrightarrow{e_1}\langle\surd,s'\rangle\quad \langle y,s\rangle\xrightarrow{e_2}\langle y',s''\rangle}{\langle x\parallel y,s\rangle\xrightarrow{\{e_1,e_2\}}\langle y',s'\cup s''\rangle} \quad\frac{\langle x,s\rangle\xrightarrow{e_1}\langle x',s'\rangle\quad \langle y,s\rangle\xrightarrow{e_2}\langle y',s''\rangle}{\langle x\parallel y,s\rangle\xrightarrow{\{e_1,e_2\}}\langle x'\between y',s'\cup s''\rangle}$$

        $$\frac{\langle x,s\rangle\xrightarrow{e_1}\langle\surd,s'\rangle\quad \langle y,s\rangle\xnrightarrow{e_2}\quad(e_1\%e_2)}{\langle x\parallel y,s\rangle\xrightarrow{e_1}\langle y,s'\rangle} \quad\frac{\langle x,s\rangle\xrightarrow{e_1}\langle x',s'\rangle\quad \langle y,s\rangle\xnrightarrow{e_2}\quad(e_1\%e_2)}{\langle x\parallel y,s\rangle\xrightarrow{e_1}\langle x'\between y,s'\rangle}$$

        $$\frac{\langle x,s\rangle\xnrightarrow{e_1}\quad \langle y,s\rangle\xrightarrow{e_2}\langle\surd,s''\rangle\quad(e_1\%e_2)}{\langle x\parallel y,s\rangle\xrightarrow{e_2}\langle x,s''\rangle} \quad\frac{\langle x,s\rangle\xnrightarrow{e_1}\quad \langle y,s\rangle\xrightarrow{e_2}\langle y',s''\rangle\quad(e_1\%e_2)}{\langle x\parallel y,s\rangle\xrightarrow{e_2}\langle x\between y',s''\rangle}$$

        $$\frac{\langle x,s\rangle\xrightarrow{e_1}\langle\surd,s'\rangle\quad \langle y,s\rangle\xrightarrow{e_2}\langle\surd,s''\rangle}{\langle x\mid y,s\rangle\xrightarrow{\gamma(e_1,e_2)}\langle\surd,effect(\gamma(e_1,e_2),s)\rangle} \quad\frac{\langle x,s\rangle\xrightarrow{e_1}\langle x',s'\rangle\quad \langle y,s\rangle\xrightarrow{e_2}\langle\surd,s''\rangle}{\langle x\mid y,s\rangle\xrightarrow{\gamma(e_1,e_2)}\langle x',effect(\gamma(e_1,e_2),s)\rangle}$$

        $$\frac{\langle x,s\rangle\xrightarrow{e_1}\langle\surd,s'\rangle\quad \langle y,s\rangle\xrightarrow{e_2}\langle y',s''\rangle}{\langle x\mid y,s\rangle\xrightarrow{\gamma(e_1,e_2)}\langle y',effect(\gamma(e_1,e_2),s)\rangle} \quad\frac{\langle x,s\rangle\xrightarrow{e_1}\langle x',s'\rangle\quad \langle y,s\rangle\xrightarrow{e_2}\langle y',s''\rangle}{\langle x\mid y,s\rangle\xrightarrow{\gamma(e_1,e_2)}\langle x'\between y',effect(\gamma(e_1,e_2),s)\rangle}$$

        $$\frac{\langle x,s\rangle\xrightarrow{e_1}\langle\surd,s'\rangle\quad (\sharp(e_1,e_2))}{\langle \Theta(x),s\rangle\xrightarrow{e_1}\langle\surd,s'\rangle} \quad\frac{\langle x,s\rangle\xrightarrow{e_2}\langle\surd,s''\rangle\quad (\sharp(e_1,e_2))}{\langle\Theta(x),s\rangle\xrightarrow{e_2}\langle\surd,s''\rangle}$$

        $$\frac{\langle x,s\rangle\xrightarrow{e_1}\langle x',s'\rangle\quad (\sharp(e_1,e_2))}{\langle\Theta(x),s\rangle\xrightarrow{e_1}\langle\Theta(x'),s'\rangle} \quad\frac{\langle x,s\rangle\xrightarrow{e_2}\langle x'',s''\rangle\quad (\sharp(e_1,e_2))}{\langle\Theta(x),s\rangle\xrightarrow{e_2}\langle\Theta(x''),s''\rangle}$$

        $$\frac{\langle x,s\rangle\xrightarrow{e_1}\langle\surd,s'\rangle \quad \langle y,s\rangle\nrightarrow^{e_2}\quad (\sharp(e_1,e_2))}{\langle x\triangleleft y,s\rangle\xrightarrow{\tau}\langle\surd,s'\rangle}
        \quad\frac{\langle x,s\rangle\xrightarrow{e_1}\langle x',s'\rangle \quad \langle y,s\rangle\nrightarrow^{e_2}\quad (\sharp(e_1,e_2))}{\langle x\triangleleft y,s\rangle\xrightarrow{\tau}\langle x',s'\rangle}$$

        $$\frac{\langle x,s\rangle\xrightarrow{e_1}\langle\surd,s\rangle \quad \langle y,s\rangle\nrightarrow^{e_3}\quad (\sharp(e_1,e_2),e_2\leq e_3)}{\langle x\triangleleft y,s\rangle\xrightarrow{e_1}\langle\surd,s'\rangle}
        \quad\frac{\langle x,s\rangle\xrightarrow{e_1}\langle x',s'\rangle \quad \langle y,s\rangle\nrightarrow^{e_3}\quad (\sharp(e_1,e_2),e_2\leq e_3)}{\langle x\triangleleft y,s\rangle\xrightarrow{e_1}\langle x',s'\rangle}$$

        $$\frac{\langle x,s\rangle\xrightarrow{e_3}\langle\surd,s'\rangle \quad \langle y,s\rangle\nrightarrow^{e_2}\quad (\sharp(e_1,e_2),e_1\leq e_3)}{\langle x\triangleleft y,s\rangle\xrightarrow{\tau}\langle\surd,s'\rangle}
        \quad\frac{\langle x,s\rangle\xrightarrow{e_3}\langle x',s'\rangle \quad \langle y,s\rangle\nrightarrow^{e_2}\quad (\sharp(e_1,e_2),e_1\leq e_3)}{\langle x\triangleleft y,s\rangle\xrightarrow{\tau}\langle x',s'\rangle}$$

        $$\frac{\langle x,s\rangle\xrightarrow{e}\langle\surd,s'\rangle}{\langle\partial_H(x),s\rangle\xrightarrow{e}\langle\surd,s'\rangle}\quad (e\notin H)\quad\frac{\langle x,s\rangle\xrightarrow{e}\langle x',s'\rangle}{\langle\partial_H(x),s\rangle\xrightarrow{e}\langle\partial_H(x'),s'\rangle}\quad(e\notin H)$$
        \caption{Transition rules of $APTC_G$}
        \label{TRForAPTCG}
    \end{table}
\end{center}

\begin{theorem}[Generalization of $APTC_G$ with respect to $BATC_G$]
$APTC_G$ is a generalization of $BATC_G$.
\end{theorem}

\begin{proof}
It follows from the following three facts.

\begin{enumerate}
  \item The transition rules of $BATC_G$ in section \ref{batcg} are all source-dependent;
  \item The sources of the transition rules $APTC_G$ contain an occurrence of $\between$, or $\parallel$, or $\mid$, or $\Theta$, or $\triangleleft$;
  \item The transition rules of $APTC_G$ are all source-dependent.
\end{enumerate}

So, $APTC_G$ is a generalization of $BATC_G$, that is, $BATC_G$ is an embedding of $APTC_G$, as desired.
\end{proof}

\begin{theorem}[Congruence of $APTC_G$ with respect to truly concurrent bisimulation equivalences]\label{CAPTCG}
(1) Pomset bisimulation equivalence $\sim_{p}$ is a congruence with respect to $APTC_G$.

(2) Step bisimulation equivalence $\sim_{s}$ is a congruence with respect to $APTC_G$.

(3) Hp-bisimulation equivalence $\sim_{hp}$ is a congruence with respect to $APTC_G$.

(4) Hhp-bisimulation equivalence $\sim_{hhp}$ is a congruence with respect to $APTC_G$.
\end{theorem}

\begin{proof}
(1) It is easy to see that pomset bisimulation is an equivalent relation on $APTC_G$ terms, we only need to prove that $\sim_{p}$ is preserved by the operators $\parallel$, $\mid$, $\Theta$, $\triangleleft$, $\partial_H$. It is trivial and we leave the proof as an exercise for the readers.

(2) It is easy to see that step bisimulation is an equivalent relation on $APTC_G$ terms, we only need to prove that $\sim_{s}$ is preserved by the operators $\parallel$, $\mid$, $\Theta$, $\triangleleft$, $\partial_H$. It is trivial and we leave the proof as an exercise for the readers.

(3) It is easy to see that hp-bisimulation is an equivalent relation on $APTC_G$ terms, we only need to prove that $\sim_{hp}$ is preserved by the operators $\parallel$, $\mid$, $\Theta$, $\triangleleft$, $\partial_H$. It is trivial and we leave the proof as an exercise for the readers.

(4) It is easy to see that hhp-bisimulation is an equivalent relation on $APTC_G$ terms, we only need to prove that $\sim_{hhp}$ is preserved by the operators $\parallel$, $\mid$, $\Theta$, $\triangleleft$, $\partial_H$. It is trivial and we leave the proof as an exercise for the readers.
\end{proof}

\begin{theorem}[Soundness of $APTC_G$ modulo truly concurrent bisimulation equivalences]\label{SAPTCG}
(1) Let $x$ and $y$ be $APTC_G$ terms. If $APTC\vdash x=y$, then $x\sim_{p} y$.

(2) Let $x$ and $y$ be $APTC_G$ terms. If $APTC\vdash x=y$, then $x\sim_{s} y$.

(3) Let $x$ and $y$ be $APTC_G$ terms. If $APTC\vdash x=y$, then $x\sim_{hp} y$.
\end{theorem}

\begin{proof}
(1) Since pomset bisimulation $\sim_p$ is both an equivalent and a congruent relation, we only need to check if each axiom in Table \ref{AxiomsForAPTCG} is sound modulo pomset bisimulation equivalence. We leave the proof as an exercise for the readers.

(2) Since step bisimulation $\sim_s$ is both an equivalent and a congruent relation, we only need to check if each axiom in Table \ref{AxiomsForAPTCG} is sound modulo step bisimulation equivalence. We leave the proof as an exercise for the readers.

(3) Since hp-bisimulation $\sim_{hp}$ is both an equivalent and a congruent relation, we only need to check if each axiom in Table \ref{AxiomsForAPTCG} is sound modulo hp-bisimulation equivalence. We leave the proof as an exercise for the readers.
\end{proof}

\begin{theorem}[Completeness of $APTC_G$ modulo truly concurrent bisimulation equivalences]\label{CAPTCG}
(1) Let $p$ and $q$ be closed $APTC_G$ terms, if $p\sim_{p} q$ then $p=q$.

(2) Let $p$ and $q$ be closed $APTC_G$ terms, if $p\sim_{s} q$ then $p=q$.

(3) Let $p$ and $q$ be closed $APTC_G$ terms, if $p\sim_{hp} q$ then $p=q$.
\end{theorem}

\begin{proof}
(1) Firstly, by the elimination theorem of $APTC_G$ (see Theorem \ref{ETAPTCG}), we know that for each closed $APTC_G$ term $p$, there exists a closed basic $APTC_G$ term $p'$, such that $APTC\vdash p=p'$, so, we only need to consider closed basic $APTC_G$ terms.

The basic terms (see Definition \ref{BTAPTCG}) modulo associativity and commutativity (AC) of conflict $+$ (defined by axioms $A1$ and $A2$ in Table \ref{AxiomsForBATCG}), and these equivalences is denoted by $=_{AC}$. Then, each equivalence class $s$ modulo AC of $+$ has the following normal form

$$s_1+\cdots+ s_k$$

with each $s_i$ either an atomic event, or an atomic guard, or of the form

$$t_1\cdot\cdots\cdot t_m$$

with each $t_j$ either an atomic event, or an atomic guard, or of the form

$$u_1\parallel\cdots\parallel u_l$$

with each $u_l$ an atomic event, or an atomic guard, and each $s_i$ is called the summand of $s$.

Now, we prove that for normal forms $n$ and $n'$, if $n\sim_{p} n'$ then $n=_{AC}n'$. It is sufficient to induct on the sizes of $n$ and $n'$.

\begin{itemize}
  \item Consider a summand $e$ of $n$. Then $\langle n,s\rangle\xrightarrow{e}\langle \surd,s'\rangle$, so $n\sim_p n'$ implies $\langle n',s\rangle\xrightarrow{e}\langle \surd,s\rangle$, meaning that $n'$ also contains the summand $e$.
  \item Consider a summand $\phi$ of $n$. Then $\langle n,s\rangle\rightarrow\langle \surd,s\rangle$, if $test(\phi,s)$ holds, so $n\sim_p n'$ implies $\langle n',s\rangle\rightarrow\langle \surd,s\rangle$, if $test(\phi,s)$ holds, meaning that $n'$ also contains the summand $\phi$.
  \item Consider a summand $t_1\cdot t_2$ of $n$,
  \begin{itemize}
    \item if $t_1\equiv e'$, then $\langle n,s\rangle\xrightarrow{e'}\langle t_2,s'\rangle$, so $n\sim_p n'$ implies $\langle n',s\rangle\xrightarrow{e'}\langle t_2',s'\rangle$ with $t_2\sim_p t_2'$, meaning that $n'$ contains a summand $e'\cdot t_2'$. Since $t_2$ and $t_2'$ are normal forms and have sizes smaller than $n$ and $n'$, by the induction hypotheses if $t_2\sim_p t_2'$ then $t_2=_{AC} t_2'$;
    \item if $t_1\equiv \phi'$, then $\langle n,s\rangle\rightarrow\langle t_2,s\rangle$, if $test(\phi',s)$ holds, so $n\sim_p n'$ implies $\langle n',s\rangle\rightarrow\langle t_2',s\rangle$ with $t_2\sim_p t_2'$, if $test(\phi',s)$ holds, meaning that $n'$ contains a summand $\phi'\cdot t_2'$. Since $t_2$ and $t_2'$ are normal forms and have sizes smaller than $n$ and $n'$, by the induction hypotheses if $t_2\sim_p t_2'$ then $t_2=_{AC} t_2'$;
    \item if $t_1\equiv e_1\parallel\cdots\parallel e_l$, then $\langle n,s\rangle\xrightarrow{\{e_1,\cdots,e_l\}}\langle t_2,s'\rangle$, so $n\sim_p n'$ implies $\langle n',s\rangle\xrightarrow{\{e_1,\cdots,e_l\}}\langle t_2',s'\rangle$ with $t_2\sim_p t_2'$, meaning that $n'$ contains a summand $(e_1\parallel\cdots\parallel e_l)\cdot t_2'$. Since $t_2$ and $t_2'$ are normal forms and have sizes smaller than $n$ and $n'$, by the induction hypotheses if $t_2\sim_p t_2'$ then $t_2=_{AC} t_2'$;
    \item if $t_1\equiv \phi_1\parallel\cdots\parallel \phi_l$, then $\langle n,s\rangle\rightarrow\langle t_2,s\rangle$, if $test(\phi_1,s),\cdots,test(\phi_l,s)$ hold, so $n\sim_p n'$ implies $\langle n',s\rangle\rightarrow\langle t_2',s\rangle$ with $t_2\sim_p t_2'$, if $test(\phi_1,s),\cdots,test(\phi_l,s)$ hold, meaning that $n'$ contains a summand $(\phi_1\parallel\cdots\parallel \phi_l)\cdot t_2'$. Since $t_2$ and $t_2'$ are normal forms and have sizes smaller than $n$ and $n'$, by the induction hypotheses if $t_2\sim_p t_2'$ then $t_2=_{AC} t_2'$.
  \end{itemize}
\end{itemize}

So, we get $n=_{AC} n'$.

Finally, let $s$ and $t$ be basic $APTC_G$ terms, and $s\sim_p t$, there are normal forms $n$ and $n'$, such that $s=n$ and $t=n'$. The soundness theorem of $APTC_G$ modulo pomset bisimulation equivalence (see Theorem \ref{SAPTCG}) yields $s\sim_p n$ and $t\sim_p n'$, so $n\sim_p s\sim_p t\sim_p n'$. Since if $n\sim_p n'$ then $n=_{AC}n'$, $s=n=_{AC}n'=t$, as desired.

(2) It can be proven similarly as (1).

(3) It can be proven similarly as (1).
\end{proof}

\begin{theorem}[Sufficient determinacy]
All related data environments with respect to $APTC_G$ can be sufficiently deterministic.
\end{theorem}

\begin{proof}
It only needs to check $effect(t,s)$ function is deterministic, and is sufficient to induct on the structure of term $t$. The new matter is the case $t=t_1\between t_2$, because $t_1$ and $t_2$ may be in race condition, the whole thing is $t_1\between t_2 = t_1\cdot t_2 + t_2\cdot t_1 + t_1\parallel t_2 + t_1\mid t_2$. We can make $effct(t)$ be sufficiently deterministic: eliminating non-determinacy caused by race condition during modeling time by use of empty event $\epsilon$. We can make $t=t_1\parallel t_2$ ($t_1\% t_2$) be $t=(\epsilon \cdot t_1)\parallel t_2$ or $t=t_1\parallel (\epsilon\cdot t_2)$ during modeling phase, and then $effct(t,s)$ becomes sufficiently deterministic.
\end{proof}

\subsubsection{Recursion}{\label{recg}}

In this subsection, we introduce recursion to capture infinite processes based on $APTC_G$. In the following, $E,F,G$ are recursion specifications, $X,Y,Z$ are recursive variables.

\begin{center}
    \begin{table}
        $$\frac{\langle t_i(\langle X_1|E\rangle,\cdots,\langle X_n|E\rangle),s\rangle\xrightarrow{\{e_1,\cdots,e_k\}}\langle\surd,s'\rangle}{\langle\langle X_i|E\rangle,s\rangle\xrightarrow{\{e_1,\cdots,e_k\}}\langle\surd,s'\rangle}$$
        $$\frac{\langle t_i(\langle X_1|E\rangle,\cdots,\langle X_n|E\rangle),s\rangle\xrightarrow{\{e_1,\cdots,e_k\}} \langle y,s'\rangle}{\langle\langle X_i|E\rangle,s\rangle\xrightarrow{\{e_1,\cdots,e_k\}} \langle y,s'\rangle}$$
        \caption{Transition rules of guarded recursion}
        \label{TRForGRG}
    \end{table}
\end{center}

\begin{theorem}[Conservitivity of $APTC_G$ with guarded recursion]
$APTC_G$ with guarded recursion is a conservative extension of $APTC_G$.
\end{theorem}

\begin{proof}
Since the transition rules of $APTC_G$ are source-dependent, and the transition rules for guarded recursion in Table \ref{TRForGRG} contain only a fresh constant in their source, so the transition rules of $APTC_G$ with guarded recursion are a conservative extension of those of $APTC_G$.
\end{proof}

\begin{theorem}[Congruence theorem of $APTC_G$ with guarded recursion]
Truly concurrent bisimulation equivalences $\sim_{p}$, $\sim_s$ and $\sim_{hp}$ are all congruences with respect to $APTC_G$ with guarded recursion.
\end{theorem}

\begin{proof}
It follows the following two facts:
\begin{enumerate}
  \item in a guarded recursive specification, right-hand sides of its recursive equations can be adapted to the form by applications of the axioms in $APTC_G$ and replacing recursion variables by the right-hand sides of their recursive equations;
  \item truly concurrent bisimulation equivalences $\sim_{p}$, $\sim_s$ and $\sim_{hp}$ are all congruences with respect to all operators of $APTC_G$.
\end{enumerate}
\end{proof}

\begin{theorem}[Elimination theorem of $APTC_G$ with linear recursion]\label{ETRecursionG}
Each process term in $APTC_G$ with linear recursion is equal to a process term $\langle X_1|E\rangle$ with $E$ a linear recursive specification.
\end{theorem}

\begin{proof}
By applying structural induction with respect to term size, each process term $t_1$ in $APTC_G$ with linear recursion generates a process can be expressed in the form of equations

$$t_i=(a_{i11}\parallel\cdots\parallel a_{i1i_1})t_{i1}+\cdots+(a_{ik_i1}\parallel\cdots\parallel a_{ik_ii_k})t_{ik_i}+(b_{i11}\parallel\cdots\parallel b_{i1i_1})+\cdots+(b_{il_i1}\parallel\cdots\parallel b_{il_ii_l})$$

for $i\in\{1,\cdots,n\}$. Let the linear recursive specification $E$ consist of the recursive equations

$$X_i=(a_{i11}\parallel\cdots\parallel a_{i1i_1})X_{i1}+\cdots+(a_{ik_i1}\parallel\cdots\parallel a_{ik_ii_k})X_{ik_i}+(b_{i11}\parallel\cdots\parallel b_{i1i_1})+\cdots+(b_{il_i1}\parallel\cdots\parallel b_{il_ii_l})$$

for $i\in\{1,\cdots,n\}$. Replacing $X_i$ by $t_i$ for $i\in\{1,\cdots,n\}$ is a solution for $E$, $RSP$ yields $t_1=\langle X_1|E\rangle$.
\end{proof}

\begin{theorem}[Soundness of $APTC_G$ with guarded recursion]\label{SAPTC_GRG}
Let $x$ and $y$ be $APTC_G$ with guarded recursion terms. If $APTC_G\textrm{ with guarded recursion}\vdash x=y$, then

(1) $x\sim_{s} y$.

(2) $x\sim_{p} y$.

(3) $x\sim_{hp} y$.
\end{theorem}

\begin{proof}
(1) Since step bisimulation $\sim_s$ is both an equivalent and a congruent relation with respect to $APTC_G$ with guarded recursion, we only need to check if each axiom in Table \ref{RDPRSP} is sound modulo step bisimulation equivalence. We leave them as exercises to the readers.

(2) Since pomset bisimulation $\sim_{p}$ is both an equivalent and a congruent relation with respect to the guarded recursion, we only need to check if each axiom in Table \ref{RDPRSP} is sound modulo pomset bisimulation equivalence. We leave them as exercises to the readers.

(3) Since hp-bisimulation $\sim_{hp}$ is both an equivalent and a congruent relation with respect to guarded recursion, we only need to check if each axiom in Table \ref{RDPRSP} is sound modulo hp-bisimulation equivalence. We leave them as exercises to the readers.
\end{proof}

\begin{theorem}[Completeness of $APTC_G$ with linear recursion]\label{CAPTC_GRG}
Let $p$ and $q$ be closed $APTC_G$ with linear recursion terms, then,

(1) if $p\sim_{s} q$ then $p=q$.

(2) if $p\sim_{p} q$ then $p=q$.

(3) if $p\sim_{hp} q$ then $p=q$.
\end{theorem}

\begin{proof}
Firstly, by the elimination theorem of $APTC_G$ with guarded recursion (see Theorem \ref{ETRecursionG}), we know that each process term in $APTC_G$ with linear recursion is equal to a process term $\langle X_1|E\rangle$ with $E$ a linear recursive specification. And for the simplicity, without loss of generalization, we do not consider empty event $\epsilon$, just because recursion with $\epsilon$ are similar to that with silent event $\tau$, please refer to the proof of Theorem \ref{CAPTCTAU} for details.

It remains to prove the following cases.

(1) If $\langle X_1|E_1\rangle \sim_s \langle Y_1|E_2\rangle$ for linear recursive specification $E_1$ and $E_2$, then $\langle X_1|E_1\rangle = \langle Y_1|E_2\rangle$.

Let $E_1$ consist of recursive equations $X=t_X$ for $X\in \mathcal{X}$ and $E_2$
consists of recursion equations $Y=t_Y$ for $Y\in\mathcal{Y}$. Let the linear recursive specification $E$ consist of recursion equations $Z_{XY}=t_{XY}$, and $\langle X|E_1\rangle\sim_s\langle Y|E_2\rangle$, and $t_{XY}$ consists of the following summands:

\begin{enumerate}
  \item $t_{XY}$ contains a summand $(a_1\parallel\cdots\parallel a_m)Z_{X'Y'}$ iff $t_X$ contains the summand $(a_1\parallel\cdots\parallel a_m)X'$ and $t_Y$ contains the summand $(a_1\parallel\cdots\parallel a_m)Y'$ such that $\langle X'|E_1\rangle\sim_s\langle Y'|E_2\rangle$;
  \item $t_{XY}$ contains a summand $b_1\parallel\cdots\parallel b_n$ iff $t_X$ contains the summand $b_1\parallel\cdots\parallel b_n$ and $t_Y$ contains the summand $b_1\parallel\cdots\parallel b_n$.
\end{enumerate}

Let $\sigma$ map recursion variable $X$ in $E_1$ to $\langle X|E_1\rangle$, and let $\pi$ map recursion variable $Z_{XY}$ in $E$ to $\langle X|E_1\rangle$. So, $\sigma((a_1\parallel\cdots\parallel a_m)X')\equiv(a_1\parallel\cdots\parallel a_m)\langle X'|E_1\rangle\equiv\pi((a_1\parallel\cdots\parallel a_m)Z_{X'Y'})$, so by $RDP$, we get $\langle X|E_1\rangle=\sigma(t_X)=\pi(t_{XY})$. Then by $RSP$, $\langle X|E_1\rangle=\langle Z_{XY}|E\rangle$, particularly, $\langle X_1|E_1\rangle=\langle Z_{X_1Y_1}|E\rangle$. Similarly, we can obtain $\langle Y_1|E_2\rangle=\langle Z_{X_1Y_1}|E\rangle$. Finally, $\langle X_1|E_1\rangle=\langle Z_{X_1Y_1}|E\rangle=\langle Y_1|E_2\rangle$, as desired.

(2) If $\langle X_1|E_1\rangle \sim_p \langle Y_1|E_2\rangle$ for linear recursive specification $E_1$ and $E_2$, then $\langle X_1|E_1\rangle = \langle Y_1|E_2\rangle$.

It can be proven similarly to (1), we omit it.

(3) If $\langle X_1|E_1\rangle \sim_{hp} \langle Y_1|E_2\rangle$ for linear recursive specification $E_1$ and $E_2$, then $\langle X_1|E_1\rangle = \langle Y_1|E_2\rangle$.

It can be proven similarly to (1), we omit it.
\end{proof}

\subsubsection{Abstraction}{\label{absg}}

To abstract away from the internal implementations of a program, and verify that the program exhibits the desired external behaviors, the silent step $\tau$ and abstraction operator $\tau_I$ are introduced, where $I\subseteq \mathbb{E}\cup G_{at}$ denotes the internal events or guards. The silent step $\tau$ represents the internal events or guards, when we consider the external behaviors of a process, $\tau$ steps can be removed, that is, $\tau$ steps must keep silent. The transition rule of $\tau$ is shown in Table \ref{TRForTauG}. In the following, let the atomic event $e$ range over $\mathbb{E}\cup\{\epsilon\}\cup\{\delta\}\cup\{\tau\}$, and $\phi$ range over $G\cup \{\tau\}$, and let the communication function $\gamma:\mathbb{E}\cup\{\tau\}\times \mathbb{E}\cup\{\tau\}\rightarrow \mathbb{E}\cup\{\delta\}$, with each communication involved $\tau$ resulting in $\delta$. We use $\tau(s)$ to denote $effect(\tau,s)$, for the fact that $\tau$ only change the state of internal data environment, that is, for the external data environments, $s=\tau(s)$.

\begin{center}
    \begin{table}
        $$\frac{}{\langle\tau,s\rangle\rightarrow\langle\surd,s\rangle}\textrm{ if }test(\tau,s)$$
        $$\frac{}{\langle\tau,s\rangle\xrightarrow{\tau}\langle\surd,\tau(s)\rangle}$$
        \caption{Transition rule of the silent step}
        \label{TRForTauG}
    \end{table}
\end{center}

In section \ref{os2}, we introduce $\tau$ into event structure, and also give the concept of weakly true concurrency. In this subsection, we give the concepts of rooted branching truly concurrent bisimulation equivalences, based on these concepts, we can design the axiom system of the silent step $\tau$ and the abstraction operator $\tau_I$.

\begin{definition}[Branching pomset, step bisimulation]\label{BPSBG}
Assume a special termination predicate $\downarrow$, and let $\surd$ represent a state with $\surd\downarrow$. Let $\mathcal{E}_1$, $\mathcal{E}_2$ be PESs. A branching pomset bisimulation is a relation $R\subseteq\langle\mathcal{C}(\mathcal{E}_1),S\rangle\times\langle\mathcal{C}(\mathcal{E}_2),S\rangle$, such that:
 \begin{enumerate}
   \item if $(\langle C_1,s\rangle,\langle C_2,s\rangle)\in R$, and $\langle C_1,s\rangle\xrightarrow{X}\langle C_1',s'\rangle$ then
   \begin{itemize}
     \item either $X\equiv \tau^*$, and $(\langle C_1',s'\rangle,\langle C_2,s\rangle)\in R$ with $s'\in \tau(s)$;
     \item or there is a sequence of (zero or more) $\tau$-transitions $\langle C_2,s\rangle\xrightarrow{\tau^*} \langle C_2^0,s^0\rangle$, such that $(\langle C_1,s\rangle,\langle C_2^0,s^0\rangle)\in R$ and $\langle C_2^0,s^0\rangle\xRightarrow{X}\langle C_2',s'\rangle$ with $(\langle C_1',s'\rangle,\langle C_2',s'\rangle)\in R$;
   \end{itemize}
   \item if $(\langle C_1,s\rangle,\langle C_2,s\rangle)\in R$, and $\langle C_2,s\rangle\xrightarrow{X}\langle C_2',s'\rangle$ then
   \begin{itemize}
     \item either $X\equiv \tau^*$, and $(\langle C_1,s\rangle,\langle C_2',s'\rangle)\in R$;
     \item or there is a sequence of (zero or more) $\tau$-transitions $\langle C_1,s\rangle\xrightarrow{\tau^*} \langle C_1^0,s^0\rangle$, such that $(\langle C_1^0,s^0\rangle,\langle C_2,s\rangle)\in R$ and $\langle C_1^0,s^0\rangle\xRightarrow{X}\langle C_1',s'\rangle$ with $(\langle C_1',s'\rangle,\langle C_2',s'\rangle)\in R$;
   \end{itemize}
   \item if $(\langle C_1,s\rangle,\langle C_2,s\rangle)\in R$ and $\langle C_1,s\rangle\downarrow$, then there is a sequence of (zero or more) $\tau$-transitions $\langle C_2,s\rangle\xrightarrow{\tau^*}\langle C_2^0,s^0\rangle$ such that $(\langle C_1,s\rangle,\langle C_2^0,s^0\rangle)\in R$ and $\langle C_2^0,s^0\rangle\downarrow$;
   \item if $(\langle C_1,s\rangle,\langle C_2,s\rangle)\in R$ and $\langle C_2,s\rangle\downarrow$, then there is a sequence of (zero or more) $\tau$-transitions $\langle C_1,s\rangle\xrightarrow{\tau^*}\langle C_1^0,s^0\rangle$ such that $(\langle C_1^0,s^0\rangle,\langle C_2,s\rangle)\in R$ and $\langle C_1^0,s^0\rangle\downarrow$.
 \end{enumerate}

We say that $\mathcal{E}_1$, $\mathcal{E}_2$ are branching pomset bisimilar, written $\mathcal{E}_1\approx_{bp}\mathcal{E}_2$, if there exists a branching pomset bisimulation $R$, such that $(\langle\emptyset,\emptyset\rangle,\langle\emptyset,\emptyset\rangle)\in R$.

By replacing pomset transitions with steps, we can get the definition of branching step bisimulation. When PESs $\mathcal{E}_1$ and $\mathcal{E}_2$ are branching step bisimilar, we write $\mathcal{E}_1\approx_{bs}\mathcal{E}_2$.
\end{definition}

\begin{definition}[Rooted branching pomset, step bisimulation]\label{RBPSBG}
Assume a special termination predicate $\downarrow$, and let $\surd$ represent a state with $\surd\downarrow$. Let $\mathcal{E}_1$, $\mathcal{E}_2$ be PESs. A rooted branching pomset bisimulation is a relation $R\subseteq\langle\mathcal{C}(\mathcal{E}_1),S\rangle\times\langle\mathcal{C}(\mathcal{E}_2),S\rangle$, such that:
 \begin{enumerate}
   \item if $(\langle C_1,s\rangle,\langle C_2,s\rangle)\in R$, and $\langle C_1,s\rangle\xrightarrow{X}\langle C_1',s'\rangle$ then $\langle C_2,s\rangle\xrightarrow{X}\langle C_2',s'\rangle$ with $\langle C_1',s'\rangle\approx_{bp}\langle C_2',s'\rangle$;
   \item if $(\langle C_1,s\rangle,\langle C_2,s\rangle)\in R$, and $\langle C_2,s\rangle\xrightarrow{X}\langle C_2',s'\rangle$ then $\langle C_1,s\rangle\xrightarrow{X}\langle C_1',s'\rangle$ with $\langle C_1',s'\rangle\approx_{bp}\langle C_2',s'\rangle$;
   \item if $(\langle C_1,s\rangle,\langle C_2,s\rangle)\in R$ and $\langle C_1,s\rangle\downarrow$, then $\langle C_2,s\rangle\downarrow$;
   \item if $(\langle C_1,s\rangle,\langle C_2,s\rangle)\in R$ and $\langle C_2,s\rangle\downarrow$, then $\langle C_1,s\rangle\downarrow$.
 \end{enumerate}

We say that $\mathcal{E}_1$, $\mathcal{E}_2$ are rooted branching pomset bisimilar, written $\mathcal{E}_1\approx_{rbp}\mathcal{E}_2$, if there exists a rooted branching pomset bisimulation $R$, such that $(\langle\emptyset,\emptyset\rangle,\langle\emptyset,\emptyset\rangle)\in R$.

By replacing pomset transitions with steps, we can get the definition of rooted branching step bisimulation. When PESs $\mathcal{E}_1$ and $\mathcal{E}_2$ are rooted branching step bisimilar, we write $\mathcal{E}_1\approx_{rbs}\mathcal{E}_2$.
\end{definition}

\begin{definition}[Branching (hereditary) history-preserving bisimulation]\label{BHHPBG}
Assume a special termination predicate $\downarrow$, and let $\surd$ represent a state with $\surd\downarrow$. A branching history-preserving (hp-) bisimulation is a weakly posetal relation $R\subseteq\langle\mathcal{C}(\mathcal{E}_1),S\rangle\overline{\times}\langle\mathcal{C}(\mathcal{E}_2),S\rangle$ such that:

 \begin{enumerate}
   \item if $(\langle C_1,s\rangle,f,\langle C_2,s\rangle)\in R$, and $\langle C_1,s\rangle\xrightarrow{e_1}\langle C_1',s'\rangle$ then
   \begin{itemize}
     \item either $e_1\equiv \tau$, and $(\langle C_1',s'\rangle,f[e_1\mapsto \tau],\langle C_2,s\rangle)\in R$;
     \item or there is a sequence of (zero or more) $\tau$-transitions $\langle C_2,s\rangle\xrightarrow{\tau^*} \langle C_2^0,s^0\rangle$, such that $(\langle C_1,s\rangle,f,\langle C_2^0,s^0\rangle)\in R$ and $\langle C_2^0,s^0\rangle\xrightarrow{e_2}\langle C_2',s'\rangle$ with $(\langle C_1',s'\rangle,f[e_1\mapsto e_2],\langle C_2',s'\rangle)\in R$;
   \end{itemize}
   \item if $(\langle C_1,s\rangle,f,\langle C_2,s\rangle)\in R$, and $\langle C_2,s\rangle\xrightarrow{e_2}\langle C_2',s'\rangle$ then
   \begin{itemize}
     \item either $e_2\equiv \tau$, and $(\langle C_1,s\rangle,f[e_2\mapsto \tau],\langle C_2',s'\rangle)\in R$;
     \item or there is a sequence of (zero or more) $\tau$-transitions $\langle C_1,s\rangle\xrightarrow{\tau^*} \langle C_1^0,s^0\rangle$, such that $(\langle C_1^0,s^0\rangle,f,\langle C_2,s\rangle)\in R$ and $\langle C_1^0,s^0\rangle\xrightarrow{e_1}\langle C_1',s'\rangle$ with $(\langle C_1',s'\rangle,f[e_2\mapsto e_1],\langle C_2',s'\rangle)\in R$;
   \end{itemize}
   \item if $(\langle C_1,s\rangle,f,\langle C_2,s\rangle)\in R$ and $\langle C_1,s\rangle\downarrow$, then there is a sequence of (zero or more) $\tau$-transitions $\langle C_2,s\rangle\xrightarrow{\tau^*}\langle C_2^0,s^0\rangle$ such that $(\langle C_1,s\rangle,f,\langle C_2^0,s^0\rangle)\in R$ and $\langle C_2^0,s^0\rangle\downarrow$;
   \item if $(\langle C_1,s\rangle,f,\langle C_2,s\rangle)\in R$ and $\langle C_2,s\rangle\downarrow$, then there is a sequence of (zero or more) $\tau$-transitions $\langle C_1,s\rangle\xrightarrow{\tau^*}\langle C_1^0,s^0\rangle$ such that $(\langle C_1^0,s^0\rangle,f,\langle C_2,s\rangle)\in R$ and $\langle C_1^0,s^0\rangle\downarrow$.
 \end{enumerate}

$\mathcal{E}_1,\mathcal{E}_2$ are branching history-preserving (hp-)bisimilar and are written $\mathcal{E}_1\approx_{bhp}\mathcal{E}_2$ if there exists a branching hp-bisimulation $R$ such that $(\langle\emptyset,\emptyset\rangle,\emptyset,\langle\emptyset,\emptyset\rangle)\in R$.

A branching hereditary history-preserving (hhp-)bisimulation is a downward closed branching hp-bisimulation. $\mathcal{E}_1,\mathcal{E}_2$ are branching hereditary history-preserving (hhp-)bisimilar and are written $\mathcal{E}_1\approx_{bhhp}\mathcal{E}_2$.
\end{definition}

\begin{definition}[Rooted branching (hereditary) history-preserving bisimulation]\label{RBHHPBG}
Assume a special termination predicate $\downarrow$, and let $\surd$ represent a state with $\surd\downarrow$. A rooted branching history-preserving (hp-) bisimulation is a weakly posetal relation $R\subseteq\langle\mathcal{C}(\mathcal{E}_1),S\rangle\overline{\times}\langle\mathcal{C}(\mathcal{E}_2),S\rangle$ such that:

 \begin{enumerate}
   \item if $(\langle C_1,s\rangle,f,\langle C_2,s\rangle)\in R$, and $\langle C_1,s\rangle\xrightarrow{e_1}\langle C_1',s'\rangle$, then $\langle C_2,s\rangle\xrightarrow{e_2}\langle C_2',s'\rangle$ with $\langle C_1',s'\rangle\approx_{bhp}\langle C_2',s'\rangle$;
   \item if $(\langle C_1,s\rangle,f,\langle C_2,s\rangle)\in R$, and $\langle C_2,s\rangle\xrightarrow{e_2}\langle C_2',s'\rangle$, then $\langle C_1,s\rangle\xrightarrow{e_1}\langle C_1',s'\rangle$ with $\langle C_1',s'\rangle\approx_{bhp}\langle C_2',s'\rangle$;
   \item if $(\langle C_1,s\rangle,f,\langle C_2,s\rangle)\in R$ and $\langle C_1,s\rangle\downarrow$, then $\langle C_2,s\rangle\downarrow$;
   \item if $(\langle C_1,s\rangle,f,\langle C_2,s\rangle)\in R$ and $\langle C_2,s\rangle\downarrow$, then $\langle C_1,s\rangle\downarrow$.
 \end{enumerate}

$\mathcal{E}_1,\mathcal{E}_2$ are rooted branching history-preserving (hp-)bisimilar and are written $\mathcal{E}_1\approx_{rbhp}\mathcal{E}_2$ if there exists a rooted branching hp-bisimulation $R$ such that $(\langle\emptyset,\emptyset\rangle,\emptyset,\langle\emptyset,\emptyset\rangle)\in R$.

A rooted branching hereditary history-preserving (hhp-)bisimulation is a downward closed rooted branching hp-bisimulation. $\mathcal{E}_1,\mathcal{E}_2$ are rooted branching hereditary history-preserving (hhp-)bisimilar and are written $\mathcal{E}_1\approx_{rbhhp}\mathcal{E}_2$.
\end{definition}

\begin{definition}[Guarded linear recursive specification]\label{GLRSG}
A linear recursive specification $E$ is guarded if there does not exist an infinite sequence of $\tau$-transitions $\langle X|E\rangle\xrightarrow{\tau}\langle X'|E\rangle\xrightarrow{\tau}\langle X''|E\rangle\xrightarrow{\tau}\cdots$, and there does not exist an infinite sequence of $\epsilon$-transitions $\langle X|E\rangle\rightarrow\langle X'|E\rangle\rightarrow\langle X''|E\rangle\rightarrow\cdots$.
\end{definition}

\begin{theorem}[Conservitivity of $APTC_G$ with silent step and guarded linear recursion]
$APTC_G$ with silent step and guarded linear recursion is a conservative extension of $APTC_G$ with linear recursion.
\end{theorem}

\begin{proof}
Since the transition rules of $APTC_G$ with linear recursion are source-dependent, and the transition rules for silent step in Table \ref{TRForTauG} contain only a fresh constant $\tau$ in their source, so the transition rules of $APTC_G$ with silent step and guarded linear recursion is a conservative extension of those of $APTC_G$ with linear recursion.
\end{proof}

\begin{theorem}[Congruence theorem of $APTC_G$ with silent step and guarded linear recursion]
Rooted branching truly concurrent bisimulation equivalences $\approx_{rbp}$, $\approx_{rbs}$ and $\approx_{rbhp}$ are all congruences with respect to $APTC_G$ with silent step and guarded linear recursion.
\end{theorem}

\begin{proof}
It follows the following three facts:
\begin{enumerate}
  \item in a guarded linear recursive specification, right-hand sides of its recursive equations can be adapted to the form by applications of the axioms in $APTC_G$ and replacing recursion variables by the right-hand sides of their recursive equations;
  \item truly concurrent bisimulation equivalences $\sim_{p}$, $\sim_s$ and $\sim_{hp}$ are all congruences with respect to all operators of $APTC_G$, while truly concurrent bisimulation equivalences $\sim_{p}$, $\sim_s$ and $\sim_{hp}$ imply the corresponding rooted branching truly concurrent bisimulations $\approx_{rbp}$, $\approx_{rbs}$ and $\approx_{rbhp}$ (see Proposition \ref{WSCBE}), so rooted branching truly concurrent bisimulations $\approx_{rbp}$, $\approx_{rbs}$ and $\approx_{rbhp}$ are all congruences with respect to all operators of $APTC_G$;
  \item While $\mathbb{E}$ is extended to $\mathbb{E}\cup\{\tau\}$, and $G$ is extended to $G\cup\{\tau\}$, it can be proved that rooted branching truly concurrent bisimulations $\approx_{rbp}$, $\approx_{rbs}$ and $\approx_{rbhp}$ are all congruences with respect to all operators of $APTC_G$, we omit it.
\end{enumerate}
\end{proof}

We design the axioms for the silent step $\tau$ in Table \ref{AxiomsForTauG}.

\begin{center}
\begin{table}
  \begin{tabular}{@{}ll@{}}
\hline No. &Axiom\\
  $B1$ & $e\cdot\tau=e$\\
  $B2$ & $e\cdot(\tau\cdot(x+y)+x)=e\cdot(x+y)$\\
  $B3$ & $x\parallel\tau=x$\\
  $G25$ & $\phi\cdot\tau=\phi$\\
  $G26$ & $\phi\cdot(\tau\cdot(x+y)+x)=\phi\cdot(x+y)$\\
\end{tabular}
\caption{Axioms of silent step}
\label{AxiomsForTauG}
\end{table}
\end{center}

\begin{theorem}[Elimination theorem of $APTC_G$ with silent step and guarded linear recursion]\label{ETTauG}
Each process term in $APTC_G$ with silent step and guarded linear recursion is equal to a process term $\langle X_1|E\rangle$ with $E$ a guarded linear recursive specification.
\end{theorem}

\begin{proof}
By applying structural induction with respect to term size, each process term $t_1$ in $APTC_G$ with silent step and guarded linear recursion generates a process can be expressed in the form of equations

$$t_i=(a_{i11}\parallel\cdots\parallel a_{i1i_1})t_{i1}+\cdots+(a_{ik_i1}\parallel\cdots\parallel a_{ik_ii_k})t_{ik_i}+(b_{i11}\parallel\cdots\parallel b_{i1i_1})+\cdots+(b_{il_i1}\parallel\cdots\parallel b_{il_ii_l})$$

for $i\in\{1,\cdots,n\}$. Let the linear recursive specification $E$ consist of the recursive equations

$$X_i=(a_{i11}\parallel\cdots\parallel a_{i1i_1})X_{i1}+\cdots+(a_{ik_i1}\parallel\cdots\parallel a_{ik_ii_k})X_{ik_i}+(b_{i11}\parallel\cdots\parallel b_{i1i_1})+\cdots+(b_{il_i1}\parallel\cdots\parallel b_{il_ii_l})$$

for $i\in\{1,\cdots,n\}$. Replacing $X_i$ by $t_i$ for $i\in\{1,\cdots,n\}$ is a solution for $E$, $RSP$ yields $t_1=\langle X_1|E\rangle$.
\end{proof}

\begin{theorem}[Soundness of $APTC_G$ with silent step and guarded linear recursion]\label{SAPTC_GTAUG}
Let $x$ and $y$ be $APTC_G$ with silent step and guarded linear recursion terms. If $APTC_G$ with silent step and guarded linear recursion $\vdash x=y$, then

(1) $x\approx_{rbs} y$.

(2) $x\approx_{rbp} y$.

(3) $x\approx_{rbhp} y$.
\end{theorem}

\begin{proof}
(1) Since rooted branching step bisimulation $\approx_{rbs}$ is both an equivalent and a congruent relation with respect to $APTC_G$ with silent step and guarded linear recursion, we only need to check if each axiom in Table \ref{AxiomsForTauG} is sound modulo rooted branching step bisimulation equivalence. We leave them as exercises to the readers.

(2) Since rooted branching pomset bisimulation $\approx_{rbp}$ is both an equivalent and a congruent relation with respect to $APTC_G$ with silent step and guarded linear recursion, we only need to check if each axiom in Table \ref{AxiomsForTauG} is sound modulo rooted branching pomset bisimulation $\approx_{rbp}$. We leave them as exercises to the readers.

(3) Since rooted branching hp-bisimulation $\approx_{rbhp}$ is both an equivalent and a congruent relation with respect to $APTC_G$ with silent step and guarded linear recursion, we only need to check if each axiom in Table \ref{AxiomsForTauG} is sound modulo rooted branching hp-bisimulation $\approx_{rbhp}$. We leave them as exercises to the readers.
\end{proof}

\begin{theorem}[Completeness of $APTC_G$ with silent step and guarded linear recursion]\label{CAPTC_GTAUG}
Let $p$ and $q$ be closed $APTC_G$ with silent step and guarded linear recursion terms, then,

(1) if $p\approx_{rbs} q$ then $p=q$.

(2) if $p\approx_{rbp} q$ then $p=q$.

(3) if $p\approx_{rbhp} q$ then $p=q$.
\end{theorem}

\begin{proof}
Firstly, by the elimination theorem of $APTC_G$ with silent step and guarded linear recursion (see Theorem \ref{ETTauG}), we know that each process term in $APTC_G$ with silent step and guarded linear recursion is equal to a process term $\langle X_1|E\rangle$ with $E$ a guarded linear recursive specification.

It remains to prove the following cases.

(1) If $\langle X_1|E_1\rangle \approx_{rbs} \langle Y_1|E_2\rangle$ for guarded linear recursive specification $E_1$ and $E_2$, then $\langle X_1|E_1\rangle = \langle Y_1|E_2\rangle$.

Firstly, the recursive equation $W=\tau+\cdots+\tau$ with $W\nequiv X_1$ in $E_1$ and $E_2$, can be removed, and the corresponding summands $aW$ are replaced by $a$, to get $E_1'$ and $E_2'$, by use of the axioms $RDP$, $A3$ and $B1$, and $\langle X|E_1\rangle = \langle X|E_1'\rangle$, $\langle Y|E_2\rangle = \langle Y|E_2'\rangle$.

Let $E_1$ consists of recursive equations $X=t_X$ for $X\in \mathcal{X}$ and $E_2$
consists of recursion equations $Y=t_Y$ for $Y\in\mathcal{Y}$, and are not the form $\tau+\cdots+\tau$. Let the guarded linear recursive specification $E$ consists of recursion equations $Z_{XY}=t_{XY}$, and $\langle X|E_1\rangle\approx_{rbs}\langle Y|E_2\rangle$, and $t_{XY}$ consists of the following summands:

\begin{enumerate}
  \item $t_{XY}$ contains a summand $(a_1\parallel\cdots\parallel a_m)Z_{X'Y'}$ iff $t_X$ contains the summand $(a_1\parallel\cdots\parallel a_m)X'$ and $t_Y$ contains the summand $(a_1\parallel\cdots\parallel a_m)Y'$ such that $\langle X'|E_1\rangle\approx_{rbs}\langle Y'|E_2\rangle$;
  \item $t_{XY}$ contains a summand $b_1\parallel\cdots\parallel b_n$ iff $t_X$ contains the summand $b_1\parallel\cdots\parallel b_n$ and $t_Y$ contains the summand $b_1\parallel\cdots\parallel b_n$;
  \item $t_{XY}$ contains a summand $\tau Z_{X'Y}$ iff $XY\nequiv X_1Y_1$, $t_X$ contains the summand $\tau X'$, and $\langle X'|E_1\rangle\approx_{rbs}\langle Y|E_2\rangle$;
  \item $t_{XY}$ contains a summand $\tau Z_{XY'}$ iff $XY\nequiv X_1Y_1$, $t_Y$ contains the summand $\tau Y'$, and $\langle X|E_1\rangle\approx_{rbs}\langle Y'|E_2\rangle$.
\end{enumerate}

Since $E_1$ and $E_2$ are guarded, $E$ is guarded. Constructing the process term $u_{XY}$ consist of the following summands:

\begin{enumerate}
  \item $u_{XY}$ contains a summand $(a_1\parallel\cdots\parallel a_m)\langle X'|E_1\rangle$ iff $t_X$ contains the summand $(a_1\parallel\cdots\parallel a_m)X'$ and $t_Y$ contains the summand $(a_1\parallel\cdots\parallel a_m)Y'$ such that $\langle X'|E_1\rangle\approx_{rbs}\langle Y'|E_2\rangle$;
  \item $u_{XY}$ contains a summand $b_1\parallel\cdots\parallel b_n$ iff $t_X$ contains the summand $b_1\parallel\cdots\parallel b_n$ and $t_Y$ contains the summand $b_1\parallel\cdots\parallel b_n$;
  \item $u_{XY}$ contains a summand $\tau \langle X'|E_1\rangle$ iff $XY\nequiv X_1Y_1$, $t_X$ contains the summand $\tau X'$, and $\langle X'|E_1\rangle\approx_{rbs}\langle Y|E_2\rangle$.
\end{enumerate}

Let the process term $s_{XY}$ be defined as follows:

\begin{enumerate}
  \item $s_{XY}\triangleq\tau\langle X|E_1\rangle + u_{XY}$ iff $XY\nequiv X_1Y_1$, $t_Y$ contains the summand $\tau Y'$, and $\langle X|E_1\rangle\approx_{rbs}\langle Y'|E_2\rangle$;
  \item $s_{XY}\triangleq\langle X|E_1\rangle$, otherwise.
\end{enumerate}

So, $\langle X|E_1\rangle=\langle X|E_1\rangle+u_{XY}$, and $(a_1\parallel\cdots\parallel a_m)(\tau\langle X|E_1\rangle+u_{XY})=(a_1\parallel\cdots\parallel a_m)((\tau\langle X|E_1\rangle+u_{XY})+u_{XY})=(a_1\parallel\cdots\parallel a_m)(\langle X|E_1\rangle+u_{XY})=(a_1\parallel\cdots\parallel a_m)\langle X|E_1\rangle$, hence, $(a_1\parallel\cdots\parallel a_m)s_{XY}=(a_1\parallel\cdots\parallel a_m)\langle X|E_1\rangle$.

Let $\sigma$ map recursion variable $X$ in $E_1$ to $\langle X|E_1\rangle$, and let $\pi$ map recursion variable $Z_{XY}$ in $E$ to $s_{XY}$. It is sufficient to prove $s_{XY}=\pi(t_{XY})$ for recursion variables $Z_{XY}$ in $E$. Either $XY\equiv X_1Y_1$ or $XY\nequiv X_1Y_1$, we all can get $s_{XY}=\pi(t_{XY})$. So, $s_{XY}=\langle Z_{XY}|E\rangle$ for recursive variables $Z_{XY}$ in $E$ is a solution for $E$. Then by $RSP$, particularly, $\langle X_1|E_1\rangle=\langle Z_{X_1Y_1}|E\rangle$. Similarly, we can obtain $\langle Y_1|E_2\rangle=\langle Z_{X_1Y_1}|E\rangle$. Finally, $\langle X_1|E_1\rangle=\langle Z_{X_1Y_1}|E\rangle=\langle Y_1|E_2\rangle$, as desired.

(2) If $\langle X_1|E_1\rangle \approx_{rbp} \langle Y_1|E_2\rangle$ for guarded linear recursive specification $E_1$ and $E_2$, then $\langle X_1|E_1\rangle = \langle Y_1|E_2\rangle$.

It can be proven similarly to (1), we omit it.

(3) If $\langle X_1|E_1\rangle \approx_{rbhb} \langle Y_1|E_2\rangle$ for guarded linear recursive specification $E_1$ and $E_2$, then $\langle X_1|E_1\rangle = \langle Y_1|E_2\rangle$.

It can be proven similarly to (1), we omit it.
\end{proof}

The unary abstraction operator $\tau_I$ ($I\subseteq \mathbb{E}\cup G_{at}$) renames all atomic events or atomic guards in $I$ into $\tau$. $APTC_G$ with silent step and abstraction operator is called $APTC_{G_{\tau}}$. The transition rules of operator $\tau_I$ are shown in Table \ref{TRForAbstractionG}.

\begin{center}
    \begin{table}
        $$\frac{\langle x,s\rangle\xrightarrow{e}\langle\surd,s'\rangle}{\langle\tau_I(x),s\rangle\xrightarrow{e}\langle\surd,s'\rangle}\quad e\notin I
        \quad\quad\frac{\langle x,s\rangle\xrightarrow{e}\langle x',s'\rangle}{\langle\tau_I(x),s\rangle\xrightarrow{e}\langle\tau_I(x'),s'\rangle}\quad e\notin I$$

        $$\frac{\langle x,s\rangle\xrightarrow{e}\langle\surd,s'\rangle}{\langle\tau_I(x),s\rangle\xrightarrow{\tau}\langle\surd,\tau(s)\rangle}\quad e\in I
        \quad\quad\frac{\langle x,s\rangle\xrightarrow{e}\langle x',s'\rangle}{\langle\tau_I(x),s\rangle\xrightarrow{\tau}\langle\tau_I(x'),\tau(s)\rangle}\quad e\in I$$
        \caption{Transition rule of the abstraction operator}
        \label{TRForAbstractionG}
    \end{table}
\end{center}

\begin{theorem}[Conservitivity of $APTC_{G_{\tau}}$ with guarded linear recursion]
$APTC_{G_{\tau}}$ with guarded linear recursion is a conservative extension of $APTC_G$ with silent step and guarded linear recursion.
\end{theorem}

\begin{proof}
Since the transition rules of $APTC_G$ with silent step and guarded linear recursion are source-dependent, and the transition rules for abstraction operator in Table \ref{TRForAbstractionG} contain only a fresh operator $\tau_I$ in their source, so the transition rules of $APTC_{G_{\tau}}$ with guarded linear recursion is a conservative extension of those of $APTC_G$ with silent step and guarded linear recursion.
\end{proof}

\begin{theorem}[Congruence theorem of $APTC_{G_{\tau}}$ with guarded linear recursion]
Rooted branching truly concurrent bisimulation equivalences $\approx_{rbp}$, $\approx_{rbs}$ and $\approx_{rbhp}$ are all congruences with respect to $APTC_{G_{\tau}}$ with guarded linear recursion.
\end{theorem}

\begin{proof}
(1) It is easy to see that rooted branching pomset bisimulation is an equivalent relation on $APTC_{G_{\tau}}$ with guarded linear recursion terms, we only need to prove that $\approx_{rbp}$ is preserved by the operators $\tau_I$. It is trivial and we leave the proof as an exercise for the readers.

(2) The cases of rooted branching step bisimulation $\approx_{rbs}$, rooted branching hp-bisimulation $\approx_{rbhp}$ can be proven similarly, we omit them.
\end{proof}

We design the axioms for the abstraction operator $\tau_I$ in Table \ref{AxiomsForAbstractionG}.

\begin{center}
\begin{table}
  \begin{tabular}{@{}ll@{}}
\hline No. &Axiom\\
  $TI1$ & $e\notin I\quad \tau_I(e)=e$\\
  $TI2$ & $e\in I\quad \tau_I(e)=\tau$\\
  $TI3$ & $\tau_I(\delta)=\delta$\\
  $TI4$ & $\tau_I(x+y)=\tau_I(x)+\tau_I(y)$\\
  $TI5$ & $\tau_I(x\cdot y)=\tau_I(x)\cdot\tau_I(y)$\\
  $TI6$ & $\tau_I(x\parallel y)=\tau_I(x)\parallel\tau_I(y)$\\
  $G27$ & $\phi\notin I\quad \tau_I(\phi)=\phi$\\
  $G28$ & $\phi\in I\quad \tau_I(\phi)=\tau$\\
\end{tabular}
\caption{Axioms of abstraction operator}
\label{AxiomsForAbstractionG}
\end{table}
\end{center}

\begin{theorem}[Soundness of $APTC_{G_{\tau}}$ with guarded linear recursion]\label{SAPTC_GABSG}
Let $x$ and $y$ be $APTC_{G_{\tau}}$ with guarded linear recursion terms. If $APTC_{G_{\tau}}$ with guarded linear recursion $\vdash x=y$, then

(1) $x\approx_{rbs} y$.

(2) $x\approx_{rbp} y$.

(3) $x\approx_{rbhp} y$.
\end{theorem}

\begin{proof}
(1) Since rooted branching step bisimulation $\approx_{rbs}$ is both an equivalent and a congruent relation with respect to $APTC_{G_{\tau}}$ with guarded linear recursion, we only need to check if each axiom in Table \ref{AxiomsForAbstractionG} is sound modulo rooted branching step bisimulation equivalence. We leave them as exercises to the readers.

(2) Since rooted branching pomset bisimulation $\approx_{rbp}$ is both an equivalent and a congruent relation with respect to $APTC_{G_{\tau}}$ with guarded linear recursion, we only need to check if each axiom in Table \ref{AxiomsForAbstractionG} is sound modulo rooted branching pomset bisimulation $\approx_{rbp}$. We leave them as exercises to the readers.

(3) Since rooted branching hp-bisimulation $\approx_{rbhp}$ is both an equivalent and a congruent relation with respect to $APTC_{G_{\tau}}$ with guarded linear recursion, we only need to check if each axiom in Table \ref{AxiomsForAbstractionG} is sound modulo rooted branching hp-bisimulation $\approx_{rbhp}$. We leave them as exercises to the readers.
\end{proof}

Though $\tau$-loops are prohibited in guarded linear recursive specifications (see Definition \ref{GLRSG}) in a specifiable way, they can be constructed using the abstraction operator, for example, there exist $\tau$-loops in the process term $\tau_{\{a\}}(\langle X|X=aX\rangle)$. To avoid $\tau$-loops caused by $\tau_I$ and ensure fairness, the concept of cluster and $CFAR$ (Cluster Fair Abstraction Rule) \cite{CFAR} are still needed.

\begin{theorem}[Completeness of $APTC_{G_{\tau}}$ with guarded linear recursion and $CFAR$]\label{CCFARG}
Let $p$ and $q$ be closed $APTC_{G_{\tau}}$ with guarded linear recursion and $CFAR$ terms, then,

(1) if $p\approx_{rbs} q$ then $p=q$.

(2) if $p\approx_{rbp} q$ then $p=q$.

(3) if $p\approx_{rbhp} q$ then $p=q$.
\end{theorem}

\begin{proof}
(1) For the case of rooted branching step bisimulation, the proof is following.

Firstly, in the proof the Theorem \ref{CAPTC_GTAUG}, we know that each process term $p$ in $APTC_G$ with silent step and guarded linear recursion is equal to a process term $\langle X_1|E\rangle$ with $E$ a guarded linear recursive specification. And we prove if $\langle X_1|E_1\rangle\approx_{rbs}\langle Y_1|E_2\rangle$, then $\langle X_1|E_1\rangle=\langle Y_1|E_2\rangle$

The only new case is $p\equiv\tau_I(q)$. Let $q=\langle X|E\rangle$ with $E$ a guarded linear recursive specification, so $p=\tau_I(\langle X|E\rangle)$. Then the collection of recursive variables in $E$ can be divided into its clusters $C_1,\cdots,C_N$ for $I$. Let

$$(a_{1i1}\parallel\cdots\parallel a_{k_{i1}i1}) Y_{i1}+\cdots+(a_{1im_i}\parallel\cdots\parallel a_{k_{im_i}im_i}) Y_{im_i}+b_{1i1}\parallel\cdots\parallel b_{l_{i1}i1}+\cdots+b_{1im_i}\parallel\cdots\parallel b_{l_{im_i}im_i}$$

be the conflict composition of exits for the cluster $C_i$, with $i\in\{1,\cdots,N\}$.

For $Z\in C_i$ with $i\in\{1,\cdots,N\}$, we define

$$s_Z\triangleq (\hat{a_{1i1}}\parallel\cdots\parallel \hat{a_{k_{i1}i1}}) \tau_I(\langle Y_{i1}|E\rangle)+\cdots+(\hat{a_{1im_i}}\parallel\cdots\parallel \hat{a_{k_{im_i}im_i}}) \tau_I(\langle Y_{im_i}|E\rangle)+\hat{b_{1i1}}\parallel\cdots\parallel \hat{b_{l_{i1}i1}}+\cdots+\hat{b_{1im_i}}\parallel\cdots\parallel \hat{b_{l_{im_i}im_i}}$$

For $Z\in C_i$ and $a_1,\cdots,a_j\in \mathbb{E}\cup\{\tau\}$ with $j\in\mathbb{N}$, we have

$(a_1\parallel\cdots\parallel a_j)\tau_I(\langle Z|E\rangle)$

$=(a_1\parallel\cdots\parallel a_j)\tau_I((a_{1i1}\parallel\cdots\parallel a_{k_{i1}i1}) \langle Y_{i1}|E\rangle+\cdots+(a_{1im_i}\parallel\cdots\parallel a_{k_{im_i}im_i}) \langle Y_{im_i}|E\rangle+b_{1i1}\parallel\cdots\parallel b_{l_{i1}i1}+\cdots+b_{1im_i}\parallel\cdots\parallel b_{l_{im_i}im_i})$

$=(a_1\parallel\cdots\parallel a_j)s_Z$

Let the linear recursive specification $F$ contain the same recursive variables as $E$, for $Z\in C_i$, $F$ contains the following recursive equation

$$Z=(\hat{a_{1i1}}\parallel\cdots\parallel \hat{a_{k_{i1}i1}}) Y_{i1}+\cdots+(\hat{a_{1im_i}}\parallel\cdots\parallel \hat{a_{k_{im_i}im_i}})  Y_{im_i}+\hat{b_{1i1}}\parallel\cdots\parallel \hat{b_{l_{i1}i1}}+\cdots+\hat{b_{1im_i}}\parallel\cdots\parallel \hat{b_{l_{im_i}im_i}}$$

It is easy to see that there is no sequence of one or more $\tau$-transitions from $\langle Z|F\rangle$ to itself, so $F$ is guarded.

For

$$s_Z=(\hat{a_{1i1}}\parallel\cdots\parallel \hat{a_{k_{i1}i1}}) Y_{i1}+\cdots+(\hat{a_{1im_i}}\parallel\cdots\parallel \hat{a_{k_{im_i}im_i}}) Y_{im_i}+\hat{b_{1i1}}\parallel\cdots\parallel \hat{b_{l_{i1}i1}}+\cdots+\hat{b_{1im_i}}\parallel\cdots\parallel \hat{b_{l_{im_i}im_i}}$$

is a solution for $F$. So, $(a_1\parallel\cdots\parallel a_j)\tau_I(\langle Z|E\rangle)=(a_1\parallel\cdots\parallel a_j)s_Z=(a_1\parallel\cdots\parallel a_j)\langle Z|F\rangle$.

So,

$$\langle Z|F\rangle=(\hat{a_{1i1}}\parallel\cdots\parallel \hat{a_{k_{i1}i1}}) \langle Y_{i1}|F\rangle+\cdots+(\hat{a_{1im_i}}\parallel\cdots\parallel \hat{a_{k_{im_i}im_i}}) \langle Y_{im_i}|F\rangle+\hat{b_{1i1}}\parallel\cdots\parallel \hat{b_{l_{i1}i1}}+\cdots+\hat{b_{1im_i}}\parallel\cdots\parallel \hat{b_{l_{im_i}im_i}}$$

Hence, $\tau_I(\langle X|E\rangle=\langle Z|F\rangle)$, as desired.

(2) For the case of rooted branching pomset bisimulation, it can be proven similarly to (1), we omit it.

(3) For the case of rooted branching hp-bisimulation, it can be proven similarly to (1), we omit it.
\end{proof}

\subsubsection{Hoare Logic for $APTC_G$}{\label{hl}}

In this section, we introduce Hoare logic for $APTC_G$. We do not introduce the preliminaries of Hoare logic, please refer to \cite{HL} for details.

A partial correct formula has the form

$$\{pre\}P\{post\}$$

where $pre$ are preconditions, $post$ are postconditions, and $P$ are programs. $\{pre\}P\{post\}$ means that $pre$ hold, then $P$ are executed and $post$ hold. We take the guards $G$ of $APTC_G$ as the language of conditions, and closed terms of $APTC_G$ as programs. For some condition $\alpha\in G$ and some data state $s\in S$, we denote $S\models \alpha[s]$ for $\langle \alpha,s\rangle\rightarrow\langle\surd,s\rangle$, and $S\models\alpha$ for $\forall s\in S, S\models \alpha[s]$, $S\models \{\alpha\}p\{\beta\}$ for all $s\in S$, $\mu\subseteq \mathbb{E}\cup G$, $S\models \alpha[s]$, $\langle p,s\rangle\xrightarrow{\mu}\langle p',s'\rangle$, $S\models\beta[s']$ with $s'\in S$. It is obvious that $S\models \{\alpha\}p\{\beta\}\Leftrightarrow \alpha p\approx_{rbp}(\approx_{rbs},\approx_{rbhp})\alpha p\beta$.

We design a proof system $H$ to deriving partial correct formulas over terms of $APTC_G$ as Table \ref{H} shows. Let $\Gamma$ be a set of conditions and partial correct formulas, we denote $\Gamma\vdash\{\alpha\}t\{\beta\}$ iff we can derive $\{\alpha\}t\{\beta\}$ in $H$, note that $t$ does not need to be closed terms. And we write $\alpha\rightarrow \beta$ for $S\models \alpha\Rightarrow S\models\beta$.

\begin{center}
    \begin{table}
        $(H1)\quad\{wp(e,\alpha)\}e\{\alpha\}\textrm{ if }e\in\mathbb{E}$

        $(H2)\quad\{\alpha\}\phi\{\alpha\cdot\phi\}\textrm{ if }\phi\in G$

        $(H3)\quad\frac{\{\alpha\}t\{\beta\}\quad\{\alpha\}t'\{\beta\}}{\{\alpha\}t+t'\{\beta\}}$

        $(H4)\quad\frac{\{\alpha\}t\{\alpha'\}\quad\{\alpha'\}t'\{\beta\}}{\{\alpha\}t\cdot t'\{\beta\}}$

        $(H5)\quad\frac{\{\alpha\}t\{\alpha'\}\quad\{\beta\}t'\{\beta'\}}{\{\alpha\parallel\beta\}t\between t'\{\alpha'\parallel\beta'\}}$

        $(H6)\quad\frac{\{\alpha\}t\{\beta\}}{\{\alpha\}\Theta(t)\{\beta\}}$

        $(H7)\quad\frac{\{\alpha\}t\{\beta\}}{\{\alpha\}\partial_H(t)\{\beta\}}$

        $(H8)\quad\frac{\{\alpha\}t\{\beta\}}{\{\alpha\}\tau_I(t)\{\beta\}}$

        $(H9)\quad\frac{\alpha\rightarrow\alpha'\quad\{\alpha'\}t\{\beta'\}\quad\beta'\rightarrow\beta}{\{\alpha\}t\{\beta\}}$

        $(H10)\quad$ For $E=\{x=t_x|x\in V_E\}$ a guarded linear recursive specification $\forall y\in V_E$ and $z\in V_E$:
        \[\frac{\{\alpha_x\}t_x\{\beta_x\}\quad\cdots\quad\{\alpha_y\}t_y\{\beta_y\}}{\{\alpha_z\}\langle z|E\rangle\{\beta_z\}}\]

        $(H10')\quad$ For $E=\{x=t_x|x\in V_E\}$ a guarded linear recursive specification $\forall y\in V_E$ and $z\in V_E$:
        \[\frac{\{\alpha_{x_1}\parallel\cdots\parallel\alpha_{x_{nx}}\}t_x\{\beta_{x_1}\parallel\cdots\parallel\beta_{x_{nx}}\}\quad\cdots\quad\{\alpha_{y_1}\parallel\cdots \parallel\alpha_{y_{ny}}\}t_y\{\beta_{y_1}\parallel\cdots\parallel\beta_{y_{ny}}\}}{\{\alpha_{z_1}\parallel\cdots\parallel\alpha_{z_{nz}}\}\langle z|E\rangle\{\beta_{z_1}\parallel\cdots\parallel\beta_{z_{nz}}\}}\]
        \caption{The proof system $H$}
        \label{H}
    \end{table}
\end{center}

\begin{theorem}[Soundness of $H$]
Let $Tr_S$ be the set of conditions that hold in $S$. Let $p$ be a closed term of $APTC_{G_{\tau}}$ with guarded linear recursion and $CFAR$, and $\alpha,\beta\in G$ be guards. Then

\begin{eqnarray}
Tr_S\vdash\{\alpha\}p\{\beta\}&\Rightarrow&APTC_{G_{\tau}}\textrm{ with guarded linear recursion and } CFAR\vdash \alpha p=\alpha p\beta\nonumber\\
&\Leftrightarrow&\alpha p\approx_{rbs}(\approx_{rbp},\approx_{rbhp})\alpha p\beta\nonumber\\
&\Leftrightarrow&S\models\{\alpha\}p\{\beta\}\nonumber
\end{eqnarray}
\end{theorem}

\begin{proof}
We only need to prove

$$Tr_S\vdash\{\alpha\}p\{\beta\}\Rightarrow APTC_{G_{\tau}}\textrm{ with guarded linear recursion and } CFAR\vdash \alpha p=\alpha p\beta$$

For $H1$-$H10$, by induction on the length of derivation, the soundness of $H1$-$H10$ are straightforward. We only prove the soundness of $H10'$.

%For the proof of $H10$ and $H10'$, it is similar to the soundness proof in \cite{HLPA}, the only difference is that the guards maybe in parallel, and we do not repeat any more, please refer to \cite{HLPA} for details.

Let $E=\{x_i=t_i(x_1,\cdots,x_n)|i=1,\cdots,n\}$ be a guarded linear recursive specification. Assume that

$$Tr_S,\{\{\alpha_1\parallel\cdots\parallel\alpha_{n_i}\}x_i\{\beta_1\parallel\cdots \parallel\beta_{n_i}\}|i=1,\cdots,n\}\vdash\{\alpha_1\parallel\cdots\parallel\alpha_{n_j}\}t_j(x_1,\cdots,x_n)\{\beta_1\parallel\cdots \parallel\beta_{n_j}\}$$

for $j=1,\cdots,n$. We would show that $APTC_{G_{\tau}}\textrm{ with guarded linear recursion and } CFAR\vdash (\alpha_1\parallel\cdots\parallel\alpha_{n_j})X_j=(\alpha_1\parallel\cdots\parallel\alpha_{n_j}) X_j(\beta_1\parallel\cdots\parallel\beta_{n_j})$.

We write recursive specifications $E'$ and $E''$ for

$$E'=\{y_i=(\alpha_1\parallel\cdots\parallel\alpha_{n_i})t_i(y_1,\cdots,y_n)|i=1,\cdots,n\}$$

$$E''=\{z_i=(\alpha_1\parallel\cdots\parallel\alpha_{n_i})t_i(z_1(\beta_1\parallel\cdots\parallel\beta_{n_1}),\cdots,z_n(\beta_1\parallel\cdots\parallel\beta_{n_n}))|i=1,\cdots,n\}$$

and would show that for $j=1,\cdots,n$,

(1) $(\alpha_1\parallel\cdots\parallel\alpha_{n_j})X_j=Y_j$;

(2) $Z_j(\beta_1\parallel\cdots\parallel\beta_{n_j})=Z_j$;

(3) $Z_j=Y_j$.

For (1), we have

\begin{eqnarray}
(\alpha_1\parallel\cdots\parallel\alpha_{n_j})X_j&=&(\alpha_1\parallel\cdots\parallel\alpha_{n_j})t_j(X_1,\cdots,X_n)\nonumber\\
&=&(\alpha_1\parallel\cdots\parallel\alpha_{n_j})t_j((\alpha_1\parallel\cdots\parallel\alpha_{n_1})X_1,\cdots,(\alpha_1\parallel\cdots\parallel\alpha_{n_n})X_n)\nonumber
\end{eqnarray}

by RDP, we have $(\alpha_1\parallel\cdots\parallel\alpha_{n_j})X_j=Y_j$.

For (2), we have

\begin{eqnarray}
Z_j(\beta_1\parallel\cdots\parallel\beta_{n_j})&=&(\alpha_1\parallel\cdots\parallel\alpha_{n_j}) t_j(Z_1(\beta_1\parallel\cdots\parallel\beta_{n_1}),\cdots,Z_n(\beta_1\parallel\cdots \parallel\beta_{n_n}))(\beta_1\parallel\cdots\parallel\beta_{n_j})\nonumber\\
&=&(\alpha_1\parallel\cdots\parallel\alpha_{n_j}) t_j(Z_1(\beta_1\parallel\cdots\parallel\beta_{n_1}),\cdots,Z_n(\beta_1\parallel\cdots \parallel\beta_{n_n}))\nonumber\\
&=&(\alpha_1\parallel\cdots\parallel\alpha_{n_j}) t_j((Z_1(\beta_1\parallel\cdots\parallel\beta_{n_1}))(\beta_1\parallel\cdots\parallel\beta_{n_1}),\cdots,(Z_n(\beta_1\parallel\cdots \parallel\beta_{n_n}))(\beta_1\parallel\cdots\parallel\beta_{n_n}))\nonumber
\end{eqnarray}

by RDP, we have $Z_j(\beta_1\parallel\cdots\parallel\beta_{n_j})=Z_j$.

For (3), we have

\begin{eqnarray}
Z_j&=&(\alpha_1\parallel\cdots\parallel\alpha_{n_j})t_j(Z_1(\beta_1\parallel\cdots\parallel\beta_{n_1}),\cdots,Z_n(\beta_1\parallel\cdots\parallel\beta_{n_n}))\nonumber\\
&=&(\alpha_1\parallel\cdots\parallel\alpha_{n_j})t_j(Z_1,\cdots,Z_n)\nonumber
\end{eqnarray}

by RDP, we have $Z_j=Y_j$.
\end{proof}

\section{Axiomatization for Hhp-Bisimilarity}{\label{ahhpb}}

Since hhp-bisimilarity is a downward closed hp-bisimilarity and can be downward closed to single atomic event, which implies bisimilarity. As Moller \cite{ILM} proven, there is not a finite sound and complete axiomatization for parallelism $\parallel$ modulo bisimulation equivalence, so there is not a finite sound and complete axiomatization for parallelism $\parallel$ modulo hhp-bisimulation equivalence either. Inspired by the way of left merge to modeling the full merge for bisimilarity, we introduce a left parallel composition $\leftmerge$ to model the full parallelism $\parallel$ for hhp-bisimilarity.

In the following subsection, we add left parallel composition $\leftmerge$ to the whole theory. Because the resulting theory is similar to the former, we only list the significant differences, and all proofs of the conclusions are left to the reader.

\subsection{$APTC$ with Left Parallel Composition}

The transition rules of left parallel composition $\leftmerge$ are shown in Table \ref{TRForLeftParallel}. With a little abuse, we extend the causal order relation $\leq$ on $\mathbb{E}$ to include the original partial order (denoted by $<$) and concurrency (denoted by $=$).

\begin{center}
    \begin{table}
        $$\frac{x\xrightarrow{e_1}\surd\quad y\xrightarrow{e_2}\surd \quad(e_1\leq e_2)}{x\leftmerge y\xrightarrow{\{e_1,e_2\}}\surd} \quad\frac{x\xrightarrow{e_1}x'\quad y\xrightarrow{e_2}\surd \quad(e_1\leq e_2)}{x\leftmerge y\xrightarrow{\{e_1,e_2\}}x'}$$
        $$\frac{x\xrightarrow{e_1}\surd\quad y\xrightarrow{e_2}y' \quad(e_1\leq e_2)}{x\leftmerge y\xrightarrow{\{e_1,e_2\}}y'} \quad\frac{x\xrightarrow{e_1}x'\quad y\xrightarrow{e_2}y' \quad(e_1\leq e_2)}{x\leftmerge y\xrightarrow{\{e_1,e_2\}}x'\between y'}$$
        \caption{Transition rules of left parallel operator $\leftmerge$}
        \label{TRForLeftParallel}
    \end{table}
\end{center}

The new axioms for parallelism are listed in Table \ref{AxiomsForLeftParallelism}.

\begin{center}
    \begin{table}
        \begin{tabular}{@{}ll@{}}
            \hline No. &Axiom\\
            $A6$ & $x+ \delta = x$\\
            $A7$ & $\delta\cdot x =\delta$\\
            $P1$ & $x\between y = x\parallel y + x\mid y$\\
            $P2$ & $x\parallel y = y \parallel x$\\
            $P3$ & $(x\parallel y)\parallel z = x\parallel (y\parallel z)$\\
            $P4$ & $x\parallel y = x\leftmerge y + y\leftmerge x$\\
            $P5$ & $(e_1\leq e_2)\quad e_1\leftmerge (e_2\cdot y) = (e_1\leftmerge e_2)\cdot y$\\
            $P6$ & $(e_1\leq e_2)\quad (e_1\cdot x)\leftmerge e_2 = (e_1\leftmerge e_2)\cdot x$\\
            $P7$ & $(e_1\leq e_2)\quad (e_1\cdot x)\leftmerge (e_2\cdot y) = (e_1\leftmerge e_2)\cdot (x\between y)$\\
            $P8$ & $(x+ y)\leftmerge z = (x\leftmerge z)+ (y\leftmerge z)$\\
            $P9$ & $\delta\leftmerge x = \delta$\\
            $C10$ & $e_1\mid e_2 = \gamma(e_1,e_2)$\\
            $C11$ & $e_1\mid (e_2\cdot y) = \gamma(e_1,e_2)\cdot y$\\
            $C12$ & $(e_1\cdot x)\mid e_2 = \gamma(e_1,e_2)\cdot x$\\
            $C13$ & $(e_1\cdot x)\mid (e_2\cdot y) = \gamma(e_1,e_2)\cdot (x\between y)$\\
            $C14$ & $(x+ y)\mid z = (x\mid z) + (y\mid z)$\\
            $C15$ & $x\mid (y+ z) = (x\mid y)+ (x\mid z)$\\
            $C16$ & $\delta\mid x = \delta$\\
            $C17$ & $x\mid\delta = \delta$\\
            $CE18$ & $\Theta(e) = e$\\
            $CE19$ & $\Theta(\delta) = \delta$\\
            $CE20$ & $\Theta(x+ y) = \Theta(x)\triangleleft y + \Theta(y)\triangleleft x$\\
            $CE21$ & $\Theta(x\cdot y)=\Theta(x)\cdot\Theta(y)$\\
            $CE22$ & $\Theta(x\leftmerge y) = ((\Theta(x)\triangleleft y)\leftmerge y)+ ((\Theta(y)\triangleleft x)\leftmerge x)$\\
            $CE23$ & $\Theta(x\mid y) = ((\Theta(x)\triangleleft y)\mid y)+ ((\Theta(y)\triangleleft x)\mid x)$\\
            $U24$ & $(\sharp(e_1,e_2))\quad e_1\triangleleft e_2 = \tau$\\
            $U25$ & $(\sharp(e_1,e_2),e_2\leq e_3)\quad e_1\triangleleft e_3 = e_1$\\
            $U26$ & $(\sharp(e_1,e_2),e_2\leq e_3)\quad e3\triangleleft e_1 = \tau$\\
            $U27$ & $e\triangleleft \delta = e$\\
            $U28$ & $\delta \triangleleft e = \delta$\\
            $U29$ & $(x+ y)\triangleleft z = (x\triangleleft z)+ (y\triangleleft z)$\\
            $U30$ & $(x\cdot y)\triangleleft z = (x\triangleleft z)\cdot (y\triangleleft z)$\\
            $U31$ & $(x\leftmerge y)\triangleleft z = (x\triangleleft z)\leftmerge (y\triangleleft z)$\\
            $U32$ & $(x\mid y)\triangleleft z = (x\triangleleft z)\mid (y\triangleleft z)$\\
            $U33$ & $x\triangleleft (y+ z) = (x\triangleleft y)\triangleleft z$\\
            $U34$ & $x\triangleleft (y\cdot z)=(x\triangleleft y)\triangleleft z$\\
            $U35$ & $x\triangleleft (y\leftmerge z) = (x\triangleleft y)\triangleleft z$\\
            $U36$ & $x\triangleleft (y\mid z) = (x\triangleleft y)\triangleleft z$\\
        \end{tabular}
        \caption{Axioms of parallelism with left parallel composition}
        \label{AxiomsForLeftParallelism}
    \end{table}
\end{center}

\begin{definition}[Basic terms of $APTC$ with left parallel composition]
The set of basic terms of $APTC$, $\mathcal{B}(APTC)$, is inductively defined as follows:
\begin{enumerate}
  \item $\mathbb{E}\subset\mathcal{B}(APTC)$;
  \item if $e\in \mathbb{E}, t\in\mathcal{B}(APTC)$ then $e\cdot t\in\mathcal{B}(APTC)$;
  \item if $t,s\in\mathcal{B}(APTC)$ then $t+ s\in\mathcal{B}(APTC)$;
  \item if $t,s\in\mathcal{B}(APTC)$ then $t\leftmerge s\in\mathcal{B}(APTC)$.
\end{enumerate}
\end{definition}

\begin{theorem}[Generalization of the algebra for left parallelism with respect to $BATC$]
The algebra for left parallelism is a generalization of $BATC$.
\end{theorem}

\begin{theorem}[Congruence theorem of $APTC$ with left parallel composition]
Truly concurrent bisimulation equivalences $\sim_{p}$, $\sim_s$, $\sim_{hp}$ and $\sim_{hhp}$ are all congruences with respect to $APTC$ with left parallel composition.
\end{theorem}

\begin{theorem}[Elimination theorem of parallelism with left parallel composition]
Let $p$ be a closed $APTC$ with left parallel composition term. Then there is a basic $APTC$ term $q$ such that $APTC\vdash p=q$.
\end{theorem}

\begin{theorem}[Soundness of parallelism  with left parallel composition modulo truly concurrent bisimulation equivalences]
Let $x$ and $y$ be $APTC$ with left parallel composition terms. If $APTC\vdash x=y$, then

\begin{enumerate}
  \item $x\sim_{s} y$;
  \item $x\sim_{p} y$;
  \item $x\sim_{hp} y$;
  \item $x\sim_{hhp} y$.
\end{enumerate}
\end{theorem}

\begin{theorem}[Completeness of parallelism with left parallel composition modulo truly concurrent bisimulation equivalences]
Let $x$ and $y$ be $APTC$ terms.

\begin{enumerate}
  \item If $x\sim_{s} y$, then $APTC\vdash x=y$;
  \item if $x\sim_{p} y$, then $APTC\vdash x=y$;
  \item if $x\sim_{hp} y$, then $APTC\vdash x=y$;
  \item if $x\sim_{hhp} y$, then $APTC\vdash x=y$.
\end{enumerate}
\end{theorem}

The transition rules of encapsulation operator are the same, and the its axioms are shown in \ref{AxiomsForEncapsulationLeft}.

\begin{center}
    \begin{table}
        \begin{tabular}{@{}ll@{}}
            \hline No. &Axiom\\
            $D1$ & $e\notin H\quad\partial_H(e) = e$\\
            $D2$ & $e\in H\quad \partial_H(e) = \delta$\\
            $D3$ & $\partial_H(\delta) = \delta$\\
            $D4$ & $\partial_H(x+ y) = \partial_H(x)+\partial_H(y)$\\
            $D5$ & $\partial_H(x\cdot y) = \partial_H(x)\cdot\partial_H(y)$\\
            $D6$ & $\partial_H(x\leftmerge y) = \partial_H(x)\leftmerge\partial_H(y)$\\
        \end{tabular}
        \caption{Axioms of encapsulation operator with left parallel composition}
        \label{AxiomsForEncapsulationLeft}
    \end{table}
\end{center}

\begin{theorem}[Conservativity of $APTC$ with respect to the algebra for parallelism with left parallel composition]
$APTC$ is a conservative extension of the algebra for parallelism with left parallel composition.
\end{theorem}

\begin{theorem}[Congruence theorem of encapsulation operator $\partial_H$]
Truly concurrent bisimulation equivalences $\sim_{p}$, $\sim_s$, $\sim_{hp}$ and $\sim_{hhp}$ are all congruences with respect to encapsulation operator $\partial_H$.
\end{theorem}

\begin{theorem}[Elimination theorem of $APTC$]
Let $p$ be a closed $APTC$ term including the encapsulation operator $\partial_H$. Then there is a basic $APTC$ term $q$ such that $APTC\vdash p=q$.
\end{theorem}

\begin{theorem}[Soundness of $APTC$ modulo truly concurrent bisimulation equivalences]
Let $x$ and $y$ be $APTC$ terms including encapsulation operator $\partial_H$. If $APTC\vdash x=y$, then

\begin{enumerate}
  \item $x\sim_{s} y$;
  \item $x\sim_{p} y$;
  \item $x\sim_{hp} y$;
  \item $x\sim_{hhp} y$.
\end{enumerate}
\end{theorem}

\begin{theorem}[Completeness of $APTC$ modulo truly concurrent bisimulation equivalences]
Let $p$ and $q$ be closed $APTC$ terms including encapsulation operator $\partial_H$,

\begin{enumerate}
  \item if $p\sim_{s} q$ then $p=q$;
  \item if $p\sim_{p} q$ then $p=q$;
  \item if $p\sim_{hp} q$ then $p=q$;
  \item if $p\sim_{hhp} q$ then $p=q$.
\end{enumerate}
\end{theorem}

\subsection{Recursion}

\begin{definition}[Recursive specification]
A recursive specification is a finite set of recursive equations

$$X_1=t_1(X_1,\cdots,X_n)$$
$$\cdots$$
$$X_n=t_n(X_1,\cdots,X_n)$$

where the left-hand sides of $X_i$ are called recursion variables, and the right-hand sides $t_i(X_1,\cdots,X_n)$ are process terms in $APTC$ with possible occurrences of the recursion variables $X_1,\cdots,X_n$.
\end{definition}

\begin{definition}[Solution]
Processes $p_1,\cdots,p_n$ are a solution for a recursive specification $\{X_i=t_i(X_1,\cdots,X_n)|i\in\{1,\cdots,n\}\}$ (with respect to truly concurrent bisimulation equivalences $\sim_s$($\sim_p$, $\sim_{hp}$, $\sim_{hhp}$)) if $p_i\sim_s (\sim_p, \sim_{hp},\sim{hhp})t_i(p_1,\cdots,p_n)$ for $i\in\{1,\cdots,n\}$.
\end{definition}

\begin{definition}[Guarded recursive specification]
A recursive specification

$$X_1=t_1(X_1,\cdots,X_n)$$
$$...$$
$$X_n=t_n(X_1,\cdots,X_n)$$

is guarded if the right-hand sides of its recursive equations can be adapted to the form by applications of the axioms in $APTC$ and replacing recursion variables by the right-hand sides of their recursive equations,

$$(a_{11}\leftmerge\cdots\leftmerge a_{1i_1})\cdot s_1(X_1,\cdots,X_n)+\cdots+(a_{k1}\leftmerge\cdots\leftmerge a_{ki_k})\cdot s_k(X_1,\cdots,X_n)+(b_{11}\leftmerge\cdots\leftmerge b_{1j_1})+\cdots+(b_{1j_1}\leftmerge\cdots\leftmerge b_{lj_l})$$

where $a_{11},\cdots,a_{1i_1},a_{k1},\cdots,a_{ki_k},b_{11},\cdots,b_{1j_1},b_{1j_1},\cdots,b_{lj_l}\in \mathbb{E}$, and the sum above is allowed to be empty, in which case it represents the deadlock $\delta$.
\end{definition}

\begin{definition}[Linear recursive specification]
A recursive specification is linear if its recursive equations are of the form

$$(a_{11}\leftmerge\cdots\leftmerge a_{1i_1})X_1+\cdots+(a_{k1}\leftmerge\cdots\leftmerge a_{ki_k})X_k+(b_{11}\leftmerge\cdots\leftmerge b_{1j_1})+\cdots+(b_{1j_1}\leftmerge\cdots\leftmerge b_{lj_l})$$

where $a_{11},\cdots,a_{1i_1},a_{k1},\cdots,a_{ki_k},b_{11},\cdots,b_{1j_1},b_{1j_1},\cdots,b_{lj_l}\in \mathbb{E}$, and the sum above is allowed to be empty, in which case it represents the deadlock $\delta$.
\end{definition}

\begin{theorem}[Conservitivity of $APTC$ with guarded recursion]
$APTC$ with guarded recursion is a conservative extension of $APTC$.
\end{theorem}

\begin{theorem}[Congruence theorem of $APTC$ with guarded recursion]
Truly concurrent bisimulation equivalences $\sim_{p}$, $\sim_s$, $\sim_{hp}$, $\sim_{hhp}$ are all congruences with respect to $APTC$ with guarded recursion.
\end{theorem}

\begin{theorem}[Elimination theorem of $APTC$ with linear recursion]
Each process term in $APTC$ with linear recursion is equal to a process term $\langle X_1|E\rangle$ with $E$ a linear recursive specification.
\end{theorem}

\begin{theorem}[Soundness of $APTC$ with guarded recursion]
Let $x$ and $y$ be $APTC$ with guarded recursion terms. If $APTC\textrm{ with guarded recursion}\vdash x=y$, then
\begin{enumerate}
  \item $x\sim_{s} y$;
  \item $x\sim_{p} y$;
  \item $x\sim_{hp} y$;
  \item $x\sim_{hhp} y$.
\end{enumerate}
\end{theorem}

\begin{theorem}[Completeness of $APTC$ with linear recursion]
Let $p$ and $q$ be closed $APTC$ with linear recursion terms, then,
\begin{enumerate}
  \item if $p\sim_{s} q$ then $p=q$;
  \item if $p\sim_{p} q$ then $p=q$;
  \item if $p\sim_{hp} q$ then $p=q$;
  \item if $p\sim_{hhp} q$ then $p=q$.
\end{enumerate}
\end{theorem}

\subsection{Abstraction}

\begin{definition}[Guarded linear recursive specification]
A recursive specification is linear if its recursive equations are of the form

$$(a_{11}\leftmerge\cdots\leftmerge a_{1i_1})X_1+\cdots+(a_{k1}\leftmerge\cdots\leftmerge a_{ki_k})X_k+(b_{11}\leftmerge\cdots\leftmerge b_{1j_1})+\cdots+(b_{1j_1}\leftmerge\cdots\leftmerge b_{lj_l})$$

where $a_{11},\cdots,a_{1i_1},a_{k1},\cdots,a_{ki_k},b_{11},\cdots,b_{1j_1},b_{1j_1},\cdots,b_{lj_l}\in \mathbb{E}\cup\{\tau\}$, and the sum above is allowed to be empty, in which case it represents the deadlock $\delta$.

A linear recursive specification $E$ is guarded if there does not exist an infinite sequence of $\tau$-transitions $\langle X|E\rangle\xrightarrow{\tau}\langle X'|E\rangle\xrightarrow{\tau}\langle X''|E\rangle\xrightarrow{\tau}\cdots$.
\end{definition}

The transition rules of $\tau$ are the same, and axioms of $\tau$ are as Table \ref{AxiomsForTauLeft} shows.

\begin{theorem}[Conservitivity of $APTC$ with silent step and guarded linear recursion]
$APTC$ with silent step and guarded linear recursion is a conservative extension of $APTC$ with linear recursion.
\end{theorem}

\begin{theorem}[Congruence theorem of $APTC$ with silent step and guarded linear recursion]
Rooted branching truly concurrent bisimulation equivalences $\approx_{rbp}$, $\approx_{rbs}$, $\approx_{rbhp}$, and $\approx_{rbhhp}$ are all congruences with respect to $APTC$ with silent step and guarded linear recursion.
\end{theorem}

\begin{center}
\begin{table}
  \begin{tabular}{@{}ll@{}}
\hline No. &Axiom\\
  $B1$ & $e\cdot\tau=e$\\
  $B2$ & $e\cdot(\tau\cdot(x+y)+x)=e\cdot(x+y)$\\
  $B3$ & $x\leftmerge\tau=x$\\
\end{tabular}
\caption{Axioms of silent step}
\label{AxiomsForTauLeft}
\end{table}
\end{center}

\begin{theorem}[Elimination theorem of $APTC$ with silent step and guarded linear recursion]
Each process term in $APTC$ with silent step and guarded linear recursion is equal to a process term $\langle X_1|E\rangle$ with $E$ a guarded linear recursive specification.
\end{theorem}

\begin{theorem}[Soundness of $APTC$ with silent step and guarded linear recursion]
Let $x$ and $y$ be $APTC$ with silent step and guarded linear recursion terms. If $APTC$ with silent step and guarded linear recursion $\vdash x=y$, then
\begin{enumerate}
  \item $x\approx_{rbs} y$;
  \item $x\approx_{rbp} y$;
  \item $x\approx_{rbhp} y$;
  \item $x\approx_{rbhhp} y$.
\end{enumerate}
\end{theorem}

\begin{theorem}[Completeness of $APTC$ with silent step and guarded linear recursion]
Let $p$ and $q$ be closed $APTC$ with silent step and guarded linear recursion terms, then,
\begin{enumerate}
  \item if $p\approx_{rbs} q$ then $p=q$;
  \item if $p\approx_{rbp} q$ then $p=q$;
  \item if $p\approx_{rbhp} q$ then $p=q$;
  \item if $p\approx_{rbhhp} q$ then $p=q$.
\end{enumerate}
\end{theorem}

The transition rules of $\tau_I$ are the same, and the axioms are shown in Table \ref{AxiomsForAbstractionLeft}.

\begin{theorem}[Conservitivity of $APTC_{\tau}$ with guarded linear recursion]
$APTC_{\tau}$ with guarded linear recursion is a conservative extension of $APTC$ with silent step and guarded linear recursion.
\end{theorem}

\begin{theorem}[Congruence theorem of $APTC_{\tau}$ with guarded linear recursion]
Rooted branching truly concurrent bisimulation equivalences $\approx_{rbp}$, $\approx_{rbs}$, $\approx_{rbhp}$ and $\approx_{rbhhp}$ are all congruences with respect to $APTC_{\tau}$ with guarded linear recursion.
\end{theorem}

\begin{center}
\begin{table}
  \begin{tabular}{@{}ll@{}}
\hline No. &Axiom\\
  $TI1$ & $e\notin I\quad \tau_I(e)=e$\\
  $TI2$ & $e\in I\quad \tau_I(e)=\tau$\\
  $TI3$ & $\tau_I(\delta)=\delta$\\
  $TI4$ & $\tau_I(x+y)=\tau_I(x)+\tau_I(y)$\\
  $TI5$ & $\tau_I(x\cdot y)=\tau_I(x)\cdot\tau_I(y)$\\
  $TI6$ & $\tau_I(x\leftmerge y)=\tau_I(x)\leftmerge\tau_I(y)$\\
\end{tabular}
\caption{Axioms of abstraction operator}
\label{AxiomsForAbstractionLeft}
\end{table}
\end{center}

\begin{theorem}[Soundness of $APTC_{\tau}$ with guarded linear recursion]
Let $x$ and $y$ be $APTC_{\tau}$ with guarded linear recursion terms. If $APTC_{\tau}$ with guarded linear recursion $\vdash x=y$, then
\begin{enumerate}
  \item $x\approx_{rbs} y$;
  \item $x\approx_{rbp} y$;
  \item $x\approx_{rbhp} y$;
  \item $x\approx_{rbhhp} y$.
\end{enumerate}
\end{theorem}

\begin{definition}[Cluster]
Let $E$ be a guarded linear recursive specification, and $I\subseteq \mathbb{E}$. Two recursion variable $X$ and $Y$ in $E$ are in the same cluster for $I$ iff there exist sequences of transitions $\langle X|E\rangle\xrightarrow{\{b_{11},\cdots, b_{1i}\}}\cdots\xrightarrow{\{b_{m1},\cdots, b_{mi}\}}\langle Y|E\rangle$ and $\langle Y|E\rangle\xrightarrow{\{c_{11},\cdots, c_{1j}\}}\cdots\xrightarrow{\{c_{n1},\cdots, c_{nj}\}}\langle X|E\rangle$, where $b_{11},\cdots,b_{mi},c_{11},\cdots,c_{nj}\in I\cup\{\tau\}$.

$a_1\leftmerge\cdots\leftmerge a_k$ or $(a_1\leftmerge\cdots\leftmerge a_k) X$ is an exit for the cluster $C$ iff: (1) $a_1\leftmerge\cdots\leftmerge a_k$ or $(a_1\leftmerge\cdots\leftmerge a_k) X$ is a summand at the right-hand side of the recursive equation for a recursion variable in $C$, and (2) in the case of $(a_1\leftmerge\cdots\leftmerge a_k) X$, either $a_l\notin I\cup\{\tau\}(l\in\{1,2,\cdots,k\})$ or $X\notin C$.
\end{definition}

\begin{center}
\begin{table}
  \begin{tabular}{@{}ll@{}}
\hline No. &Axiom\\
  $CFAR$ & If $X$ is in a cluster for $I$ with exits \\
           & $\{(a_{11}\leftmerge\cdots\leftmerge a_{1i})Y_1,\cdots,(a_{m1}\leftmerge\cdots\leftmerge a_{mi})Y_m, b_{11}\leftmerge\cdots\leftmerge b_{1j},\cdots,b_{n1}\leftmerge\cdots\leftmerge b_{nj}\}$, \\
           & then $\tau\cdot\tau_I(\langle X|E\rangle)=$\\
           & $\tau\cdot\tau_I((a_{11}\leftmerge\cdots\leftmerge a_{1i})\langle Y_1|E\rangle+\cdots+(a_{m1}\leftmerge\cdots\leftmerge a_{mi})\langle Y_m|E\rangle+b_{11}\leftmerge\cdots\leftmerge b_{1j}+\cdots+b_{n1}\leftmerge\cdots\leftmerge b_{nj})$\\
\end{tabular}
\caption{Cluster fair abstraction rule}
\label{CFARLeft}
\end{table}
\end{center}

\begin{theorem}[Soundness of $CFAR$]
$CFAR$ is sound modulo rooted branching truly concurrent bisimulation equivalences $\approx_{rbs}$, $\approx_{rbp}$, $\approx_{rbhp}$ and $\approx_{rbhhp}$.
\end{theorem}

\begin{theorem}[Completeness of $APTC_{\tau}$ with guarded linear recursion and $CFAR$]
Let $p$ and $q$ be closed $APTC_{\tau}$ with guarded linear recursion and $CFAR$ terms, then,
\begin{enumerate}
  \item if $p\approx_{rbs} q$ then $p=q$;
  \item if $p\approx_{rbp} q$ then $p=q$;
  \item if $p\approx_{rbhp} q$ then $p=q$;
  \item if $p\approx_{rbhhp} q$ then $p=q$.
\end{enumerate}
\end{theorem}

\section{Conclusions}{\label{con}}

Now, let us conclude this paper. We try to find the algebraic laws for true concurrency, as a uniform logic for true concurrency \cite{LTC1} \cite{LTC2} already existed. There are simple comparisons between Hennessy and Milner (HM) logic for bisimulation equivalence and the uniform logic for truly concurrent bisimulation equivalences, the algebraic laws \cite{ALNC}, $ACP$ \cite{ACP} for bisimulation equivalence, and \emph{what} for truly concurrent bisimulation equivalences, which is still missing.

Following the above idea, we find the algebraic laws for true concurrency, which is called $APTC$, an algebra for true concurrency. Like $ACP$, $APTC$ also has four modules: $BATC$ (Basic Algebra for True Concurrency), $APTC$ (Algebra for Parallelism in True Concurrency), recursion and abstraction, and we prove the soundness and completeness of their algebraic laws modulo truly concurrent bisimulation equivalences. And we show its applications in verification of behaviors of system in a truly concurrent flavor, and its modularity by extending a new renaming operator and a new shadow constant into it.

Unlike $ACP$, in $APTC$, the parallelism is a fundamental computational pattern, and cannot be steadied by other computational patterns. We establish a whole theory which has correspondences to $ACP$.

In future, we will pursue the wide applications of $APTC$ in verifications of the behavioral correctness of concurrent systems.

\newpage

%\appendix
%\section{appendix 1}
%
%appendix 1

\label{lastpage}

\end{document}